%% file: main.tex
\documentclass[letterpaper,twocolumn,10pt]{article}
\usepackage{usenix-2020-09}
\usepackage{amsmath}
\usepackage{amssymb}
\usepackage{enumitem}
\usepackage{numprint}
\usepackage{booktabs}
\npthousandsep{,}
\usepackage{graphicx} 
\usepackage{caption,subcaption}
\usepackage{balance}
\usepackage{amsthm}
\newtheorem{theorem}{Theorem}
\theoremstyle{definition}

\microtypecontext{spacing=nonfrench}

\begin{document}
\thispagestyle{headings}
\markright{\hfill To appear in USENIX Security Symposium, 2022\hfill}

\title{Poisoning Attacks to Local Differential Privacy Protocols for Key-Value Data}

\author{
{\rm Yongji Wu,}\ \ \
{\rm Xiaoyu Cao,}\ \ \
{\rm Jinyuan Jia,}\ \ \
{\rm Neil Zhenqiang Gong}\ \ \
\\
{Duke University}\\
{\{yongji.wu769, xiaoyu.cao, jinyuan.jia, neil.gong\}@duke.edu}
}

\maketitle

\begin{abstract}

Local Differential Privacy (LDP) protocols enable an untrusted server to perform privacy-preserving, federated data analytics. 
Various LDP protocols have been developed for different types of data such as categorical data, numerical data, and key-value data. Due to their distributed settings, LDP protocols are fundamentally vulnerable to \emph{poisoning attacks}, in which fake users manipulate the server's analytics results via sending carefully crafted data to the server. However, existing poisoning attacks focused on LDP protocols for simple data types such as categorical and numerical data, leaving the security of LDP protocols for more advanced data types such as key-value data unexplored. 

In this work, we aim to bridge the gap by introducing novel poisoning attacks to LDP protocols for key-value data. In such a LDP protocol, a server aims to simultaneously estimate the frequency and mean value of each key among some users, each of whom possesses a set of key-value pairs. 
Our poisoning attacks aim to simultaneously maximize the frequencies and mean values of some attacker-chosen target keys via sending carefully crafted data from some fake users to the sever. Specifically, since our attacks have two objectives, we formulate them as a two-objective optimization problem. 
Moreover, we propose a method to approximately solve the two-objective optimization problem, from which we obtain the optimal crafted data the fake users should send to the server.
We demonstrate the effectiveness of our attacks to three LDP protocols for key-value data both theoretically and empirically.  We also explore two defenses  against our attacks, which are effective in some scenarios but have limited effectiveness in other scenarios. Our results highlight the needs for new defenses against our poisoning attacks.

\end{abstract}

%%%%%%%%%%%%%%%%%%%%%%%%%%%%

\section{Introduction}

Nowadays, many Internet services rely on users' data. 
However, it poses significant challenges to  users' privacy for a server to collect raw data from users. 
\emph{Local Differential Privacy (LDP)}~\cite{erlingsson2014rappor} aims to address the challenges. Specifically, LDP is a variant of differential privacy~\cite{dwork2006calibrating} under a local setting, where each user locally perturbs his/her data before sending it to an untrusted server. The server aggregates the perturbed data and obtains the statistics of interest. LDP ensures that even if the server is compromised, users' privacy is still well-protected. Due to its promising resilience against untrusted server, LDP has been widely deployed by  Internet giants such as Google~\cite{erlingsson2014rappor}, Apple~\cite{appledf2017}, and Microsoft~\cite{ding2017collecting}. 

Moreover, LDP protocols have been proposed for different types of data, such as categorical data~\cite{wang2017locally,erlingsson2014rappor,jia2019calibrate,qin2016heavy,wang2018locally,wang2019locally},  numerical data~\cite{ding2017collecting,duchi2018minimax}, multidimensional data~\cite{wang2019collecting,zhang2018calm}, and key-value data~\cite{ye2019privkv,gu2020pckv}. 
For instance, in recommender systems, each user rates a set of items (e.g., products), where an item and a rating can be viewed as a key and a value, respectively. Thus, each user possesses a set of key-value pairs. In current recommender systems,  users  send their raw key-value pairs to the server. However, given access to users' raw key-value pairs, an untrusted server can infer users' sensitive attributes (e.g., gender, age, sexual orientation) via attribute inference attacks~\cite{gong2016you,jia2017attriinfer}. LDP protocols enable a server to   collect  frequency (i.e., popularity) and mean value (i.e., mean rating) of each key from users without accessing their raw key-value pairs and thus protect users' rating-behavior privacy. 
The collected frequencies and mean values can be used to rank keys and make recommendations to users.  

However, due to the distributed settings, LDP protocols are vulnerable to \emph{poisoning attacks}~\cite{cheu2019manipulation,cao2019data}, in which an attacker injects fake users into the system and manipulates the server's analytics results via sending carefully crafted data from the fake users to the server. Specifically, Cheu et al.~\cite{cheu2019manipulation} showed that poisoning attacks can degrade the overall performance for indiscriminate items, while Cao et al.~\cite{cao2019data} showed that poisoning attacks can promote  attacker-chosen target items in LDP protocols for frequency estimation and heavy hitter identification. However, these studies focused on simple data types such as categorical data and numerical data, in which each user possesses a single categorical item or numerical value. The security of LDP protocols for more advanced data types such as key-value data is largely unexplored.  

In this work, we aim to bridge this gap. Specifically, we perform a systematic study on poisoning attacks to LDP protocols for key-value data. 
In our poisoning attacks, an attacker aims to simultaneously promote the estimated frequencies and mean values for some attacker-chosen \emph{target keys}. An attacker can inject some fake users into the system and send carefully crafted data to the server to achieve the attack goals. Our attacks pose severe security threats to LDP protocols for key-value data. 
For example, when such a LDP protocol is deployed to collect popularity and mean ratings of mobile apps in a mobile-app recommender system, an attacker can use our attacks to promote a malicious app's popularity and mean rating such that it may be recommended to more people.

However, different from the poisoning attacks to LDP protocols for simple data types~\cite{cheu2019manipulation,cao2019data}, poisoning attacks to the LDP protocols for key-value data face new challenges. Specifically, key-value data are inherently heterogeneous, i.e., keys are categorical  and values are numerical. Moreover,   there are  correlations between  the keys and the values. In particular,  the estimated mean value of a key depends on the estimated frequency  of the key. 
Furthermore, each user may possess more than one key-value pair, while each user  only has a single item or numerical value in LDP protocols for categorical and numerical data. Therefore, existing poisoning attacks are insufficient for  LDP protocols for key-value data. 

To address the challenges, we formulate our poisoning attacks as a \emph{two-objective optimization problem}, which explicitly captures the attacker's two objectives on promoting both the estimated frequencies and mean values of the target keys. Specifically, we define the \emph{frequency gain} (or \emph{mean gain}) as the difference between the total estimated frequency (or mean value) of the target keys before and after attack. The expected frequency gain and expected mean gain are the two objective functions in our two-objective optimization problem, where the expectation is taken over the randomness in a LDP protocol. Moreover, we propose a method, called \emph{maximal gain attack (M2GA)}, to approximately solve the two-objective optimization problem. The solution corresponds to the crafted data fake users should send to the server. Specifically, M2GA can exactly maximize the expected frequency gain and approximately maximize the expected mean gain.  

To demonstrate the effectiveness of  M2GA, we also propose two baseline poisoning attacks, called \emph{random message attack ({RMA})} and \emph{random key-value pair attack ({RKVA})}. In RMA, each fake user sends a random message in the   domain allowed by the LDP protocol to the server, while in RKVA, each fake user picks a random target key, associates the largest allowable value with it, and perturbs the key-value pair  following the LDP protocol before sending it to the server.

We apply our attacks to three state-of-the-art LDP protocols for key-value data, e.g., PrivKVM~\cite{ye2019privkv}, PCKV-UE~\cite{gu2020pckv}, and PCKV-GRR~\cite{gu2020pckv}. Moreover, we evaluate our attacks both theoretically and empirically. Theoretically, we derive the expected frequency gains of our attacks exactly. However, it is challenging to derive the expected mean gains exactly because they involve divisions of random variables. To address the challenge, we derive the expected mean gains approximately via relaxing the divisions of random variables. We note that prior work~\cite{cheu2019manipulation,cao2019data} found security-privacy trade-offs in LDP protocols for categorical and numerical data, i.e., such a LDP protocol is more vulnerable to poisoning attacks when it is more privacy-preserving. One interesting finding from our theoretical analysis is that, such security-privacy trade-off does not necessarily hold in LDP protocols for key-value data. For instance, in M2GA to  PrivKVM~\cite{ye2019privkv}, the expected frequency gain increases (i.e., more vulnerable to M2GA) as the privacy budget decreases (i.e., more privacy-preserving) when an attacker selects one target key; the expected frequency gain does not depend on the privacy budget when an attacker selects two target keys; and the expected frequency gain decreases as the privacy budget decreases when an attacker selects more than two target keys. Empirically, we evaluate our attacks on multiple datasets. Our results show that M2GA can successfully promote the estimated frequencies and mean values of the target keys, and that M2GA substantially outperforms the two baseline attacks.

We also explore two defenses against our poisoning attacks. Specifically, in one defense, the server uses one-class  classifier to detect fake users via treating users' data sent to the server as their features. 
PrivKVM  requires multiple communication rounds between the users and the server. Therefore, in our second defense, the server detects fake users in PrivKVM via checking the consistency of their data  sent to the server in multiple rounds. Our intuition is that a fake user sends highly correlated data to the server in multiple rounds, while a genuine user does not. 
Our empirical results show that our defenses are effective in some scenarios. For instance, when the fraction of  fake users and the number of target keys are small, M2GA achieves negligible frequency gains and mean gains when the second defense is deployed. 
However, the defenses are ineffective in other scenarios, e.g., when the fraction of fake users or the number of target keys is large for the second defense, which highlights the needs for new defense mechanisms against our attacks.

Our contributions can be summarized as follows:
\begin{itemize}[leftmargin=*]
    \item To the best of our knowledge, we are the first to study  poisoning attacks to LDP protocols for key-value data. 
    \item We formulate our attacks as a two-objective optimization problem,     which aims to maximize both the expected frequency gain and expected mean gain of the target keys. 
    \item We evaluate our attacks on three state-of-the-art LDP protocols for key-value data both theoretically and empirically. 
    \item We investigate two defenses against our attacks. Our results show that the defenses can defend against our attacks in some scenarios but not in others, which highlights that new defenses are needed to mitigate our attacks.
\end{itemize}

%%%%%%%%%%%%%%%%%%%%%%%%%%%%

\section{Related Work}
\label{sec:related}

\paragraph{Poisoning Attacks to LDP} Two concurrent studies~\cite{cheu2019manipulation,cao2019data} proposed poisoning attacks to LDP protocols for categorical and numerical data. In these LDP protocols, each user holds a single item or numerical value, and a server aims to estimate the frequencies of items or identify heavy hitters that have the largest item frequencies.  Cheu et al.~\cite{cheu2019manipulation} showed that an attacker can downgrade the accuracy of the estimated item frequencies or the identified heavy hitters for indiscriminate items via injecting fake users into the system. Cao et al.~\cite{cao2019data} showed that an attacker can increase the estimated frequencies for attacker-chosen target items or promote them to be identified as heavy hitters. In particular, Cao et al. formulated their poisoning attacks as a single-objective optimization problem, where the objective function is to maximize the frequency gains for the target items. As we discussed in Introduction, these poisoning attacks are insufficient for LDP protocols for key-value data. 

In particular, our work differs from \cite{cao2019data} in the following aspects. First, we formulate a two-objective optimization problem for key-value data instead of the single-objective one. Second, our solutions to the optimization problems are different. Third, we propose different defenses against the poisoning attacks. Fourth, we observe different privacy-security trade-off. Specifically, Cao et al. \cite{cao2019data} found that when the privacy guarantee is stronger, a protocol becomes less secure to poisoning attacks. We do not necessarily observe such privacy-security trade-off both theoretically (in some cases) and empirically for LDP protocols for key-value data. 

\paragraph{Poisoning Attacks to ML}
Poisoning attacks to machine learning systems have been studied extensively~\cite{li2016data,ji2017backdoor,liutrojaning2018,gu2017badnets,chen2017targeted,ji2018model,yang2017fake,fang2018poisoning,fang2020influence,fang2020local,nelson2008exploiting,munoz2017towards,biggio2012poisoning,carlini2021poisoning,hidano2020exposing}. 
In these attacks, an attacker  manipulates the training phase of a machine learning system via poisoning some carefully selected training examples or tampering the training process.  For instance,  training-data poisoning attacks have been studied for  support vector machines~\cite{biggio2012poisoning}, neural networks~\cite{liutrojaning2018,gu2017badnets,chen2017targeted}, and recommender systems~\cite{li2016data,yang2017fake,fang2018poisoning,fang2020influence,huang2021data}. Training-process poisoning attacks have been studied for federated learning~\cite{bhagoji2019analyzing,fang2020local,bagdasaryan2020backdoor}. Our poisoning attacks differ from these ones because the computational process of LDP protocols is significantly different from that of machine learning training phases.

%%%%%%%%%%%%%%%%%%%%%%%%%%%%

\section{Preliminaries}
\label{sec:prelim}

\begin{table}[]
    \centering
    \begin{tabular}{|c|c|}
    \hline
       symbol  & representation \\\hline
       $n$  & \# genuine users\\\hline
       $m$ & {\# fake users}\\\hline
       $\beta$ & fraction of fake users\\\hline
       $\mathcal{K}$ & dictionary of keys\\\hline
       $d$ & \# keys\\\hline
       $\langle k, v\rangle$ & key-value pair\\\hline
       $f_k$ & frequency of $k$\\\hline
       $m_k$ & {mean value of $k$}\\\hline
       $\epsilon$ & privacy budget\\\hline
       $\ell$ & padding length\\\hline
       $r$ & \# target keys\\\hline
       $G_f$ & frequency gain\\\hline
       $G_m$ & mean gain\\\hline
    \end{tabular}
    \caption{{Notations used in this work.}}
    \label{tab:notations}
\end{table}

Before we dive into details, we summarize the important notations we use in Table \ref{tab:notations}.

\subsection{LDP Protocols for Key-Value Data}
Suppose we have $n$ users, we have a dictionary $\mathcal{K}$ of $d$ keys (i.e., $\mathcal{K}=\{1,2,\cdots, d\}$), and each user possesses a set of KV pairs $\langle k, v\rangle$, where $k\in \mathcal{K}$ and $v \in [-1,1]$. Note that, without loss of generality, we assume the values are transformed into the range $[-1,1]$. A server aims to estimate the frequency and mean value of each key among the $n$ users. The frequency of a key is the fraction of users who possess the key, while the mean value of a key is the average of the values in the  KV pairs that contain the key.  Formally, the true frequency $f_k$ and mean value $m_k$ for each key $k$ are defined as follows: 
\begin{equation*}
f_{k}=\frac{\sum_{u =1}^{n} \mathbb{I}_{S_{u}}(\langle k, \cdot\rangle)}{n}, \quad m_{k}=\frac{\sum_{u \in \{1,\cdots,n\},\langle k, v\rangle \in S_{u}} v}{n \cdot f_{k}},    
\end{equation*}
where $\mathcal{S}_u$ is the set of KV pairs possessed by user $u$ and $\mathbb{I}_{S_{u}}(\langle k, \cdot\rangle)$ is an indicator function that equals 1 if  one KV pair in $S_{u}$ contains the key $k$ and equals 0 otherwise.

\paragraph{Framework of LDP Protocols for Key-Value Data} In LDP protocols, each user randomly perturbs its KV pairs and sends the perturbed data (called \emph{message}) to the server. Roughly speaking, in LDP, any two sets of KV pairs are perturbed to the same message with close probabilities. State-of-the-art LDP protocols~\cite{ye2019privkv,gu2020pckv} for key-value data consist of the following three key steps.

\begin{itemize}[leftmargin=*]
    \item \textbf{Sample:} A user randomly samples a key from the dictionary and constructs a KV pair based on the sampled key. 
    \item \textbf{Perturb:} The user perturbs the constructed KV pair to obtain the message that should be sent to the server. 
    \item \textbf{Aggregate:} The server estimates the frequency and mean value of each key via aggregating the messages from all users. We denote by $\hat{f}_k$ and $\hat{m}_k$ the estimated frequency and mean value of a key $k$.  
\end{itemize}

Next, we briefly review three state-of-the-art LDP protocols for key-value data, i.e., PrivKVM~\cite{ye2019privkv}, PCKV-UE~\cite{gu2020pckv}, and PCKV-GRR~\cite{gu2020pckv}.  

\subsection{PrivKVM}
PrivKVM utilizes an iterative procedure, where the aforementioned three steps are performed for $N_{\text{iter}}$ rounds. Specifically, after each round, the server has an estimated mean value $\hat{m}_k$ for each key $k$, which is used  to construct messages for the users who do not possess the key $k$ in the next round. Next, we describe the three steps in each round.

\paragraph{Sample} For each user, PrivKVM samples a key $k$ from the dictionary uniformly at random. If the user possesses $k$, then the Sample step returns the user's KV pair $\langle k,v \rangle$, otherwise the Sample step constructs a KV pair $\langle k, v=\hat{m}_k\rangle$ ($\hat{m}_k$ is the estimated mean value in the previous round and is set to 0 in the first round). The value $v$ in the KV pair is then discretized to $v^*=1$ with a probability of $\frac{1+v}{2}$ and $v^*=-1$ with a probability $\frac{1-v}{2}$. Finally, the Sample step returns a KV pair $\langle k, v^*\rangle$  and a flag indicating whether $k$ is possessed by the user or not.

\paragraph{Perturb} First, the user perturbs the discretized value $v^*$ to be $v^{\prime}$ based on  the following rule:
\begin{equation}
    v^{\prime}=\left\{\begin{array}{ll}
v^* & \text { w.p. } \frac{e^{\epsilon_2}}{1+e^{\epsilon_2}} \\
-v^* & \text { w.p. } \frac{1}{1+e^{\epsilon_2}}
\end{array}\right.,
\end{equation}
where w.p. is short for with probability. 
Then, the user further perturbs the $\langle k, v^{\prime}\rangle$ pair to be $\langle k_p, v^{\prime}_p\rangle$. 
Specifically, if the user possesses the key $k$, then $\langle k_p, v^{\prime}_p\rangle$ is obtained  based on the following perturbation rule:
\begin{equation}
    \left\langle k_p, v^{\prime}_p\right\rangle=\left\{\begin{array}{ll}
\left\langle 1, v^{\prime}\right\rangle & \text { w.p. } \frac{e^{\epsilon_{1}}}{1+e^{\epsilon_{1}}} \\
\langle 0,0\rangle & \text { w.p. } \frac{1}{1+e^{\epsilon_{1}}}
\end{array}\right..
\end{equation}
If the user does not have $k$, then $\langle k_p, v^{\prime}_p\rangle$ is obtained as follows:
\begin{equation}
    \left\langle k_p, v^{\prime}_p\right\rangle=\left\{\begin{array}{ll}
\langle 0,0\rangle & \text { w.p. } \frac{e^{\epsilon_{1}}}{1+e^{\epsilon_{1}}} \\
\left\langle 1, v^{\prime}\right\rangle & \text { w.p. } \frac{1}{1+e^{\epsilon_1}}
\end{array}\right..
\end{equation}
Finally, the user sends the pair $\langle k_p, v^{\prime}_p\rangle$ and the index of the key $k$ to the server.

\paragraph{Aggregate}
We denote by $n_k$ the number of users reporting the index of key $k$ and the tuple $\langle 1, \cdot \rangle$. Then, the server computes the estimated frequency of $k$ as follows:
\begin{equation}
    \hat{f}_{k}=\frac{p-1+n_k/n}{2 p-1}, 
    \label{eq:freq_estimate_privkvm}
\end{equation}
 where  $p=\frac{e^{\epsilon_{1}}}{e^{\epsilon_{1}}+1}$. 
Then, the server counts the number of users $n_1^k$ (or $n_{-1}^k$) that report the index of key $k$ and the tuple $\langle 1,1 \rangle$ (or $\langle 1,-1 \rangle$). 
The server computes the estimated mean value of $k$ as follows: 
\begin{equation}
    \hat{m}_{k}=\frac{\hat{n}_{1}^k-\hat{n}_{-1}^k}{n_k},
    \label{eq:mean_estimate_privkvm}
\end{equation}
where $\hat{n}_{1}^k$ and $\hat{n}_{-1}^k$ are defined as follows:
\begin{align}
\hat{n}_{1}^k &=\frac{p-1}{2 p-1} \cdot n_k+\frac{n_{1}^k}{2 p-1}, \\
\hat{n}_{-1}^k &=\frac{p-1}{2 p-1} \cdot n_k+\frac{n_{-1}^k}{2 p-1}, 
\end{align}
where $p=\frac{e^{\epsilon_{2}}}{e^{\epsilon_{2}}+1}$. 
We note that frequency estimation is only conducted in the first round, while mean estimation uses the results after  $N_{\text{iter}}$ rounds.  The privacy budget  $\epsilon_1$ is only allocated to the first round, while the privacy budget  $\epsilon_2$ is equally allocated for each round. Specifically,  we have $\epsilon_1=\frac{\epsilon}{2}$ and $\epsilon_2=\frac{\epsilon}{2N_{\text{iter}}}$, where $\epsilon$ is the overall privacy budget.

\subsection{PCKV-UE and PCKV-GRR}
\label{back:pckv}
PCKV-UE and PCKV-GRR are  two protocols from the PCKV family~\cite{gu2020pckv}. PCKV improves PrivKVM by utilizing a \emph{padding-and-sampling} strategy in the Sample step to reduce the variance of frequency and mean value estimation. Moreover, unlike PrivKVM that performs aforementioned three steps for multiple rounds, PCKV only requires a single round. 
The two protocols PCKV-UE and PCKV-GRR mainly differ in the Perturb step and the Aggregate step, while sharing a common Sample step. Specifically, we have the following workflow:

\paragraph{Sample} Suppose a user $u$ has a set of KV pairs $\mathcal{S}_u$. If $|\mathcal{S}_u| < \ell$, where $\ell$ is called \emph{padding length} and is a parameter of the protocols, then the user $u$ pads the set $\mathcal{S}_u$ with dummy KV pairs $\{\langle d + 1,0 \rangle,\langle d + 2,0 \rangle, \dots, \langle d +l-|\mathcal{S}_u|,0 \rangle\}$. 

Note that the maximum number of dummy KV pairs is $\ell$ when $\mathcal{S}_u$ is an empty set. After the padding, a random KV pair $\langle k,v \rangle$ is drawn from the padded set. The value $v$ is then discretized in the same way as PrivKVM, i.e., the value $v$ is discretized to $v^*=1$ with a probability of $\frac{1+v}{2}$ and $v^*=-1$ with a probability $\frac{1-v}{2}$.

\paragraph{Perturb} We denote $d'=d+l$ and $\mathcal{K}^{\prime} = \{1,2,\cdots,d+l\}$ (the dictionary with dummy keys). The Perturb steps for PCKV-UE and PKCV-GRR are as follows:
\begin{itemize}[leftmargin=*]
    \item \textbf{PCKV-UE:} PCKV-UE leverages Unary Encoding (UE) to perturb KV pairs.
    In particular, a perturbed vector $\mathbf{y} \in \{1, -1, 0\}^{d^{\prime}}$ is sent to the server, where $y[i]$ contains value information of key $i$ and is obtained as follows:
    \begin{align}
     \mathbf{y}[k] & =\left\{\begin{array}{lll}
            v^*, & \text { w.p. } & a \cdot p \\
            -v^*, & \text { w.p. } & a \cdot(1-p)\\
            0, & \text { w.p. } & 1-a
            \end{array} \right.,\\
    \mathbf{y}[i] &=\left\{\begin{array}{lll}
            1, & \text { w.p. } & b / 2 \\
            -1, & \text { w.p. } & b / 2  \\
            0, & \text { w.p. } & 1-b
            \end{array}\right.,i \in \mathcal{K}^{\prime}\setminus \{k\},
    \end{align}    
    where $a$, $b$, and $p$ are as follows:
    \begin{align}
    a=\frac{1}{2}, b=\frac{2}{e^{\epsilon}+3}, p=e^{\epsilon} /\left(e^{\epsilon}+1\right).
\end{align}
    \item \textbf{PCKV-GRR:} PCKV-GRR leverages Generalized Random Response (GRR) to perturb KV pairs. Specifically, the KV pair $\langle k,v^* \rangle$ is randomly perturbed into $\langle k^{\prime}, v^{\prime} \rangle$ as follows:
    \begin{equation}
        \left\langle k^{\prime}, v^{\prime}\right\rangle=\left\{\begin{array}{lll}
        \langle k, v^*\rangle, & \text { w.p. } & a \cdot p \\
        \langle k,-v^*\rangle, & \text { w.p. } & a \cdot(1-p) \\
        \langle i, 1\rangle , & \text { w.p. } & b \cdot 0.5 \\
        \langle i,-1\rangle , & \text { w.p. } & b \cdot 0.5
        \end{array}\right.,
    \end{equation}
    where $i \in \mathcal{K}^{\prime} \setminus \{k\}$ and  $a$, $b$, and $p$ are as follows:
\begin{align}
    a=\frac{\ell\left(e^{\varepsilon}-1\right)+2}{\ell\left(e^{\varepsilon}-1\right)+2 d^{\prime}}, b=\frac{1-a}{d^{\prime}-1}, p=\frac{\ell\left(e^{\varepsilon}-1\right)+1}{\ell\left(e^{\varepsilon}-1\right)+2}.
\end{align}
The perturbed KV pair $\langle k^{\prime}, v^{\prime} \rangle$ is sent to the server. 
\end{itemize}

\paragraph{Aggregate} Due to the difference in Perturb step, the Aggregate steps for PCKV-UE and PCKV-GRR are also different. Given a key $k$, we respectively use $n^k_1$ and $n^k_{-1}$ to denote the number of users that support the KV pairs $\langle k,1 \rangle$ and $\langle k,-1 \rangle$. In particular, they can be computed as follows:
\begin{itemize}[leftmargin=*]
    \item \textbf{PCKV-UE}: Recall that, in PCKV-UE, $\mathbf{y}[k]$ contains the value information of the key $k$. We say $\mathbf{y}[k]$ supports $\langle k,1 \rangle$ (or $\langle k,-1 \rangle$)  if 
    $\mathbf{y}[k]=1$ (or $\mathbf{y}[k]=-1$). 
    Then, we can compute $n^k_1$ (or $n^k_{-1}$) as the number of users whose perturbed vectors satisfy $\mathbf{y}[k]=1$ (or $\mathbf{y}[k]=-1$).
    \item \textbf{PCKV-GRR}: In PCKV-GRR, each user sends a single perturbed KV pair $\langle k^{\prime}, v^{\prime} \rangle$ to the server. Similar to PCKV-UE, we say $\langle k^{\prime}, v^{\prime} \rangle$ supports $\langle k,1 \rangle$ (or $\langle k,-1 \rangle$)  if $k'=k$ and $v^{\prime}=1$ (or $v^{\prime}=-1$). Then, we can compute $n^k_1$ (or $n^k_{-1}$) as the number of users whose perturbed KV pairs satisfy $k'=k$ and $v^{\prime}=1$ (or $v^{\prime}=-1$).
    
\end{itemize}
Given $n^k_1$ and $n^k_{-1}$, the server can estimate the frequency of key $k$ as follows:
\begin{equation}
    \hat{f}_{k}=\frac{\left(n^k_{1}+n^k_{-1}\right) / n-b}{a-b} \cdot \ell,
    \label{eq:freq_estimate_pckv}
\end{equation}
The estimated mean value of the key $k$ is computed as follows:
\begin{equation}
    \hat{m}_{k}=\ell\left(\hat{n}^k_{1}-\hat{n}^k_{-1}\right) /\left(n \hat{f}_{k}\right),
    \label{eq:mean_estimate_pckv}
\end{equation}
where
\begin{align}
 &\left[\begin{array}{l}
    \hat{n}^k_{1} \\
    \hat{n}^k_{-1}
    \end{array}\right]=A^{-1}\left[\begin{array}{l}
    n^k_{1}-n b / 2 \\
    n^k_{-1}-n b / 2
    \end{array}\right], \\ 
    &A=\left[\begin{array}{cc}
    a p-\frac{b}{2} & a(1-p)-\frac{b}{2} \\
    a(1-p)-\frac{b}{2} & a p-\frac{b}{2}
    \end{array}\right].
\end{align}

We note that in all the three LDP protocols, the server can clip the estimated frequency $\hat{f}_k$ to be $\frac{1}{n}$ if it is smaller than  $\frac{1}{n}$ and to be 1 if it is larger than 1. Moreover, the server can clip the support counts $\hat{n}_1^k$ and $\hat{n}_{-1}^k$ into the range of $[0,\frac{n \hat{f}_k}{\ell}]$ in PCKV-UE and PCKV-GRR, as well as the range of $[0, n_k]$ in PrivKVM, before using them to estimate the mean value. 

%%%%%%%%%%%%%%%%%%%%%%%%%%%%

\section{Threat Model}

\paragraph{Attacker's capability and background knowledge}
We assume that the attacker is able to inject some fake users into the system. Previous measurement study~\cite{thomas2013trafficking} has shown that an attacker can easily obtain a large number of fake/compromised users in online web services such as Twitter and Facebook. Specifically, we assume that the attacker has access to $m$ fake users. Together with the $n$ genuine users, the server estimates frequencies and mean values of keys among the $n+m$ users. For each fake user, the attacker can arbitrarily craft its message sent to the server. 
An attacker has access to the parameters  of the LDP protocol since the LDP protocol is executed on a user side. 
Specifically, an attacker has access to the dictionary of keys, as well as the implementation details of the Sample and Perturb steps of the LDP protocol.

\paragraph{Attacker's goal} An attacker aims to promote some target keys. We assume $r$ target keys and denote them as a set $\mathbb{T} = \{k_1,k_2,\cdots,k_r\}$. The attacker aims to increase the estimated frequencies and  mean values of the target keys via sending carefully crafted messages from the fake users to the server. 
Without loss of generality, we assume the $m$ fake users have IDs $n+1, n+2, \cdots, n+m$.  We denote the set of  messages the fake users send to the server as $\mathbb{Y} = \{y_i\}_{i=n+1}^{n+m}$, where $y_i$ is the message fake user $i$ sends to the server. We denote by  $\hat{f}_k$ and $\tilde{f}_k$ the estimated frequency of key $k$ among the $n$ genuine users and all the $n+m$ users, respectively. 
Moreover, we denote by $G_f(\mathbb{Y})=\sum_{k\in\mathbb{T}} \mathbb{E}[\Delta\hat{f}_k]$ the \emph{frequency gain} of the target keys, where $\Delta\hat{f}_k=\tilde{f}_k -\hat{f}_k$ and the expectation is taken over the randomness in a LDP protocol. 

Similarly, we denote by  $\hat{m}_k$ and $\tilde{m}_k$ the estimated mean value of key $k$ among the $n$ genuine users and all the $n+m$ users, respectively. Furthermore, we denote by   $G_m(\mathbb{Y})=\sum_{k\in\mathbb{T}} \mathbb{E}[\Delta\hat{m}_k]$ the \emph{mean gain} of the target keys, where $\Delta\hat{m}_k=\tilde{m}_k -\hat{m}_k$ and the expectation is taken over the randomness in a LDP protocol. An attacker aims to simultaneously maximize the frequency gain and mean gain via carefully crafting the messages $\mathbb{Y}$. We propose to formulate such an attack goal as the following two-objective optimization problem: 
\begin{equation}
    \max_{\mathbb{Y}} \ \begin{bmatrix}
        G_f(\mathbb{Y}) \\
        G_m(\mathbb{Y})
    \end{bmatrix}.
    \label{eq:combined_opt_goal_attack}
\end{equation}

Note that we consider the target keys are weighted equally for simplicity. However, our formulation can be extended to the scenario where  the attacker  assigns different weights to different target keys in the frequency and mean gains. 
A method to solve the  two-objective optimization problem is a poisoning attack to a LDP protocol for key-value data.

%%%%%%%%%%%%%%%%%%%%%%%%%%%%

\section{Our Attacks}
\label{sec:attack}
We first introduce our three attacks and then apply them to PrivKVM, PCKV-UE, and PCKV-GRR. 

\subsection{Three Attacks}
We propose  \emph{Maximal Gain Attack (M2GA)}, which solves the two-objective optimization problem to construct the optimal messages the fake users should send to the server.  To show the effectiveness of M2GA, we also propose two baseline attacks: \emph{Random Message Attack (\textit{RMA})} and \emph{Random Key-Value Pair Attack (\textit{RKVA})}. Next, we describe them one by one. 

\subsubsection{M2GA}
Our idea is to unify the frequency gain and mean gain for different LDP protocols under the same framework, based on which we transform the two-objective optimization problem to be one that is easier to solve. 

\paragraph{Unifying the frequency gain} We first observe that the estimated frequency $\hat{f}_k$  can be unified as Eq.~\ref{eq:freq_estimate_pckv}. In particular, as discussed in Section~\ref{back:pckv}, PCKV-UE and PCKV-GRR use Eq.~\ref{eq:freq_estimate_pckv} to calculate $\hat{f}_k$. We can also use Eq.~\ref{eq:freq_estimate_pckv} to calculate $\hat{f}_k$ in PrivKVM, where the parameters $a$, $b$, and $l$ are set as follows: 
\begin{equation}
    a=\frac{e^{\epsilon_{1}}}{e^{\epsilon_{1}}+1}, b=\frac{1}{e^{\epsilon_{1}}+1},\ell=1.
    \label{eq:param_freq_estimate_privkvm}
\end{equation}

Therefore, we can represent the frequency gain $G_f(\mathbb{Y})$ as:
\begin{align*}
    G_f(\mathbb{Y})&= \sum_{k\in \mathbb{T}} \mathbb{E}\left[ \tilde{f}_k - \hat{f}_k \right] \\
    &=\sum_{k\in \mathbb{T}} \ell\left\{\mathbb{E}\left[\frac{(n^k_{1}+n^k_{-1} + \tilde{n}^k_{1} + \tilde{n}^k_{-1}) / (n+m) -b}{a-b}\right]\right. \\
    & \left.- \mathbb{E} \left[\frac{(n^k_{1}+n^k_{-1}) / n -b}{a-b} \right]\right\},
\end{align*}
where $n^k_{1}$ and $n^k_{-1}$ respectively are the support counts of $\langle k, 1 \rangle$ and $\langle k, -1\rangle$ among the $n$ genuine users, while $\tilde{n}^k_{1}$ and $\tilde{n}^k_{-1}$ are the ones among the $m$ fake users. We note that the messages $\mathbb{Y}$ only affect the term $\sum_{k\in \mathbb{T}} \mathbb{E}[\frac{(\tilde{n}^k_{1} + \tilde{n}^k_{-1})}{(n+m)(a-b)}]$ and the denominator $(n+m)(a-b)$ is irrelevant in the optimization for a given setting of LDP protocol. 
Therefore, optimizing the frequency gain is equivalent to optimizing the following:
\begin{equation}
    \max_{\mathbb{Y}} \sum_{k\in \mathbb{T}}  (\mathbb{E}[\tilde{n}^k_{1}] + \mathbb{E}[\tilde{n}^k_{-1}]).
    \label{eq:freq_attack_goal}
\end{equation}
Moreover, the frequency gain can be simplified as follows:
\begin{equation}
    G_{f}(\mathbb{Y}) =  \frac{\ell}{(n+m)(a-b)}\sum_{k\in \mathbb{T}} (\mathbb{E}[\tilde{n}^k_{1}] + \mathbb{E}[\tilde{n}^k_{-1}]) - c,
    \label{eq:freq_gain}
\end{equation}
where $c=\sum_{k\in\mathbb{T}}\frac{m \ell (n^k_1+n^k_{-1})}{n(n+m)(a-b)} = \frac{m\ell}{n+m}(f_{\mathbb{T}} + \frac{rb}{a-b})$. $f_{\mathbb{T}}=\sum_{k\in\mathbb{T}} f_k$ is the sum of the true frequencies of all target keys, which is a constant.

\paragraph{Unifying the mean gain} Similar to frequency estimation, the estimated mean value can also be unified in the following equation: 
\begin{equation}
    \hat{m}_k = \frac{\left(n^k_{1}-n^k_{-1}\right)(a-b)}{a(2 p-1)\left(n^k_{1}+n^k_{-1}-n b\right)},
    \label{eq:mean_estimation_ver2_pckv}
\end{equation}
where the parameters $a$, $b$, $p$, and $l$ are described in Section~\ref{back:pckv} for PCKV-UE and PCKV-GRR, and they are set as follows for PrivKVM: 
\begin{equation}
    a=1,b=0,p=\frac{e^{\epsilon_{2}}}{e^{\epsilon_{2}}+1}, l=1.
    \label{eq:param_mean_estimate_privkvm}
\end{equation}

Then, we can represent the mean gain $G_m(\mathbb{Y})$ as follows:
\begin{small}
\begin{align}
    G_m(\mathbb{Y}) &=\sum_{k\in \mathbb{T}} \mathbb{E}\left[\tilde{m}_k - \hat{m}_k\right] \nonumber\\
    &= \sum_{k\in \mathbb{T}}\left\{\mathbb{E}\left[ \frac{\left(n^k_{1}-n^k_{-1} + \tilde{n}^k_{1}-\tilde{n}^k_{-1} \right)(a-b)}{a(2 p-1)\left(n^k_{1}+n^k_{-1}+ \tilde{n}^k_{1}+\tilde{n}^k_{-1} - (n+m)b\right)}   \right]\right. \nonumber\\
    &\  \left. - \mathbb{E} \left[ \frac{\left(n^k_{1}-n^k_{-1}\right)(a-b)}{a(2 p-1)\left(n^k_{1}+n^k_{-1}-n b\right)} \right]\right\}.
\end{align}
\end{small}

However, unlike the frequency gain, it is non-trivial to compute the two expectations above because they involve divisions between random variables. Specifically, since $n^k_{1}$ and $n^k_{2}$ are random variables, both the numerator and the denominator are random variables. To address the challenge, we propose to use the  first-order Taylor expansion of functions of random variables~\cite{Casella1990StatisticalI} to approximately compute $G_m$. Specifically, given two random variables $X$ and $Y$, the first-order Taylor expansion means the following: 
\begin{equation}
    \mathbb{E}\left[\frac{X}{Y}\right] \approx \frac{\mathbb{E}[X]}{\mathbb{E}[Y]}.
\end{equation}
Note that we have the following:
$$
\mathbb{E}[n^k_{1} - n^k_{-1}]=n \frac{f_{k}}{\ell} a(2 p-1) m_{k},
$$
$$
\mathbb{E}[n^k_{1} + n^k_{-1}-nb]= n \frac{f_{k}}{\ell}(a-b),
$$
where $f_k$ and $m_k$ are the true frequency and mean value of $k$. 
Thus, based on the first-order Taylor expansion, we have: 
\begin{small}
\begin{equation}
    G_m \approx \sum_{k\in \mathbb{T}} \left(\frac{a-b}{a(2p-1)}\frac{n f_k a (2p-1) m_k/\ell +  \mathbb{E}[\tilde{n}^k_1] - \mathbb{E}[\tilde{n}^k_{-1}]}{n f_k (a-b) / \ell+ \mathbb{E}[\tilde{n}^k_1] + \mathbb{E}[\tilde{n}^k_{-1}] - mb} - m_k\right).
    \label{eq:mean_gain}
\end{equation}
\end{small}
For simplicity, we denote $c_1^k=n f_k a (2p-1) m_k/\ell$ and $c_2^k=n f_k (a-b) / \ell -mb$. Then, we approximate optimizing the mean gain as follows:
\begin{equation}
    \max_{\mathbb{Y}} \sum_{k\in \mathbb{T}} \frac{\mathbb{E}[\tilde{n}_{1}^k] - \mathbb{E}[\tilde{n}_{-1}^k] + c_1^k}{\mathbb{E}[\tilde{n}_{1}^k] + \mathbb{E}[\tilde{n}_{-1}^k] + c_2^k}.
    \label{eq:mean_attack_goal}
\end{equation}

\paragraph{Reformulated two-objective optimization problem} 
By combining Eq.~\ref{eq:freq_attack_goal} and~\ref{eq:mean_attack_goal}, we  re-formulate our  two-objective optimization problem as follows:
\begin{equation}
    \max_{\mathbb{Y}} \ \begin{bmatrix}
         \sum_{k\in \mathbb{T}}  (\mathbb{E}[\tilde{n}^k_{1}] + \mathbb{E}[\tilde{n}^k_{-1}]) \\
         \sum_{k\in \mathbb{T}} \frac{\mathbb{E}[\tilde{n}_1^k] - \mathbb{E}[\tilde{n}_{-1}^k] + c_1^k}{\mathbb{E}[\tilde{n}_1^k] + \mathbb{E}[\tilde{n}_{-1}^k] + c_2^k}
    \end{bmatrix}.
    \label{eq:combined_opt_goal}
\end{equation}

\subsubsection{RMA}
In this baseline attack, each fake user picks a  message uniformly at random from the message domain allowed by a LDP protocol and sends it to the server. 

\subsubsection{RKVA}
RMA does not consider any information about the target keys. 
Different from RMA, RKVA considers the target keys. Specifically, each fake user picks a random target key $k$, pairs it with an extreme value 1,  and the constructed KV pair is viewed as the fake user's KV pair. Then, the constructed KV pair is processed by the LDP protocol and the resulting message is sent to the server.

\subsection{Attacking PrivKVM}
\subsubsection{M2GA}
Recall that PrivKVM is an iterative procedure, in which the frequency estimation is performed in the first round while the mean estimation is performed in each round and the estimated mean values in the last round are used. 
Solving Eq.~\ref{eq:combined_opt_goal} exactly is non-trivial. Therefore, we propose a two-step approximate solution, which  first optimizes the frequency gain and then approximately optimizes the mean gain.

Since each user in PrivKVM sends a key index and a tuple  to the server, we craft such message for each fake user such that  Eq.~\ref{eq:freq_attack_goal} is maximized. For PrivKVM, a fake user can only inject a single key to be counted by the server. That is, a fake user can only increase either $\tilde{n}_1^k$ or $\tilde{n}_{-1}^k$ for a single key $k$. Therefore, for each fake user, we randomly select a target key $k$, and send the index $k$ and the tuple $\langle 1, \cdot \rangle$ to the server, where the reported value does not influence the frequency gain and we will discuss it for optimizing the mean gain. Thus, we have  $\sum_{k \in \mathbb{T}} (\mathbb{E}[\tilde{n}_{1}^{k}]+\mathbb{E}[\tilde{n}_{-1}^{k}])=m$, and we have the frequency gain as $G_f = \frac{m}{(n+m)(a-b)} -c$, where the parameters $a,b$ are defined in Eq.~\ref{eq:param_freq_estimate_privkvm}. In practice, for each target key,  $\frac{m}{r}$ fake users send messages including the target key to the server. 

For mean estimation, we attack each round of PrivKVM. Specifically, in PrivKVM, the value sent to the server is either 1 or -1. Therefore, to increase the estimated mean, a fake user always sends value 1 for a target key.
In other words, we have $\mathbb{E}[\tilde{n}^k_1]=\frac{m}{r}$ and  $\mathbb{E}[\tilde{n}^k_{-1}] = 0$ for each target key $k$. The mean gain is as follows: $G_{m} \approx \sum_{k \in \mathbb{T}} \frac{a-b}{a(2 p-1)} \frac{n f_{k} a(2 p-1) m_{k} / \ell+m / r}{n f_{k}(a-b) / \ell+ m / r -mb)}-m_{k}$, where $a,b,p$ are given in Eq.~\ref{eq:param_mean_estimate_privkvm}.

To summarize, each fake user sends a random target key and value 1 to the server in each round of PrivKVM.

\subsubsection{RMA}
In RMA, each fake user randomly chooses a key $k$ from the entire dictionary. Then, the fake user randomly chooses a tuple to report to the server. Specifically, $\langle 0, 0\rangle$ is chosen with a probability of $\frac{1}{2}$, while $\langle 1, -1\rangle$ and $\langle 1, 1\rangle$ are each chosen with probability $\frac{1}{4}$.

Therefore, we have a probability of $\frac{1}{2d}$ that the message of a fake user  supports key $k$, and the KV pairs $\langle k, 1\rangle$ and $\langle k, -1 \rangle$ would be supported with equal probabilities. Thus, we have $\mathbb{E}[\tilde{n}_{1}^{k}] = \mathbb{E}[\tilde{n}_{-1}^{k}] = \frac{m}{4d}$. By plugging the values into Eq.~\ref{eq:freq_gain} and Eq.~\ref{eq:mean_gain}, we have the frequency gain of $G_f = \frac{mr}{2(n+m)(a-b)d} - c$, and the mean gain  of $G_m\approx \sum_{k \in \mathbb{T}} \frac{a-b}{a(2 p-1)} \frac{n f_{k} a(2 p-1) m_{k}}{n f_{k}(a-b) + m / (2d) -m b} - m_k$. Again, we note that the parameters $a,b,p$ are different in $G_f$ and $G_m$.

\subsubsection{RKVA}
In RKVA, each fake user picks a target key $k$  uniformly at random and  the fake user's tuple is $\langle 1, 1\rangle$. 
This tuple is perturbed according to the {Perturb} step of the PrivKVM protocol. The perturbed tuple still supports $k$ with a probability of $\frac{e^{\epsilon_{1}}}{e^{\epsilon_{1}}+1}$, and the value is inverted with a probability of $\frac{1}{1+e^{\epsilon_2}}$. 

Therefore, we have $\mathbb{E}[\tilde{n}_{1}^{k}] = \mathbb{E}[\tilde{n}_{-1}^{k}] = \frac{m e^{\epsilon_{1}}}{2r(e^{\epsilon_{1}}+1)}$. 
$\mathbb{E}[\tilde{n}_{1}^{k}] = \frac{m e^{\epsilon_{1}}e^{\epsilon_2}}{r(e^{\epsilon_{1}}+1)(1+e^{\epsilon_2})}$ and $\mathbb{E}[\tilde{n}_{-1}^{k}] = \frac{m e^{\epsilon_{1}}}{r(e^{\epsilon_{1}}+1)(1+e^{\epsilon_2})}$.
The frequency gain is $G_f=\frac{m e^{\epsilon_1}}{(n+m)(a-b)(e^{\epsilon_1} + 1)} - c$, and the mean gain is $G_m \approx \sum_{k \in \mathbb{T}} \frac{a-b}{a(2 p-1)} \frac{n f_{k} a(2 p-1) m_{k} + m e^{\epsilon_{1}}(e^{\epsilon_2}-1)/r(e^{\epsilon_{1}}+1)(1+e^{\epsilon_2})}{n f_{k}(a-b) +  m e^{\epsilon_{1}} / (r(e^{\epsilon_{1}}+1)) -m b} - m_k$.

\subsection{Attacking PCKV-UE}
\subsubsection{M2GA}
In PCKV-UE, each user sends a vector of length $d + \ell$ to the server, and each dimension is checked independently on whether it supports the corresponding key. Therefore, a single user could support multiple keys. 

For each fake user, we put a 1 or -1 in all the dimensions corresponding to the target keys. Therefore, a single fake user can increase $\tilde{n}_1^k$ or $\tilde{n}_{-1}^k$ for all $k\in\mathbb{T}$. For the remaining dimensions, if we simply leave them as 0, the server may easily detect that these messages are from fake users. To address this issue, we sample some dimensions and set them to 1 or -1, such that the vectors we craft for the fake users would have the same number of 1 bits and -1 bits as the expectation of the genuine users'. Specifically, if a genuine user samples a KV pair $\langle \cdot, 1 \rangle$ (or $\langle \cdot, -1 \rangle$) to report, the perturbed vector would have $\lfloor a p + (d^{\prime} - 1) (b / 2) \rfloor$ 1 (or -1) bits and $a (1 - p) + (d^{\prime} - 1) (b / 2)$ -1 (or 1) bits on expectation. We note that this form of disguise does not affect the frequency gain and mean gain for the target keys.

Therefore, we have $\tilde{n}_1^k + \tilde{n}_{-1}^k= m$ for each target key $k$. The frequency gain is $G_f=\frac{ mr\ell}{(n+m)(a-b)} - c$. To further maximize the mean gain, we solve the optimization problem of Eq.~\ref{eq:mean_attack_goal} in the similar way as for PrivKVM. Specifically, we only need to maximize $\tilde{n}_{1}^{k}-\tilde{n}_{-1}^{k}$ for each target key $k$ under the constraints of $\tilde{n}_{1}^{k} \geq 0, \tilde{n}_{-1}^{k} \geq 0$ and $\tilde{n}_1^k + \tilde{n}_{-1}^k= m$. Therefore, we have the following optimal solution $\tilde{n}^k_1 = m$ and $\tilde{n}^k_{-1} = 0$.  
Thus, we obtain the  mean gain  as: $G_{m} \approx \sum_{k \in \mathbb{T}} \frac{a-b}{a(2 p-1)} \frac{n f_{k} a(2 p-1) m_{k} / \ell+m}{n f_{k}(a-b) / \ell+m(1-b)}-m_{k}$.

To summarize, each fake user sets the dimensions corresponding to the target keys to 1 in its vector. Moreover, to evade possible detection, each fake user randomly samples some other dimensions of its vector and set them to be 1 or -1 such that  the vector has the  same number of 1 bits and -1 bits as the expected 1 bits and -1 bits in a  genuine user's vector.

\subsubsection{RMA}
 In PCKV-UE, a message is a vector. 
For each fake user, we randomly sample the value of each dimension of the vector. Specifically, each dimension is randomly set to 1, -1, or 0 with an equal probability of $\frac{1}{3}$. Therefore, we have $\mathbb{E}[\tilde{n}_{1}^{k}]= \mathbb{E}[\tilde{n}_{-1}^{k}] = \frac{m}{3}$. The frequency gain is $G_f=\frac{2 mr\ell}{3(n+m)(a-b)} - c$. The mean gain  is $G_{m} \approx \sum_{k \in \mathbb{T}} \frac{a-b}{a(2 p-1)} \frac{n f_{k} a(2 p-1) m_{k} / \ell}{n f_{k}(a-b) / \ell+ 2m/3 - mb)}-m_{k}$. 

\subsubsection{RKVA}
For each fake user, a random target key $k$ is sampled from $\mathbb{T}$ and paired with a value of 1. The constructed KV pair is then perturbed by the PCKV-UE protocol. As the perturbation is independent for each dimension, we only need to focus on the dimensions corresponding to the target keys in $\mathbb{T}$. Similar to  PrivKVM, if $k$ is selected, its value remains the same with probability $ap$ and gets inversed to -1 with a probability of $a(1-p)$. 
However, different from PrivKVM, if a key other than $k$ is selected, the perturbed  vector  supports $k$ with  probability $b$ and the value is 1 or -1 with equal probability $\frac{b}{2}$. Therefore, we have $\mathbb{E}[\tilde{n}_{1}^{k}] =\frac{map+m(r-1)b/2}{r}$ and $\mathbb{E}[\tilde{n}_{-1}^{k}] =\frac{ma(1-p)+m(r-1)b/2}{r}$. 
The frequency gain and the mean gain are respectively as follows: $G_f=\frac{ma\ell + m(r-1)b\ell}{(n+m)(a-b)} - c$ and $G_{m} \approx \sum_{k \in \mathbb{T}} \frac{n f_{k} m_{k} / \ell + m/r}{n f_{k} / \ell+ m/r}-m_{k}$.

\subsection{Attacking PCKV-GRR}
\subsubsection{M2GA}
In PCKV-GRR, each user sends a KV pair as the message to the server. A user supports a KV pair if and only if the message he/she sent to the server is exactly the KV pair. Therefore,  similar to PrivKVM, a fake user can only increase the support count $\tilde{n}_1^k$ or $\tilde{n}_{-1}^k$ of a single target key. Similar to PrivKVM, we have $\sum_{k \in \mathbb{T}} \mathbb{E}\left[\tilde{n}_{1}^{k}+\tilde{n}_{-1}^{k}\right]=m$. The frequency gain is  $G_{f}=\frac{m\ell}{(n+m)(a-b)}-c$. 

The mean gain for each target key is then maximized in the same way as each round of PrivKVM, where we set $\tilde{n}_{1}^{k}=\frac{m}{r}$ and  $\tilde{n}_{-1}^{k}=0$, i.e., we set all the values in the KV pairs sent to the server to 1. Thus, the mean gain is  $G_m \approx \sum_{k \in \mathbb{T}} \frac{a-b}{a(2 p-1)} \frac{n f_{k} a(2 p-1) m_{k}/\ell+m / r}{n f_{k}(a-b)/\ell+ m/r - mb}-m_{k}$.

To summarize,  each fake user sends a random target key and value 1 to the server.

\subsubsection{RMA}
For each fake user, we randomly select a target key and  set its corresponding value as -1 or 1 uniformly at random , which is  the KV pair sent to the server. Therefore, we have $\mathbb{E}[\tilde{n}_{1}^{k}] = \mathbb{E}[\tilde{n}_{-1}^{k}] = \frac{m}{2d^{\prime}}$. 
We have the frequency gain as $G_f = \frac{mr\ell}{(n+m)(a-b)d^{\prime}} - c$, while the mean gain as $G_m^k\approx \sum_{k \in \mathbb{T}} \frac{a-b}{a(2 p-1)} \frac{n f_{k} a(2 p-1) m_{k}/\ell}{n f_{k}(a-b)/\ell + m / d^{\prime} -m b} - m_k$.

\subsubsection{RKVA}
For each fake user, we randomly select a target key $k$ from $\mathbb{T}$ and choose its value 1 to construct a KV pair. The KV pair is then perturbed according to the {Perturb} step of the PCKV-GRR protocol. The KV pair after perturbation still keeps $k$ as its key with a probability of $a$, and a  key other than $k$ gets perturbed to $k$ with a probability of $b$. 
Therefore, we get $\mathbb{E}[\tilde{n}_{1}^{k}]=\frac{map + m(r-1)b/2}{r}$ and $\mathbb{E}[\tilde{n}_{-1}^{k}]=\frac{ma(1-p) + m(r-1)b/2}{r}$. 
We have $G_f=\frac{ma\ell+m(r-1)b\ell }{(n+m)(a-b)}- c$ and $G_{m} \approx \sum_{k \in \mathbb{T}} \frac{n f_{k} m_{k} / \ell + m/r}{n f_{k} / \ell+ m/r}-m_{k}$.

\begin{table*}
    \caption{Frequency gains of the three attacks for PrivKVM, PCKV-UE, and PCKV-GRR.  $\beta=\frac{m}{n}$ is the fraction of fake users, $f_{\mathbb{T}}=\sum_{k \in \mathbb{T}} f_{k}$ is the sum of true frequencies of the target keys, $\epsilon$ is the privacy budget, $d$ is the total number of keys, $\ell$ is the padding length, $d^{\prime}=d+\ell$ is the padded dictionary size, and $r$ is the number of target keys.}
    \label{tab:expected_freq_gain}
    \centering
    \begin{tabular}{lccc}
        \toprule
         & PrivKVM & PCKV-UE & PCKV-GRR \\
        \midrule
        M2GA & $\frac{\beta}{1+\beta} \left[1-f_{\mathbb{T}} + \frac{2-r}{e^{\epsilon/2} -1}\right]$ & $\frac{\beta \ell}{1+\beta}\left[2r-f_{\mathbb{T}}+\frac{4r}{e^{\epsilon}-1}\right]$ & $\frac{\beta }{1+\beta}\left[(1-f_\mathbb{T})\ell + \frac{2(d^{\prime}-r)}{e^{\epsilon}-1} \right]$\\[6pt]
        
        RMA & $\frac{\beta}{1+\beta} \left[\frac{(e^{\epsilon/2}-2d+1)r}{2(e^{\epsilon/2}-1)d}-f_\mathbb{T}\right]$ & $\frac{\beta\ell}{1+\beta}\left[\frac{4e^{\epsilon}r}{3(e^\epsilon -1)}-f_\mathbb{T}\right]$ & $\frac{\beta (r-f_\mathbb{T}d^{\prime}) \ell}{(1+\beta)d^{\prime}}$ \\[6pt]
        
        RKVA & $\frac{\beta}{1+\beta}\left[ 1-f_\mathbb{T}+\frac{1-r}{e^{\epsilon/2}-1}\right]$ & $\frac{\beta \ell}{1+\beta}\left(1-f_\mathbb{T}\right)$ & $\frac{\beta \ell}{1+\beta}\left(1-f_\mathbb{T}\right)$\\
        \bottomrule
    \end{tabular}
\end{table*}

\begin{table*}
    \caption{Approximate mean gains of the three attacks for PrivKVM, PCKV-UE, and PCKV-GRR.  $f_k$ is the true frequency of key $k$, $m_k$ is the true mean value of $k$, $\epsilon_2=\frac{\epsilon}{2 N_{\text{iter}}}$ is the privacy budget in each round of PrivKVM, and  $N_{\text{iter}}$ is the number of rounds.}
    \label{tab:expected_mean_gain}
    \centering
    \begin{tabular}{lccc}
        \toprule
         & PrivKVM & PCKV-UE & PCKV-GRR \\
        \midrule
        M2GA & $\sum_{k \in \mathbb{T}} \frac{f_k m_k r + \frac{e^{\epsilon_2}+1}{e^{\epsilon_2}-1}\beta}{f_k r + \beta}-m_k$ & $\sum_{k \in \mathbb{T}}\frac{2\beta\ell (e^{\epsilon}+1) + (e^\epsilon -1)f_k m_k }{2\beta\ell (e^\epsilon+1) + (e^\epsilon -1)f_k} - m_k$ & $\sum_{k \in \mathbb{T}}\frac{(e^\epsilon-1)(\beta\ell+f_k m_k r) + 2\beta d^{\prime}}{\beta[(e^\epsilon -1)\ell + 2(d^{\prime}-r)]+(e^\epsilon-1)f_k r} - m_k$\\[6pt]
        
        RMA & $\sum_{k \in \mathbb{T}}\frac{2f_k m_k d}{2 f_k d+\beta}-m_k$ & $\sum_{k \in \mathbb{T}}\frac{3(e^\epsilon-1)f_k m_k}{3(e^\epsilon-1)f_k+4e^\epsilon \beta\ell } - m_k$ & $\sum_{k \in \mathbb{T}}\frac{f_k m_k d^{\prime}}{f_k d^{\prime} + \beta \ell} - m_k$ \\[6pt]
        
        RKVA & $\sum_{k \in \mathbb{T}}\frac{f_k m_kr(e^{\epsilon/2}+1) + e^{\epsilon/2}\beta}{f_k(e^{\epsilon/2}+1) + e^{\epsilon/2}\beta} - m_k$ & $\sum_{k \in \mathbb{T}}\frac{f_k m_k r + \beta\ell}{f_k r + \beta\ell} - m_k$ & $\sum_{k \in \mathbb{T}}\frac{f_k m_k r + \beta\ell}{f_k r + \beta\ell} - m_k$\\
        \bottomrule
    \end{tabular}
\end{table*}

\subsection{Theoretical Analysis}
\label{sec:theo_ana}
Table~\ref{tab:expected_freq_gain} summarizes the analytical forms of the frequency gains of our three attacks for the three LDP protocols, while Table~\ref{tab:expected_mean_gain} summarizes the analytical forms of the approximate mean gains of our three attacks for the three LDP protocols. 
We have replaced the protocol-dependent parameters $a$, $b$, and $p$ in the analytical forms. 
We note that we do not consider clipping the estimated frequency $\hat{f}_k$ and support counts $n^k_1, n^k_{-1}$ when deriving the analytical forms. 

First, we observe that  M2GA outperforms  RMA and RKVA.  This is because M2GA crafts the fake users' messages via solving a two-objective optimization problem. In particular, our two objectives are non-convex. However, M2GA achieves the optimal frequency gain for the three LDP protocols in all cases, as we discussed when describing M2GA for each protocol. Moreover, in a given execution of a LDP protocol, M2GA also achieves the optimal mean gain if there is only one target key $k$ (i.e., $r=1$) and $n_1^k\ge n_{-1}^k>\frac{(n+m)b}{2}$ (Appendix \ref{appendix:proof} shows the proof).
Second, we observe that the  frequency gain of an attack  increases as the fraction of fake users increases. However, we do not have this observation for mean gains. We suspect the reason is that the mean gain depends on the estimated frequency and that we approximate the mean gains via Taylor expansion. Third, the frequency gain is larger when the total true frequencies $f_\mathbb{T}$ of the target keys is smaller. This is because the frequency gain is the difference between the estimated frequencies before and after attack. Moreover, the approximate mean gain becomes larger when the true mean value of each target key becomes smaller.  

Prior work~\cite{cheu2019manipulation,cao2019data} observed  trade-off between security against poisoning attacks and privacy in LDP protocols for categorical and numerical data, i.e., such a LDP protocol is more vulnerable to poisoning attacks if it uses a smaller privacy budget. Our fourth observation is that such security-privacy trade-off does not necessarily hold in LDP protocols for key-value data. In particular, while we observe such security-privacy trade-off for the frequency gains of M2GA to PCKV-UE and PCKV-GRR, how the privacy budget $\epsilon$ influences the frequency gain of  M2GA to PrivKVM depends on $r$, the number of target keys. Specifically, the frequency gain of  M2GA to PrivKVM increases, does not change, and decreases as the privacy budget $\epsilon$ decreases when $r=1$, $r=2$, and $r>2$, respectively.  The \emph{approximate} mean gain of M2GA in Table \ref{tab:expected_mean_gain} has such security-privacy trade-off as we can verify that the derivative of the approximate mean gain of M2GA with respect to $\epsilon$ is negative. However, we do not necessarily observe such security-privacy trade-off for the \emph{true} mean gains of M2GA in our experiments. This is because the approximate mean gains are obtained using Taylor expansion of the true ones and the protocols clip frequencies and support counts in practice.

%%%%%%%%%%%%%%%%%%%%%%%%%%%%

\section{Evaluation}

\subsection{Experimental Setup}

\subsubsection{Datasets} 
\begin{table}
    \caption{Dataset statistics. \#records indicates the total number of KV pairs in a dataset, while the 90th-percentile refers to that of the number of KV pairs possessed per user.}
    \label{tab:datasets}
    \centering
    \resizebox{0.98\columnwidth}{!}{
    \begin{tabular}{lrrrr}
        \toprule
        Dataset & \#users & \#keys & \#records & 90th-percentile \\
        \midrule
        Synthetic & \numprint{100000} & \numprint{100} & \numprint{100000} & 1.0 \\
        Clothing & \numprint{105508} & \numprint{5850} & \numprint{192198} & 3.0 \\        
        TalkingData & \numprint{60822} & \numprint{320} & \numprint{1327468} & 34.0 \\
        MovieLens-1M & \numprint{943} & \numprint{1682} & \numprint{100000} & 244.4 \\
        \bottomrule
    \end{tabular}
    }
\end{table}

\begin{figure}[!t]
    \centering 
    
    \subfloat{  \includegraphics[width=0.15\textwidth]{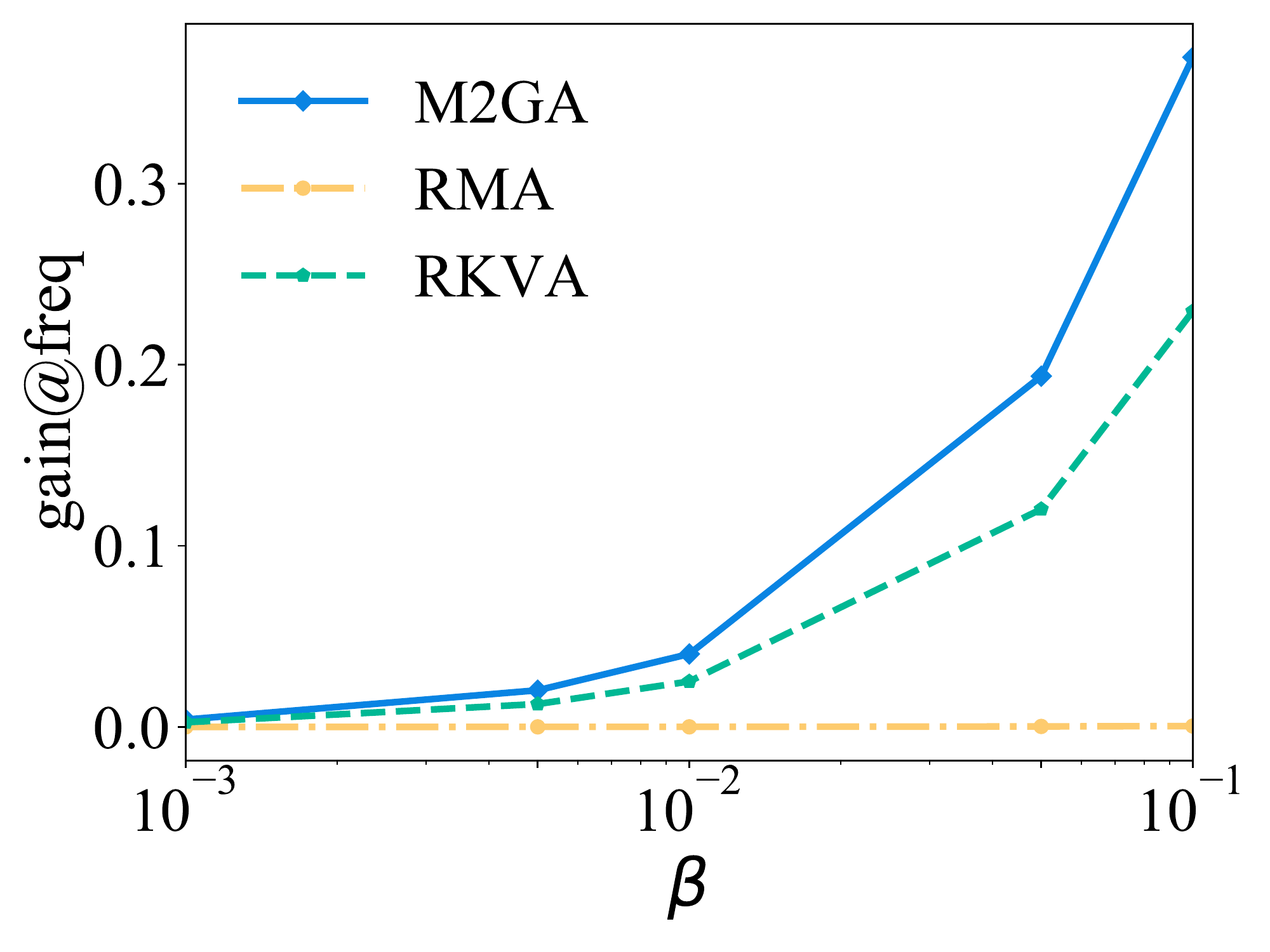}}
    \subfloat{  \includegraphics[width=0.15\textwidth]{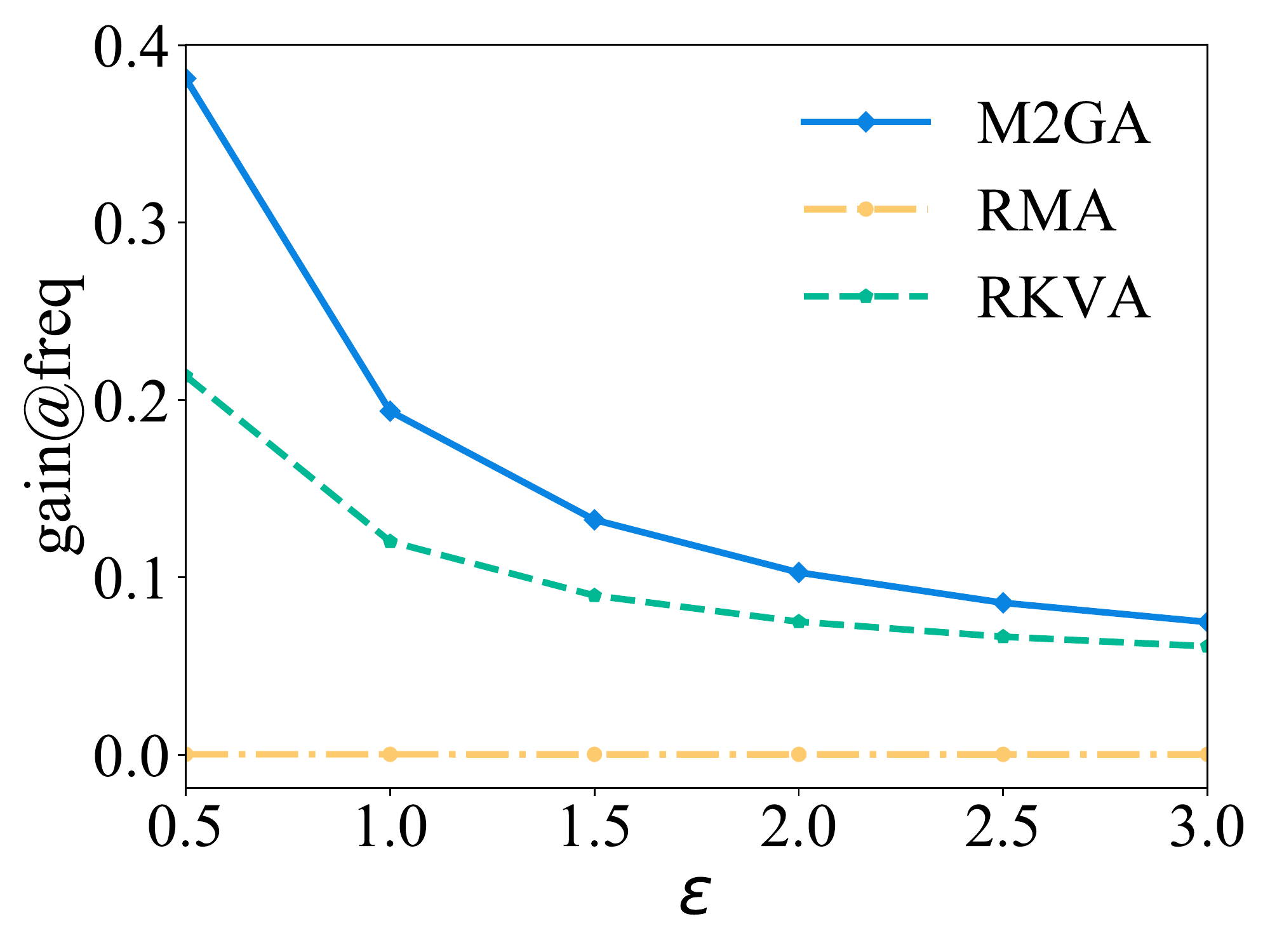}}
    \subfloat{  \includegraphics[width=0.15\textwidth]{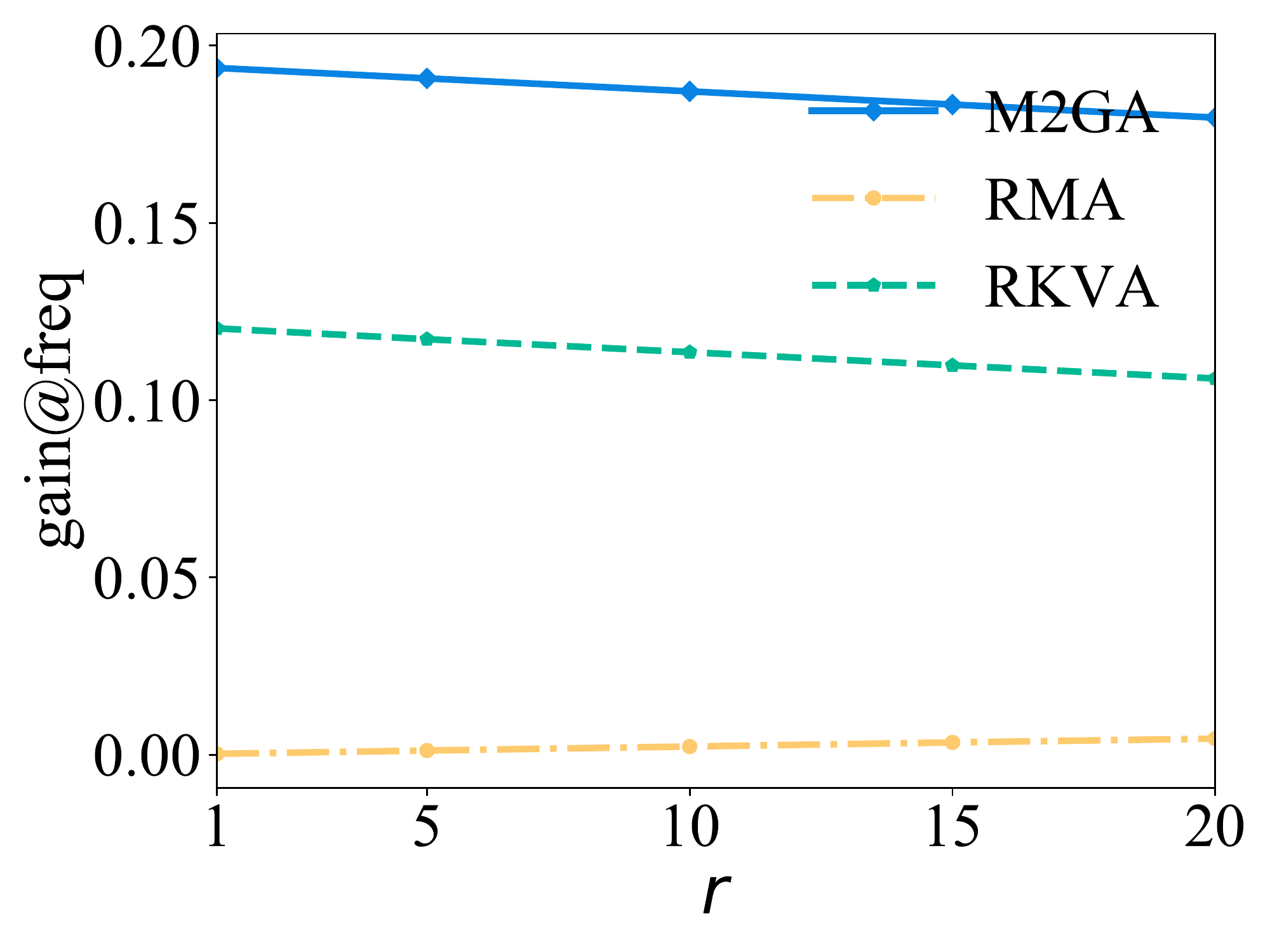}}

    \subfloat{  \includegraphics[width=0.15\textwidth]{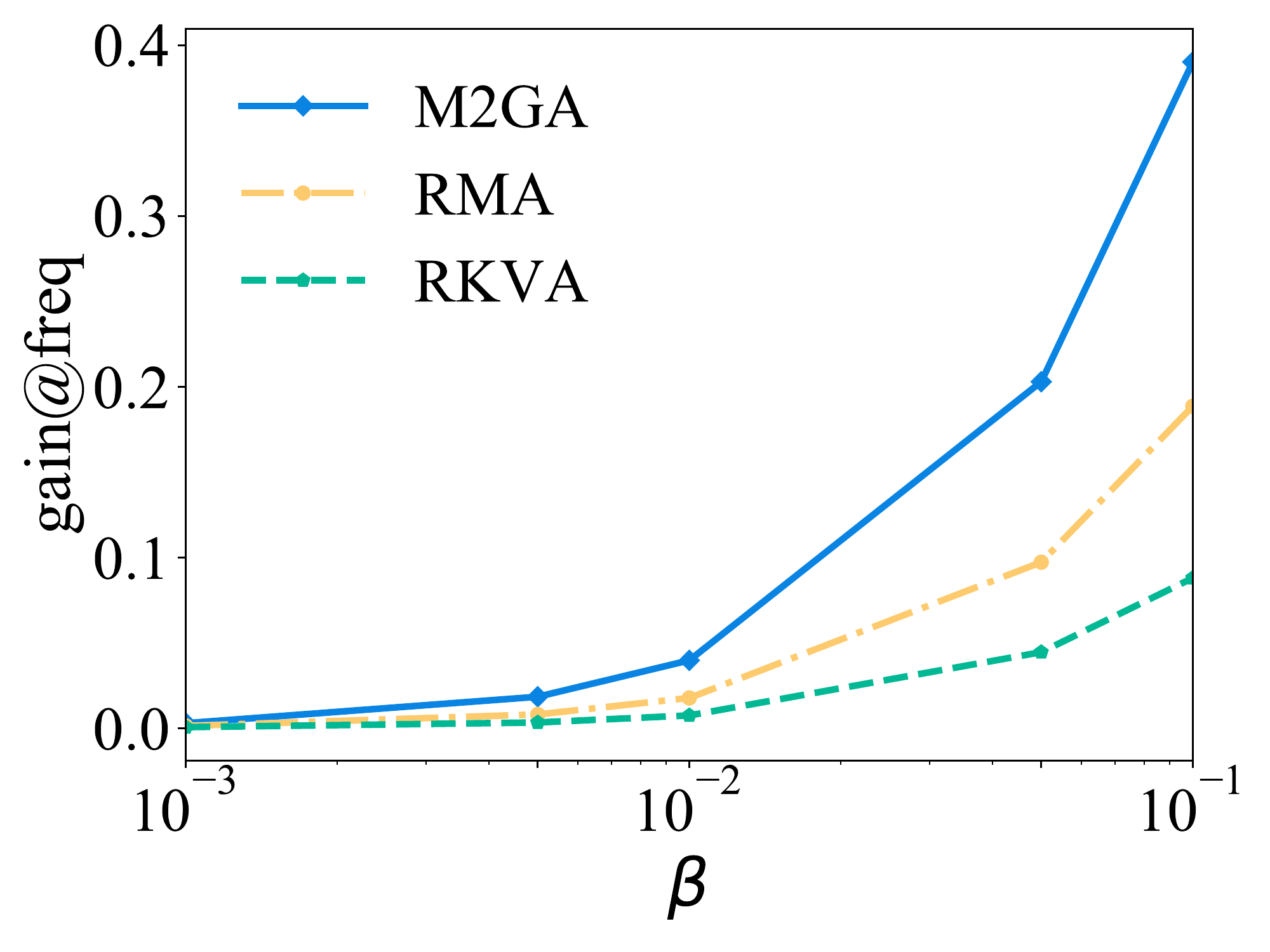}}
    \subfloat{  \includegraphics[width=0.15\textwidth]{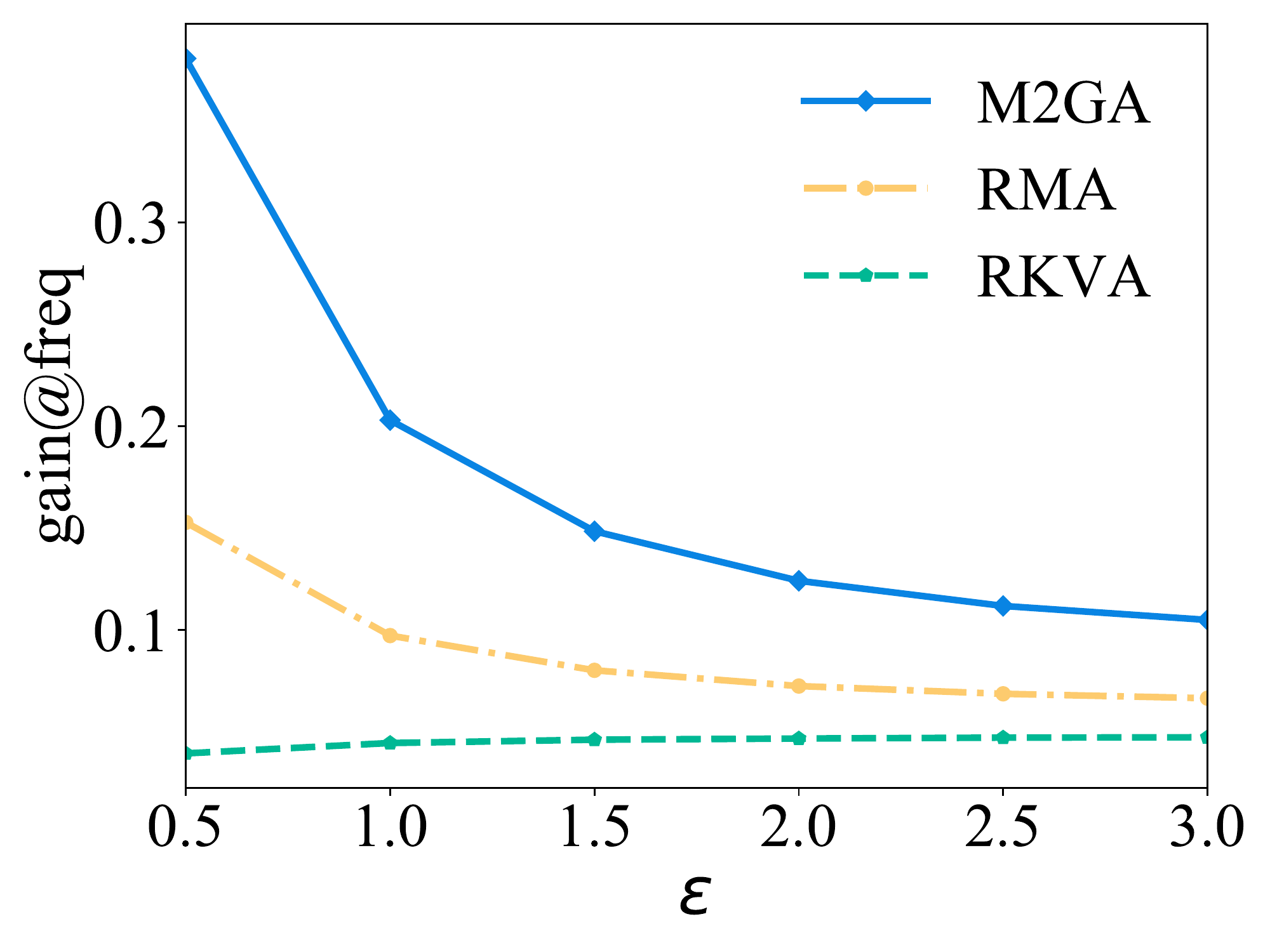}}
    \subfloat{  \includegraphics[width=0.15\textwidth]{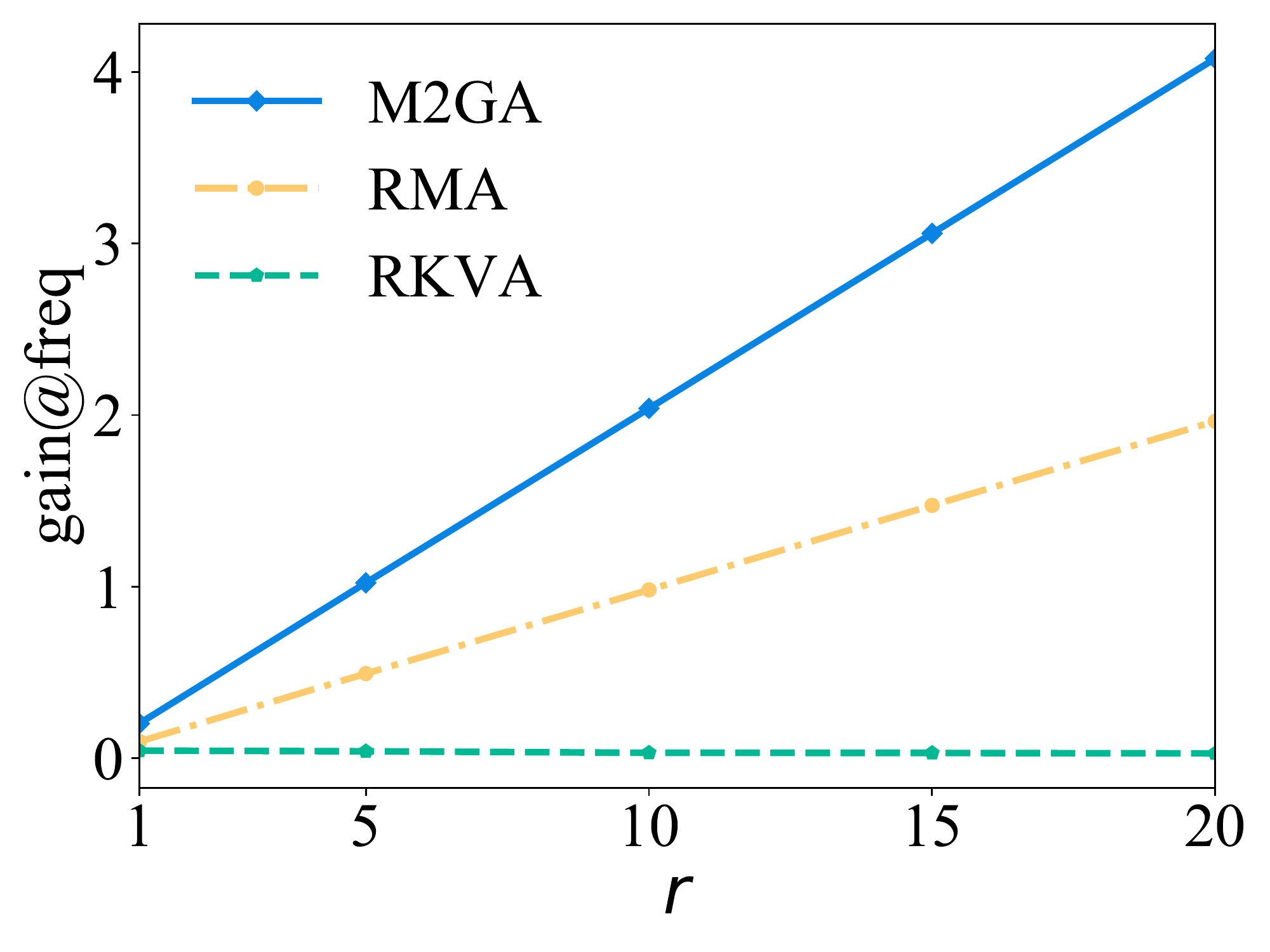}}

    \subfloat{  \includegraphics[width=0.15\textwidth]{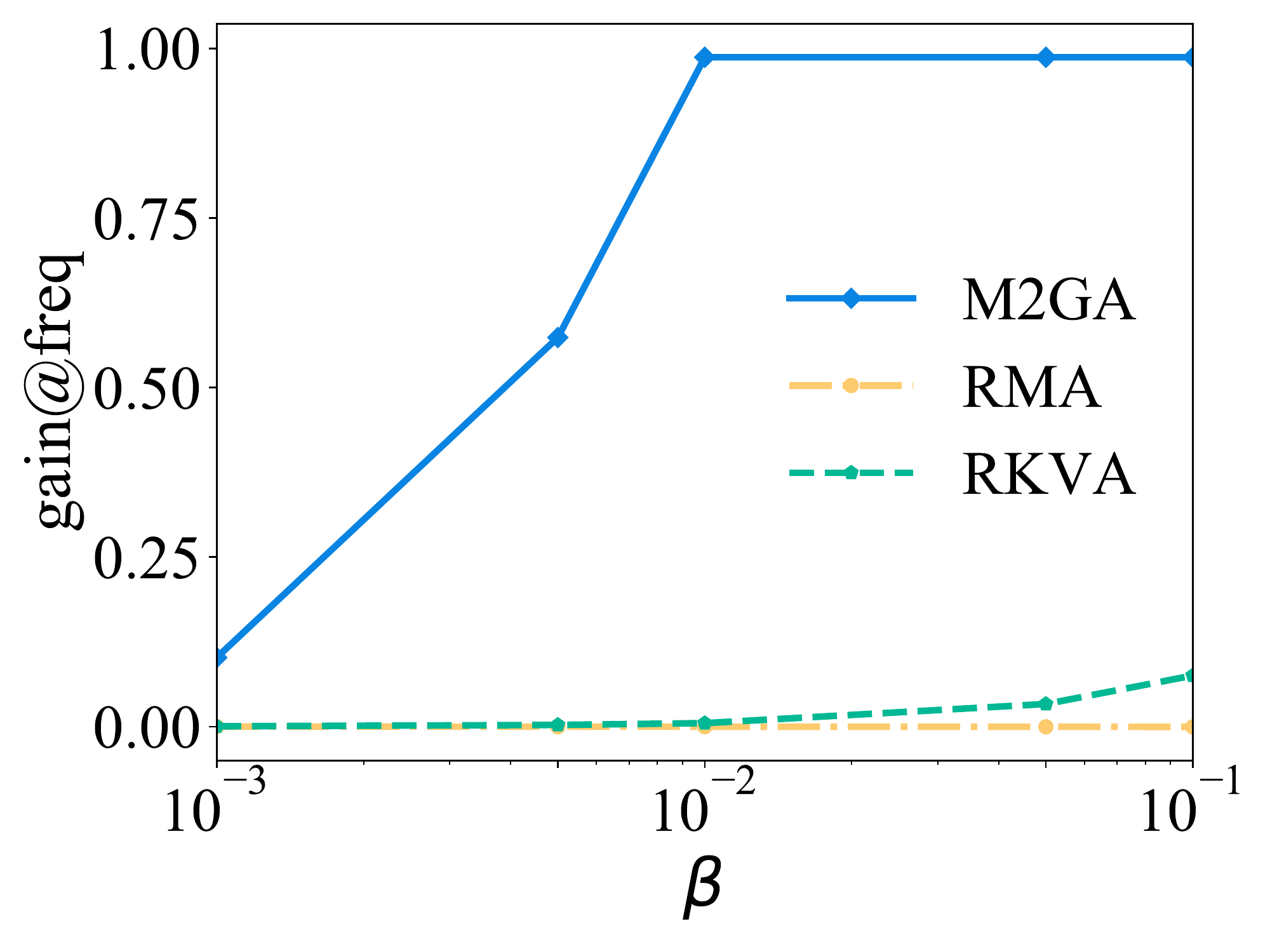}}
    \subfloat{  \includegraphics[width=0.15\textwidth]{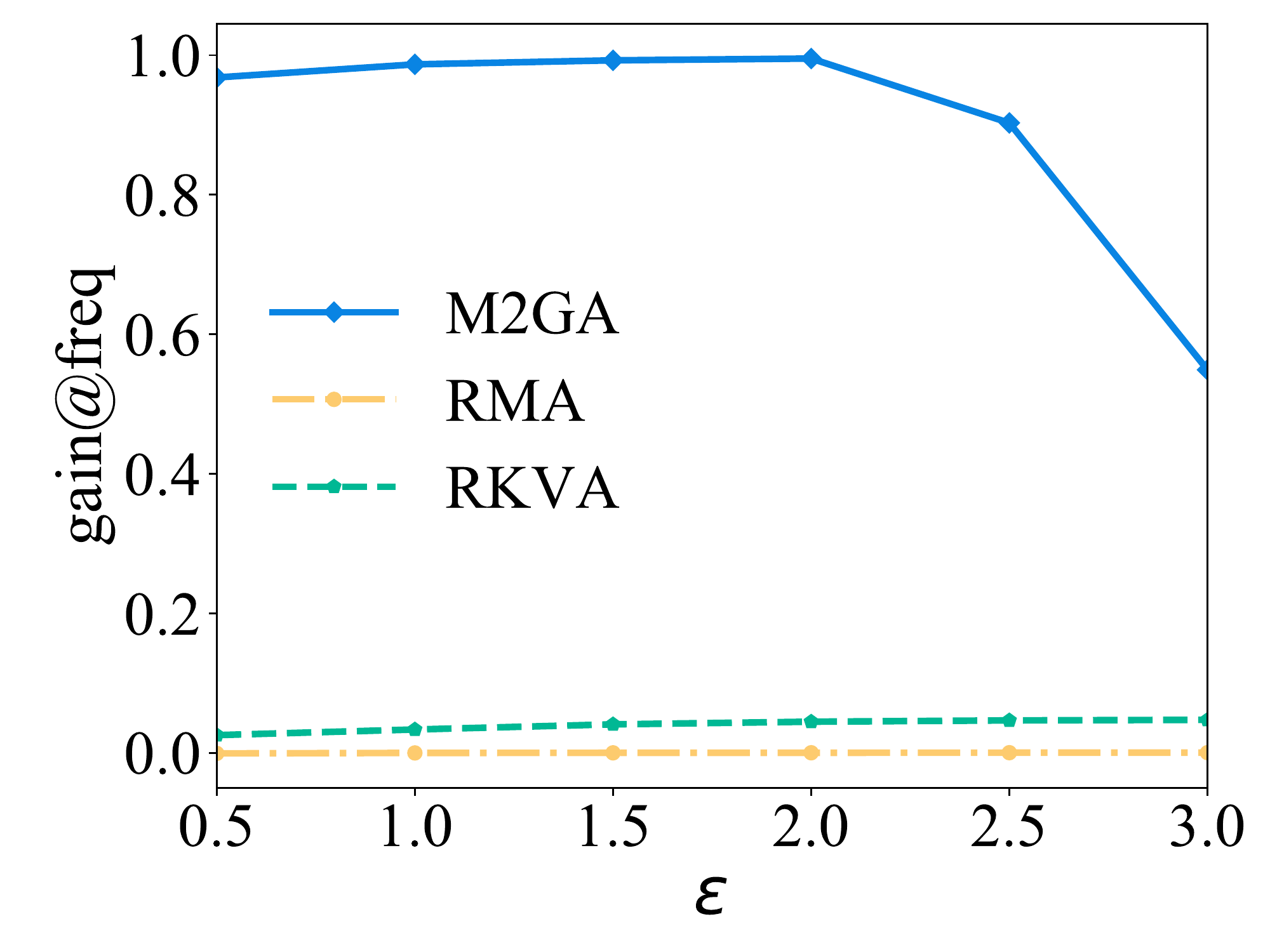}}
    \subfloat{  \includegraphics[width=0.15\textwidth]{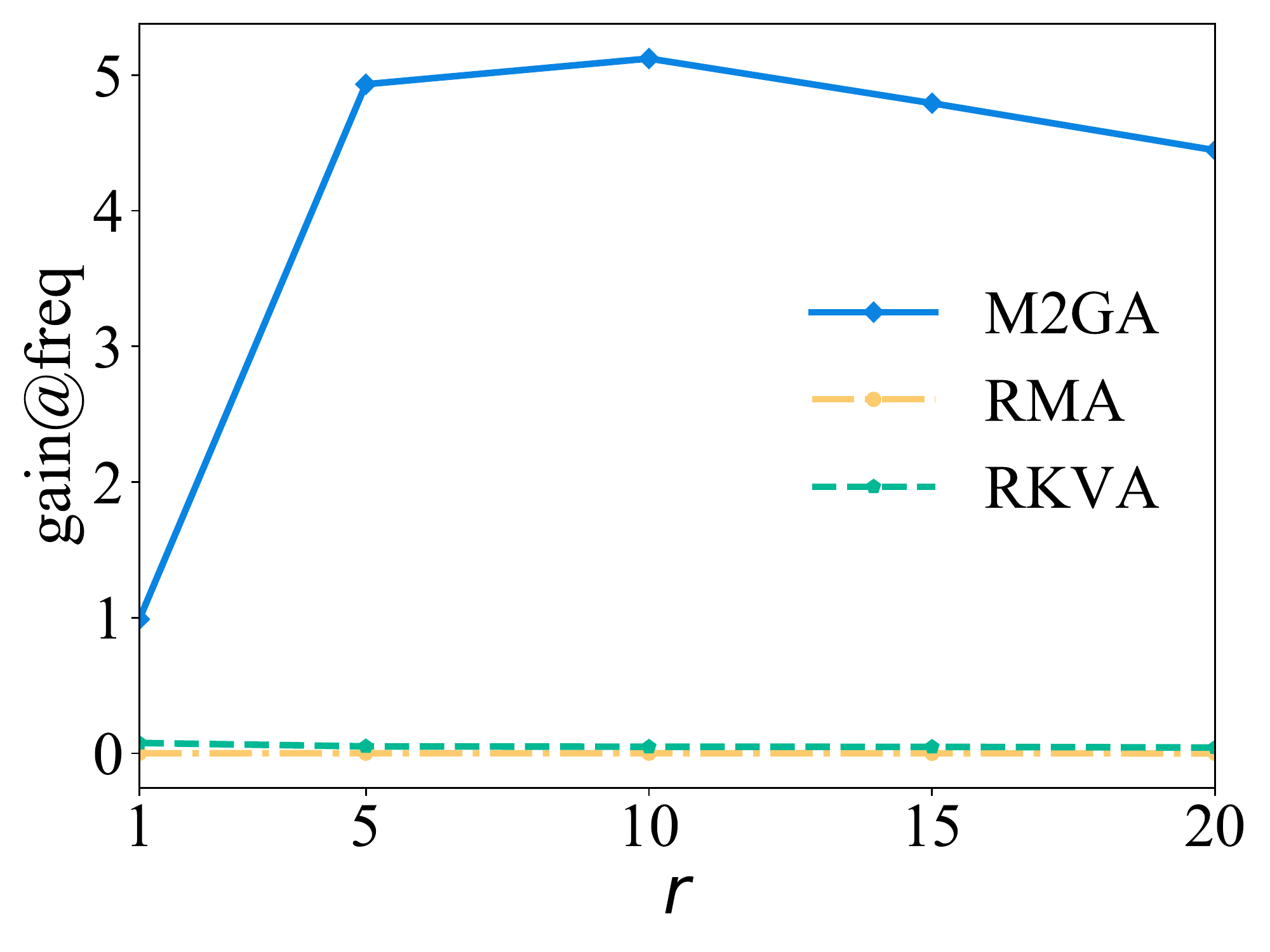}}

\caption{Impact of different parameters ($\beta, \epsilon, r$) on the frequency gains on Synthetic. The three rows are for PrivKVM, PCKV-UE, and PCKV-GRR, respectively.}
\label{fig:exp_attack_freq_synthetics}
\end{figure}

We evaluate our three attacks, i.e., M2GA, RMA, and RKVA, on a synthetic dataset and three real-world datasets. The statistics of the four datasets are shown in Table~\ref{tab:datasets}.
\begin{itemize}[leftmargin=*]
    \item \textbf{Synthetic}: Following~\cite{gu2020pckv,ye2019privkv}, we create a synthetic dataset to evaluate our attacks. In particular, we generate $10^5$ users and $100$ keys. Each user has a single KV pair. The keys and the values follow a zero-mean Gaussian distribution, where the standard deviation is 15 for keys and 1 for values.
    \item \textbf{Clothing}~\cite{Clothing}: This is a clothing fit dataset for product size recommendation. It contains users' rating scores for different products. We treat each product as a key and view each rating score as a value. Note that each user may have multiple pairs of $ \langle \text{product},\text{rating score} \rangle$.
    \item \textbf{TalkingData}~\cite{TalkingData}: This dataset contains mobile apps downloaded by users on their mobile devices. In particular, we treat each category of  mobile apps as a key and view the number of apps downloaded by a user in a category as a value.  A user may have multiple KV pairs. 
    \item \textbf{MovieLens-1M}~\cite{harper2015movielens}: This dataset contains users' rating scores for different movies. Each movie is a key and each rating score is a value.  A user may rate multiple movies. 
\end{itemize}

We scale the values in each dataset such that they fall into the range of $[-1,1]$. 

\begin{figure}[!t]
    \centering 
    
    \subfloat{  \includegraphics[width=0.15\textwidth]{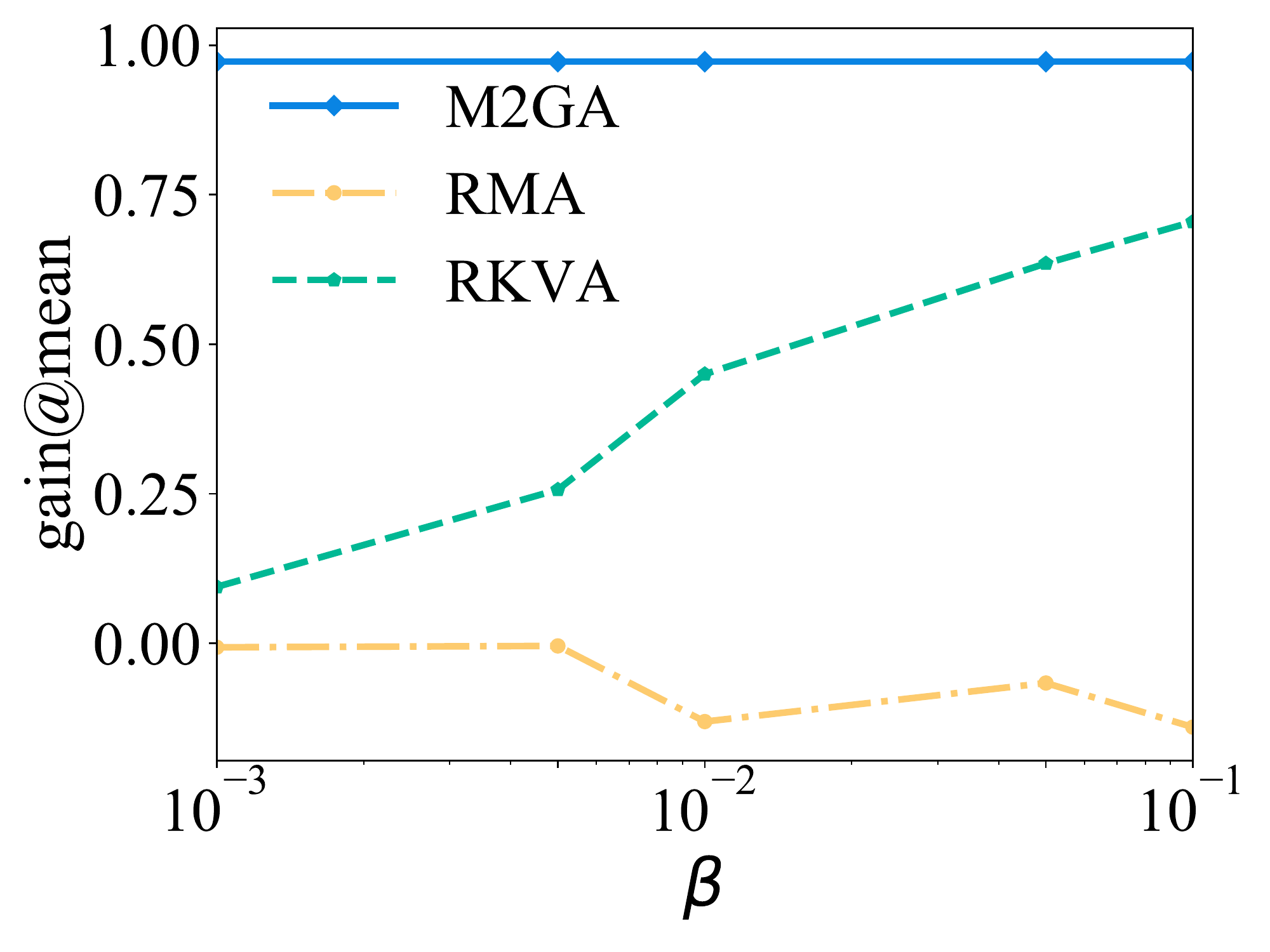}}
    \subfloat{  \includegraphics[width=0.15\textwidth]{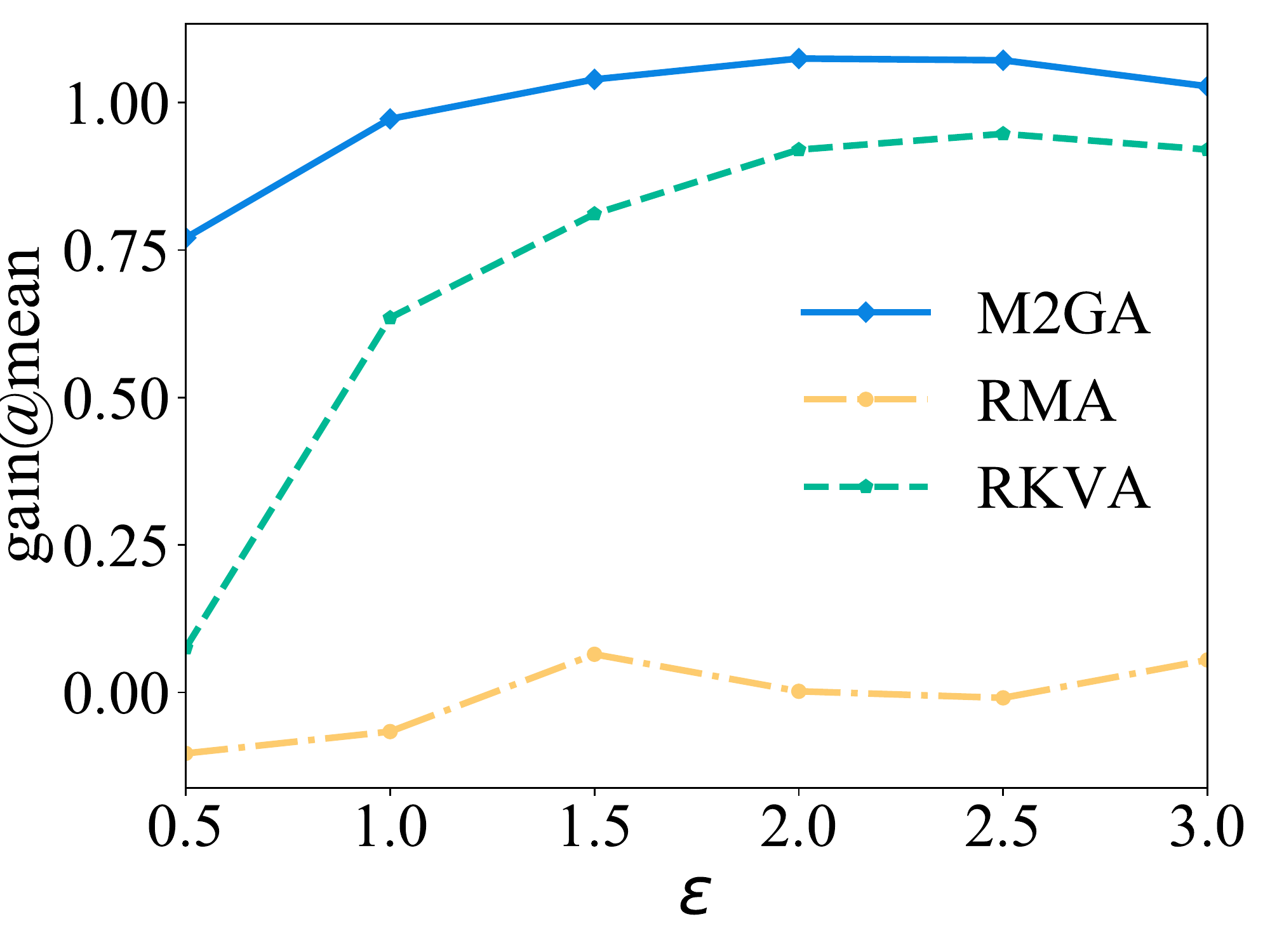}}
    \subfloat{ \includegraphics[width=0.15\textwidth]{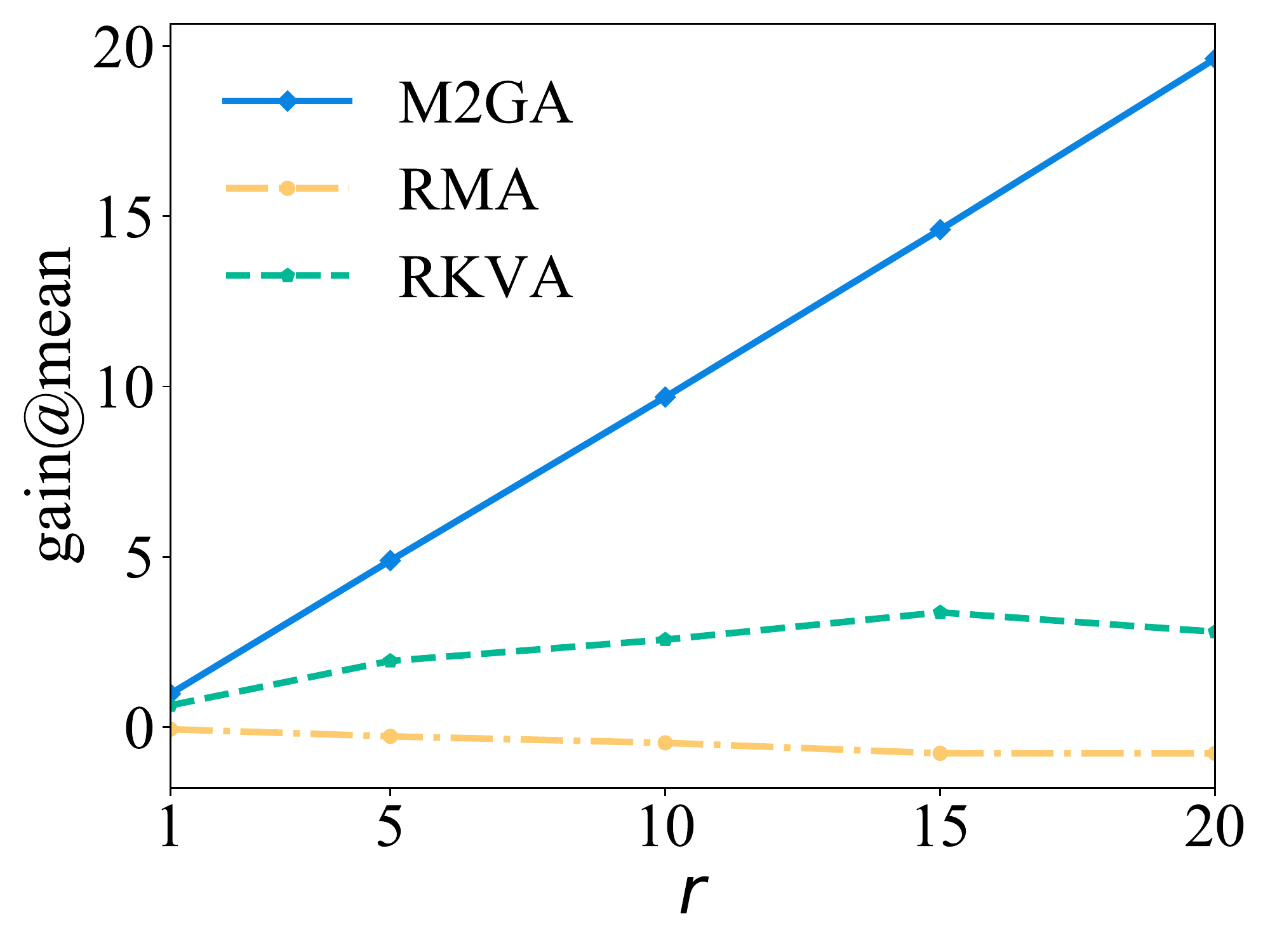}}

    \subfloat{  \includegraphics[width=0.15\textwidth]{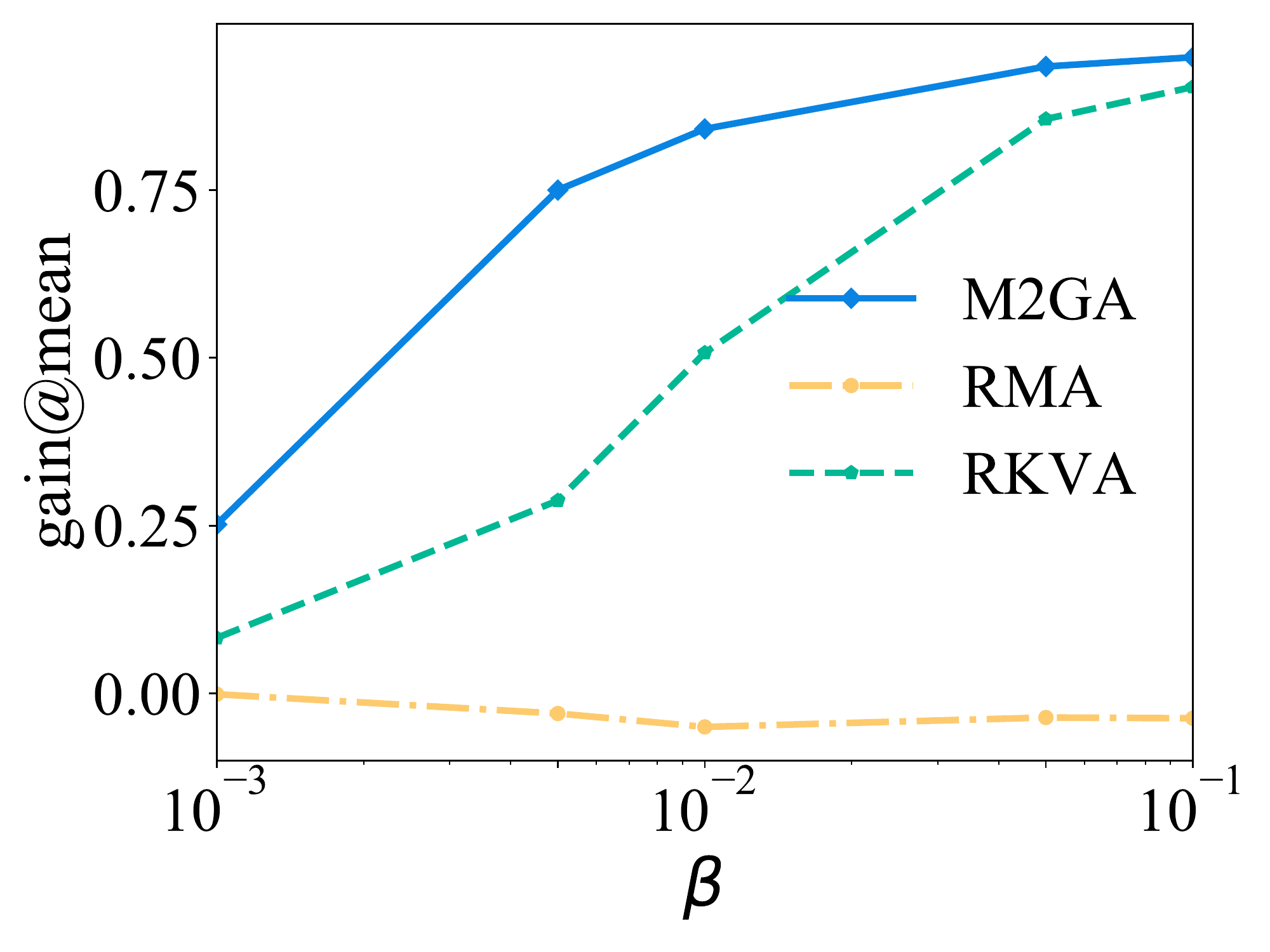}}
    \subfloat{  \includegraphics[width=0.15\textwidth]{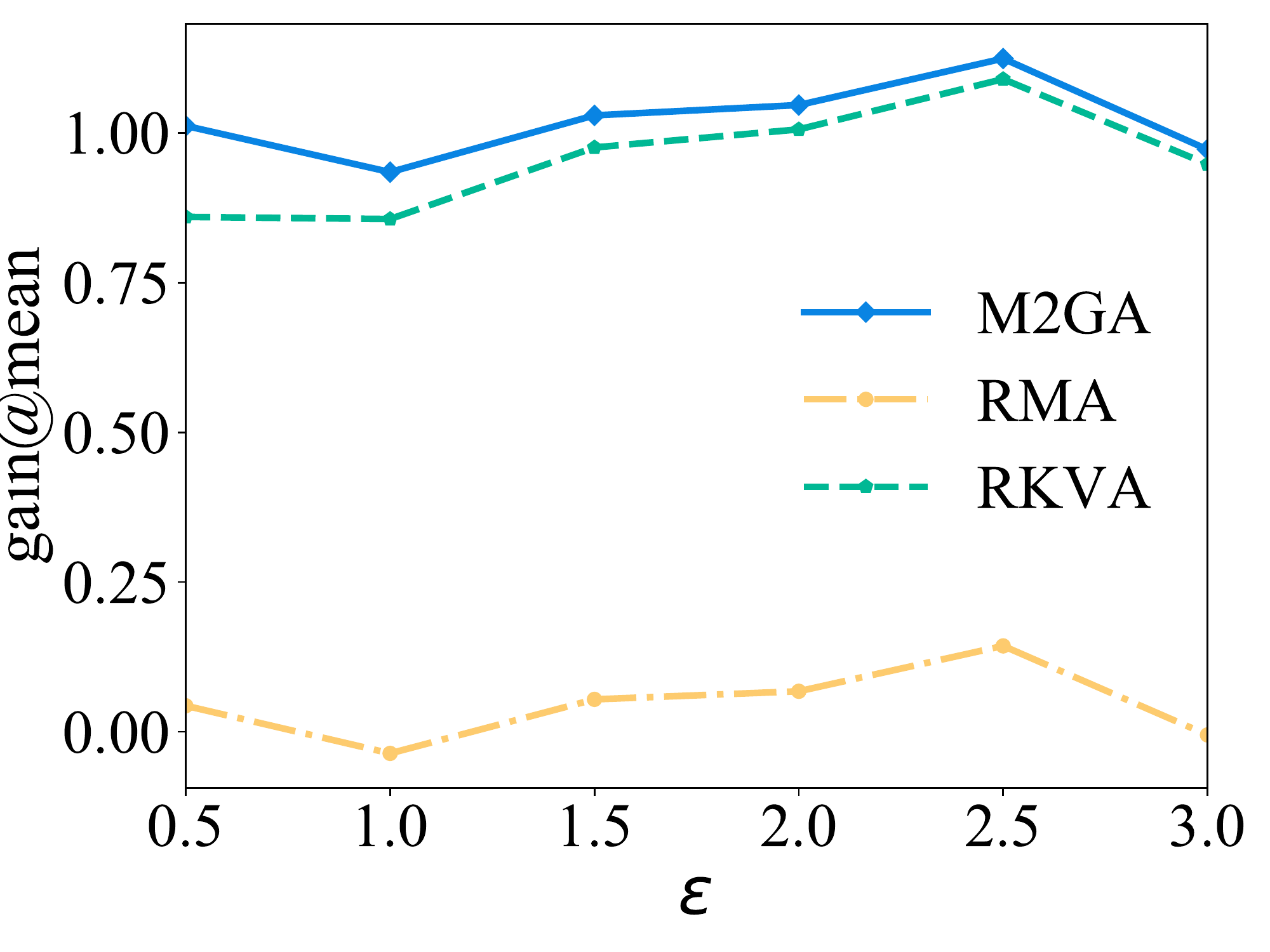}}
    \subfloat{  \includegraphics[width=0.15\textwidth]{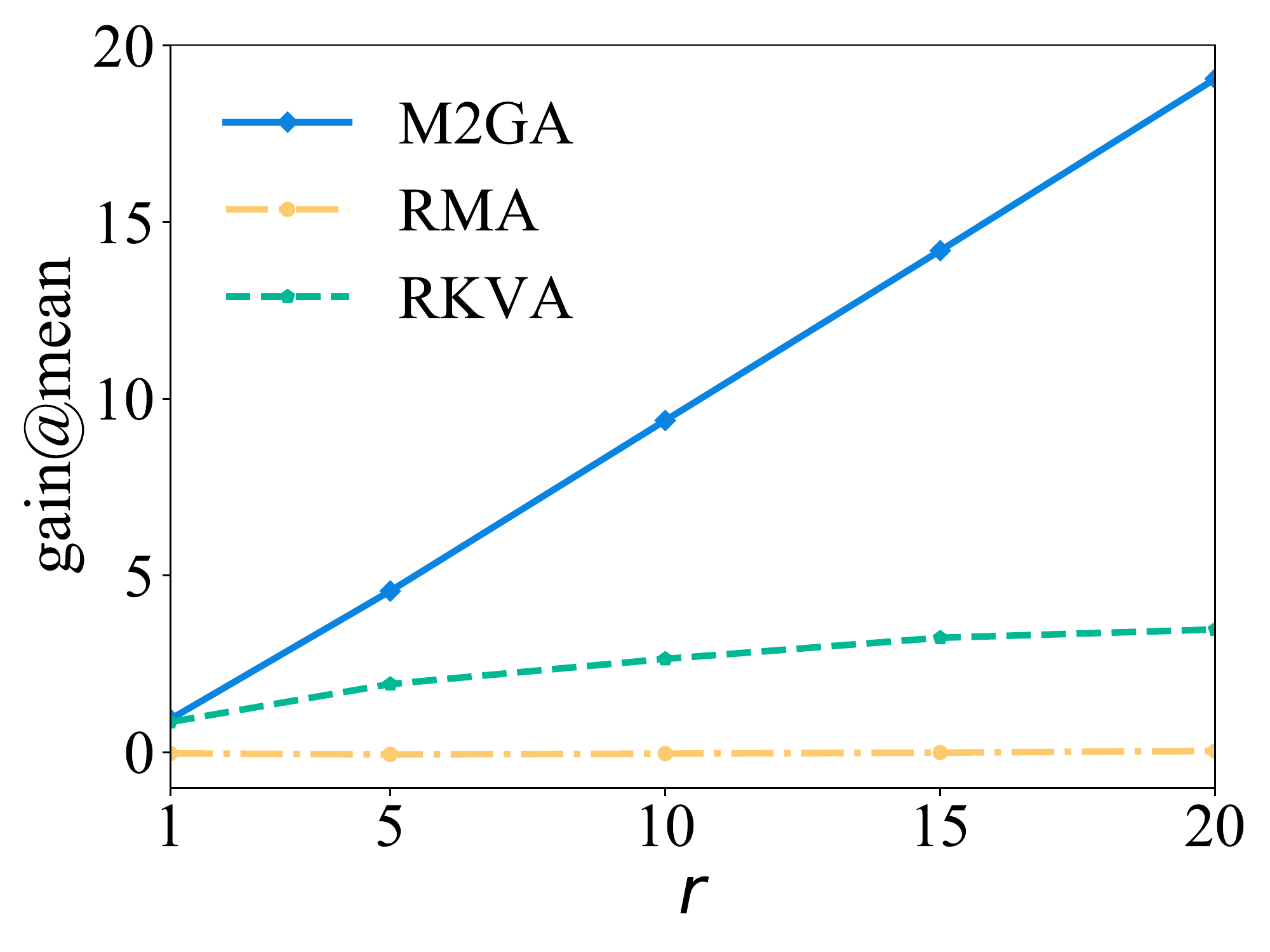}}
    
    \subfloat{  \includegraphics[width=0.15\textwidth]{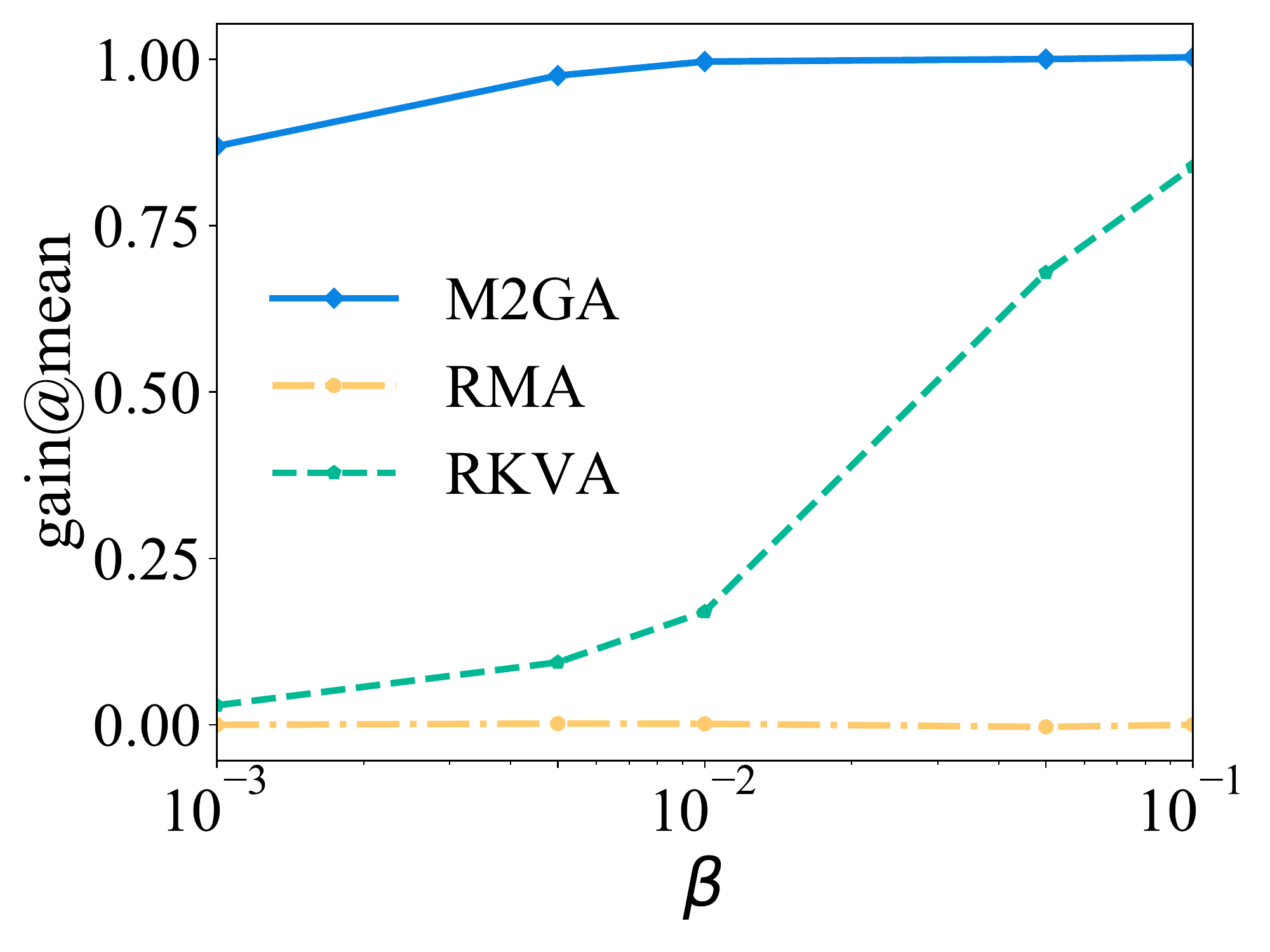}}
    \subfloat{  \includegraphics[width=0.15\textwidth]{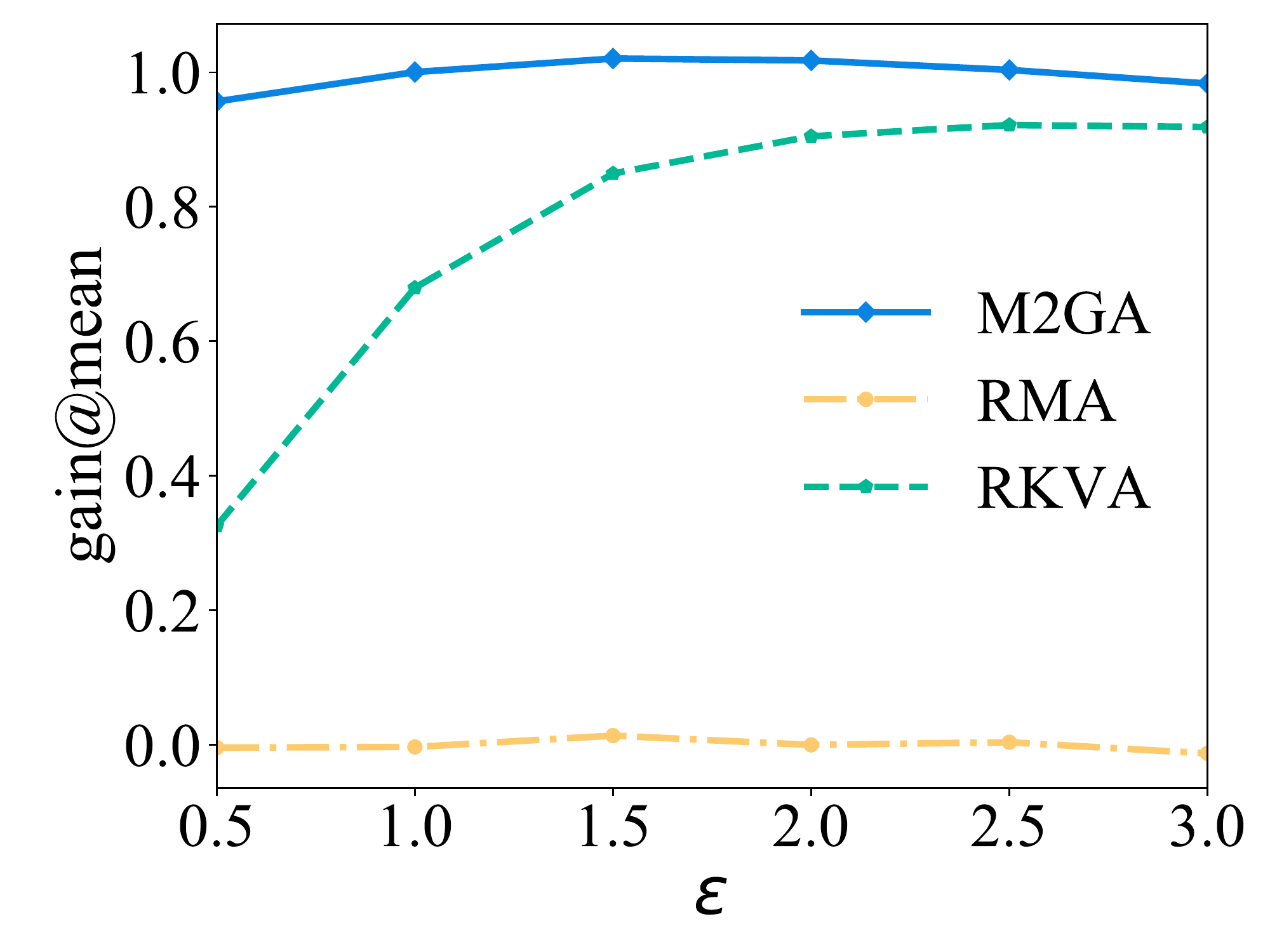}}
    \subfloat{  \includegraphics[width=0.15\textwidth]{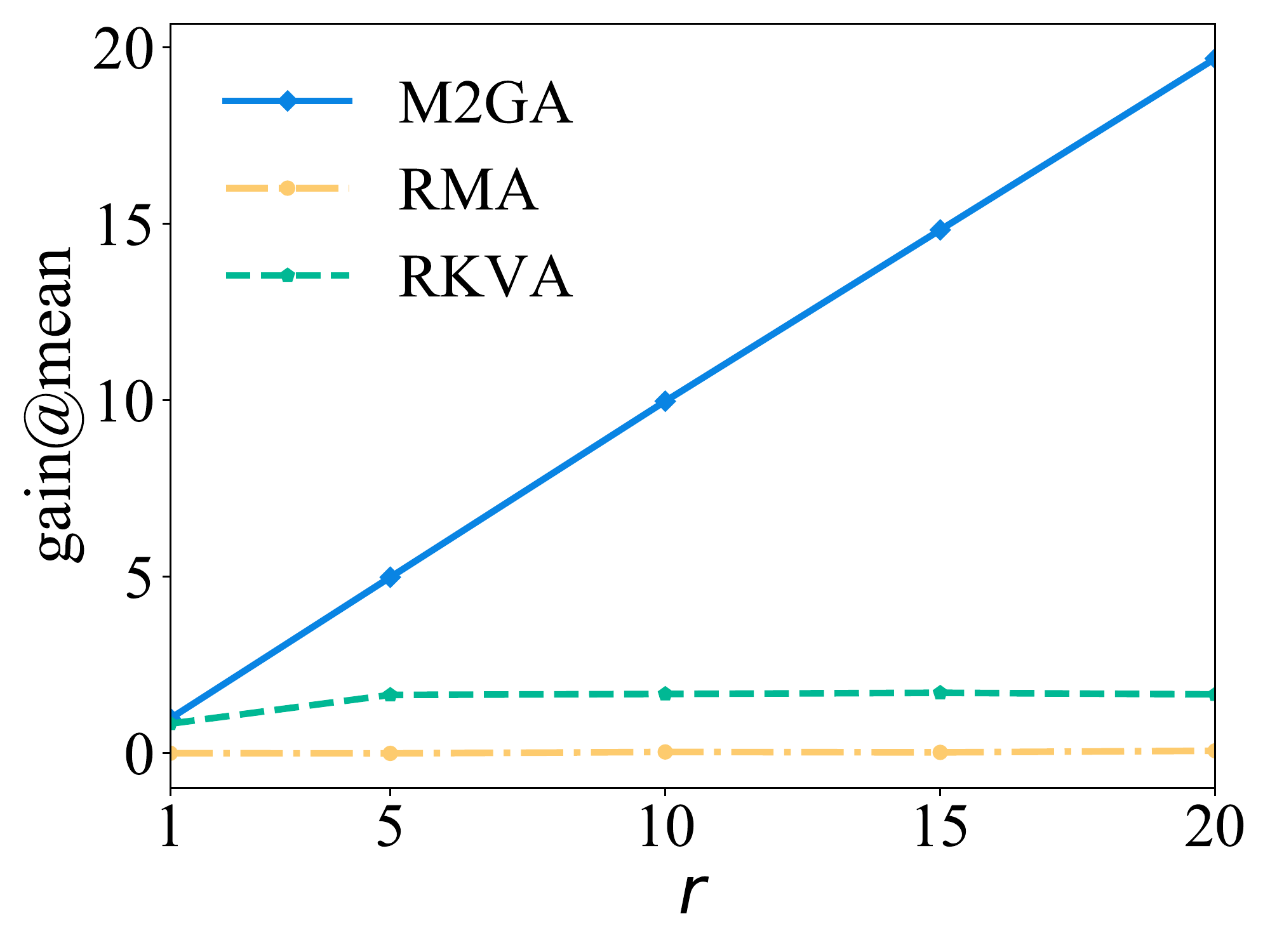}}

\caption{Impact of different parameters ($\beta, \epsilon, r$) on the mean gains on Synthetic. The three rows are for PrivKVM, PCKV-UE, and PCKV-GRR, respectively.}
\label{fig:exp_attack_mean_synthetics}
\end{figure}

\subsubsection{Evaluation Metrics}
\label{evaluation_metric_exp_sec}
\paragraph{gain@freq and gain@mean} We use \emph{frequency gain (gain@freq)} and \emph{mean gain (gain@mean)} of a set of target keys  as the evaluation metrics. In particular, given a set of target keys $\mathbb{T}$, gain@freq is computed as $\sum_{k \in \mathbb{T}} \mathbb{E}[\Delta \hat{f}_{k}]$ and gain@mean is computed as $\sum_{k \in \mathbb{T}} \mathbb{E}[\Delta \hat{m}_{k}]$, where $\Delta \hat{f}_{k}$ and $\Delta \hat{m}_{k}$ respectively measure the frequency gain and mean gain for the target key $k$. Note that frequency gain and mean gain involve expectations. In our experiments, we average the results over 100 trials to compute the expectations. Since in our experiments, we clip the estimated frequencies and support counts in the LDP protocols, the frequency gains may not be the same as those in Table~\ref{tab:expected_freq_gain}. 

\paragraph{ASR for recommender systems} We also consider recommender system as a downstream application. Specifically, the server first collects the frequency and mean value (i.e., average rating score) of each item/key from users using LDP protocols and then recommends top-$t$ items to all users based on the statistics. In this downstream application, the attacker's goal is to promote the target items/keys to be among the top-$t$ items  recommended by the system. Therefore, we use \emph{attack success rate (ASR)} as our metric, which we define as the fraction of target items that are in the $t$ recommended items after attack. Note that the target items are not among the $t$ recommended ones before attack.

We consider three different cases of recommender systems, i.e., frequency-based recommender system (Case 1), score-based recommender system (Case 2), and frequency-score-based recommender system (Case 3). In Case 1, the recommender system recommends the most popular $t$ items, i.e., the $t$ items with the largest estimated frequencies. Ties are broken by selecting the item with higher estimated average rating score. In Case 2, the recommender system recommends $t$ items with the highest estimated average rating scores. Ties are broken by selecting the item with larger estimated frequency. In Case 3, the recommender system considers both the popularity and the average rating score of an item. Specifically, the recommender system calculates the product of the estimated frequency and (uncalibrated) average rating score of each item, and recommends the $t$ items with the largest products. Ties are broken randomly. Roughly speaking, the product of the estimated frequency and average rating score of an item is the item's estimated total rating scores.

\begin{figure}[!t]
    \centering 
    
    \subfloat{  \includegraphics[width=0.15\textwidth]{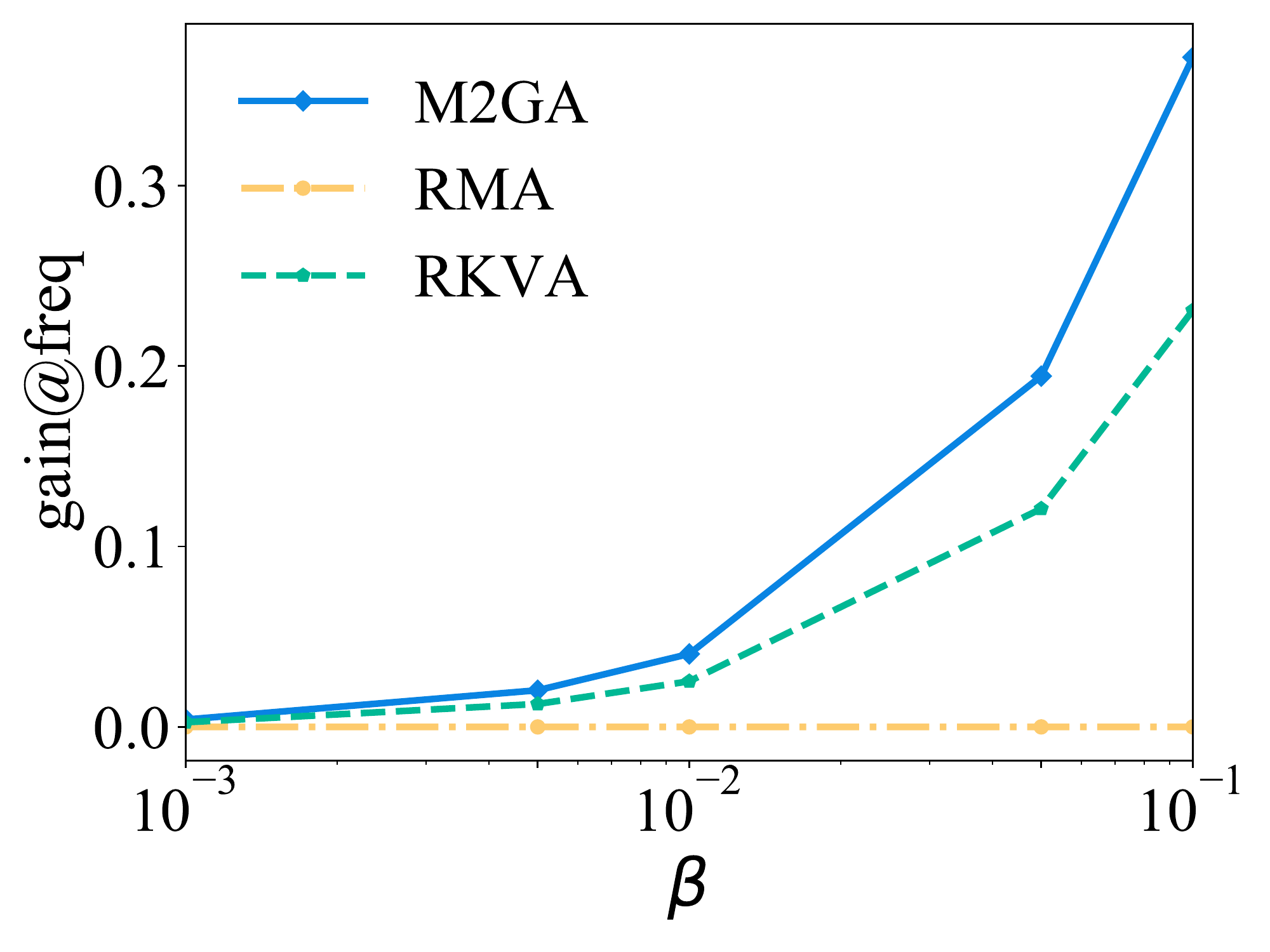}}
   \subfloat{  \includegraphics[width=0.15\textwidth]{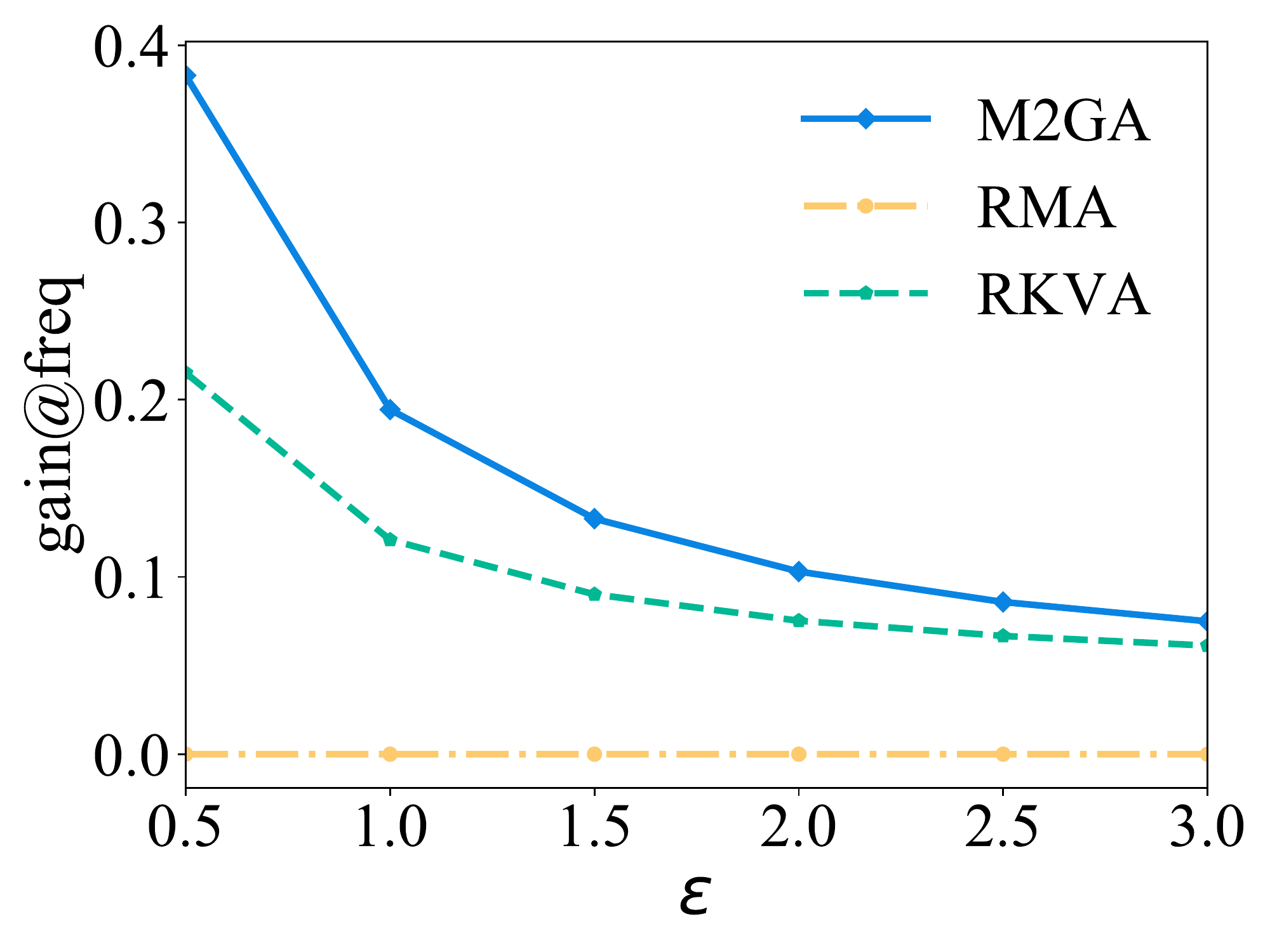}}
    \subfloat{  \includegraphics[width=0.15\textwidth]{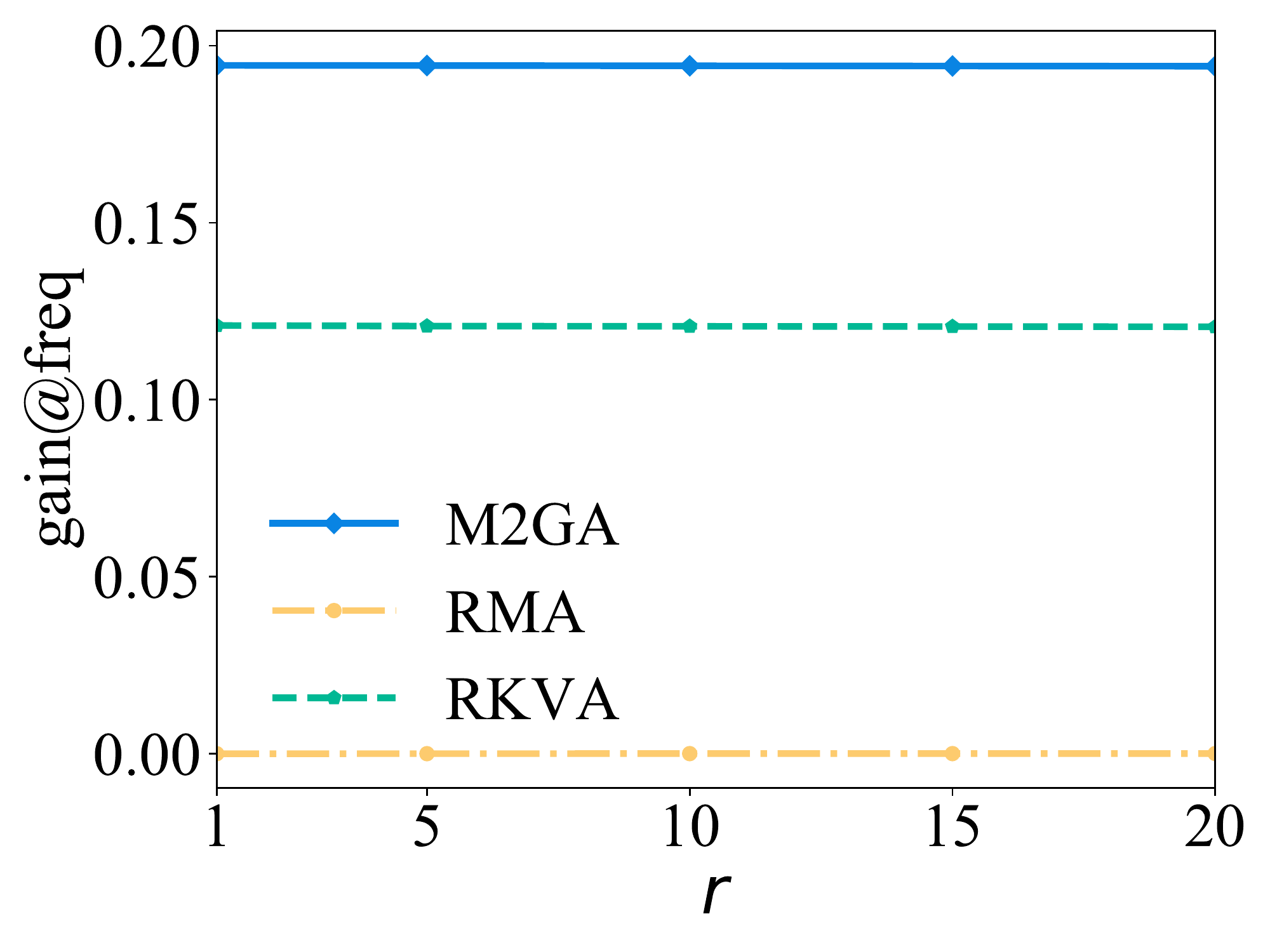}}
    
    \subfloat{  \includegraphics[width=0.15\textwidth]{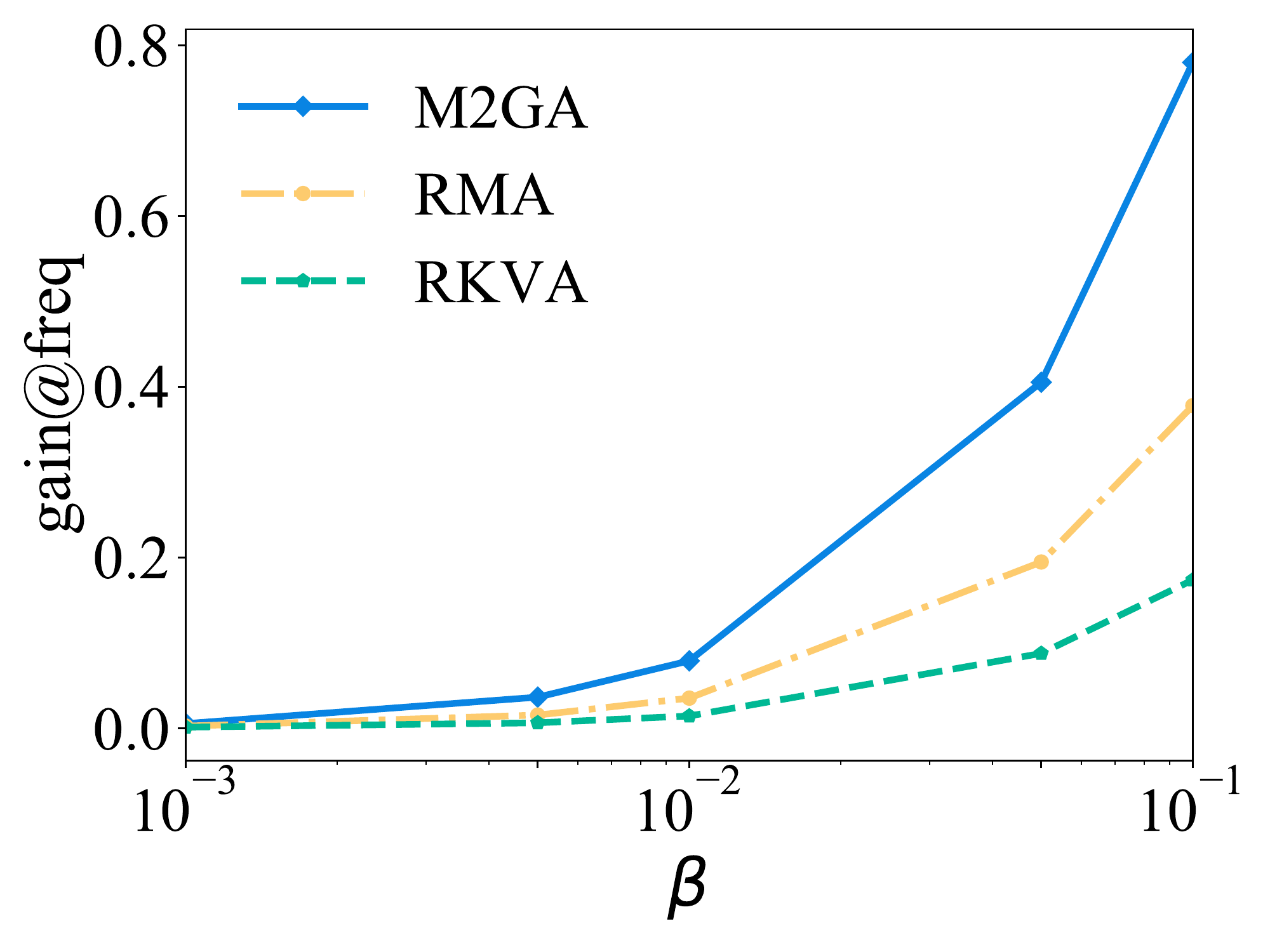}}
    \subfloat{  \includegraphics[width=0.15\textwidth]{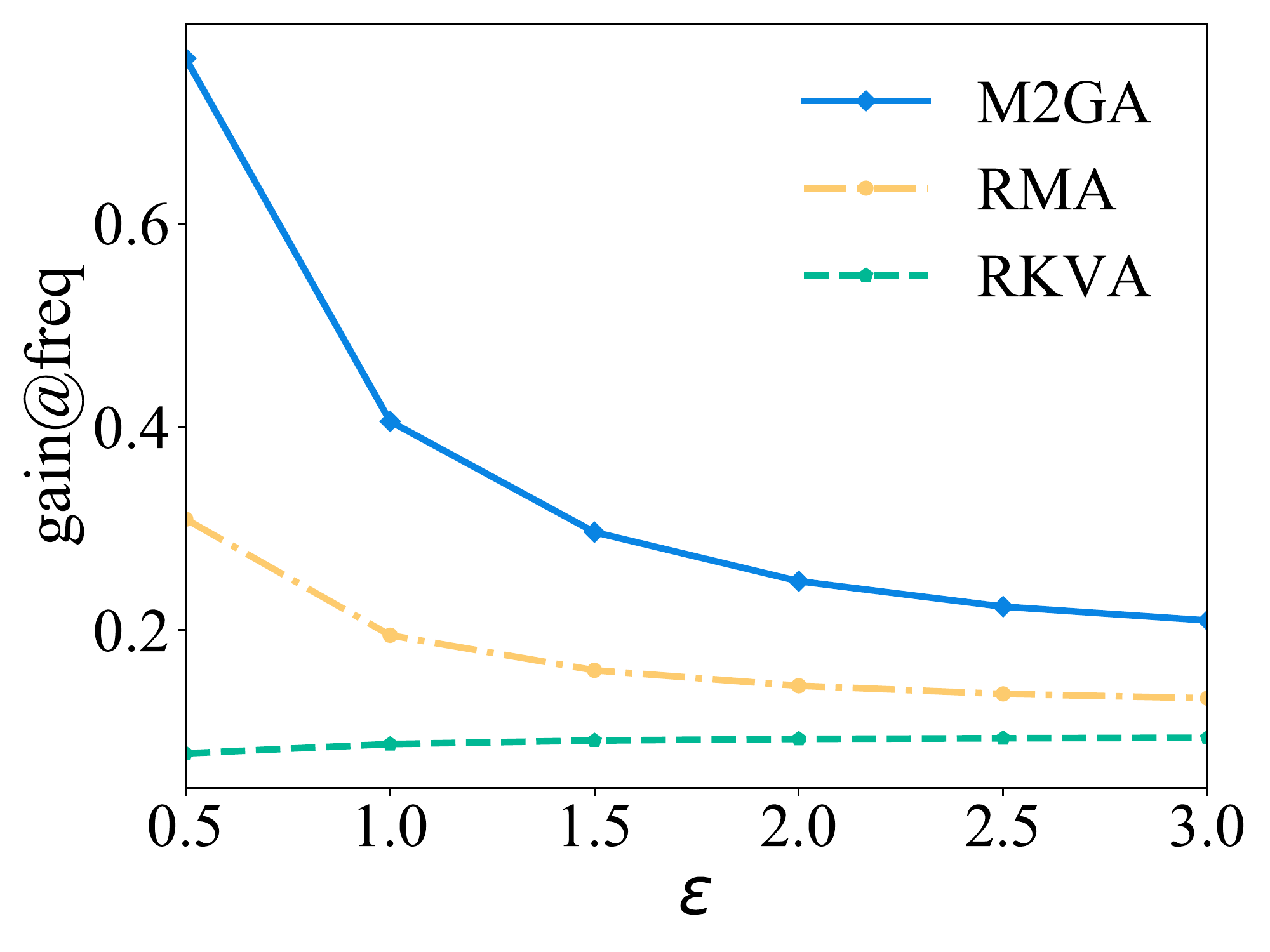}}
    \subfloat{  \includegraphics[width=0.15\textwidth]{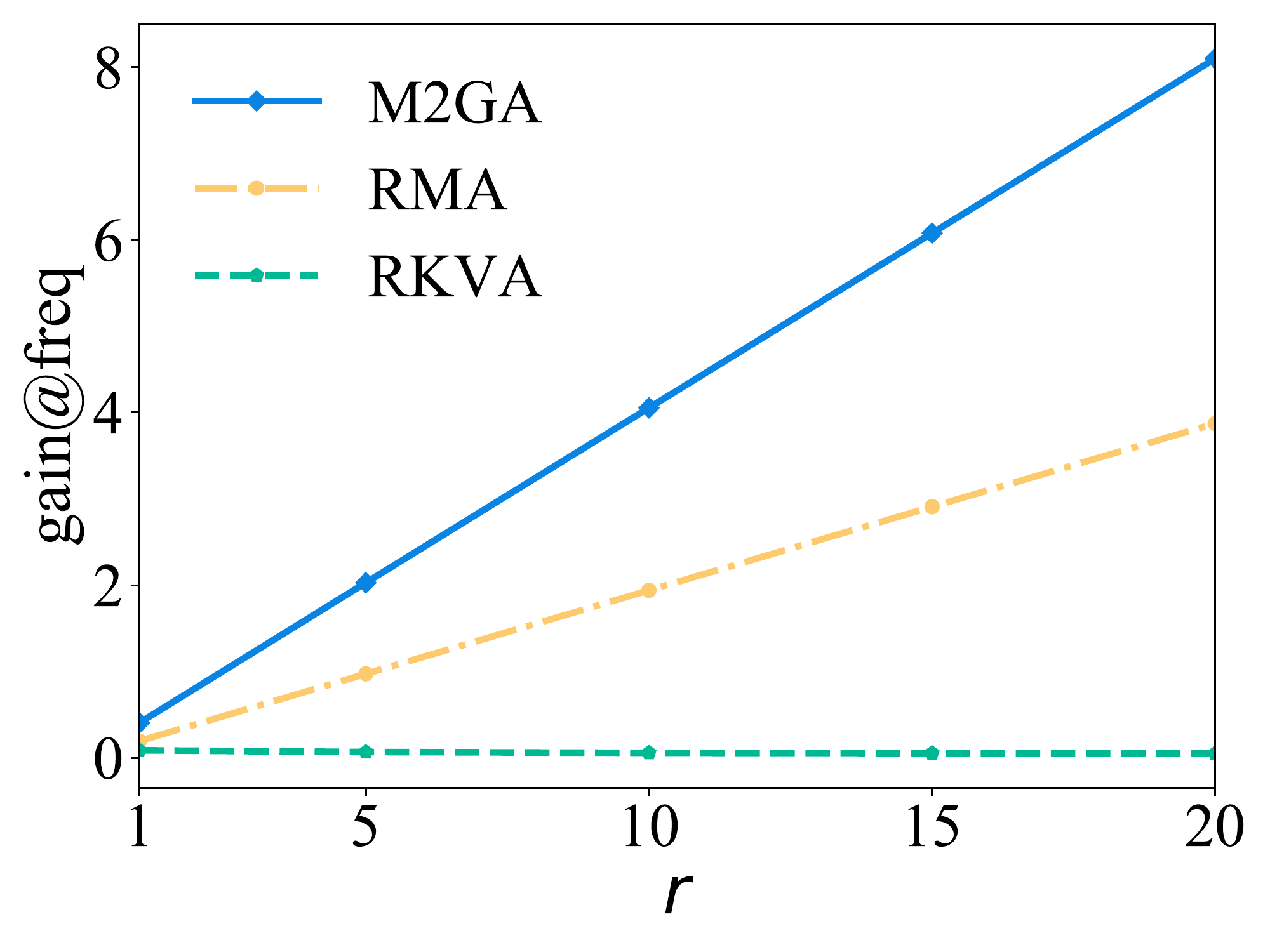}}

    \subfloat{  \includegraphics[width=0.15\textwidth]{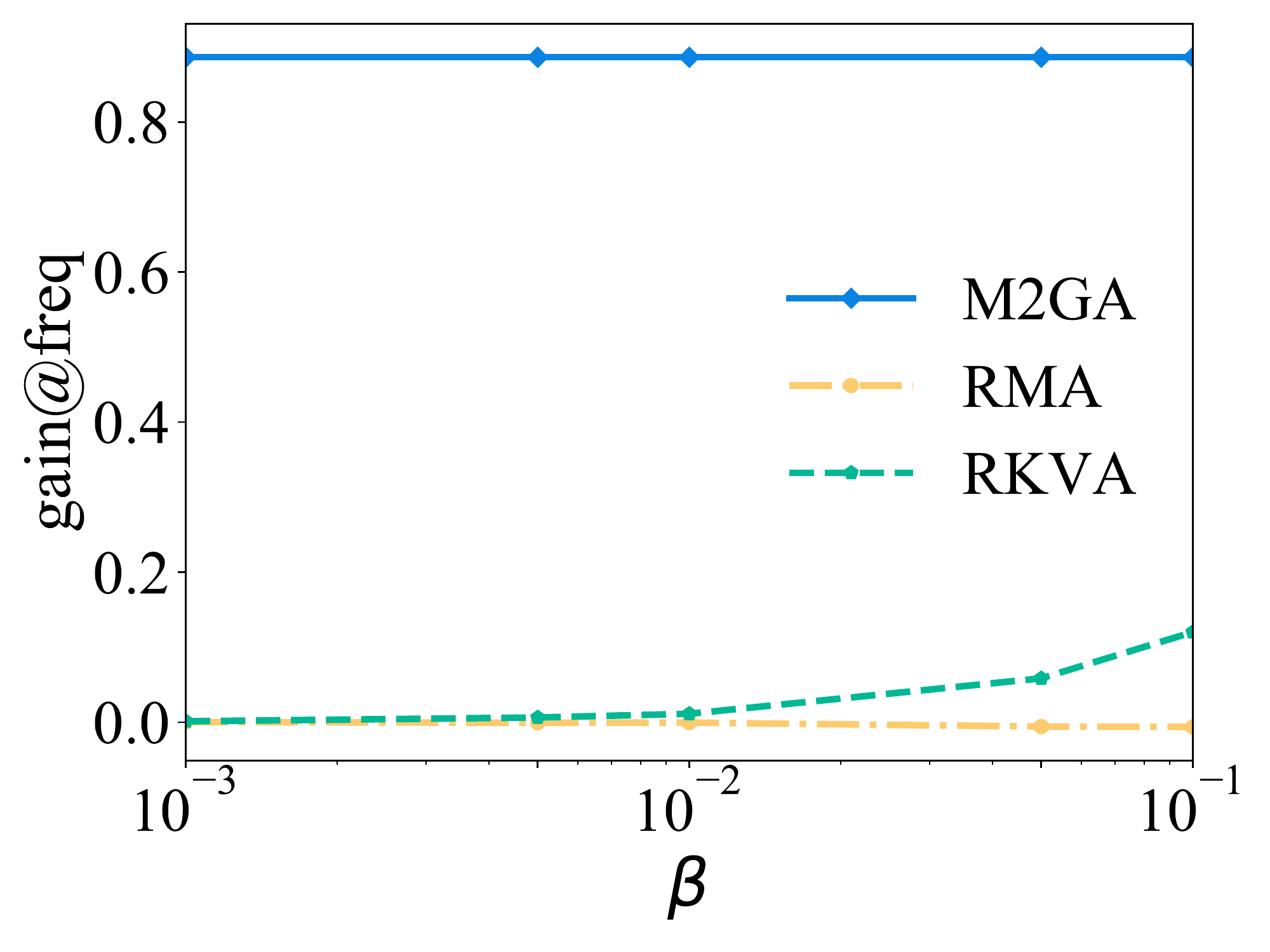}}
    \subfloat{  \includegraphics[width=0.15\textwidth]{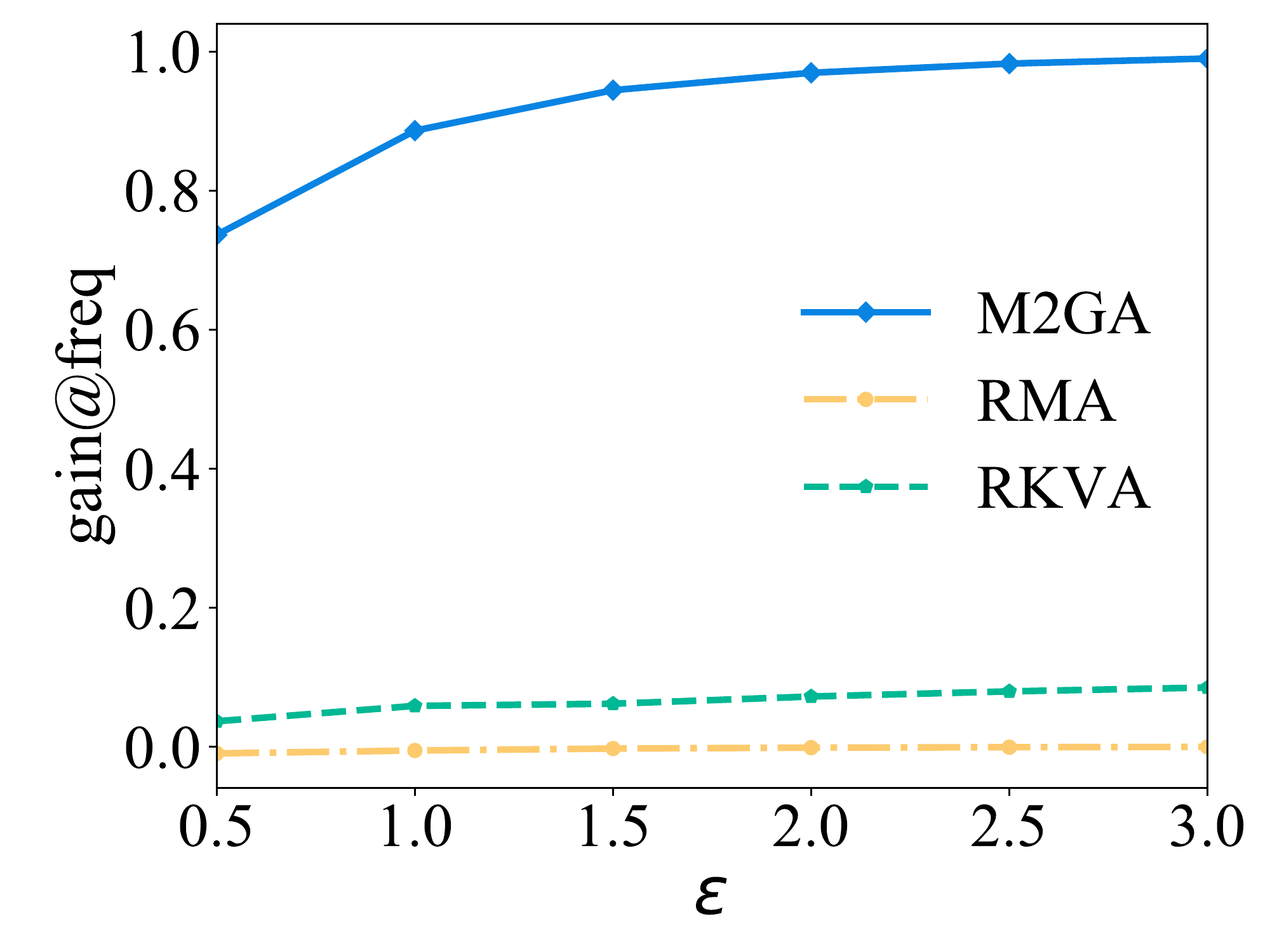}}
    \subfloat{  \includegraphics[width=0.15\textwidth]{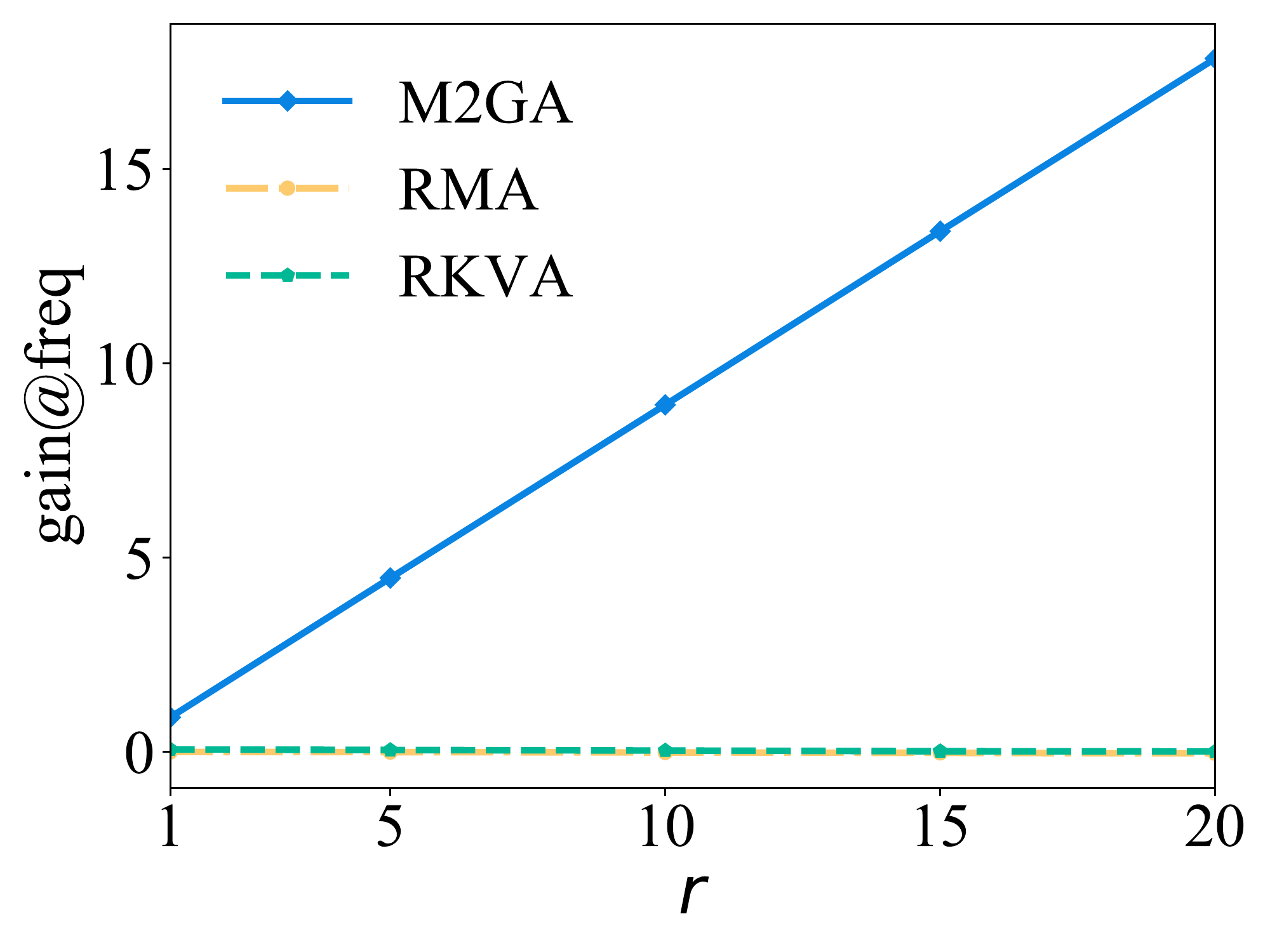}}

\caption{Impact of different parameters ($\beta, \epsilon, r$) on the frequency gains on Clothing. The three rows are for PrivKVM, PCKV-UE, and PCKV-GRR, respectively.}
\label{fig:exp_attack_freq_clothing}
\end{figure}

\begin{figure}[!t]
    \centering 
    
    \subfloat{   \includegraphics[width=0.15\textwidth]{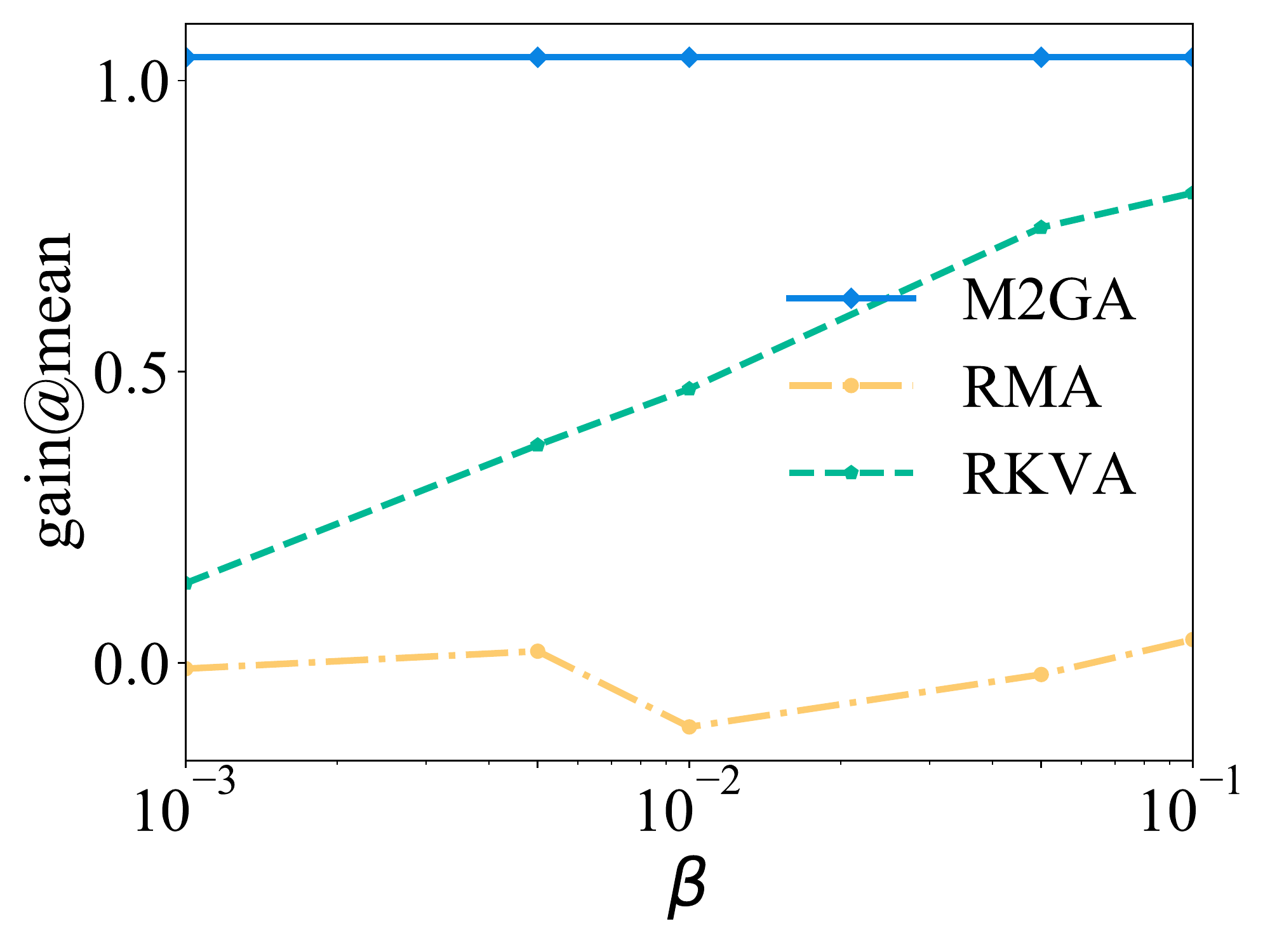}}
    \subfloat{   \includegraphics[width=0.15\textwidth]{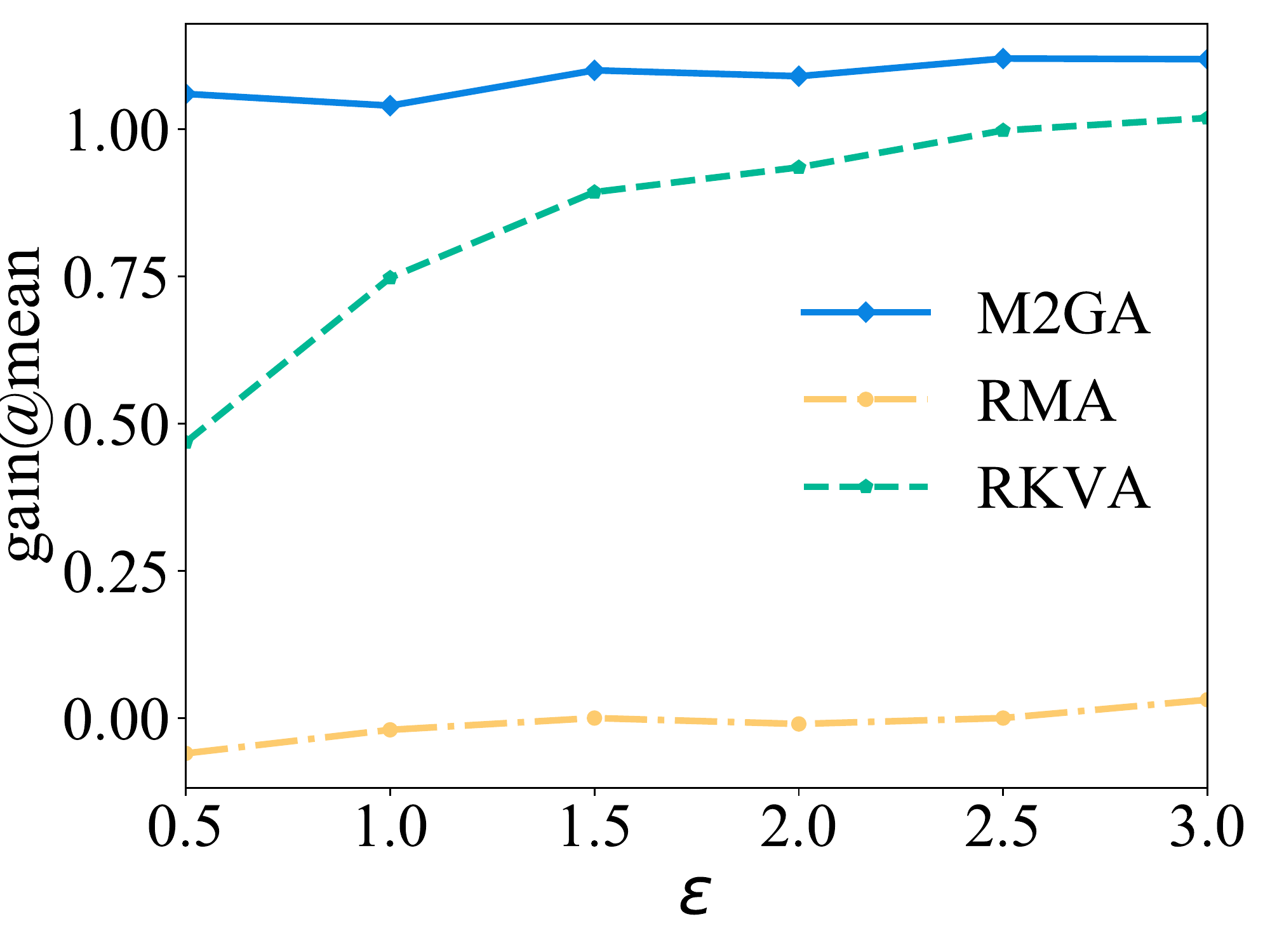}}
    \subfloat{   \includegraphics[width=0.15\textwidth]{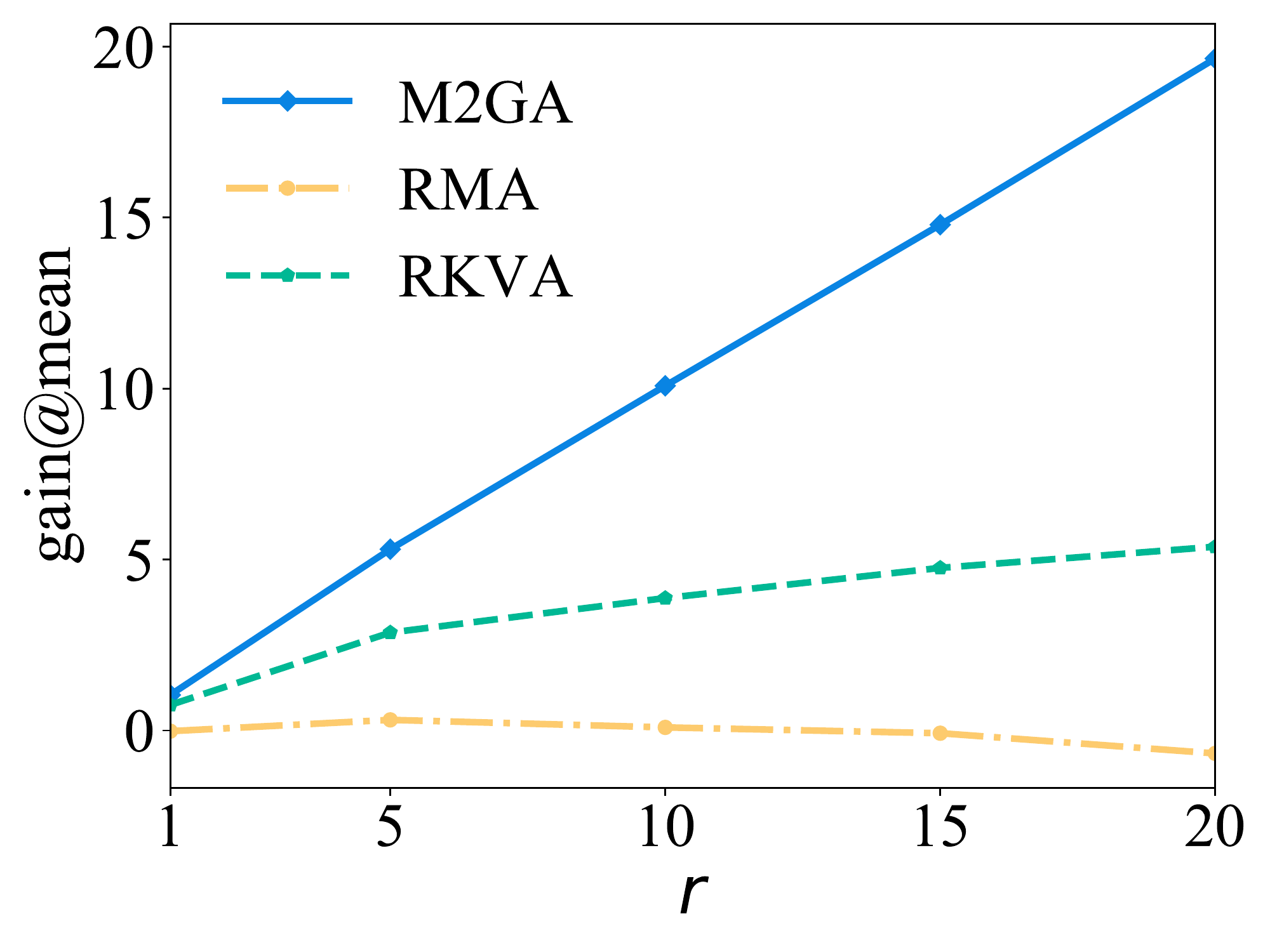}}
    
    \subfloat{   \includegraphics[width=0.15\textwidth]{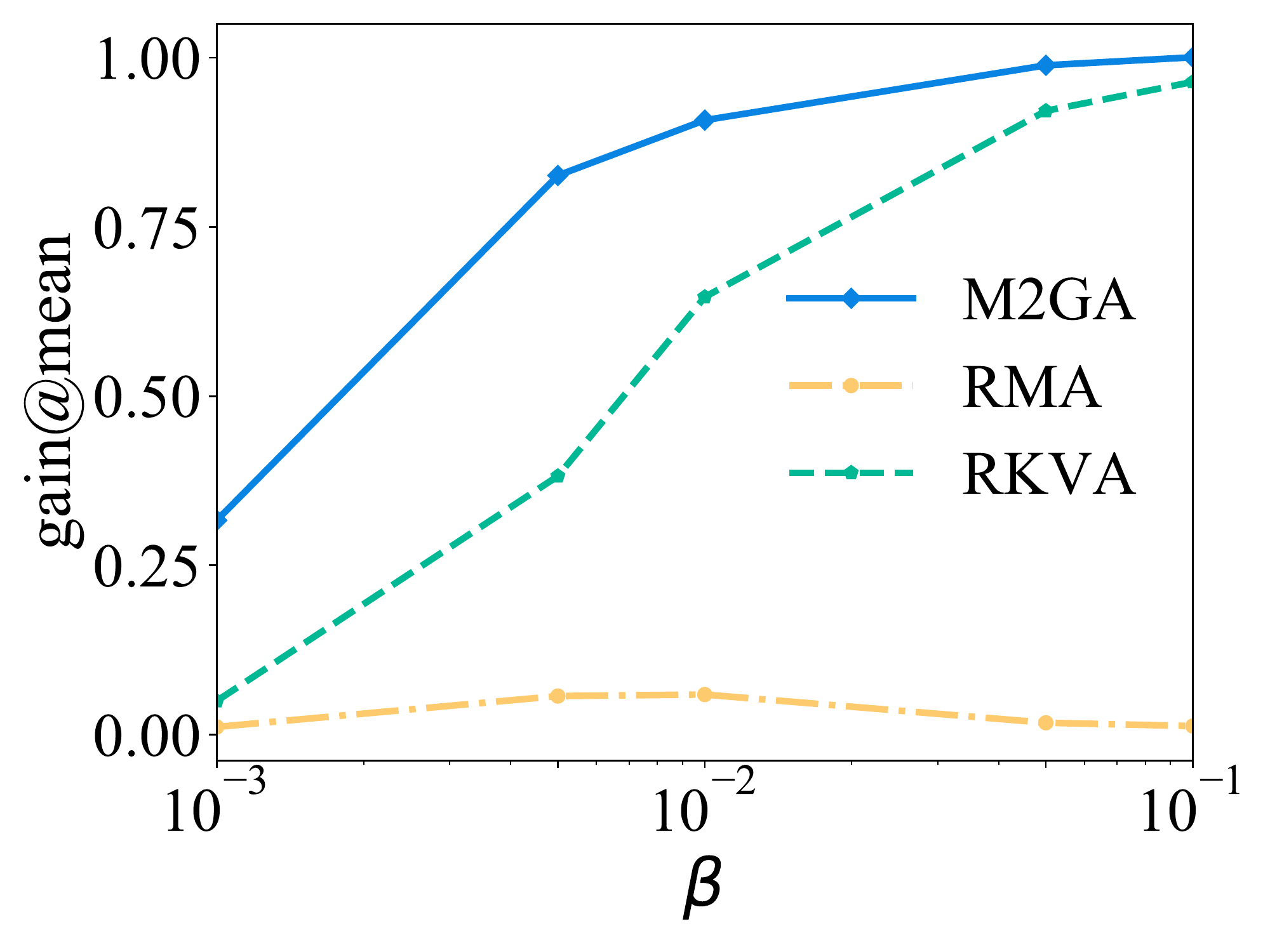}}
    \subfloat{   \includegraphics[width=0.15\textwidth]{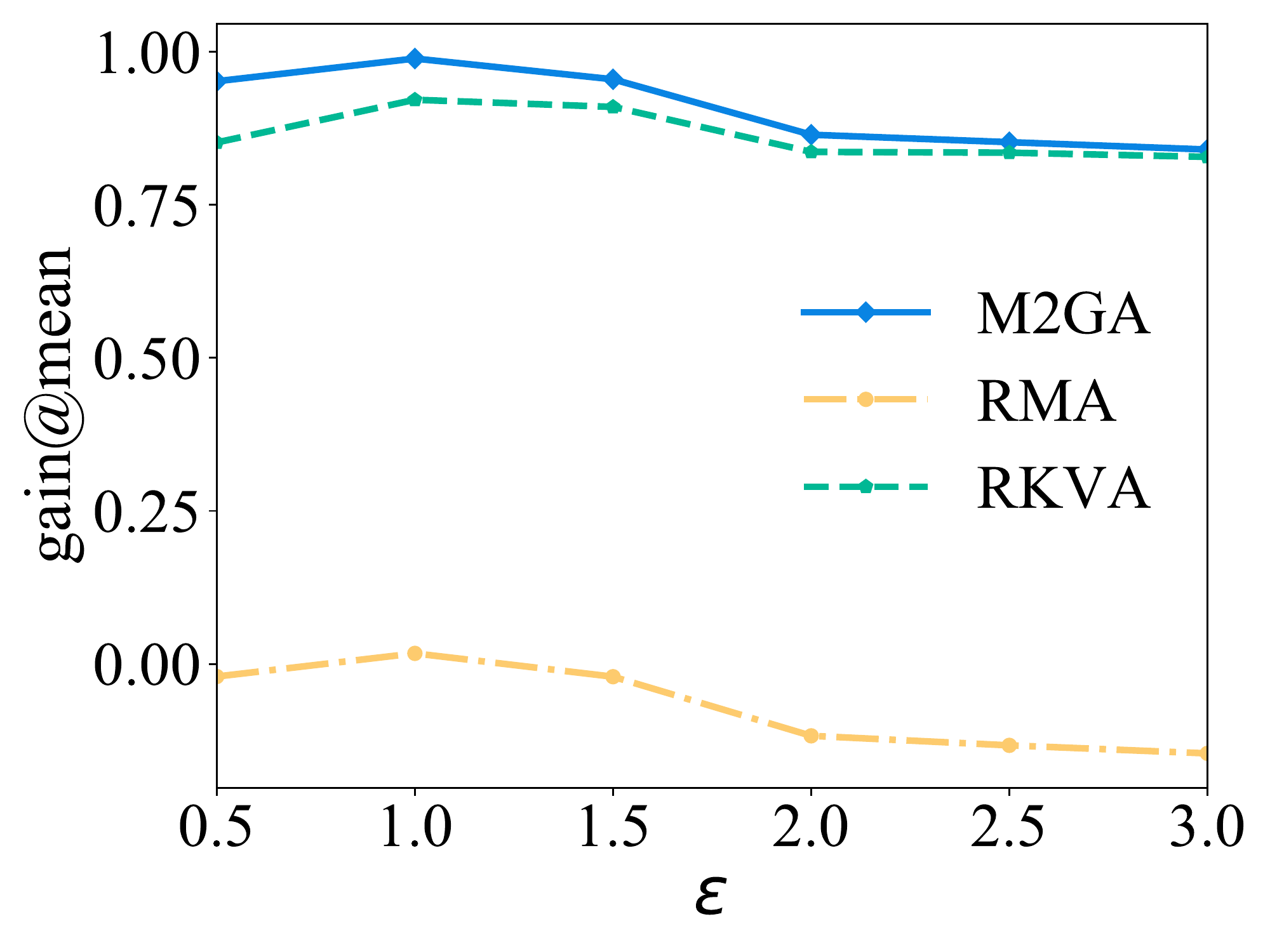}}
    \subfloat{   \includegraphics[width=0.15\textwidth]{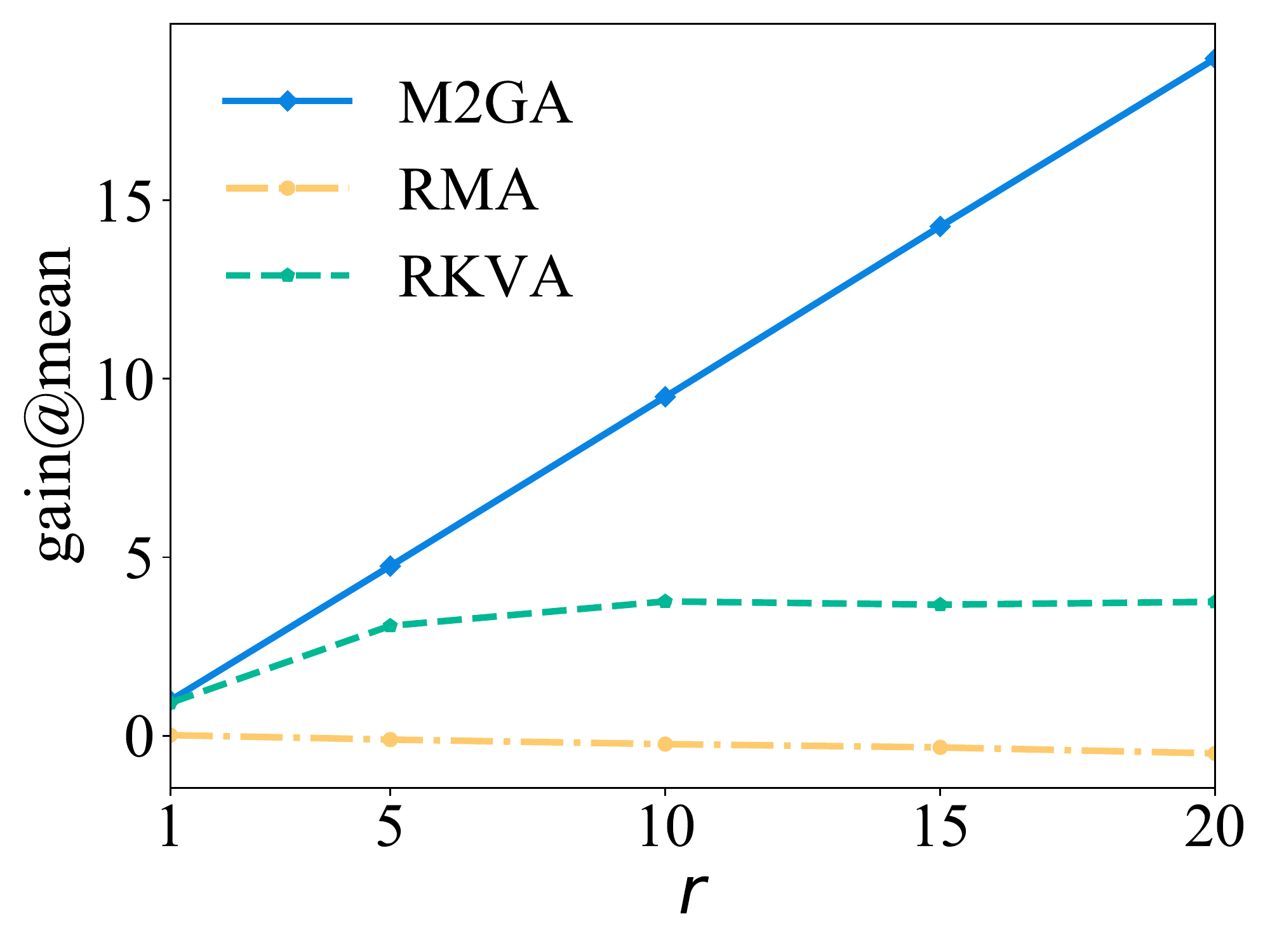}}

    \subfloat{   \includegraphics[width=0.15\textwidth]{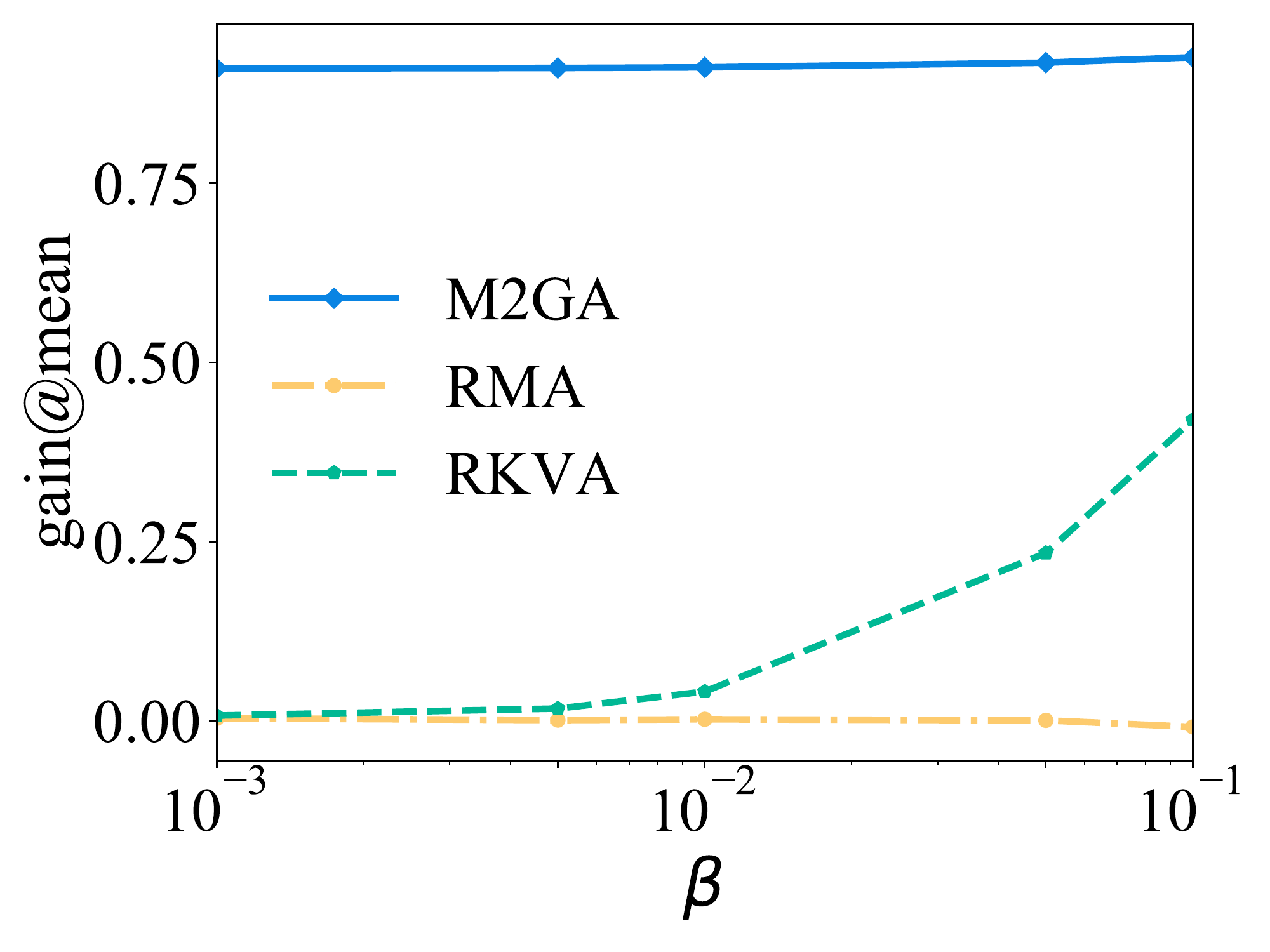}}
    \subfloat{   \includegraphics[width=0.15\textwidth]{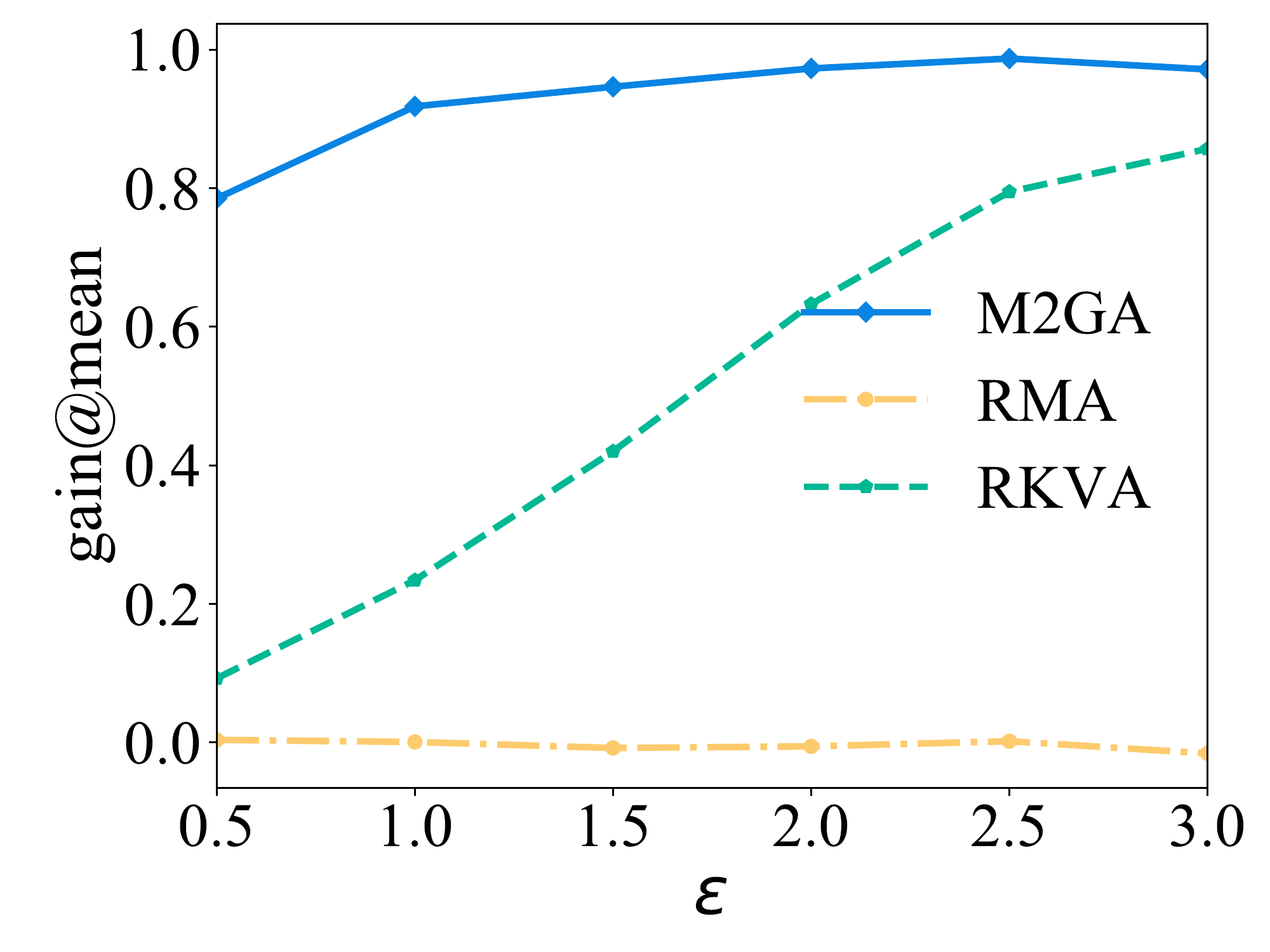}}
    \subfloat{   \includegraphics[width=0.15\textwidth]{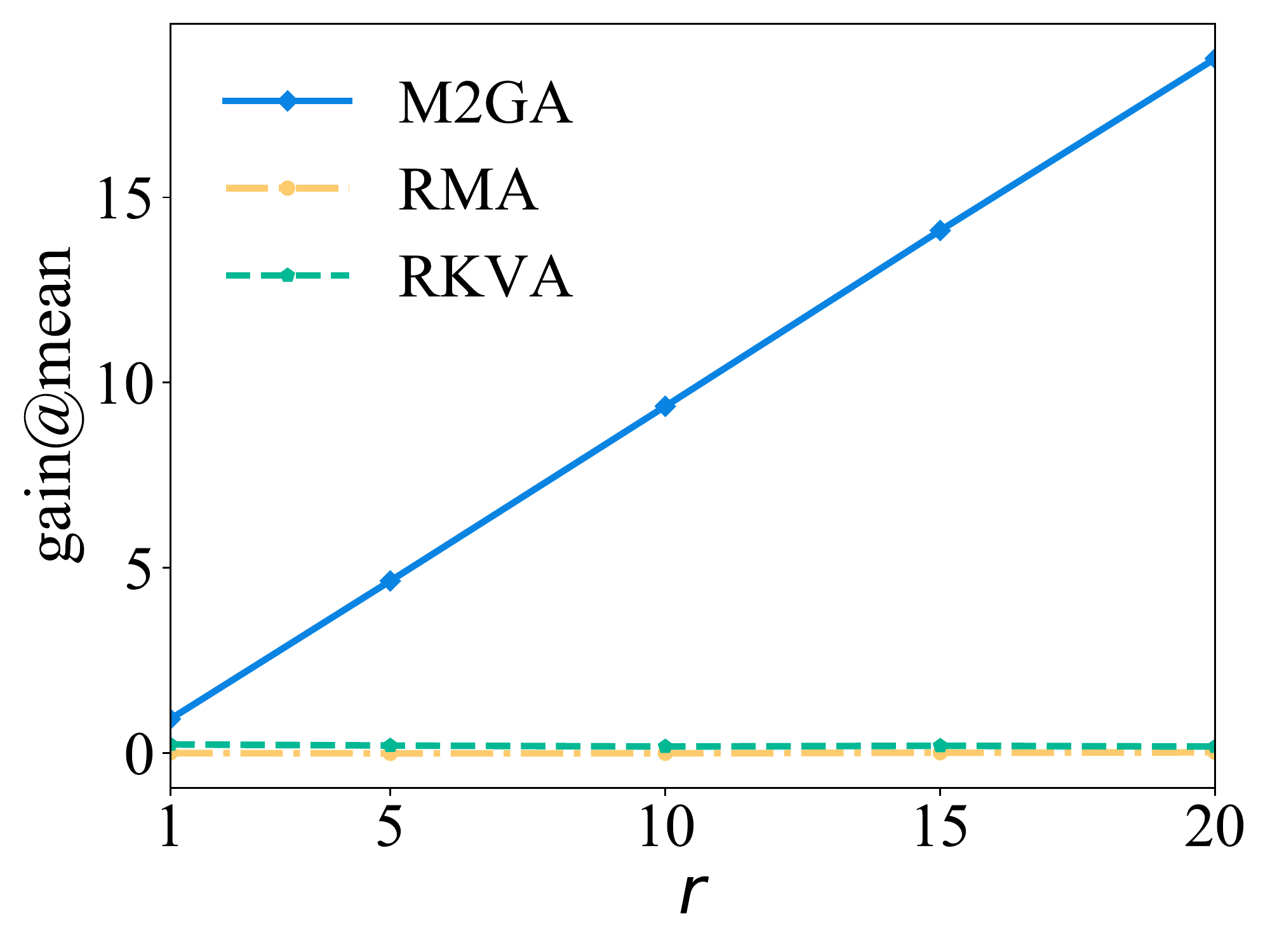}}
    
\caption{Impact of different parameters ($\beta, \epsilon, r$) on the  mean gains on Clothing. The three rows are for PrivKVM, PCKV-UE, and PCKV-GRR, respectively.}
\label{fig:exp_attack_mean_clothing}
\end{figure}

\begin{figure}[!tb]
    \centering 
    
    \subfloat{\includegraphics[width=0.15\textwidth]{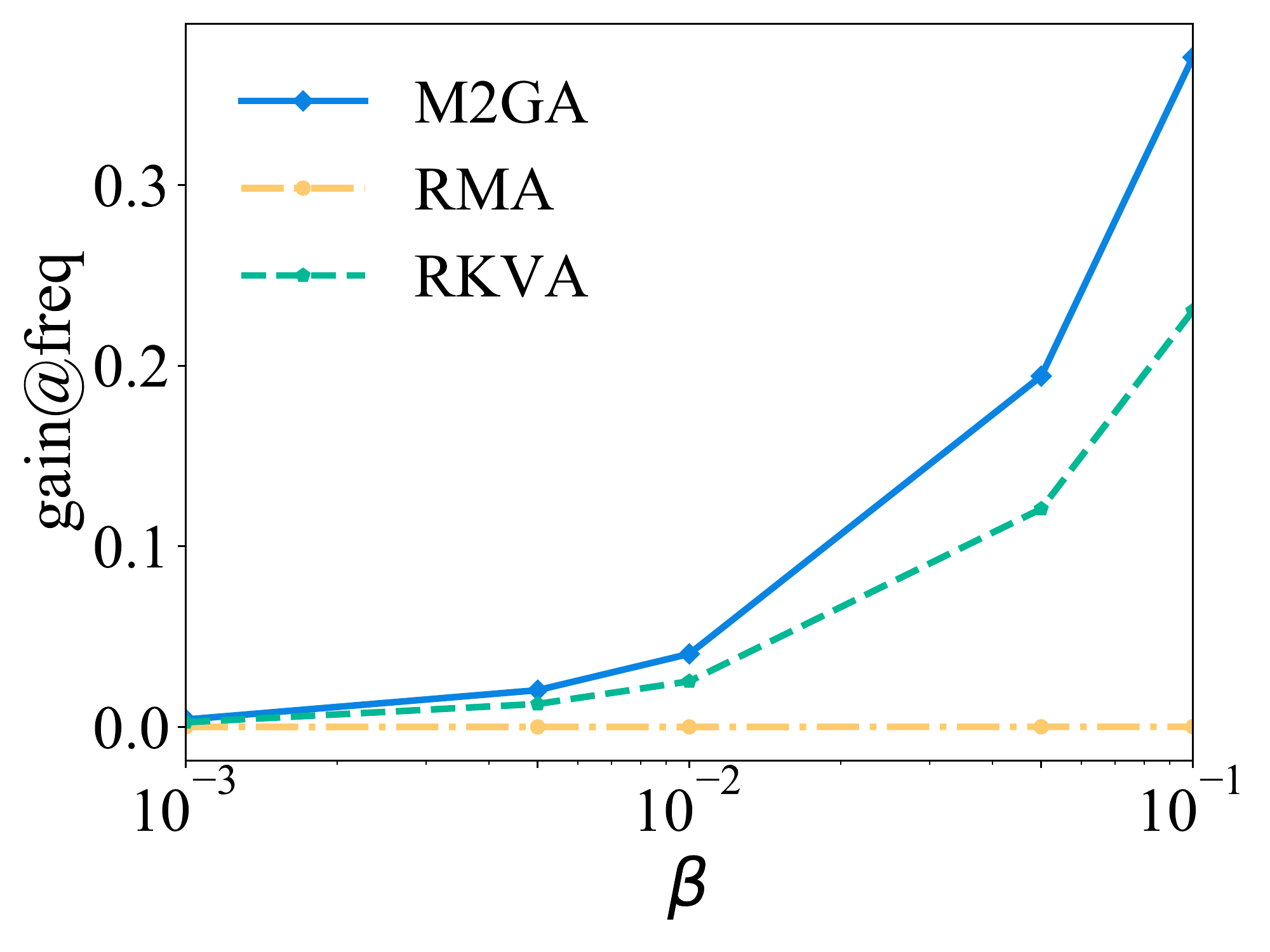}}
    \subfloat{\includegraphics[width=0.15\textwidth]{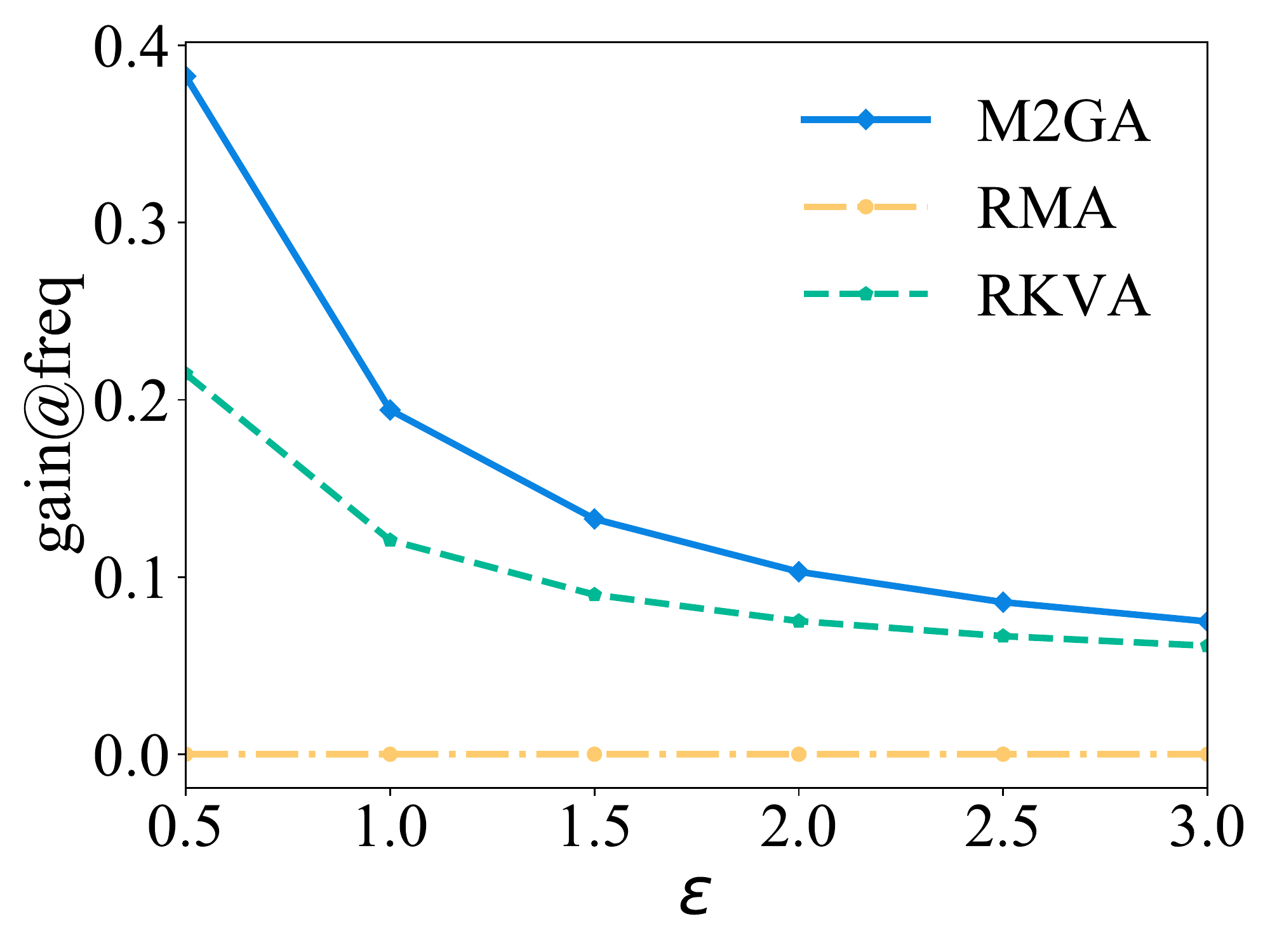}\label{fig:talkingdata_privkvm_epsilon}}
    \subfloat{\includegraphics[width=0.15\textwidth]{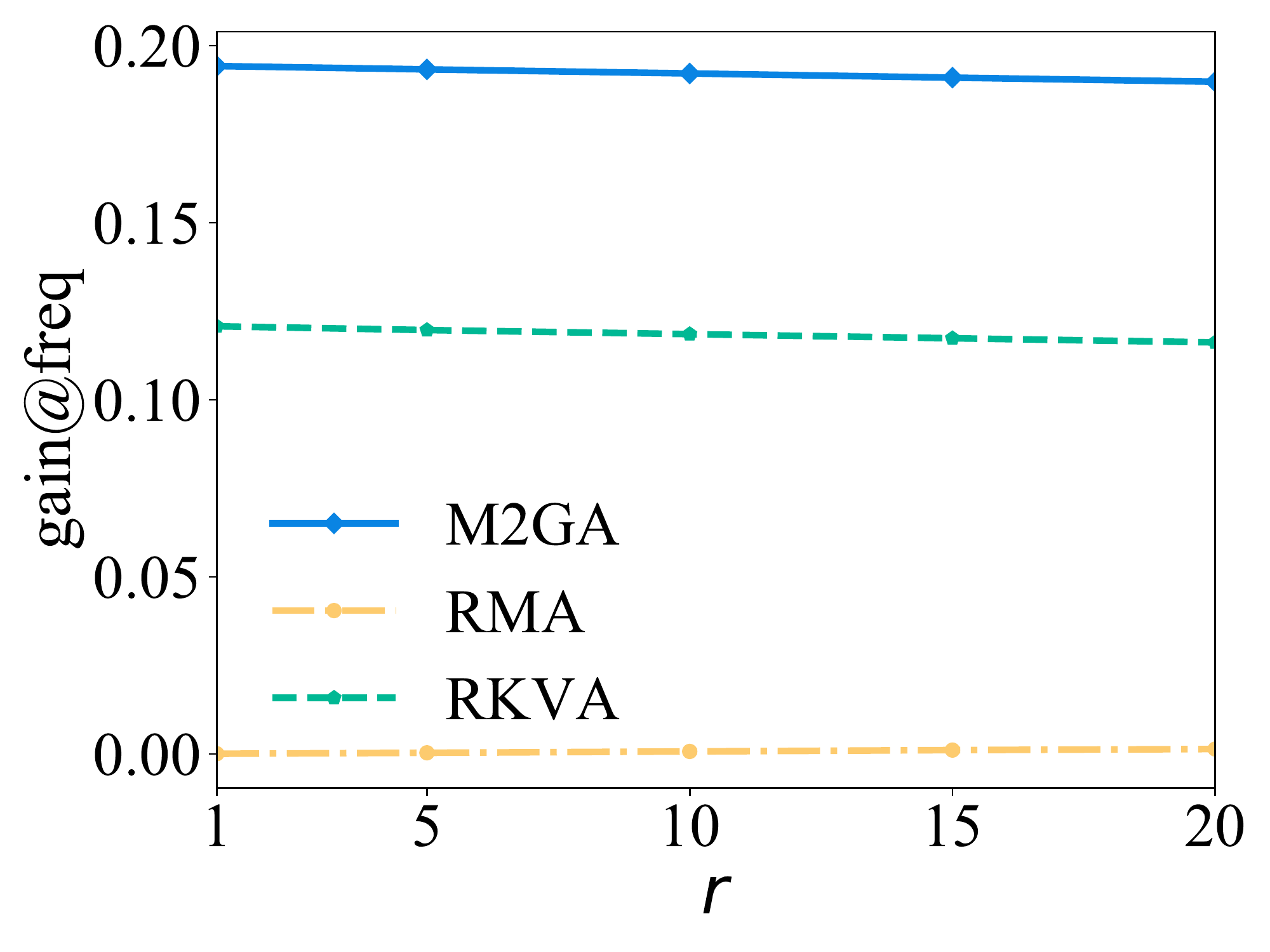}}
    
    \subfloat{\includegraphics[width=0.15\textwidth]{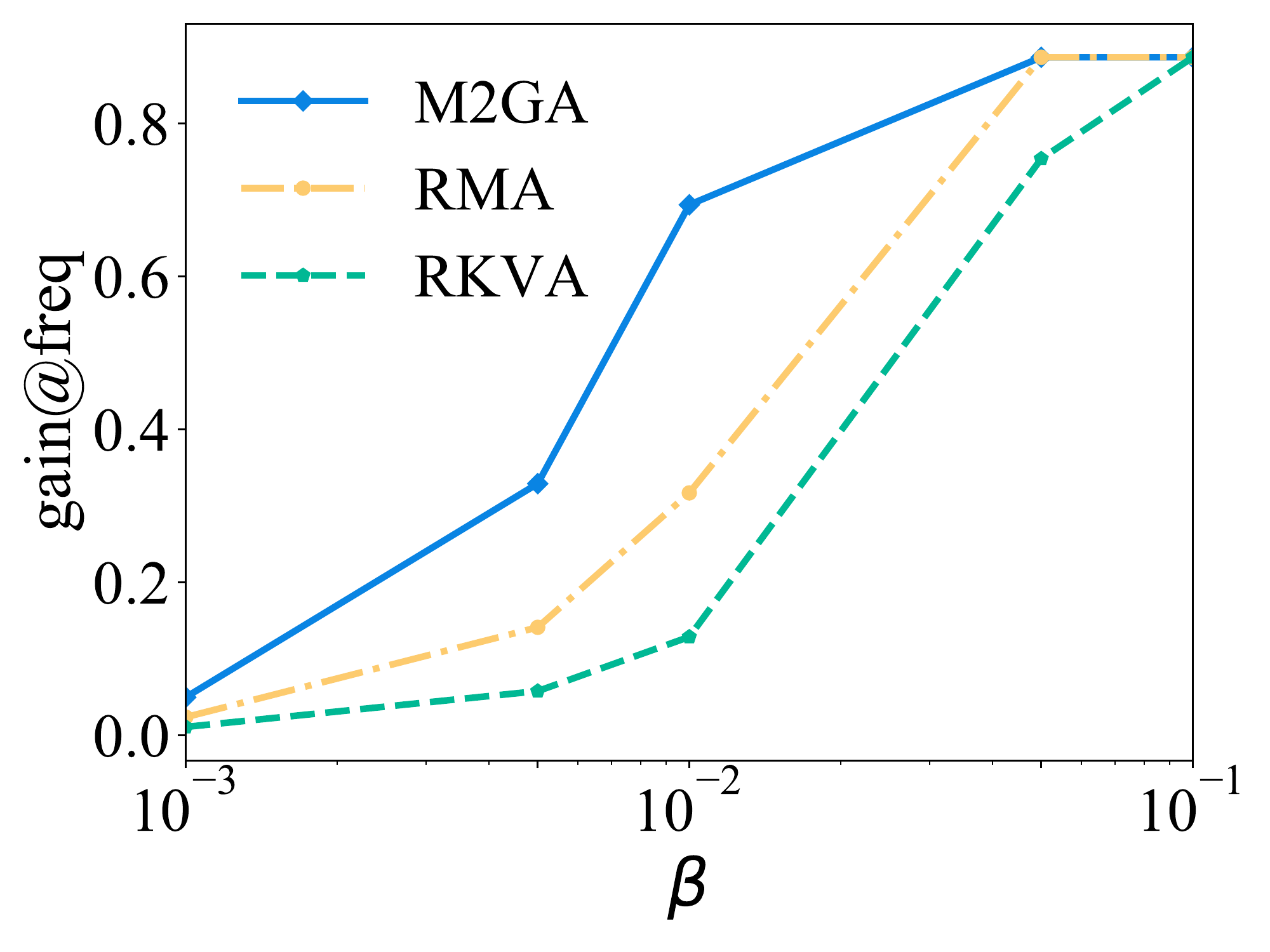}\label{fig:talkingdata_pckv_ue_freq_beta}}
    \subfloat{\includegraphics[width=0.15\textwidth]{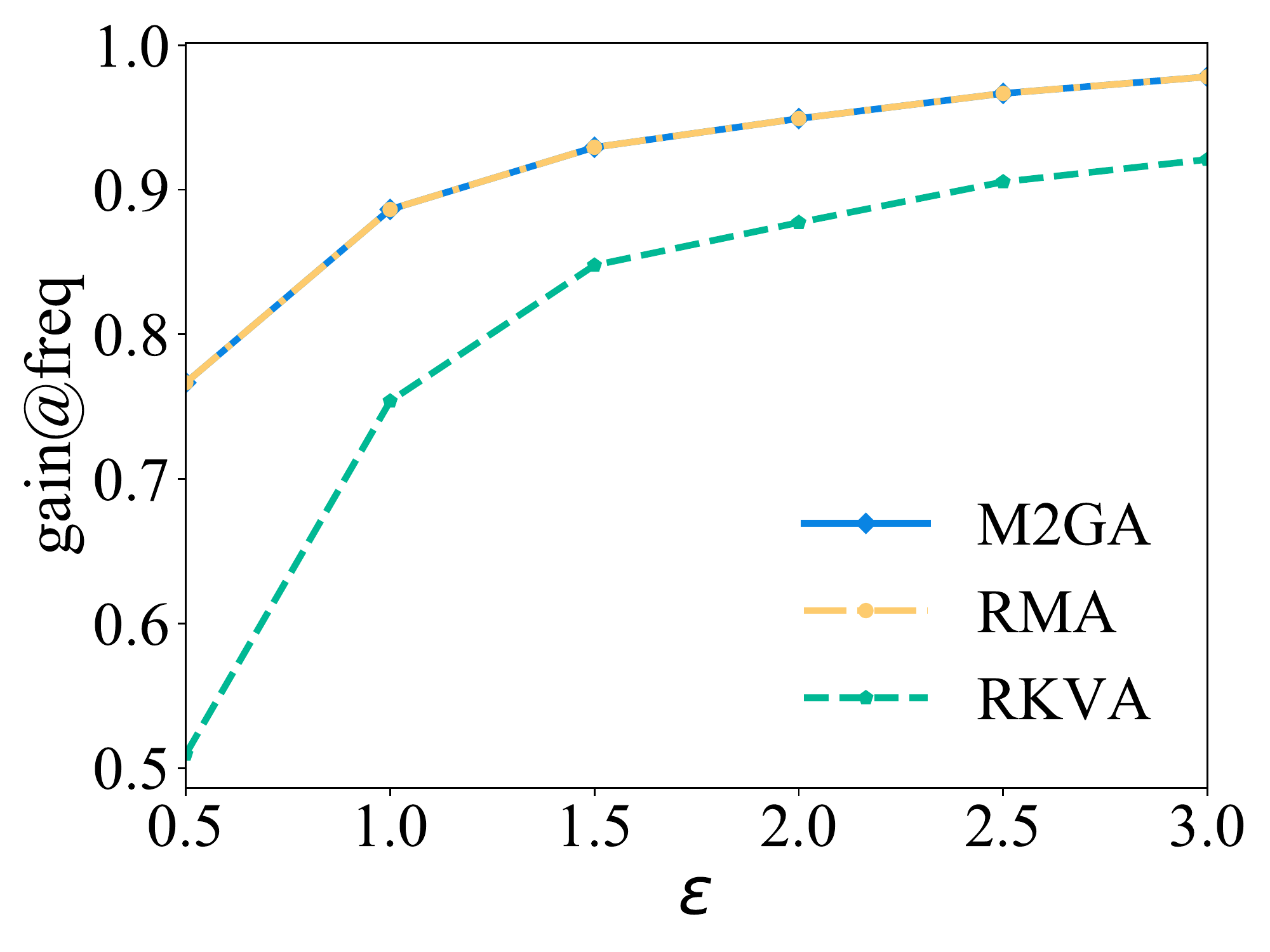}
      \label{fig:talkingdata_pckv_ue_freq_epsilon}}
    \subfloat{\includegraphics[width=0.15\textwidth]{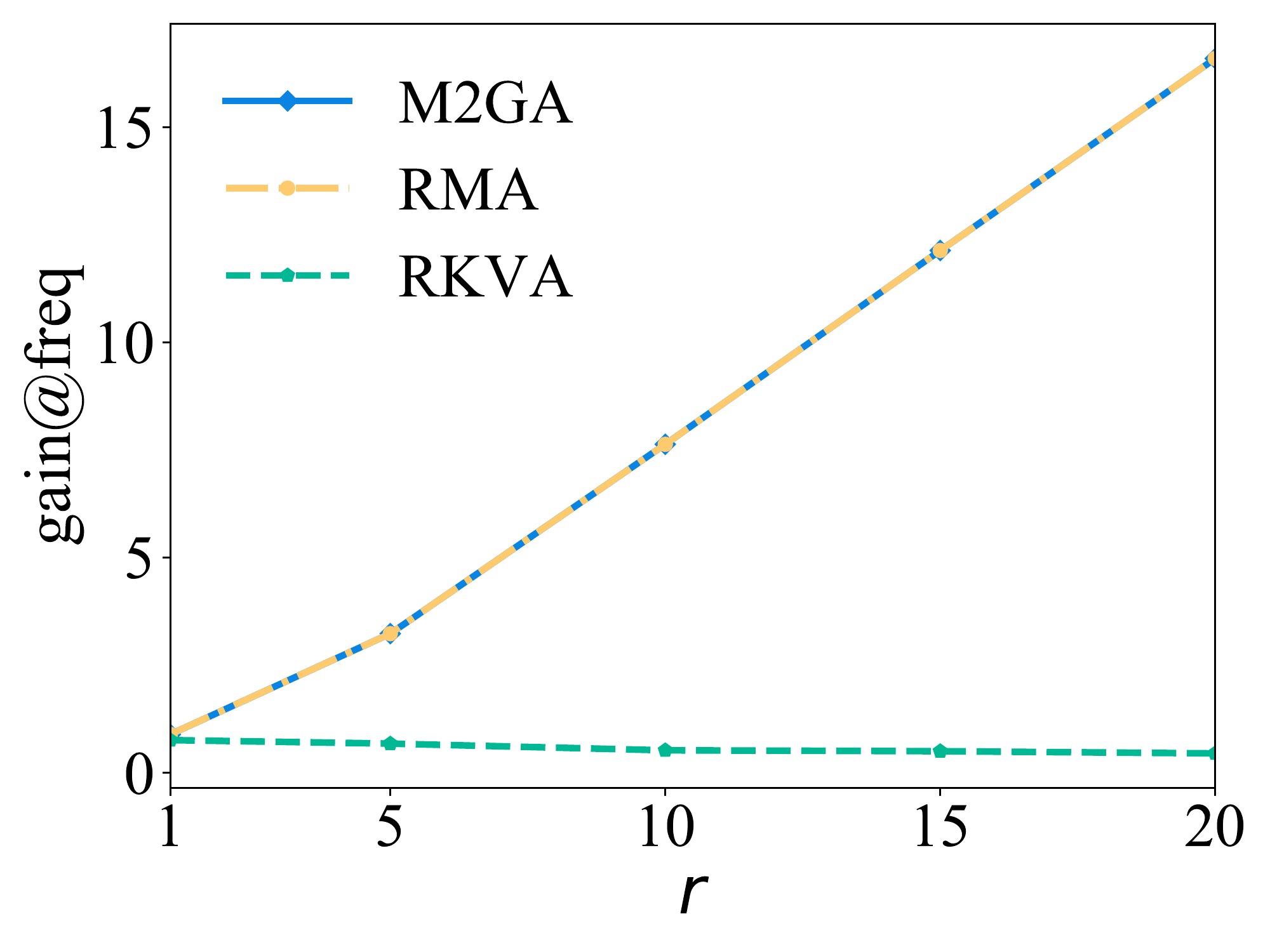}
      \label{fig:talkingdata_pckv_ue_freq_r}}

    \subfloat{\includegraphics[width=0.15\textwidth]{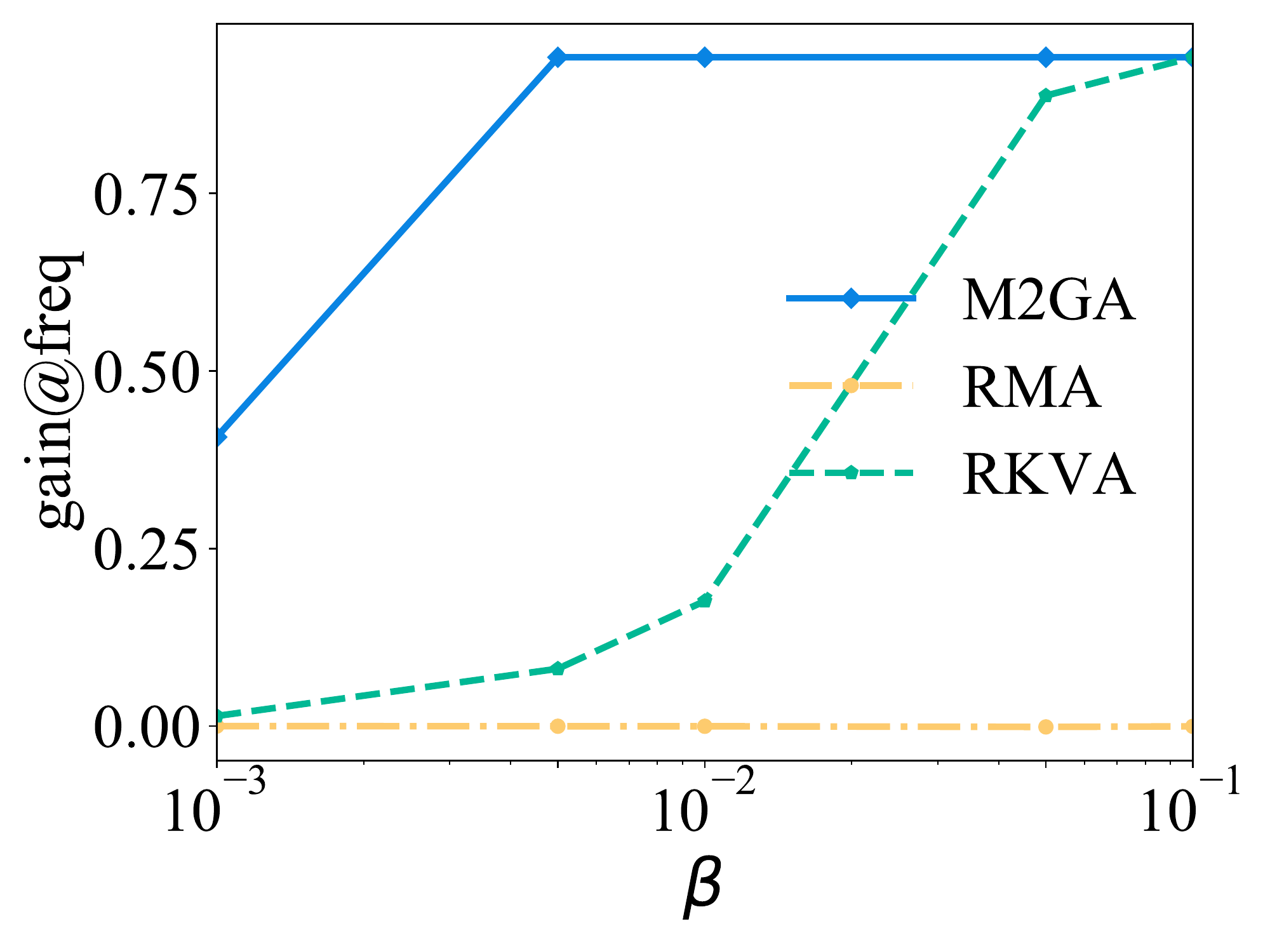}}
    \subfloat{\includegraphics[width=0.15\textwidth]{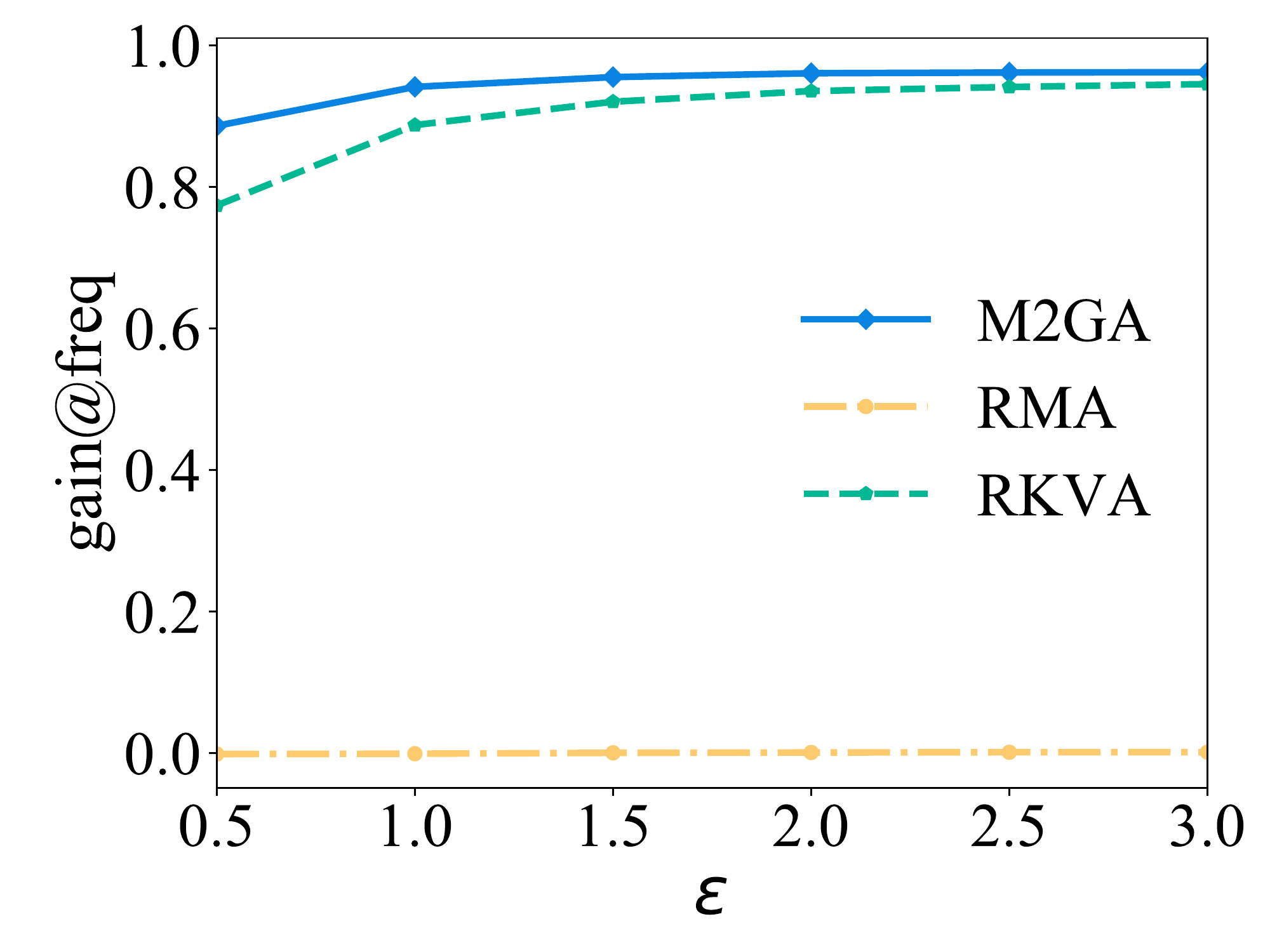}}
    \subfloat{\includegraphics[width=0.15\textwidth]{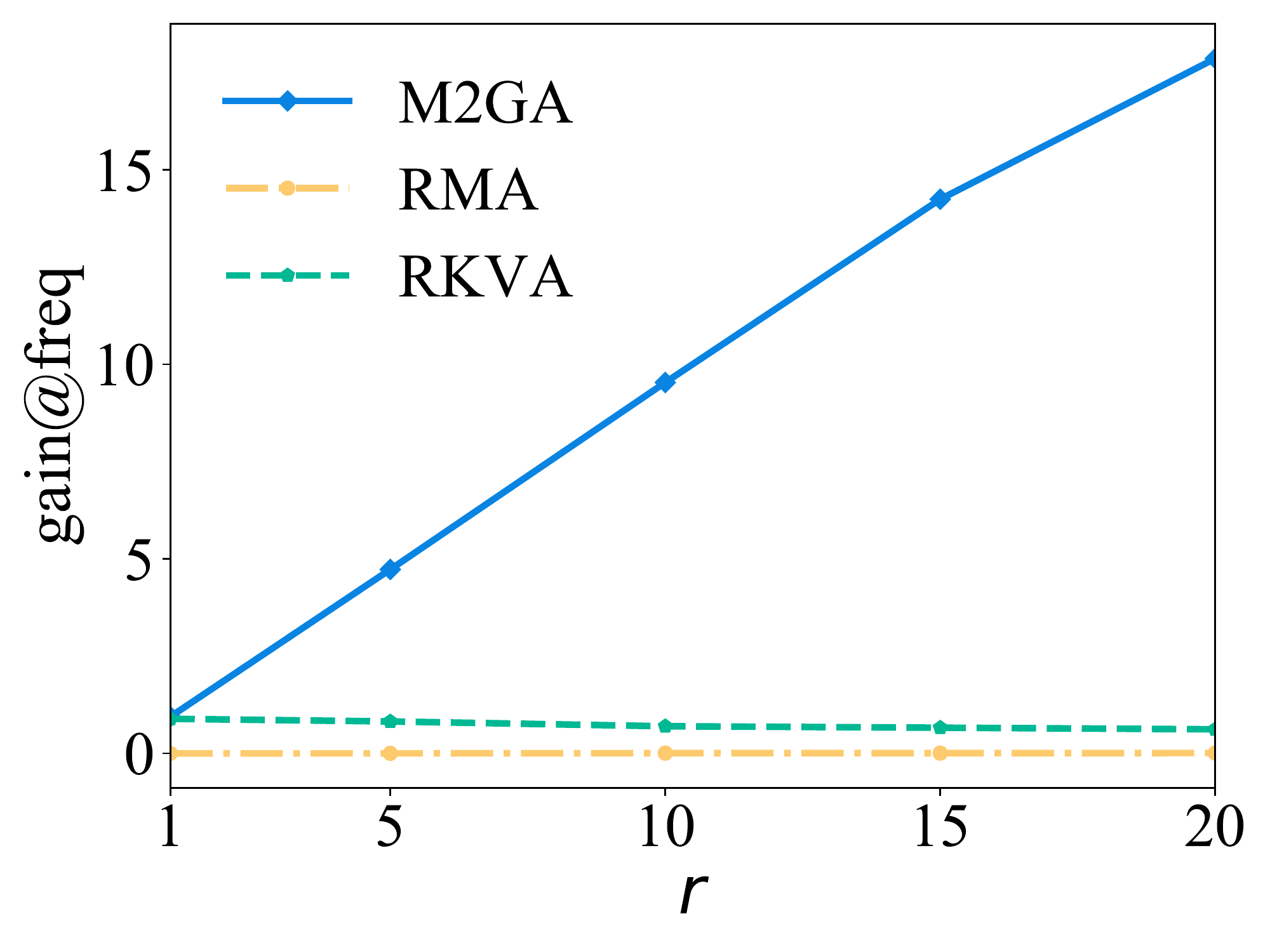}}

\caption{Impact of different parameters ($\beta, \epsilon, r$) on the  frequency gains on TalkingData. The three rows are for PrivKVM, PCKV-UE, and PCKV-GRR, respectively.}
\label{fig:exp_attack_freq_talkingdata}
\end{figure}

\begin{figure}[!tb]
    \centering 
    \vspace{1mm}
    \subfloat{\includegraphics[width=0.15\textwidth]{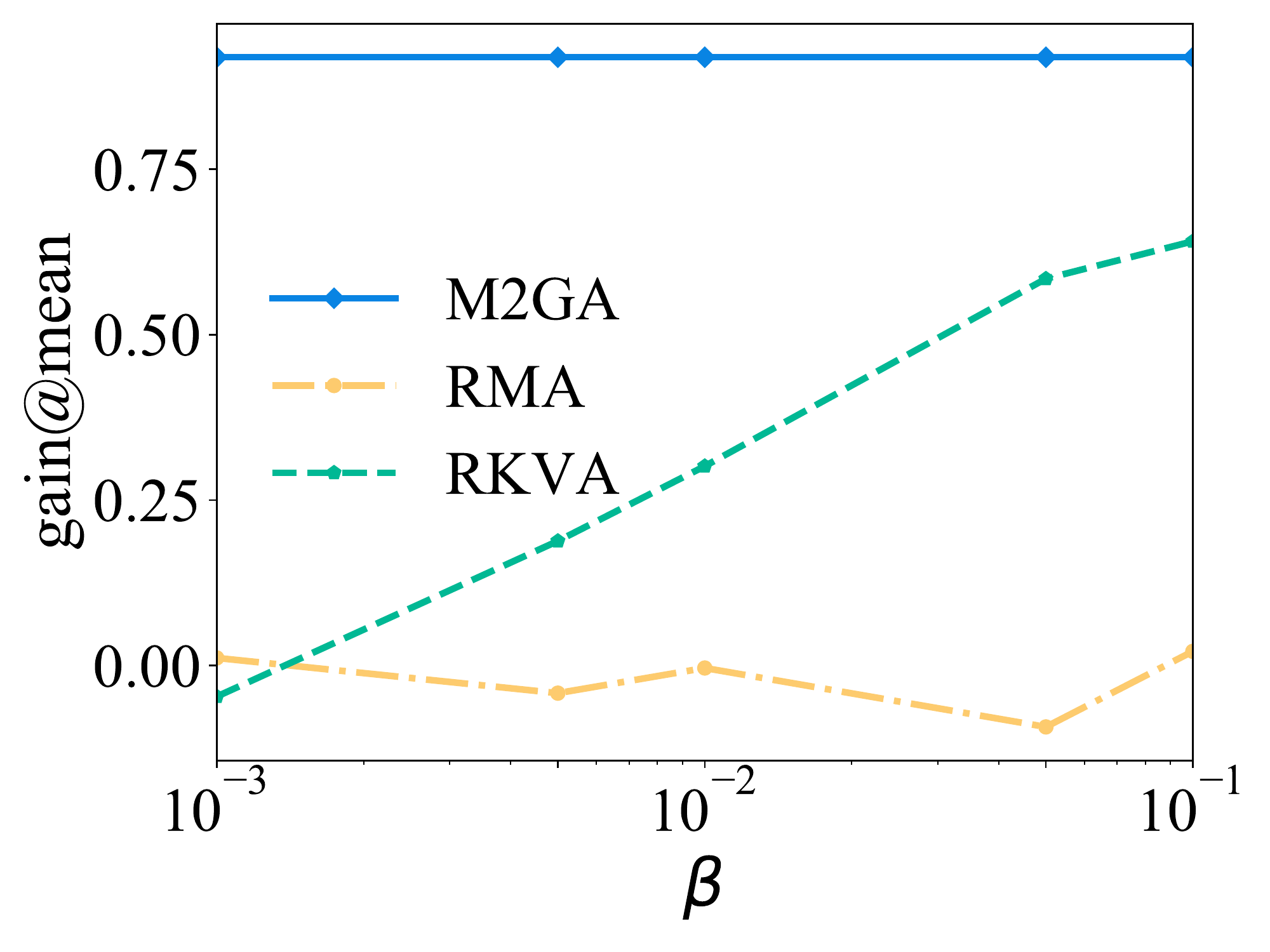}}
    \subfloat{ \includegraphics[width=0.15\textwidth]{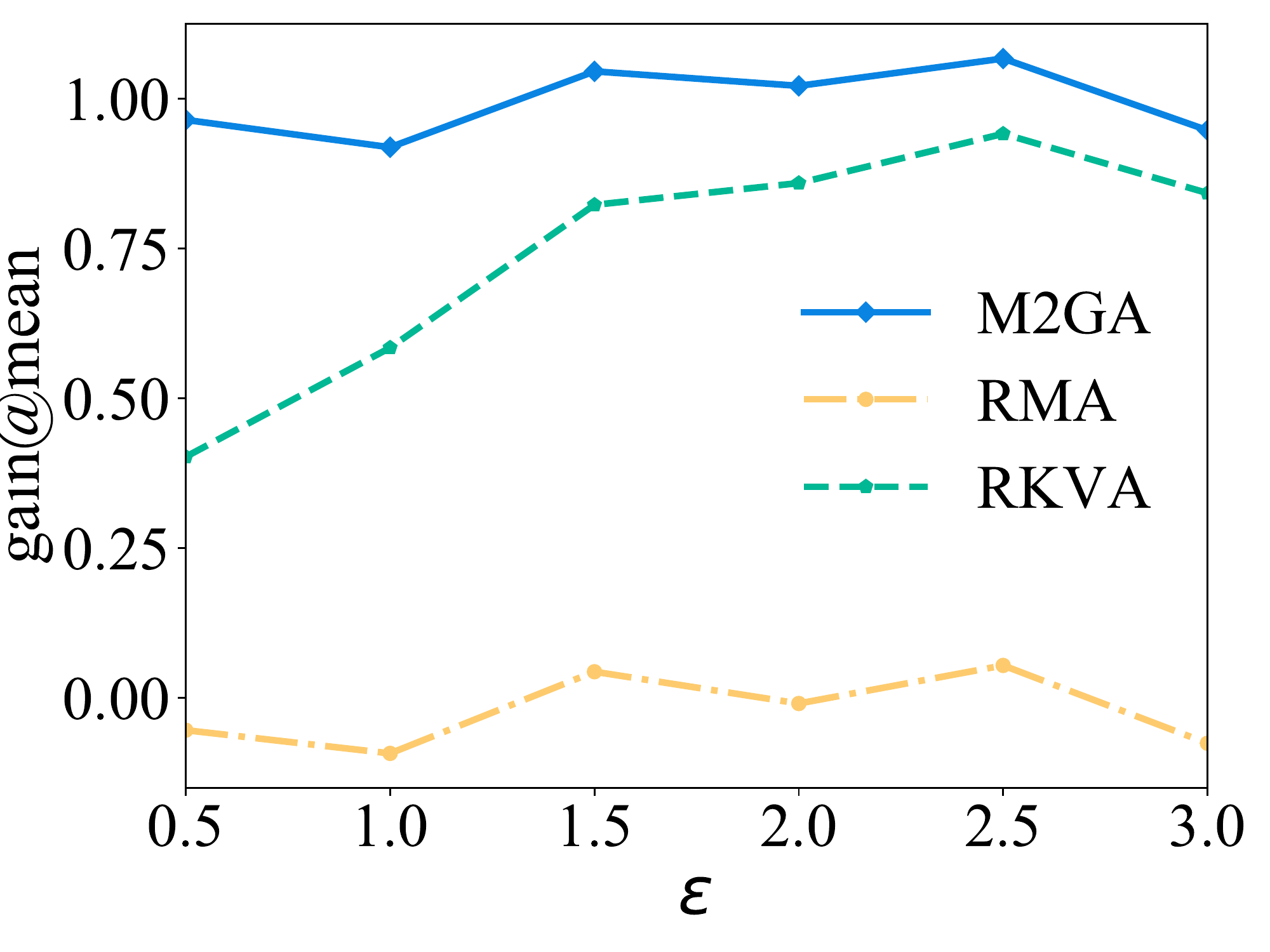}}
    \subfloat{  \includegraphics[width=0.15\textwidth]{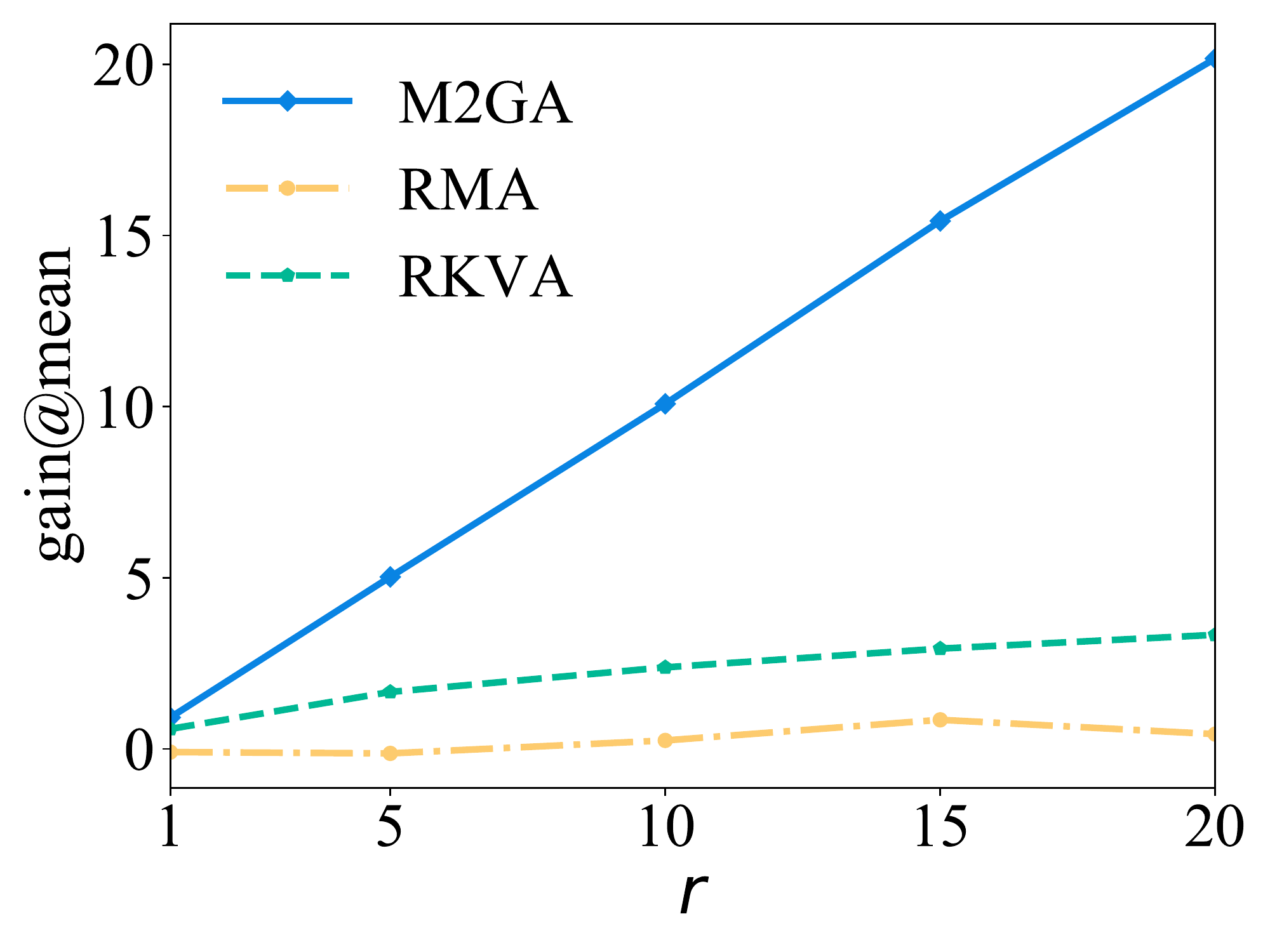}}

    \subfloat{  \includegraphics[width=0.15\textwidth]{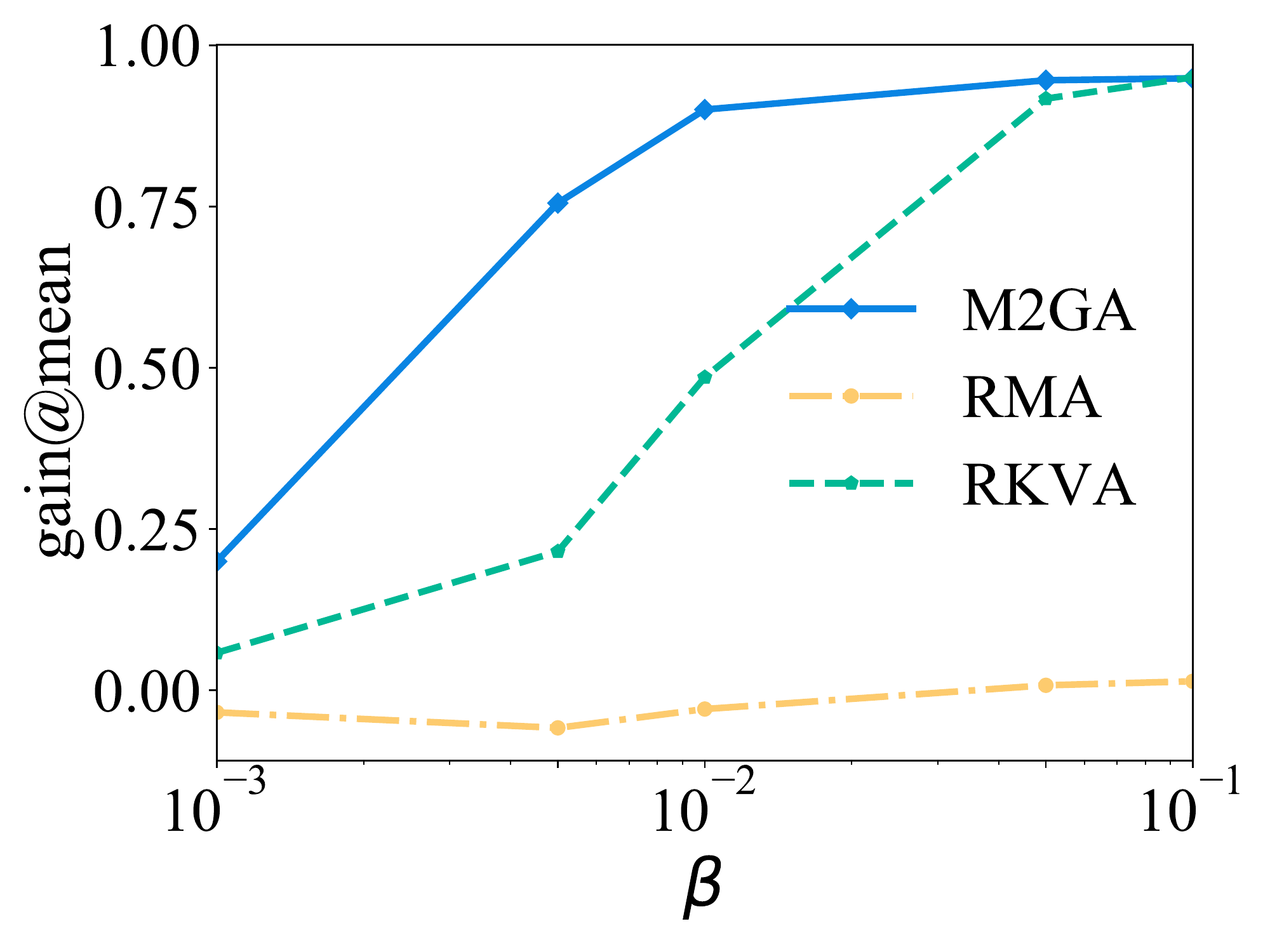}}
    \subfloat{  \includegraphics[width=0.15\textwidth]{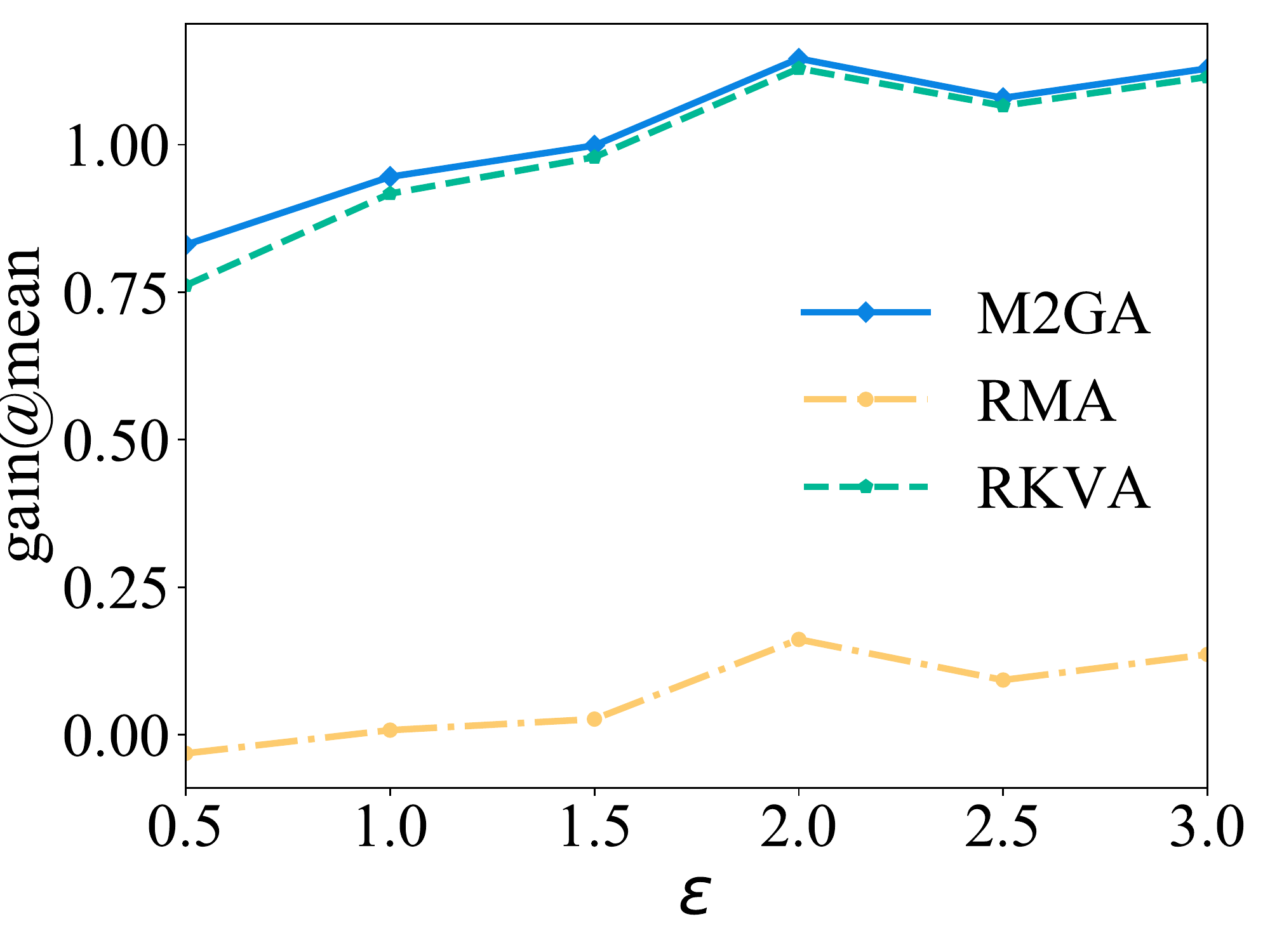}}
    \subfloat{  \includegraphics[width=0.15\textwidth]{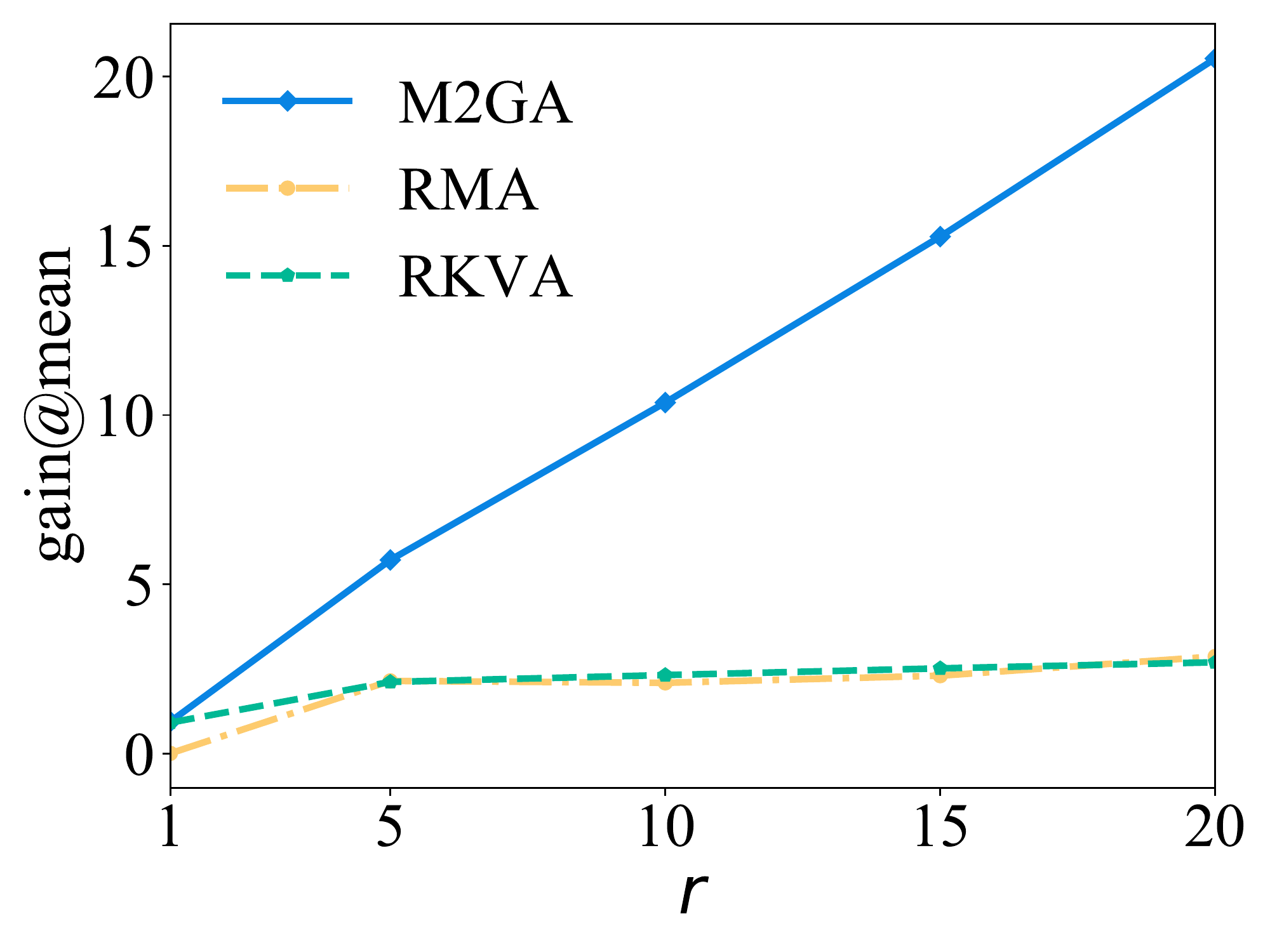}}

    \subfloat{  \includegraphics[width=0.15\textwidth]{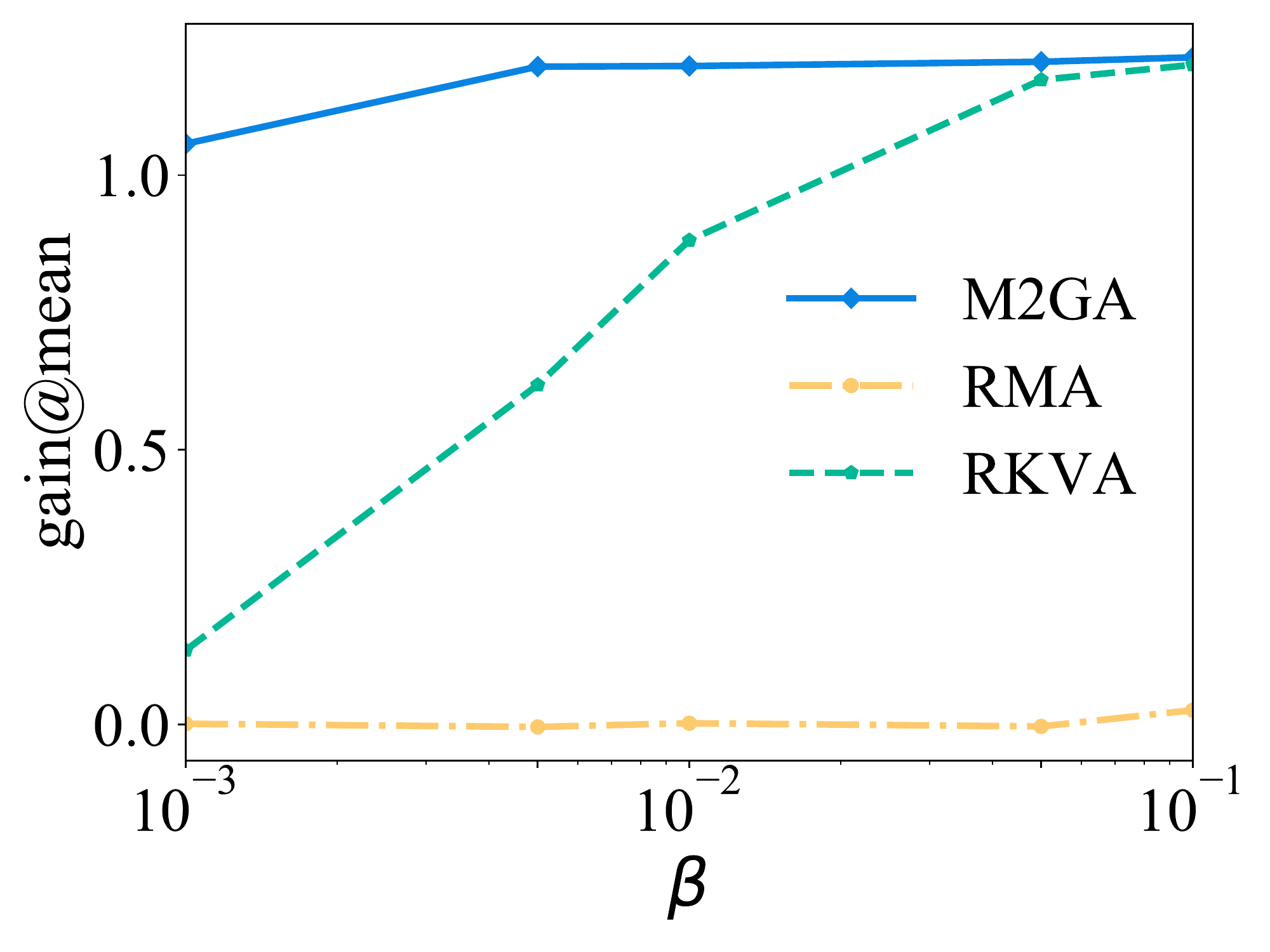}}
    \subfloat{  \includegraphics[width=0.15\textwidth]{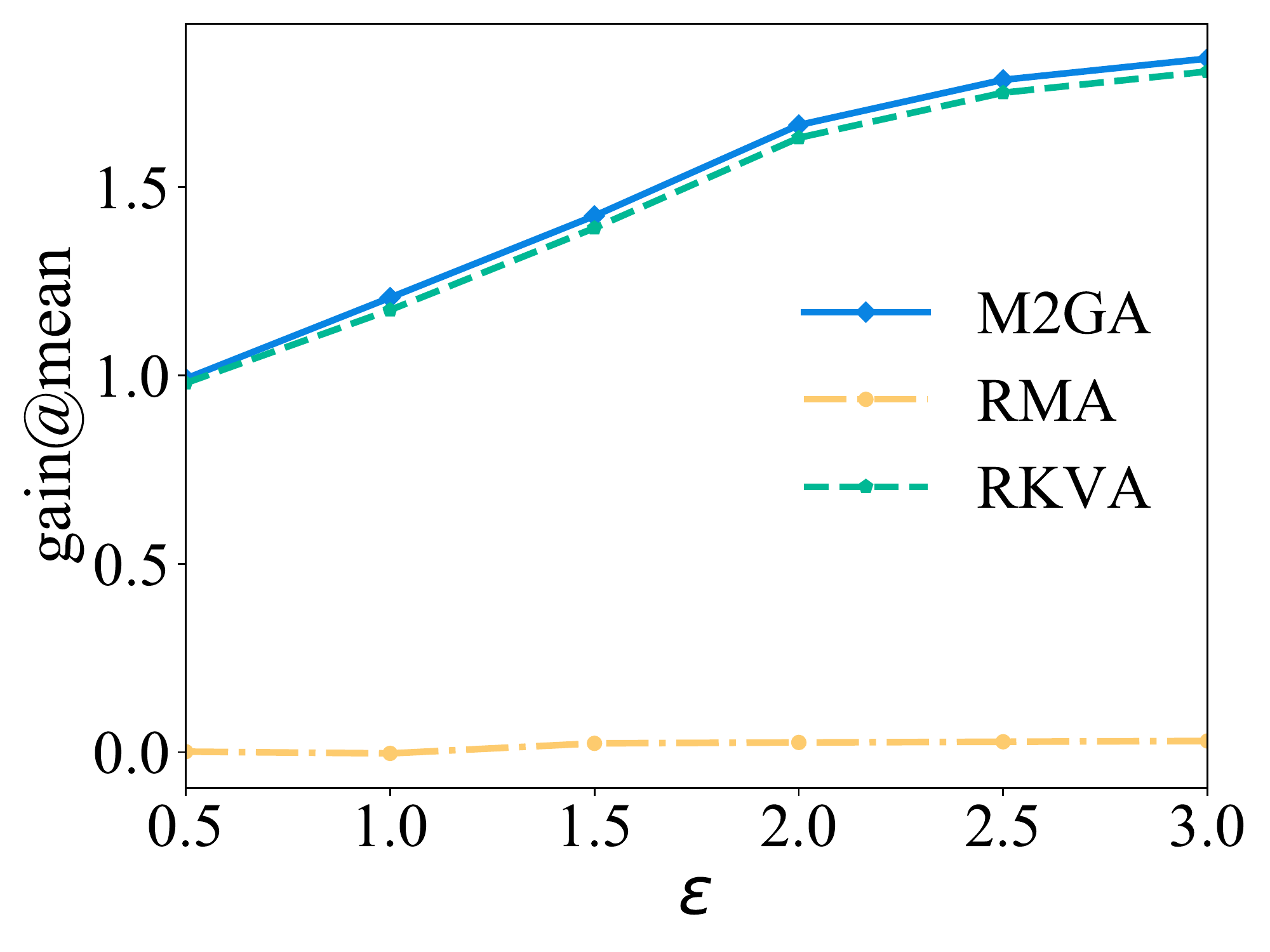}}
    \subfloat{  \includegraphics[width=0.15\textwidth]{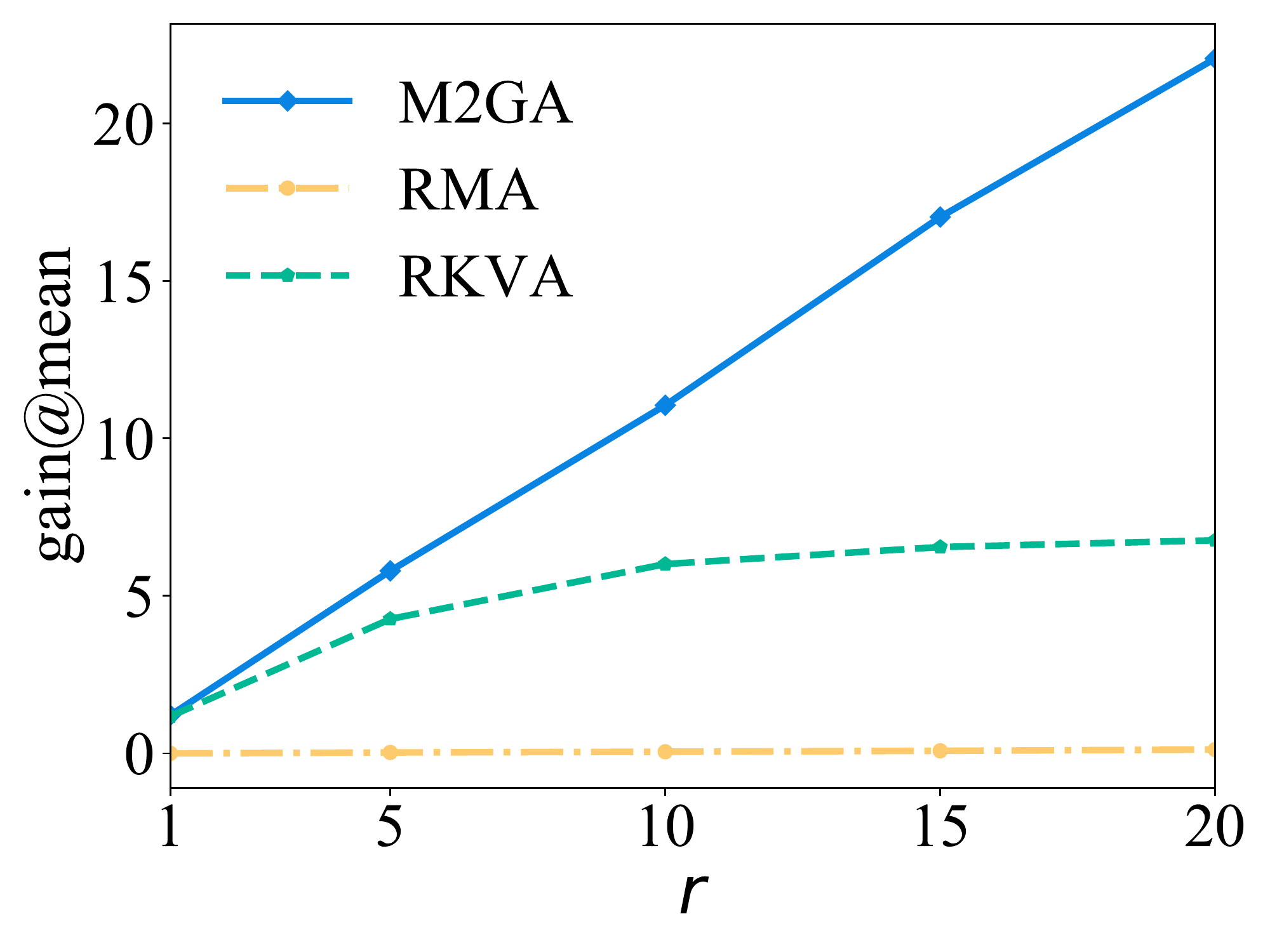}}

\caption{Impact of different parameters ($\beta, \epsilon, r$) on the  mean gains on TalkingData. The three rows are for PrivKVM, PCKV-UE, and PCKV-GRR, respectively.}
\label{fig:exp_attack_mean_talkingdata}
\end{figure}

\begin{figure}[!tb]
    \centering 
    \subfloat{  \includegraphics[width=0.15\textwidth]{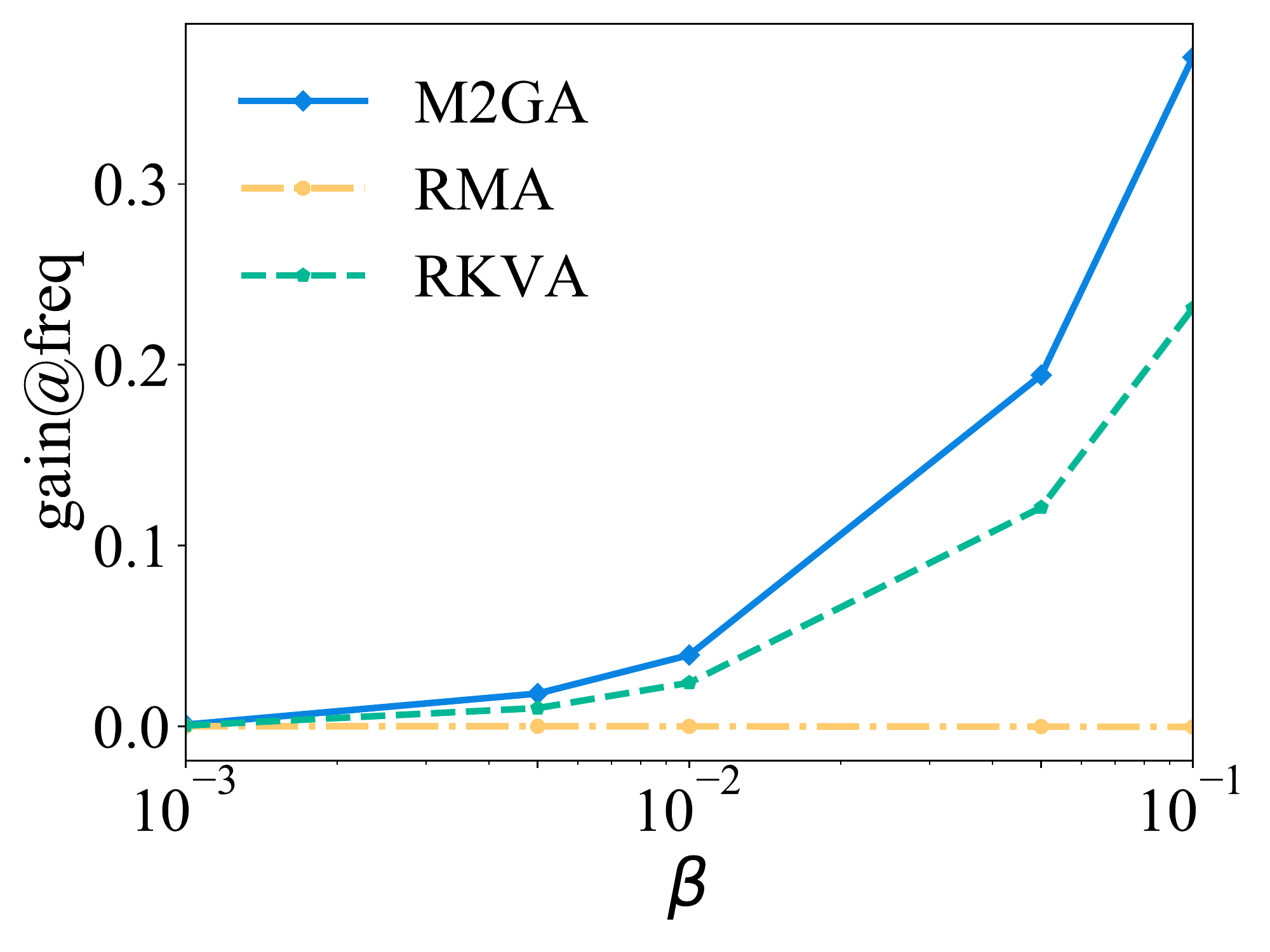}}
    \subfloat{  \includegraphics[width=0.15\textwidth]{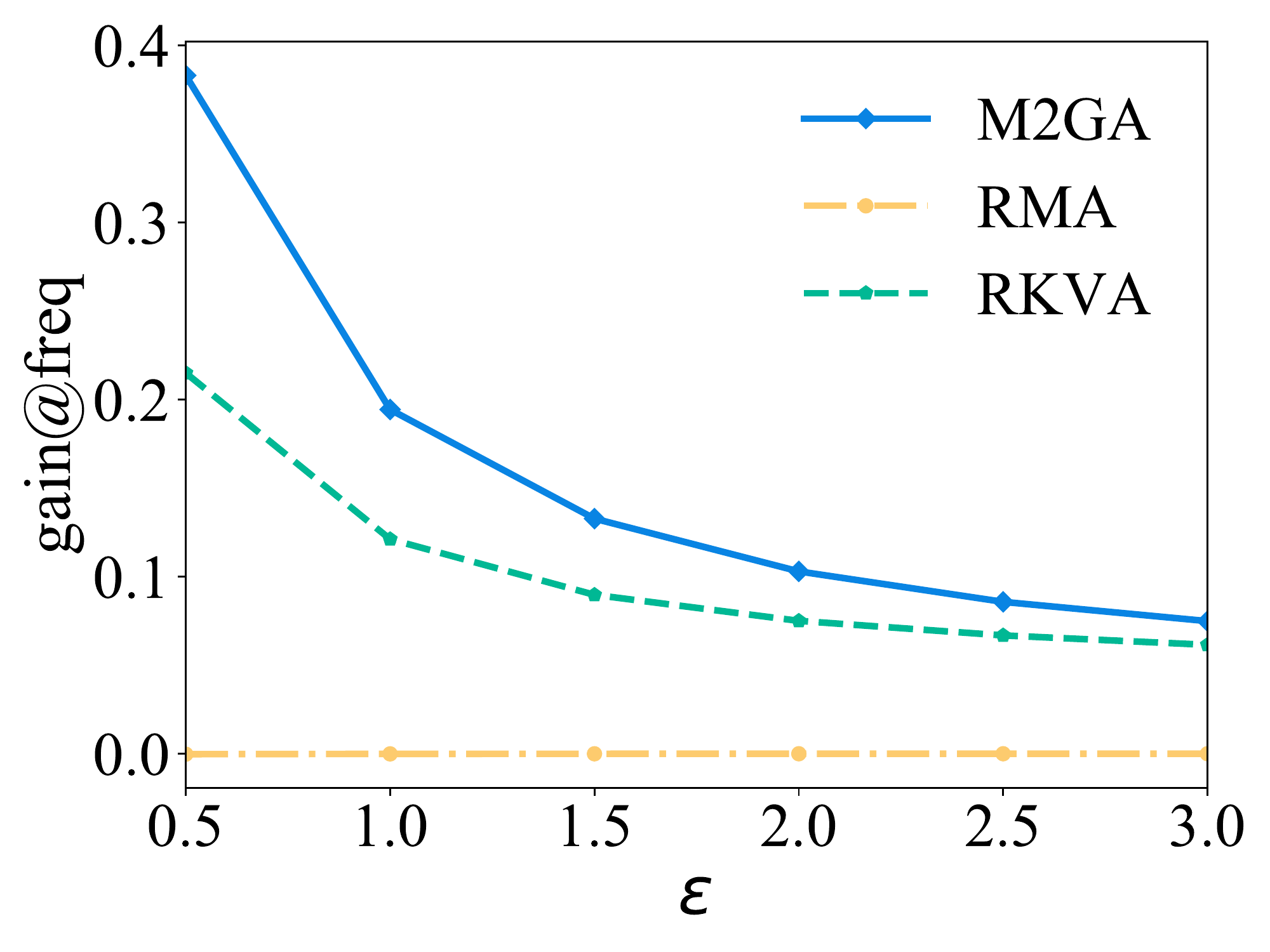}
      \label{fig:movielens_privkvm_epsilon}}
    \subfloat{  \includegraphics[width=0.15\textwidth]{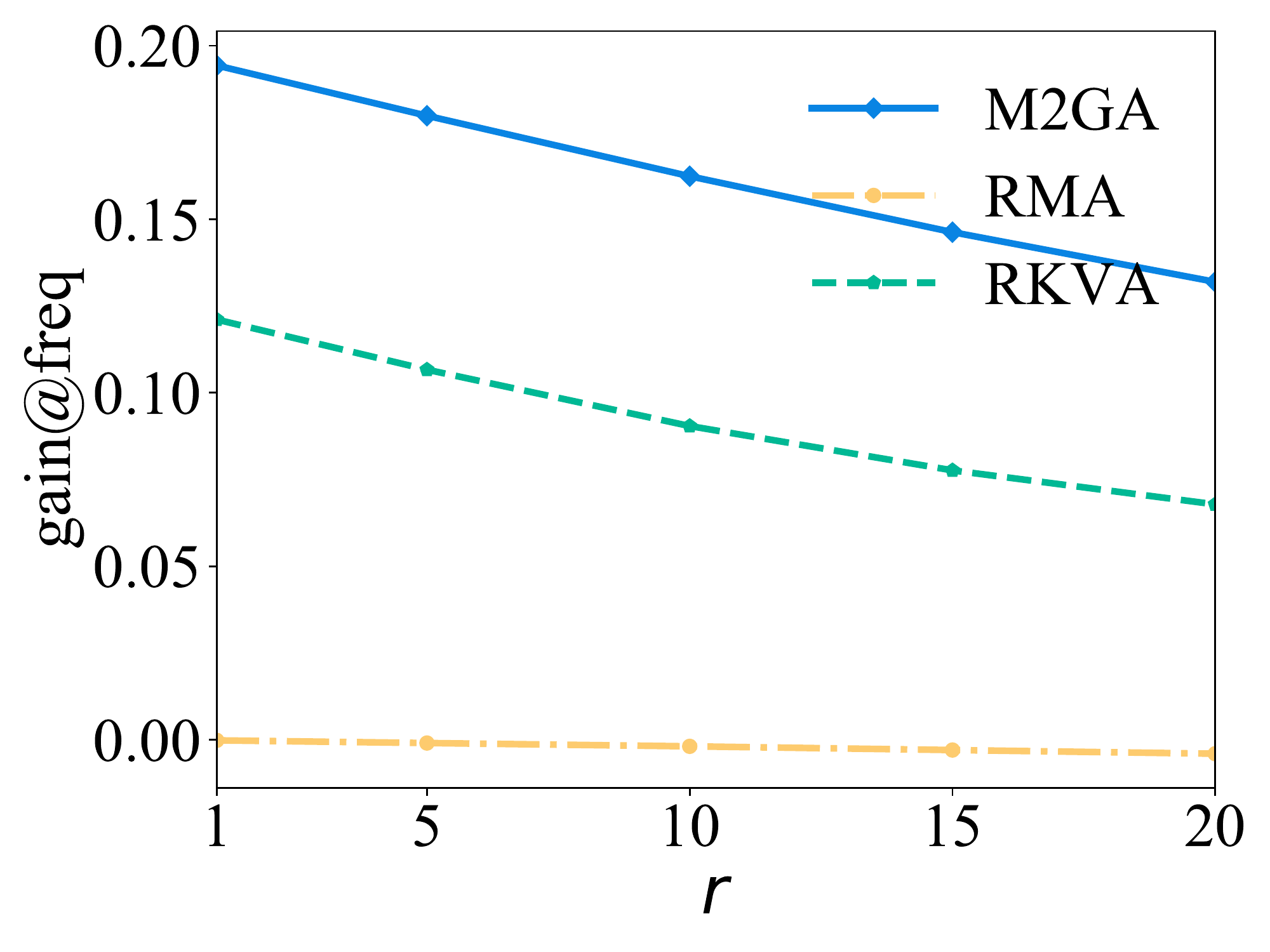}}

    \subfloat{  \includegraphics[width=0.15\textwidth]{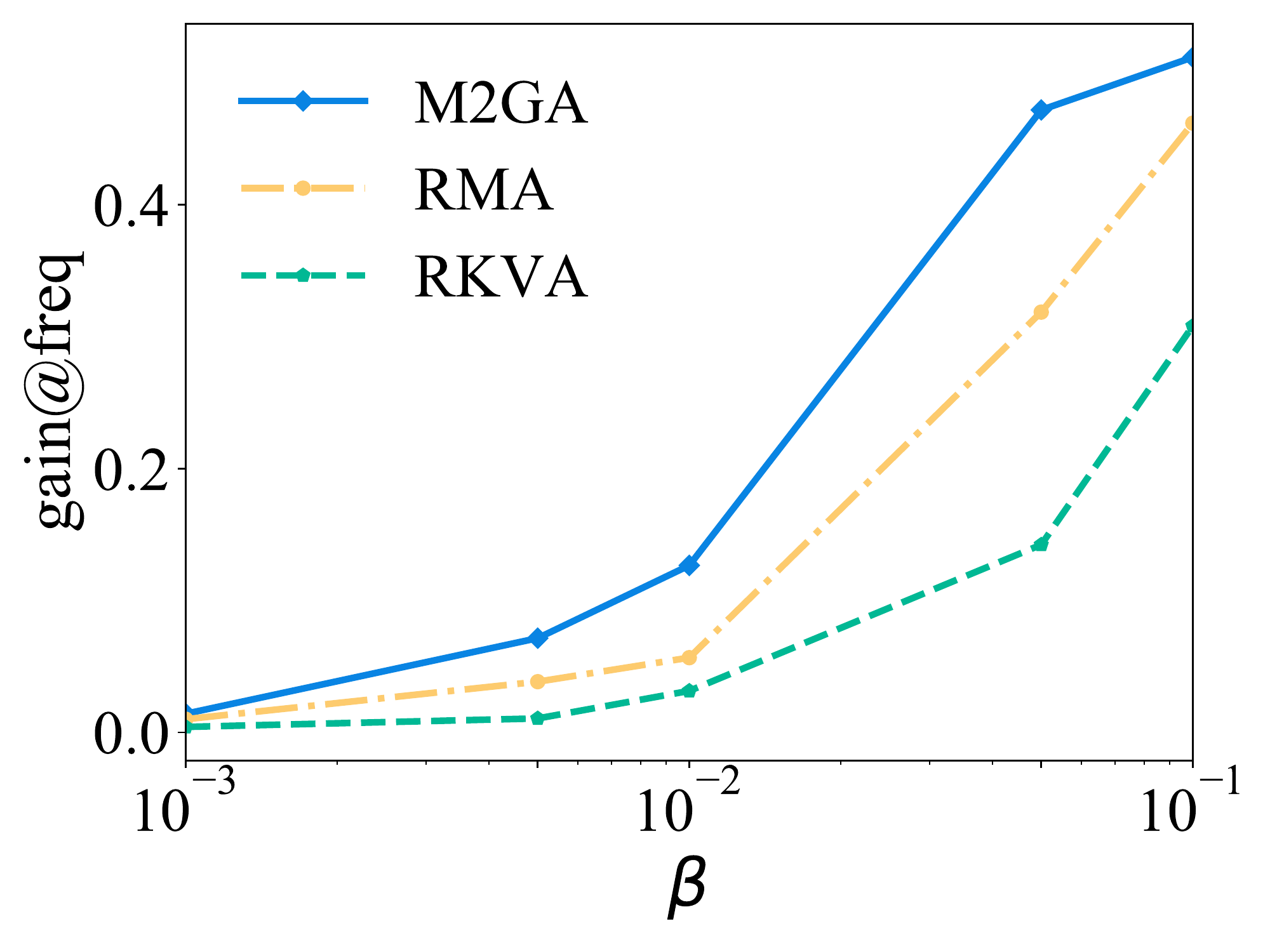}}
    \subfloat{  \includegraphics[width=0.15\textwidth]{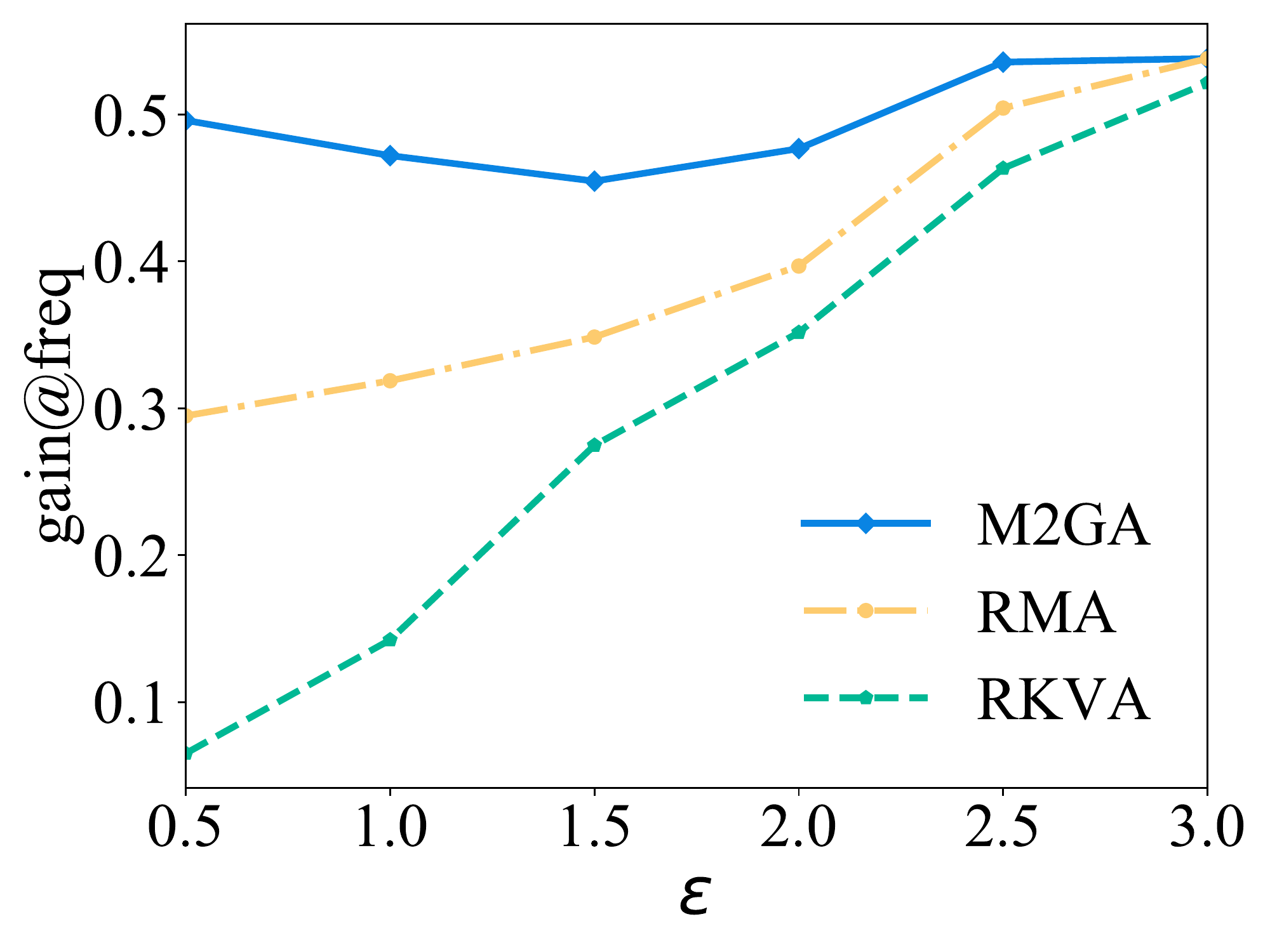}}
    \subfloat{ \includegraphics[width=0.15\textwidth]{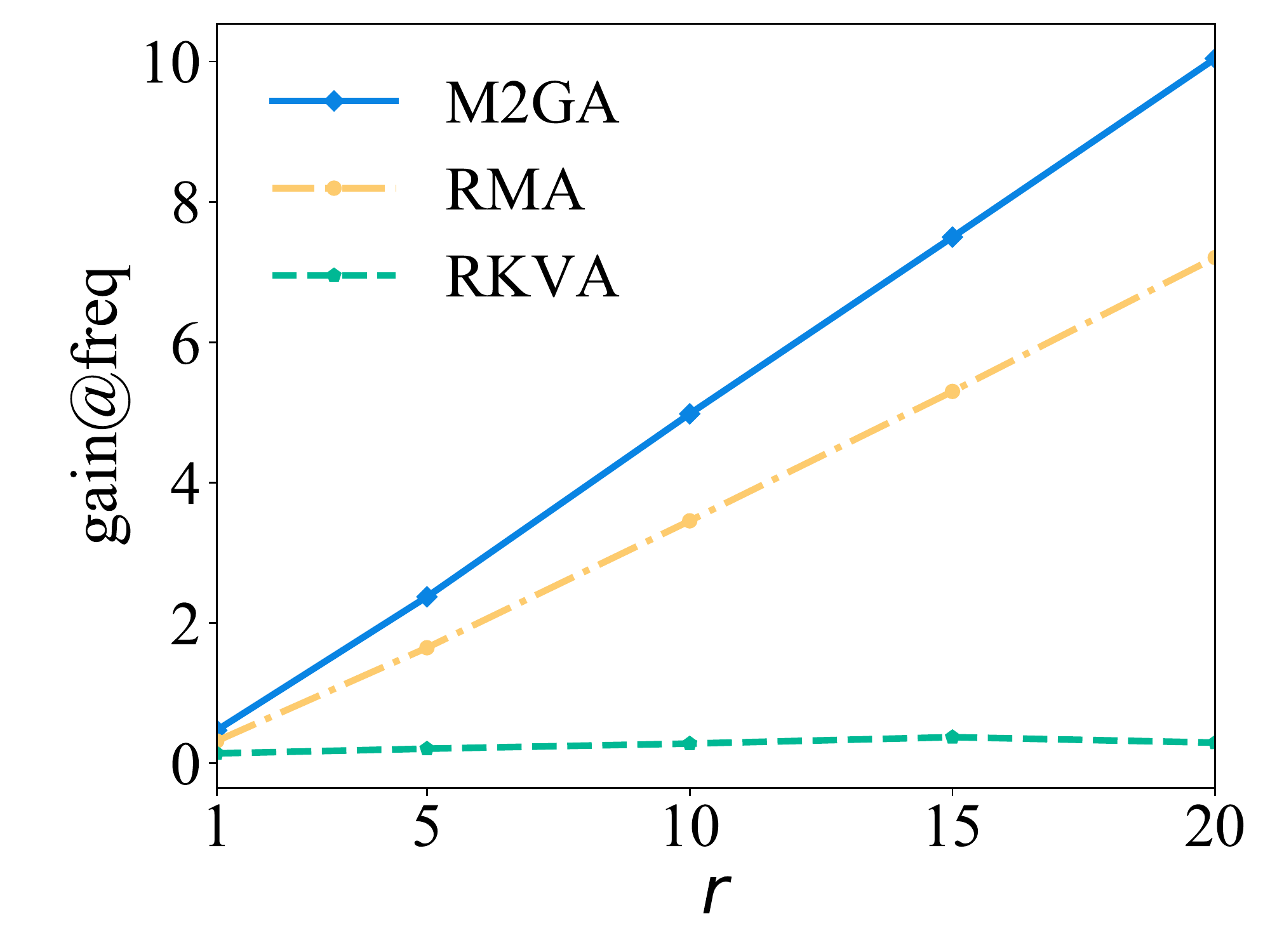}}
    
    \subfloat{  \includegraphics[width=0.15\textwidth]{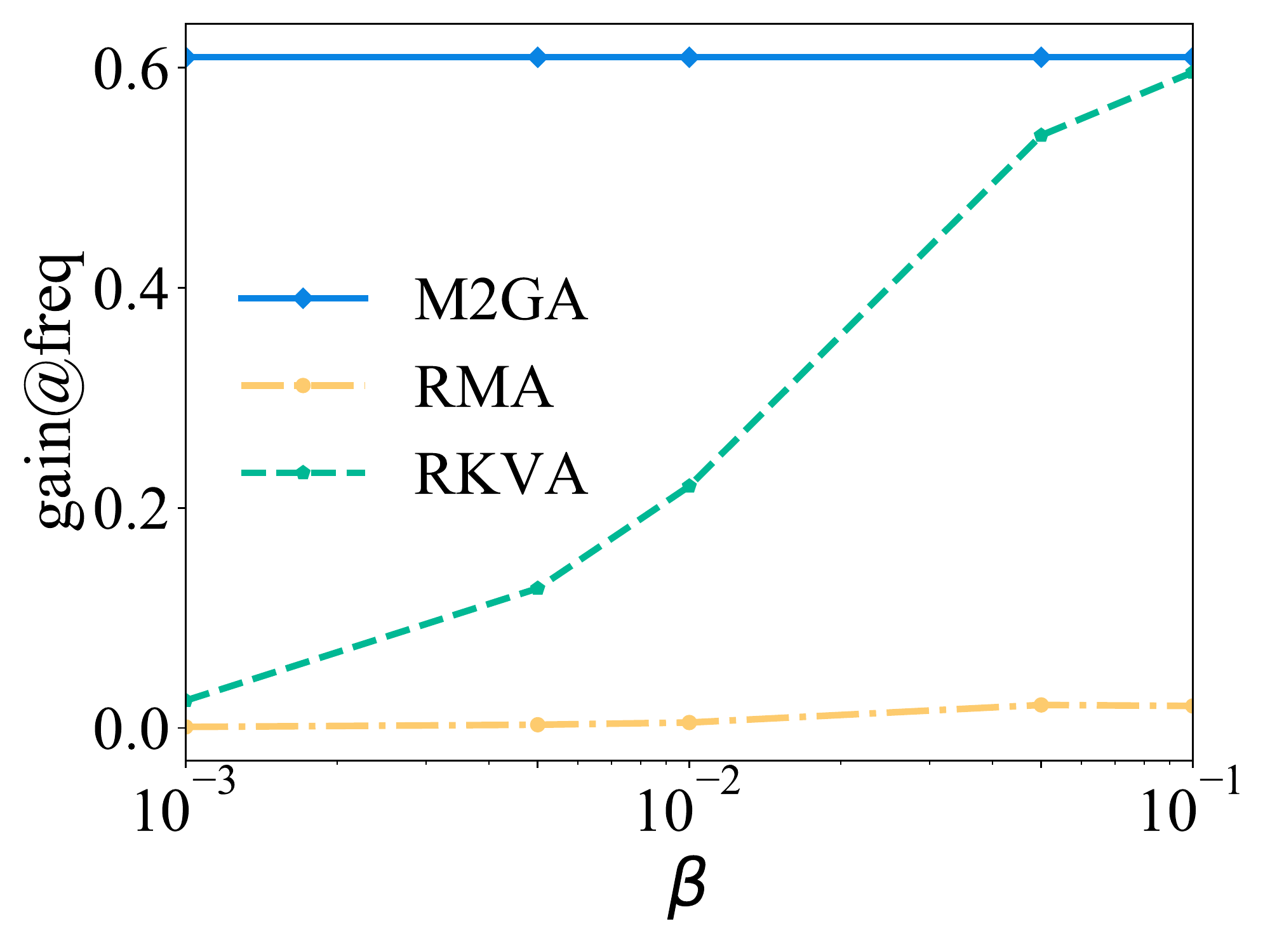}}
    \subfloat{  \includegraphics[width=0.15\textwidth]{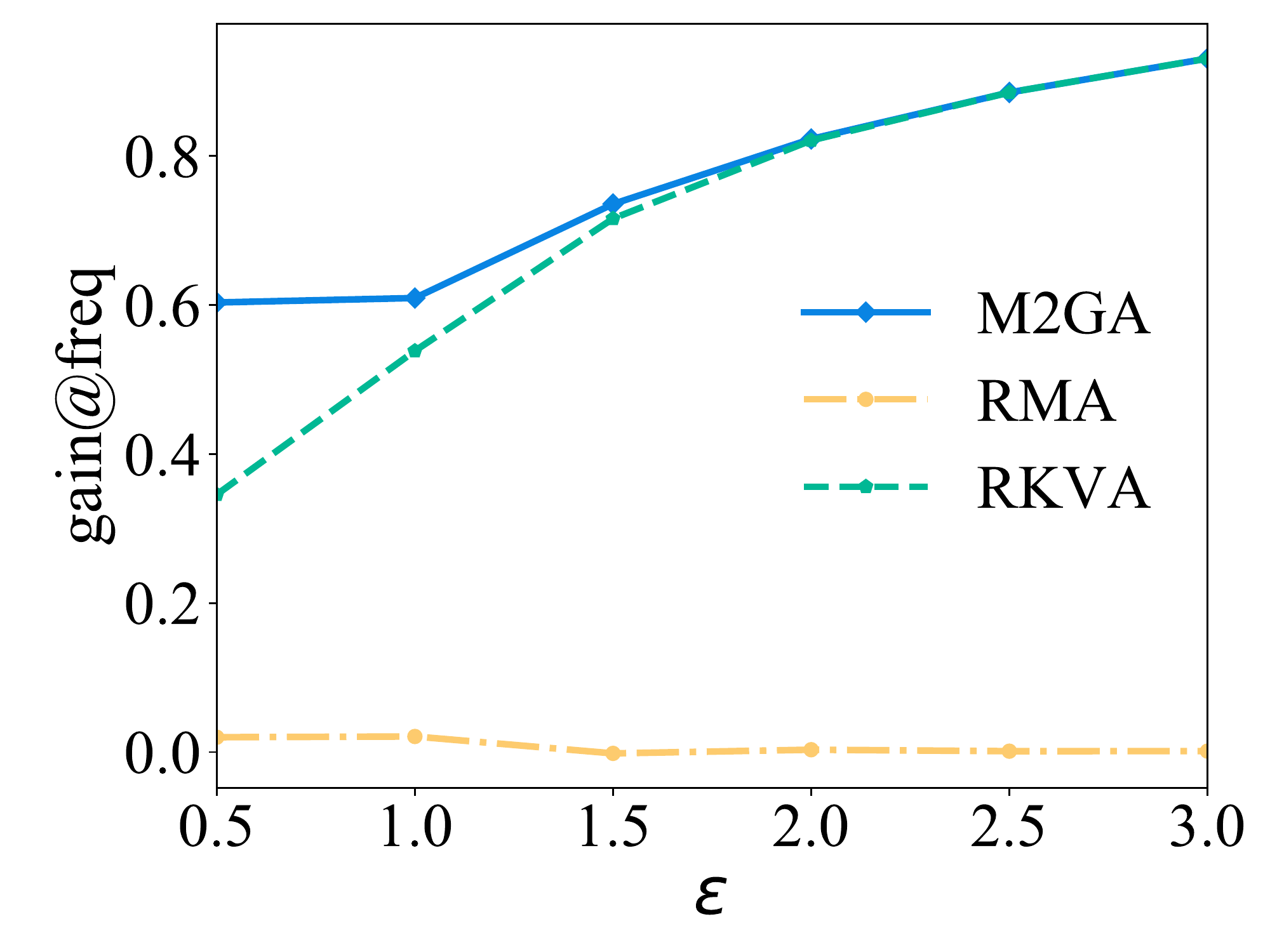}}
    \subfloat{ \includegraphics[width=0.15\textwidth]{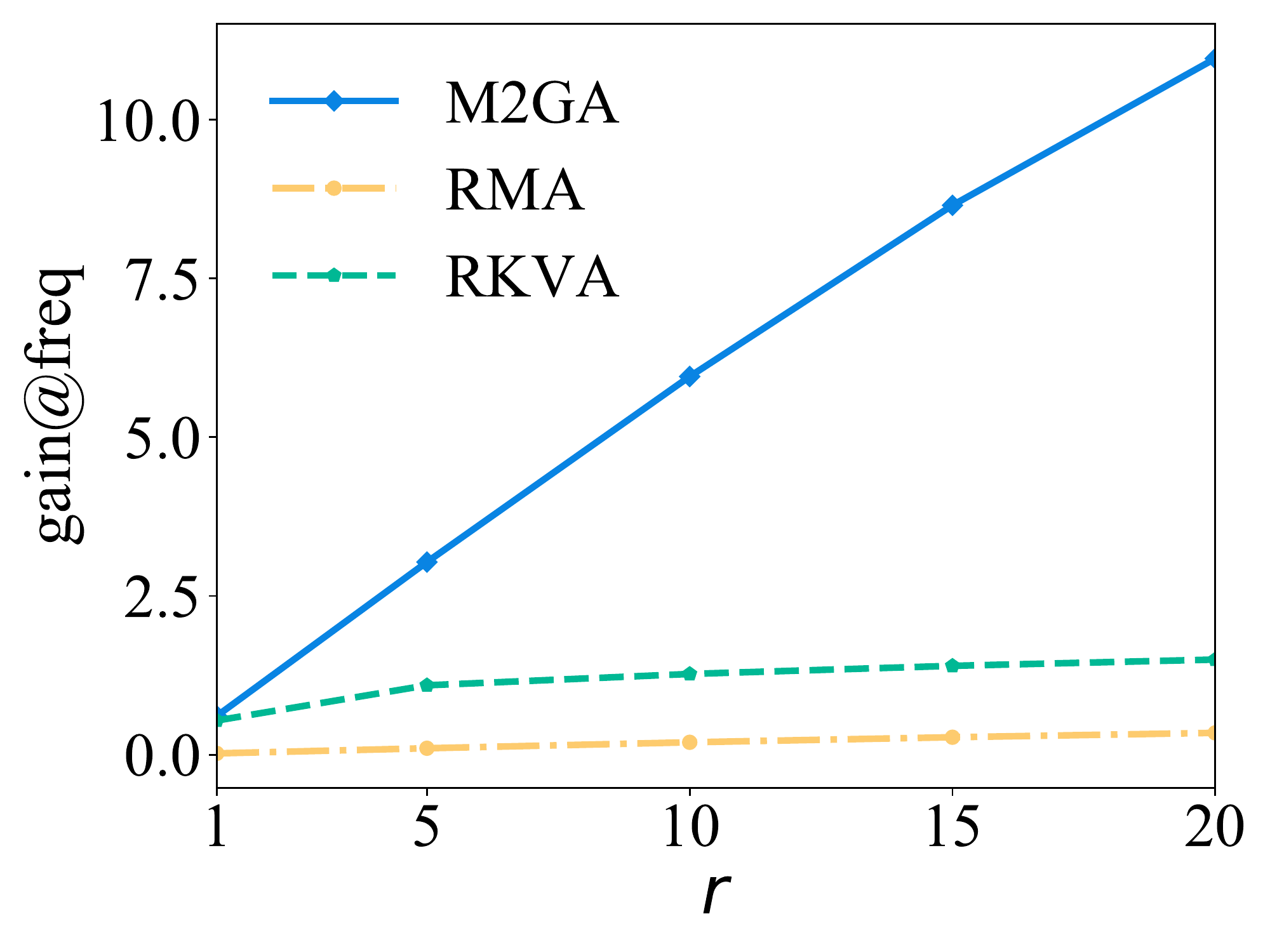}}

\caption{Impact of different parameters ($\beta, \epsilon, r$) on the  frequency gains on MovieLens-1M. The three rows are for PrivKVM, PCKV-UE, and PCKV-GRR, respectively.}
\label{fig:exp_attack_freq_movielens}
\end{figure}

\begin{figure}[tb]
    \centering 
    
     \subfloat{  \includegraphics[width=0.15\textwidth]{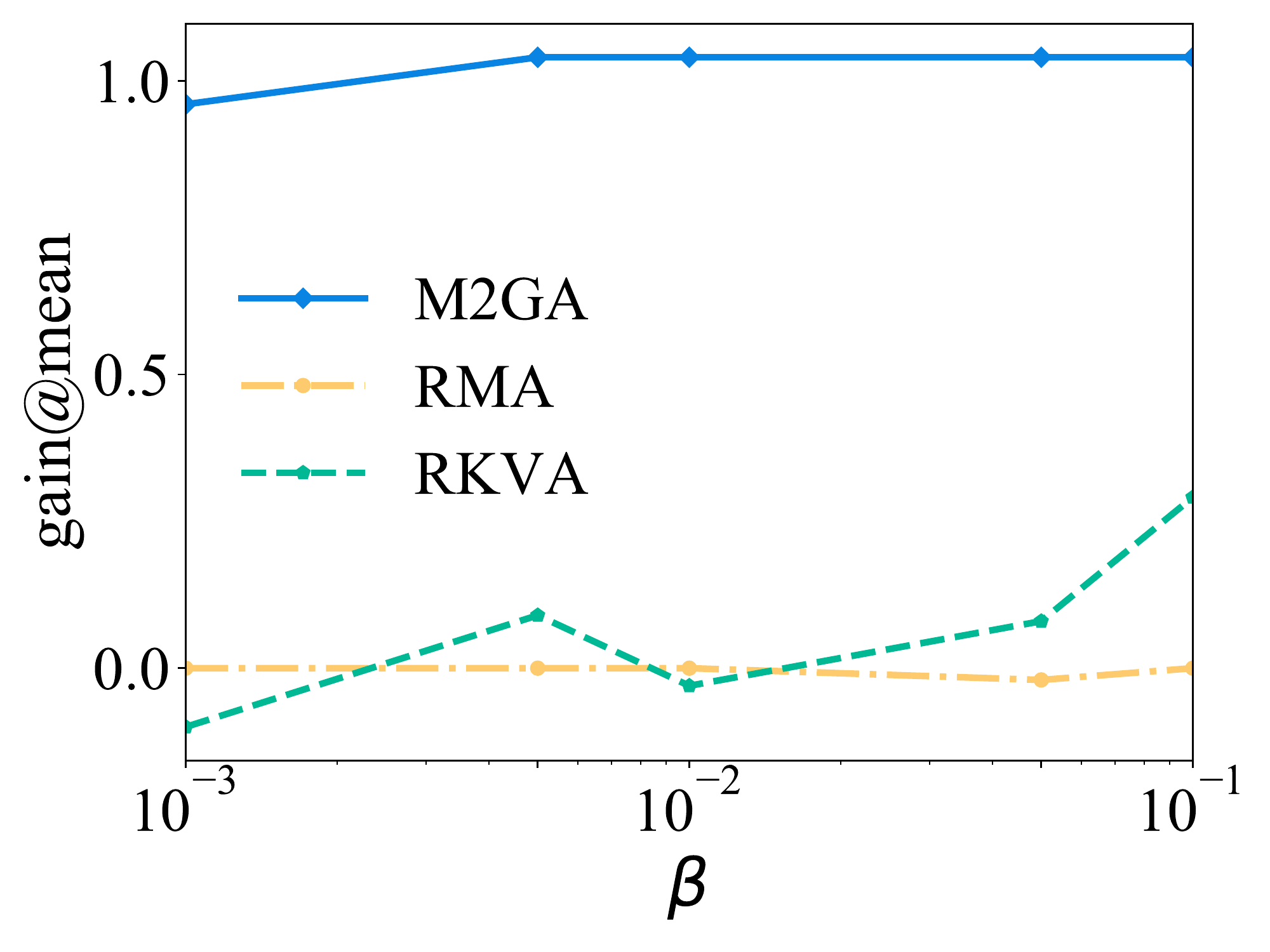}}
    \subfloat{   \includegraphics[width=0.15\textwidth]{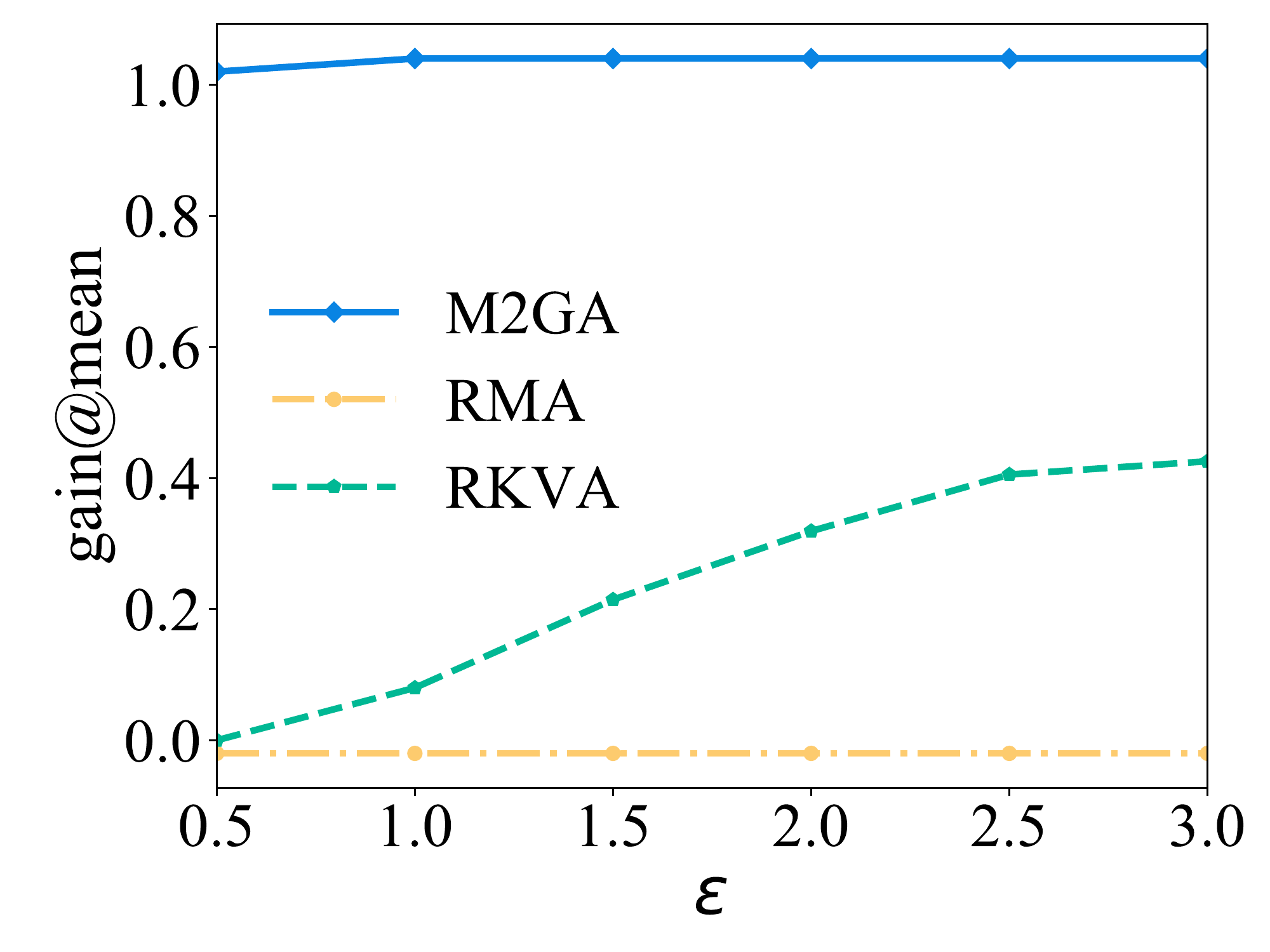}}
    \subfloat{   \includegraphics[width=0.15\textwidth]{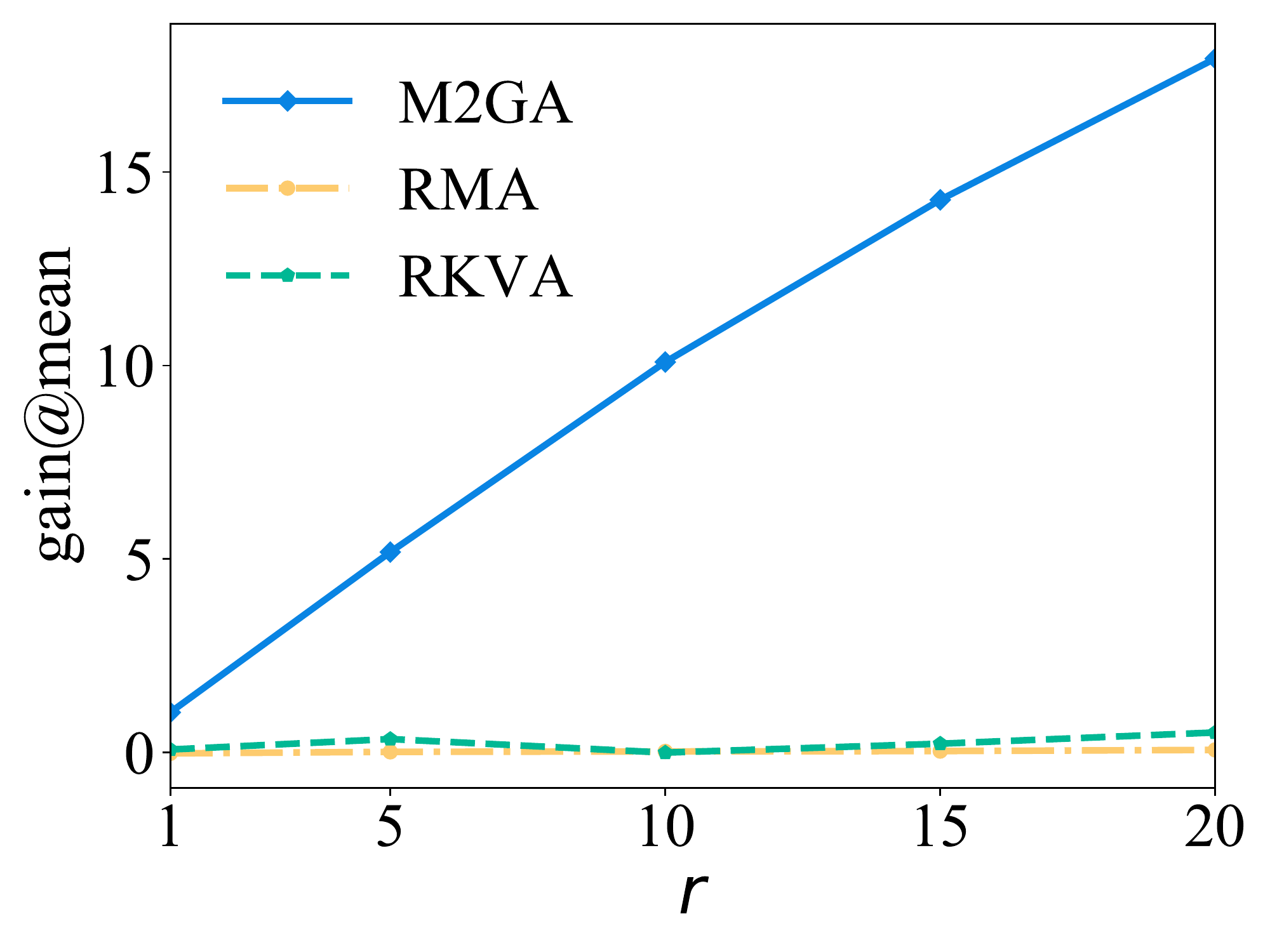}}

    \subfloat{    \includegraphics[width=0.15\textwidth]{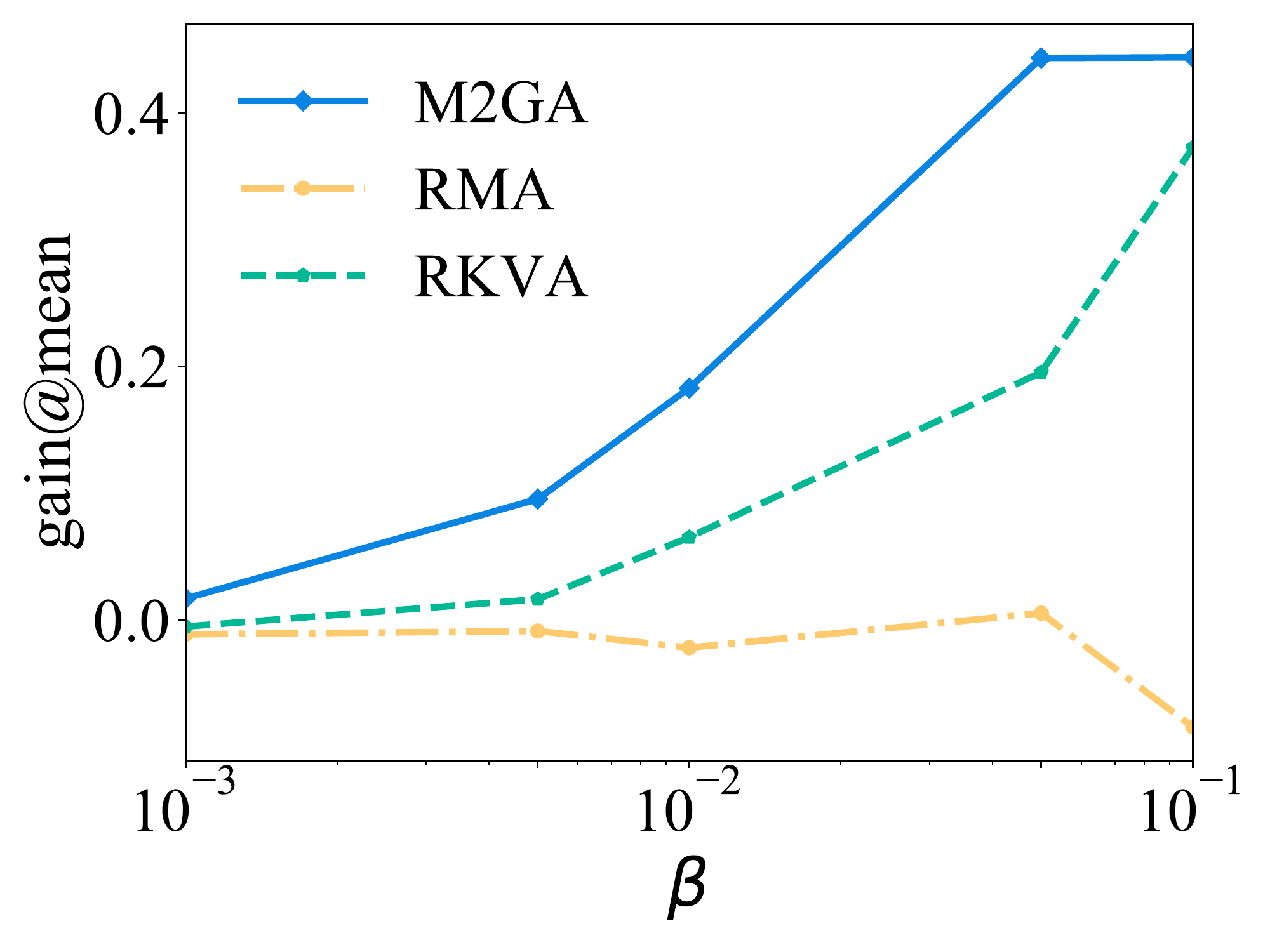}}
    \subfloat{  \includegraphics[width=0.15\textwidth]{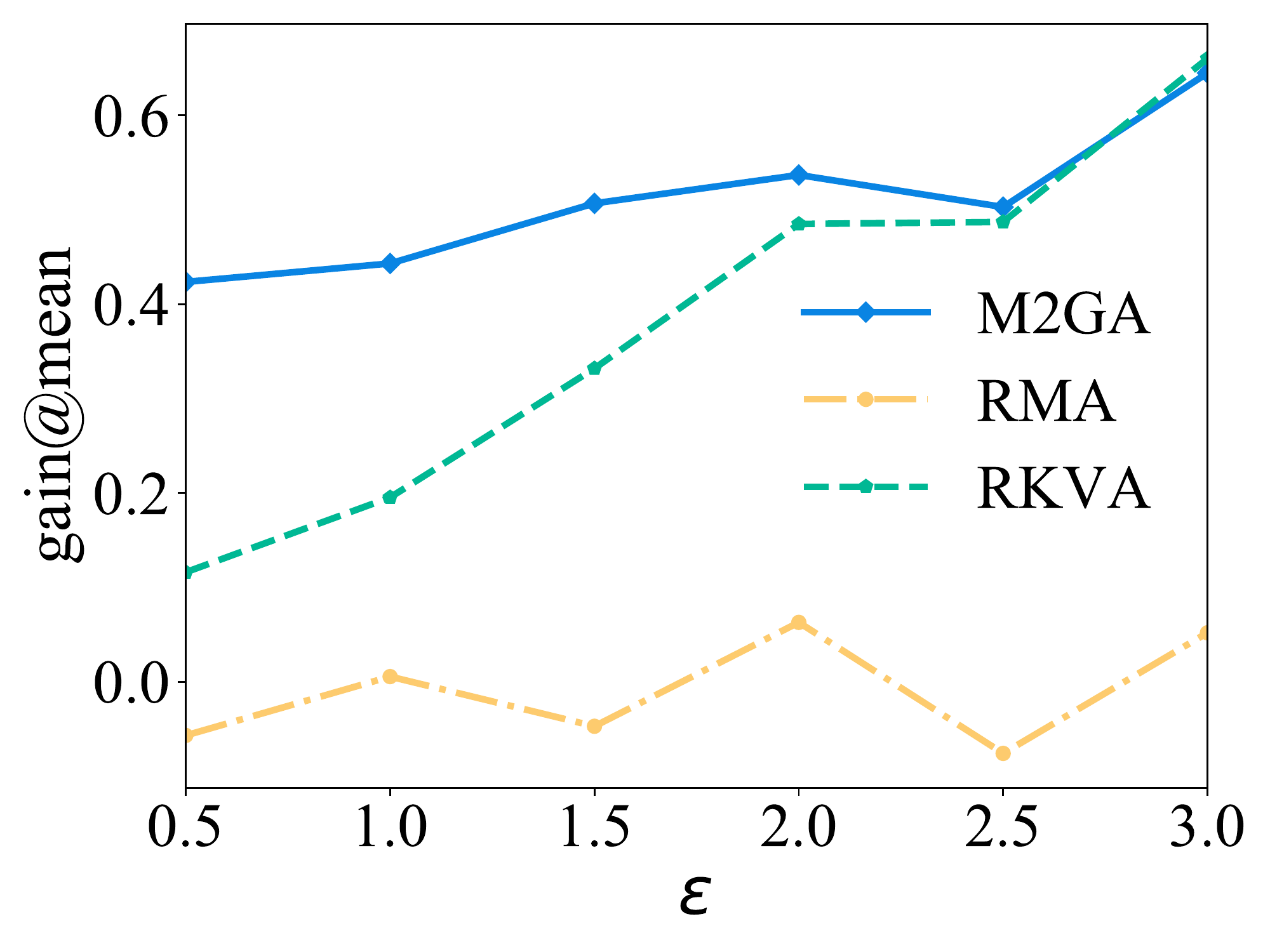}}
    \subfloat{   \includegraphics[width=0.15\textwidth]{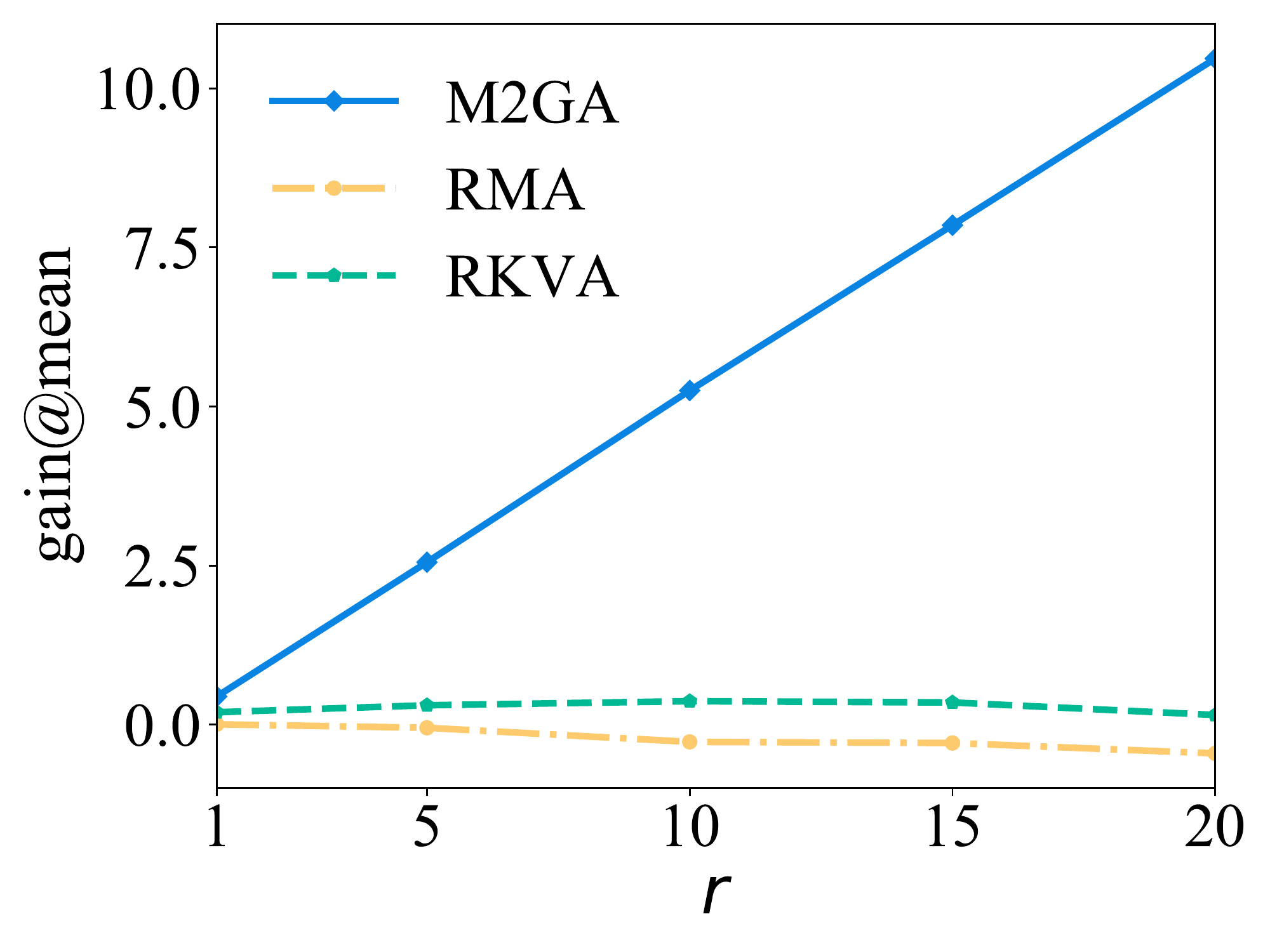}}

    \subfloat{  \includegraphics[width=0.15\textwidth]{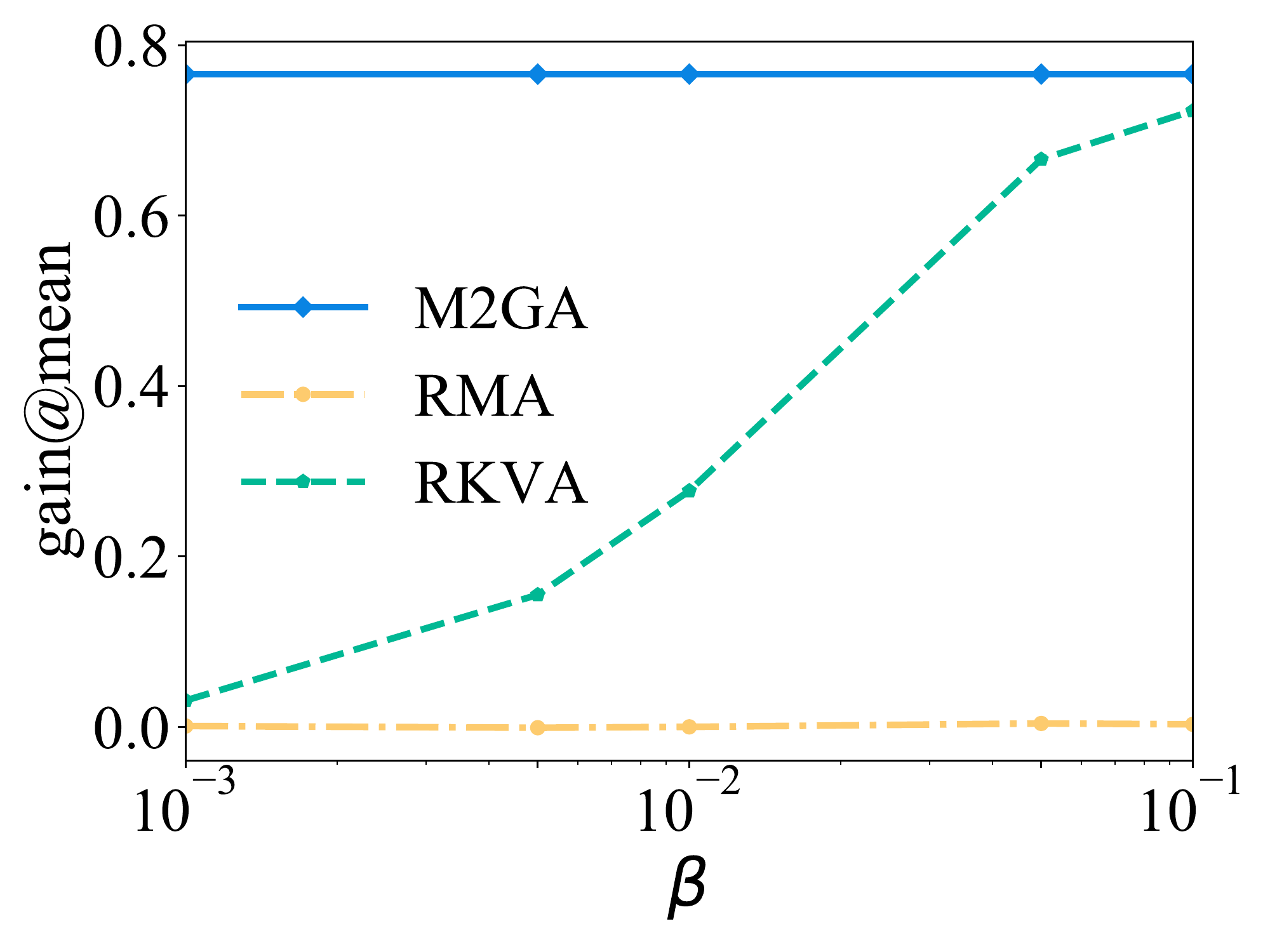}}
    \subfloat{    \includegraphics[width=0.15\textwidth]{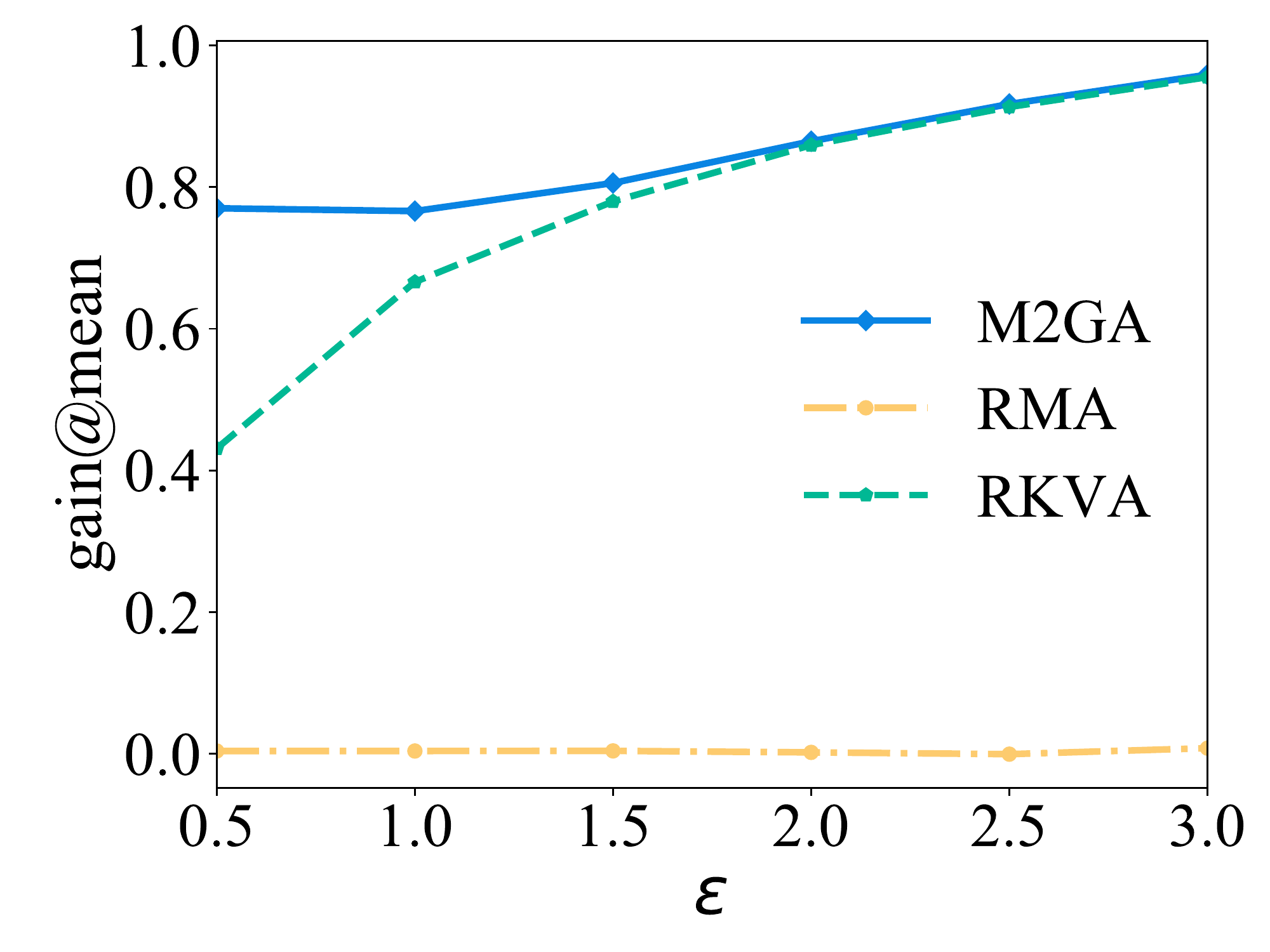}}
    \subfloat{    \includegraphics[width=0.15\textwidth]{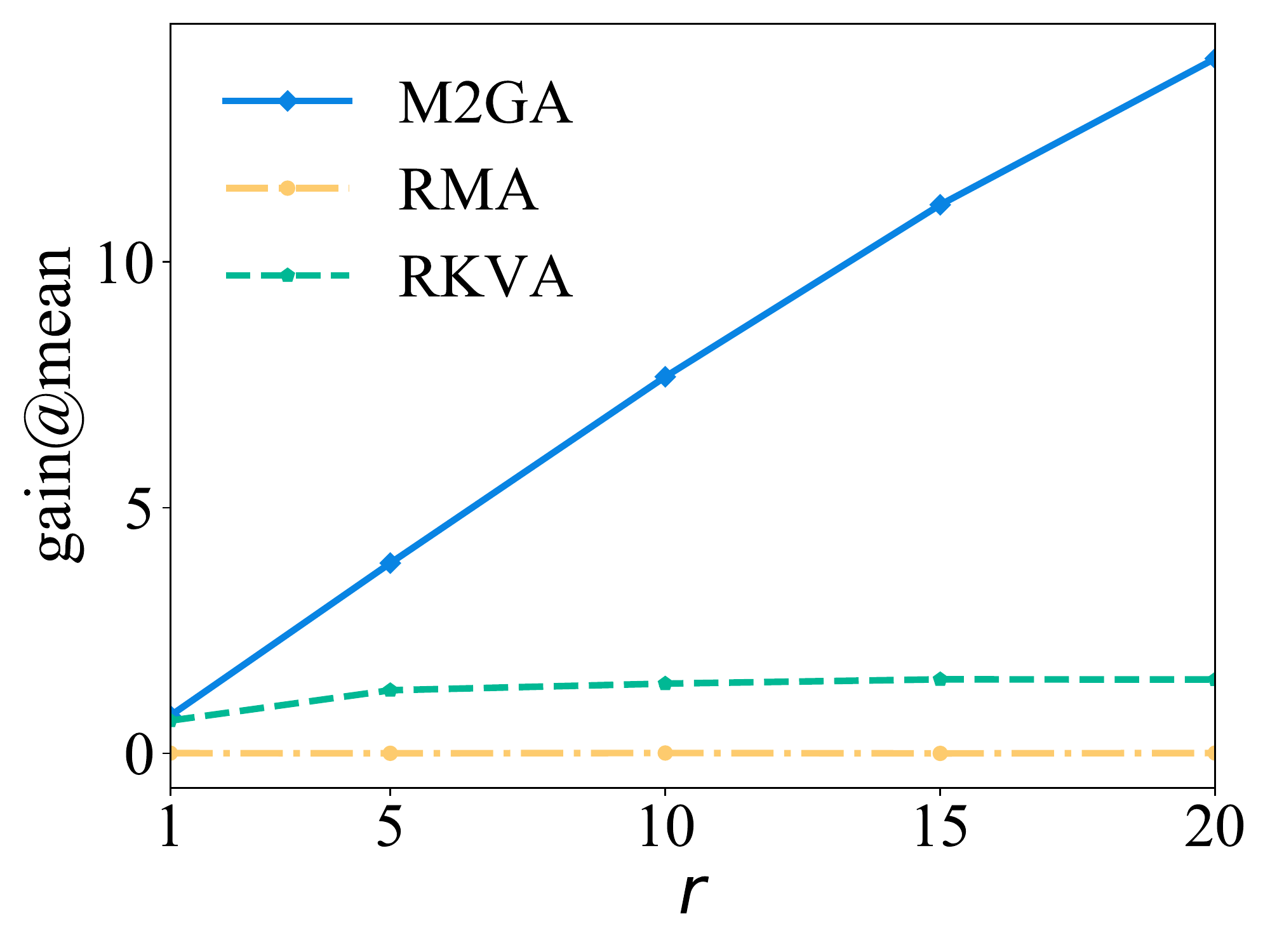}}

\caption{Impact of different parameters ($\beta, \epsilon, r$) on the  mean gains on MovieLens-1M. The three rows are for PrivKVM, PCKV-UE, and PCKV-GRR, respectively.}
\label{fig:exp_attack_mean_movielens}
\end{figure}

\begin{figure*}[!t]
    \centering 
    \subfloat{\includegraphics[width=0.24\textwidth]{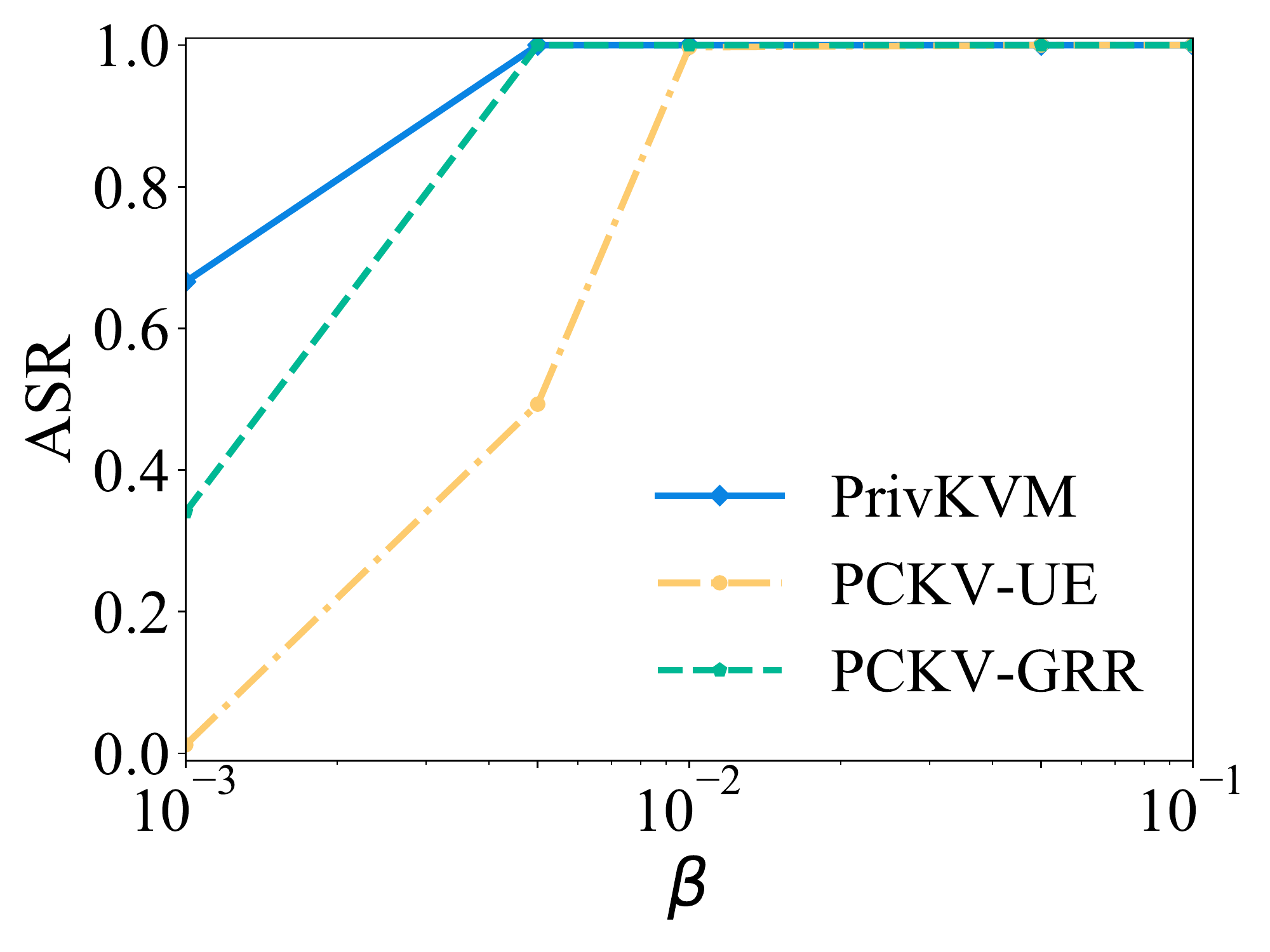}}
    \subfloat{\includegraphics[width=0.24\textwidth]{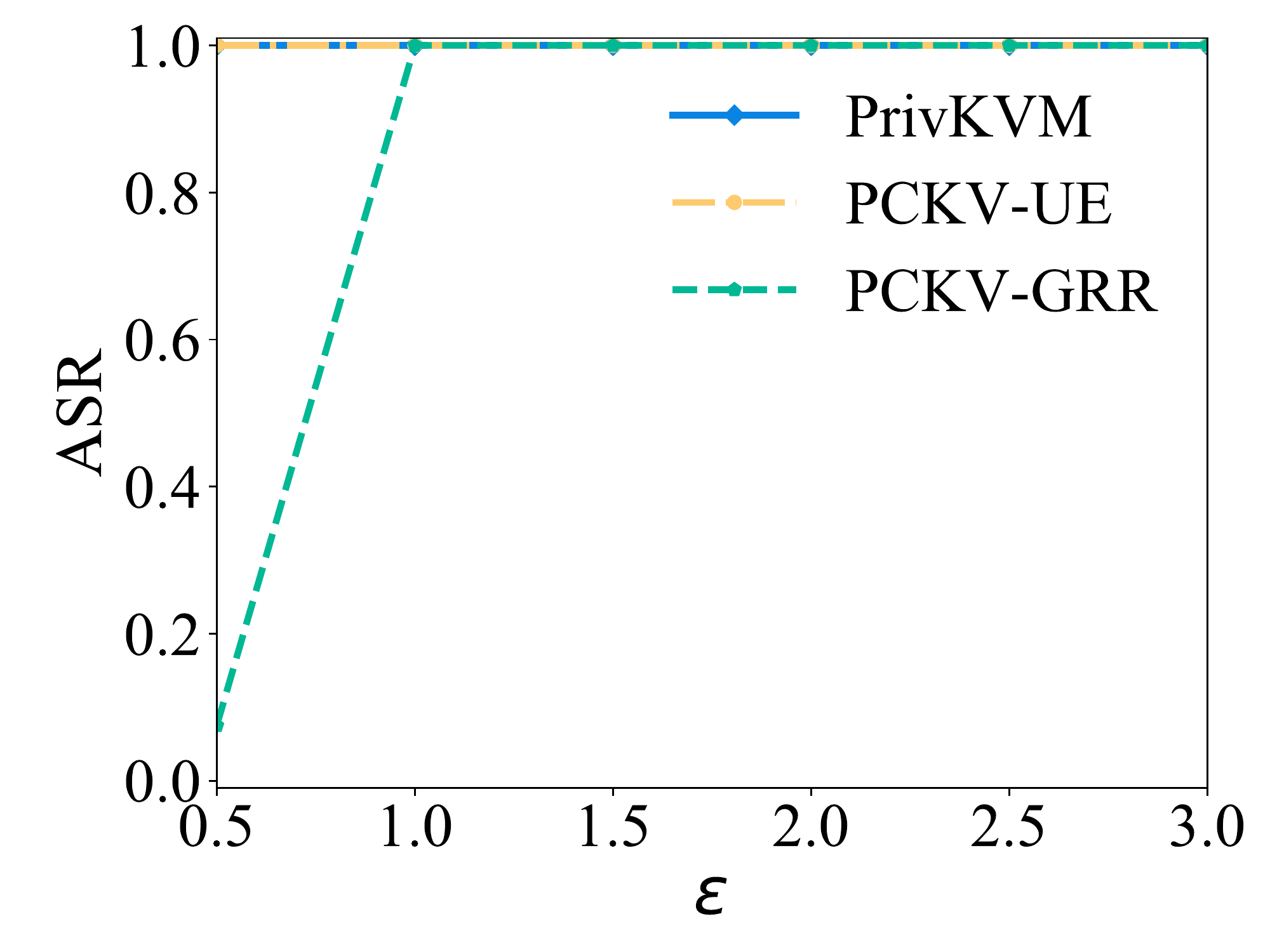}}
    \subfloat{\includegraphics[width=0.24\textwidth]{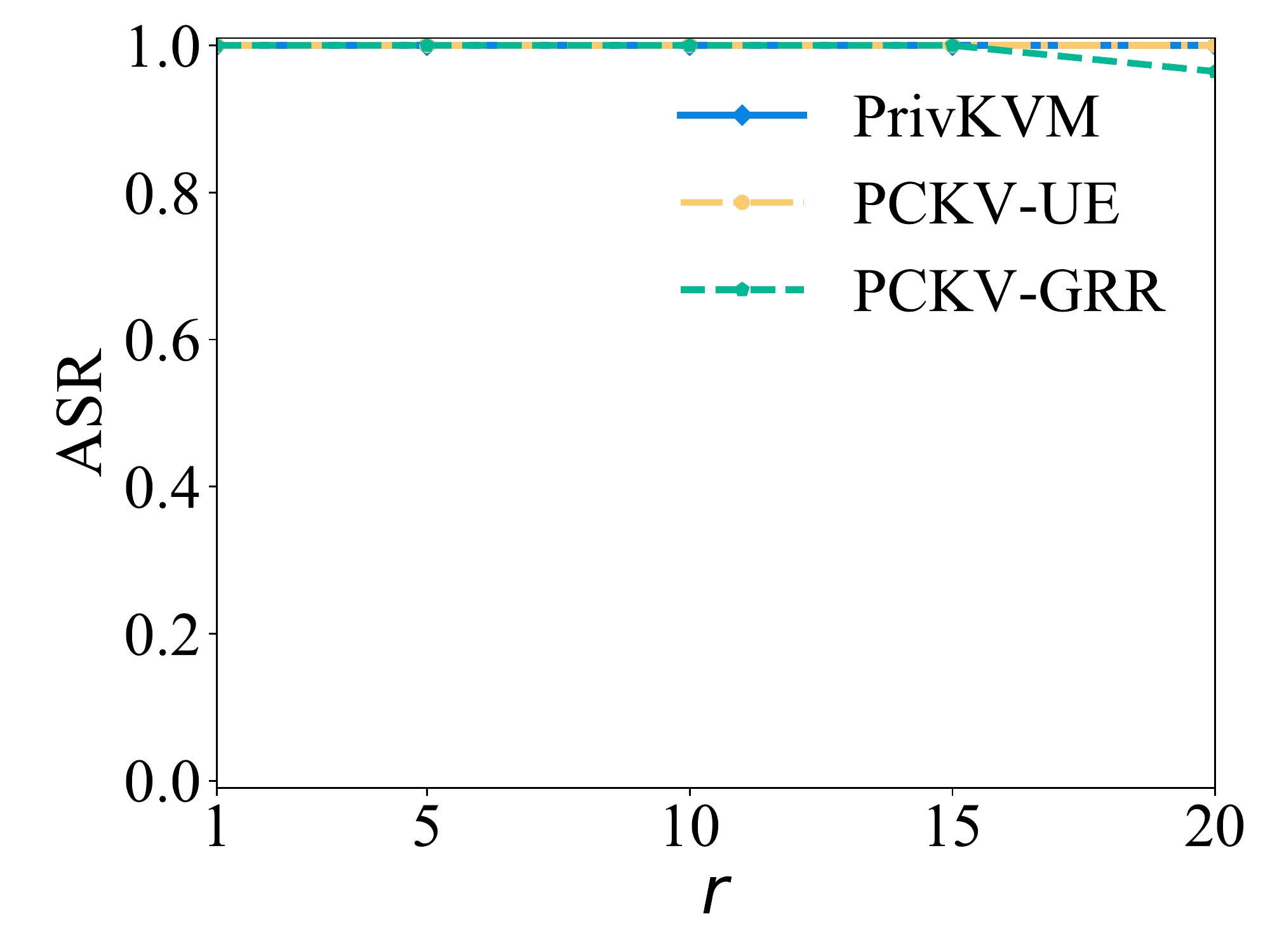}}
    \subfloat{\includegraphics[width=0.24\textwidth]{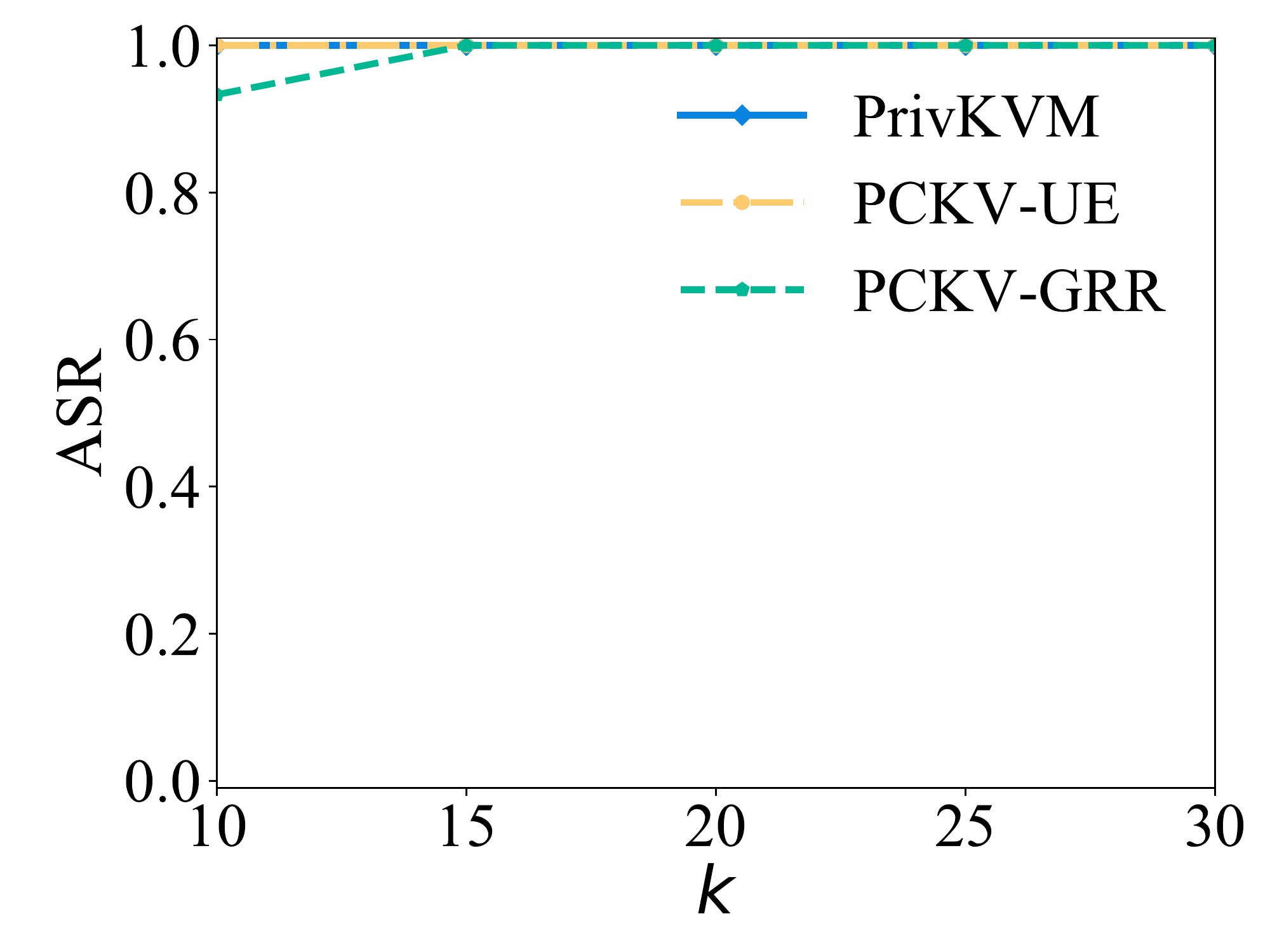}}
    \vspace{2mm}
    
    \subfloat{\includegraphics[width=0.24\textwidth]{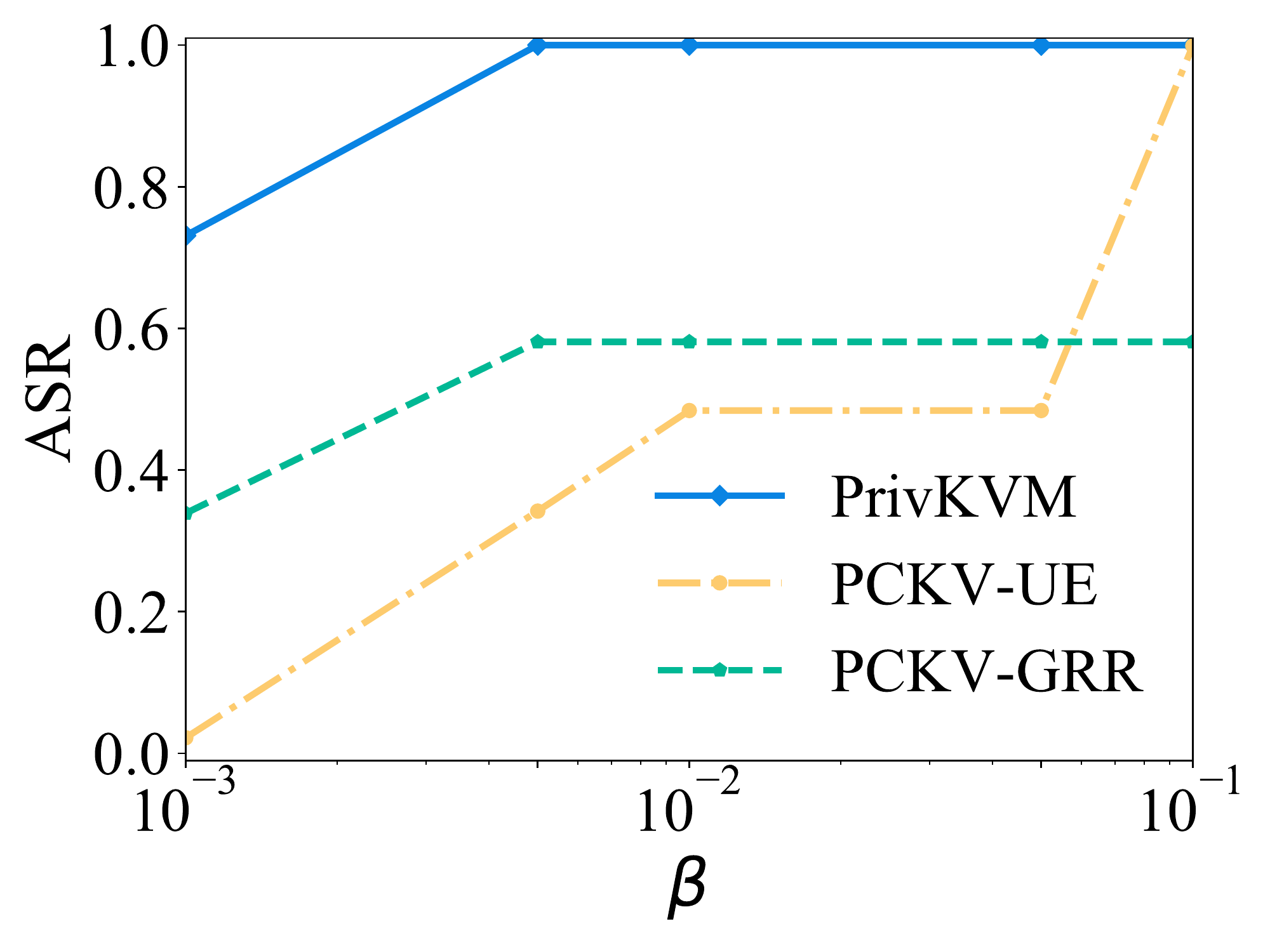}}
    \subfloat{\includegraphics[width=0.24\textwidth]{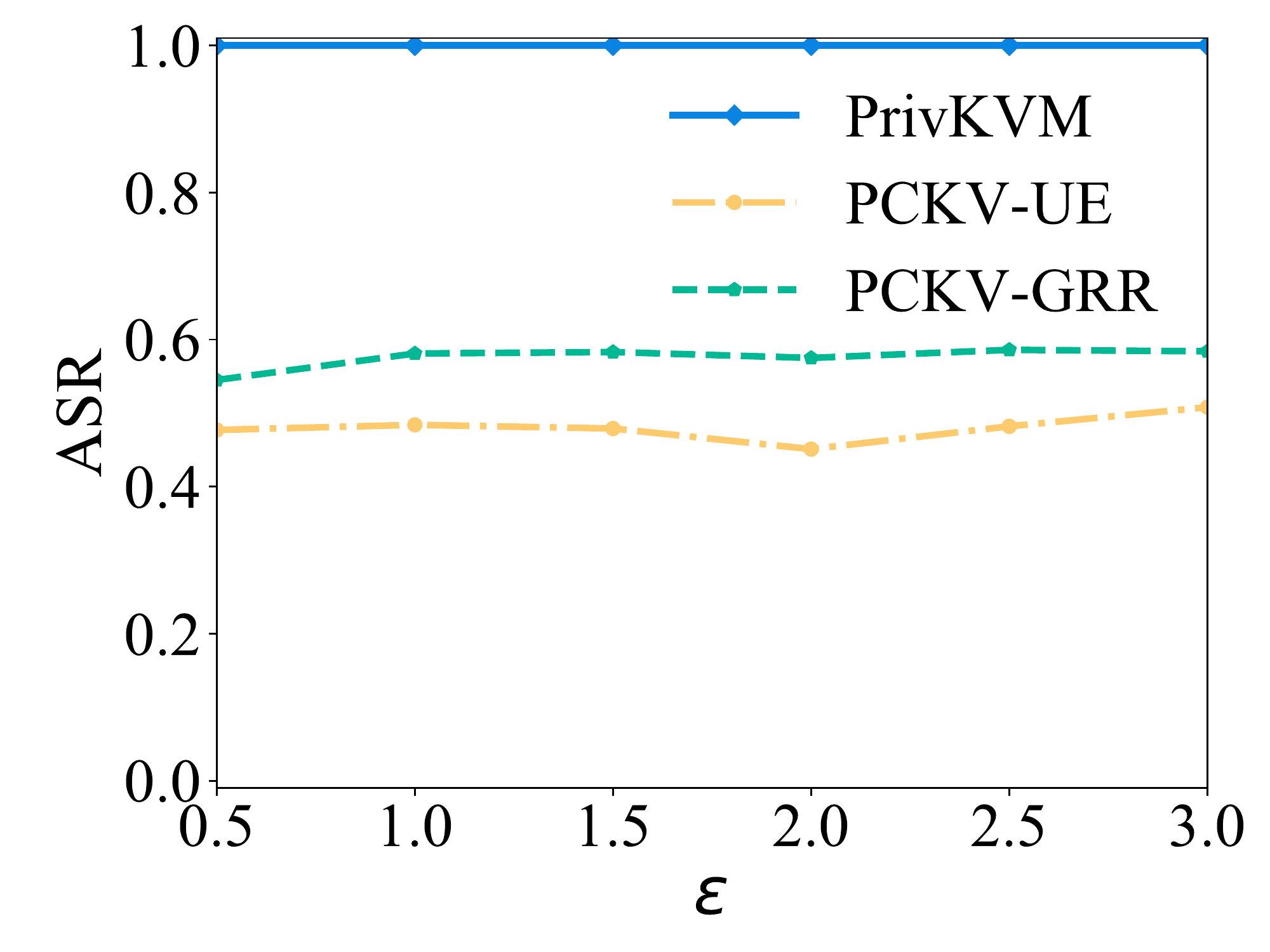}}
    \subfloat{\includegraphics[width=0.24\textwidth]{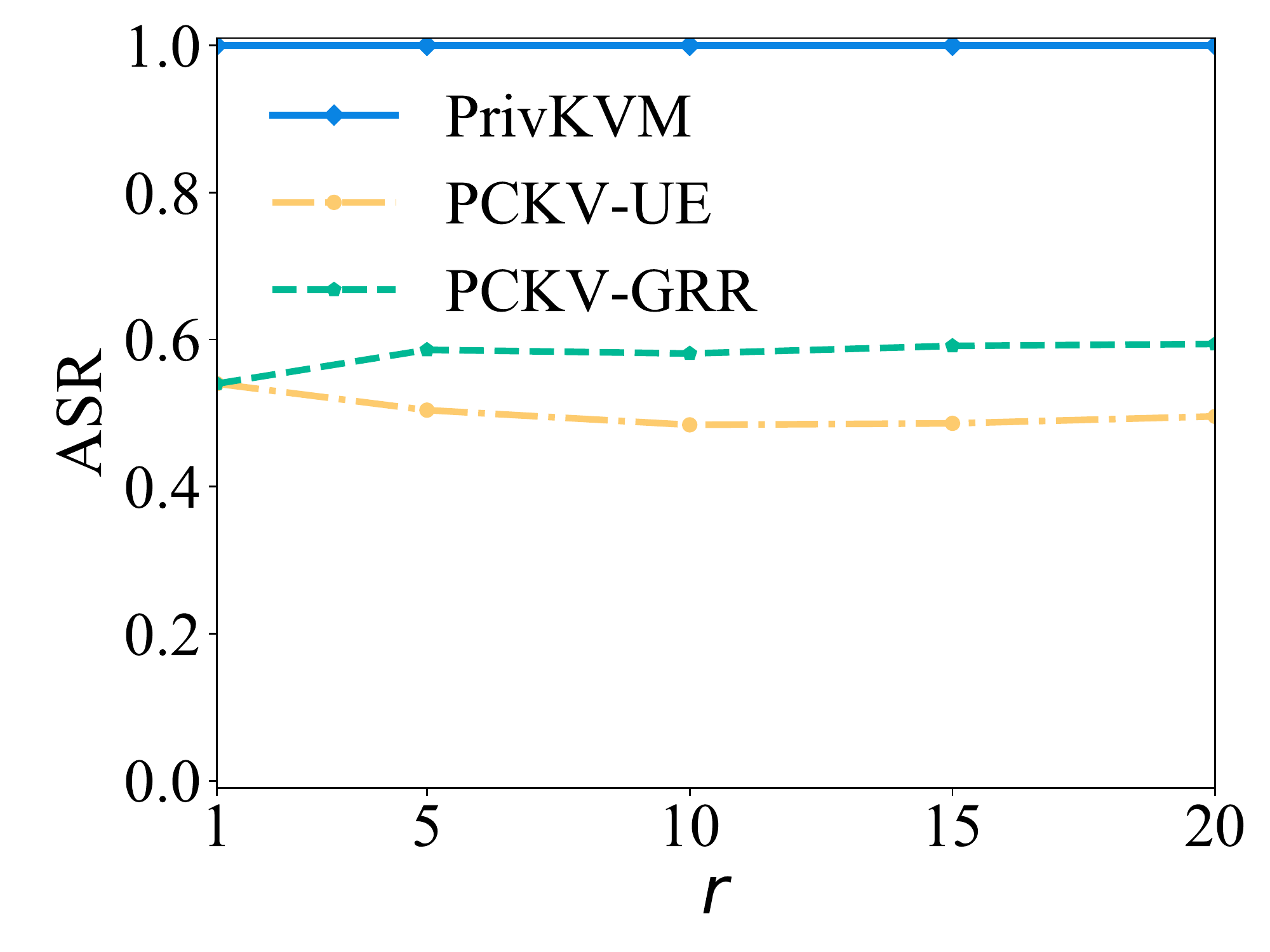}}
    \subfloat{\includegraphics[width=0.24\textwidth]{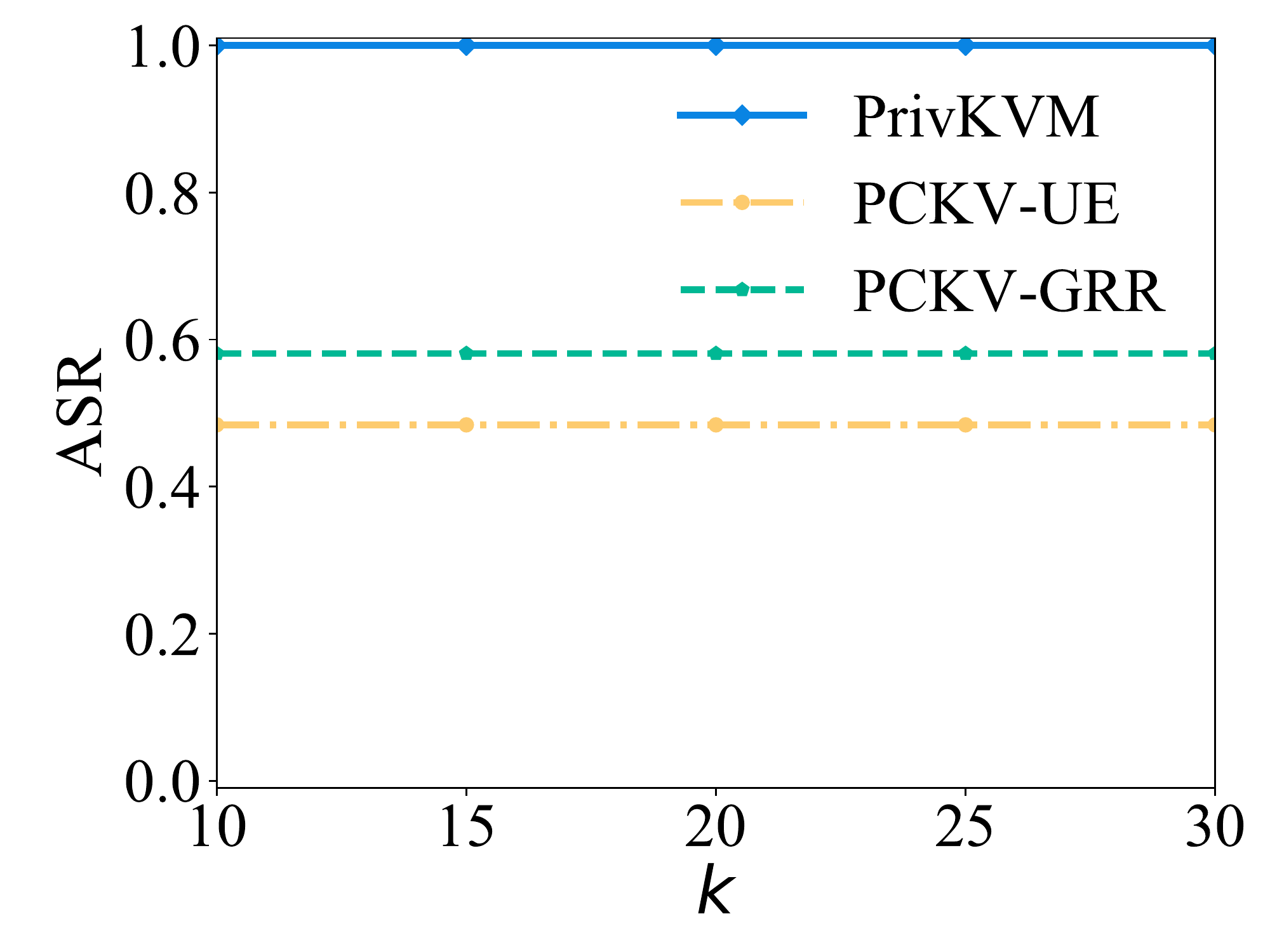}}
     \vspace{2mm}
     
    \subfloat{\includegraphics[width=0.24\textwidth]{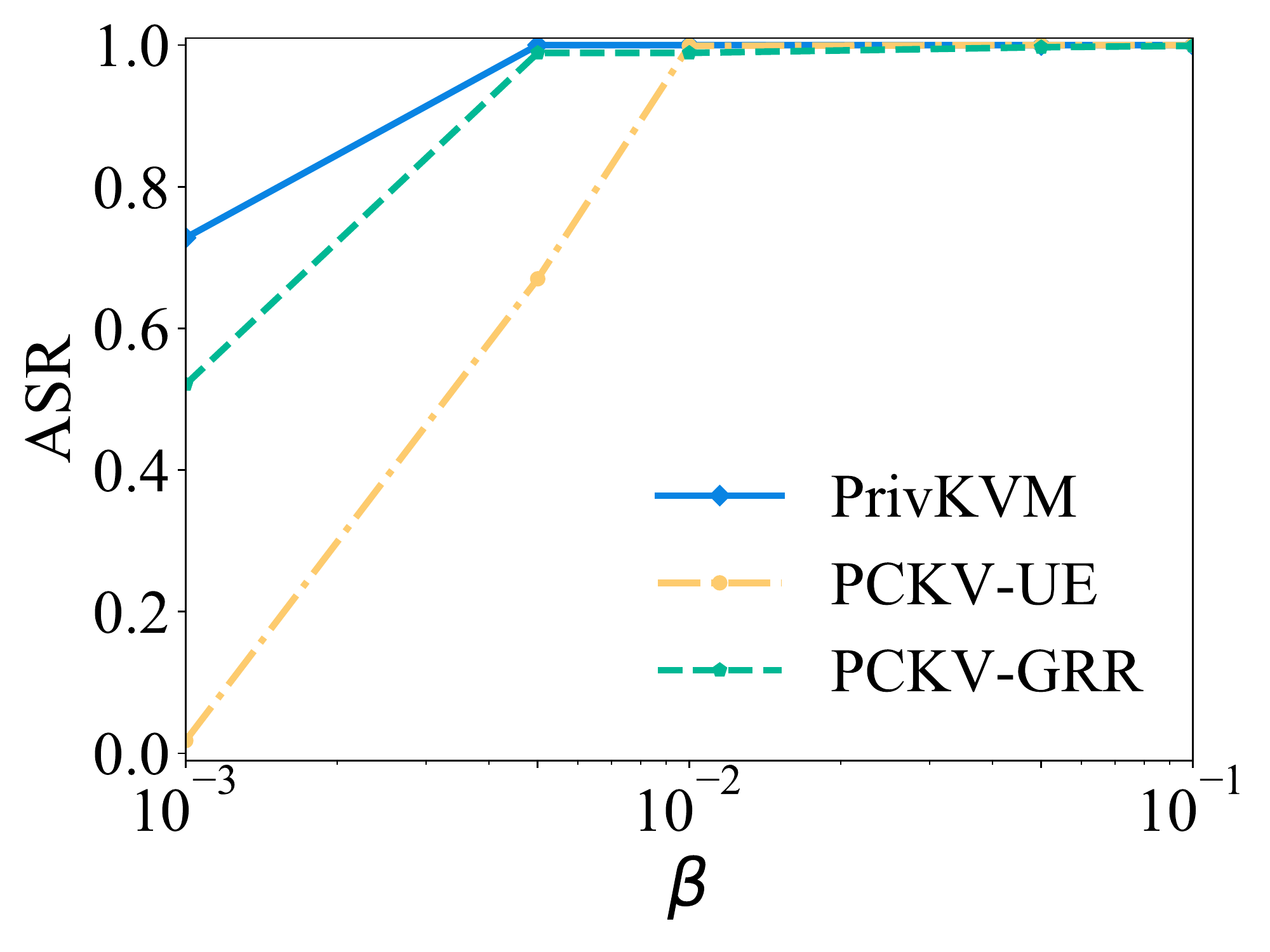}}
    \subfloat{\includegraphics[width=0.24\textwidth]{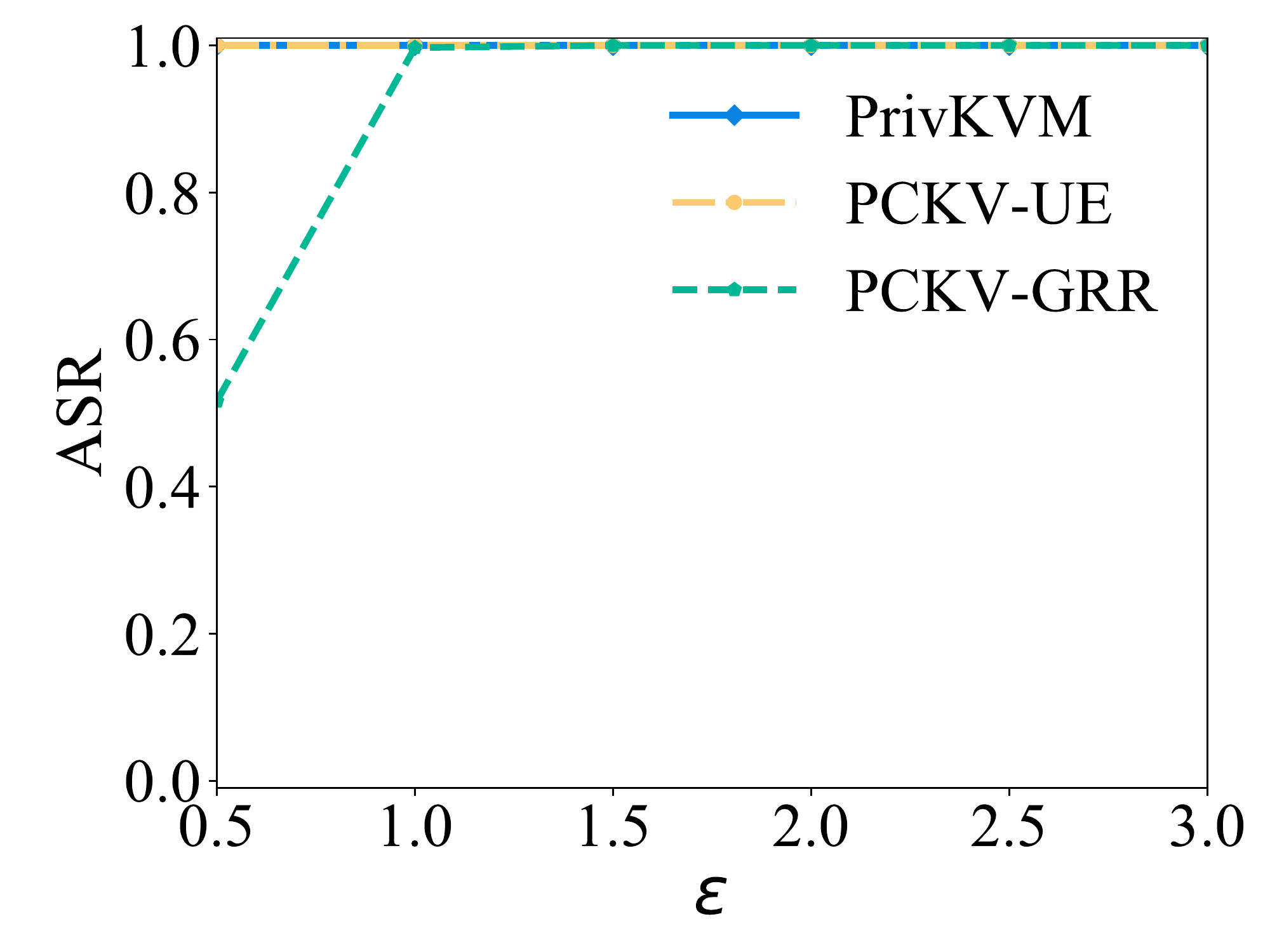}}
    \subfloat{\includegraphics[width=0.24\textwidth]{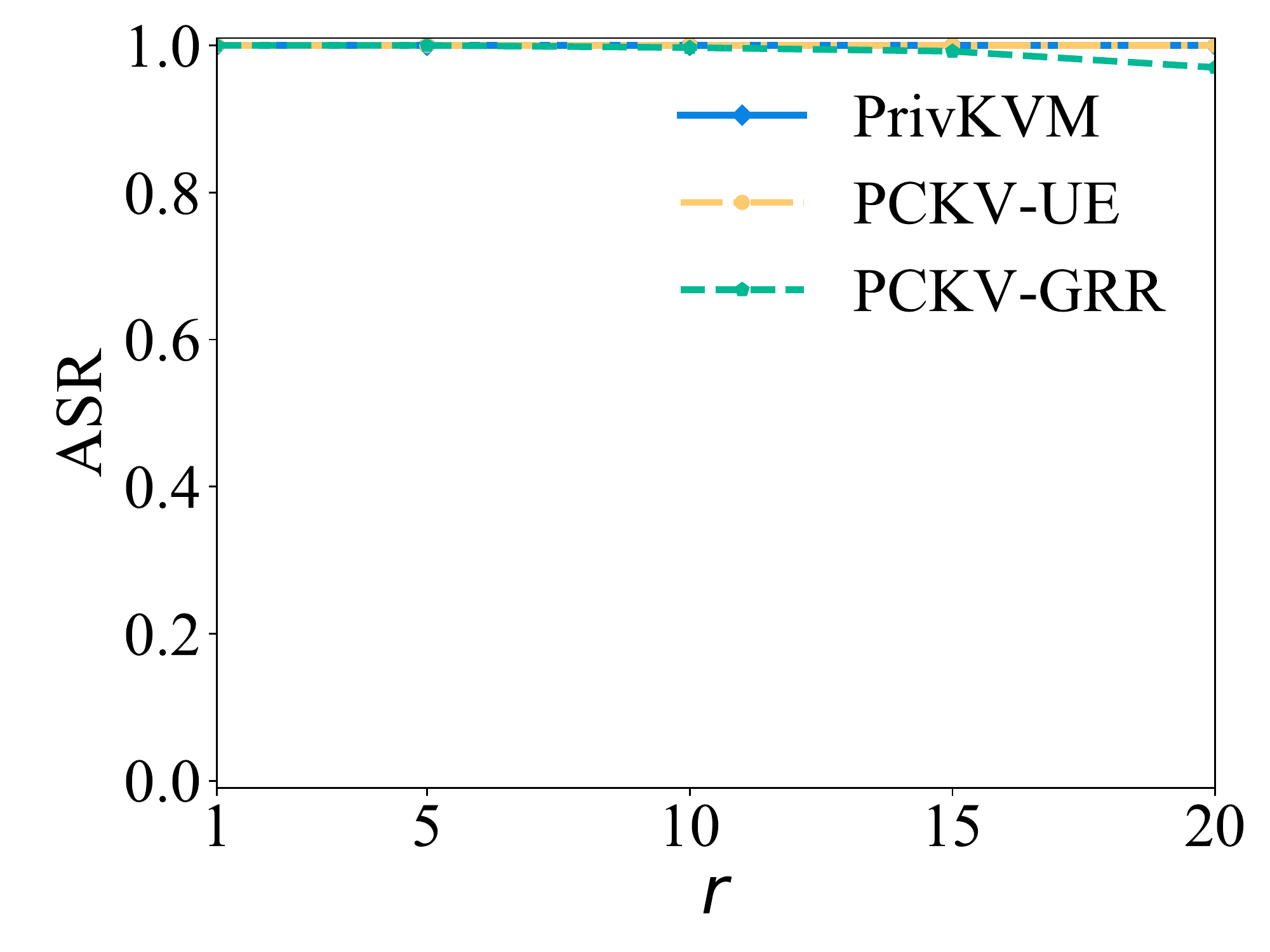}}
    \subfloat{\includegraphics[width=0.24\textwidth]{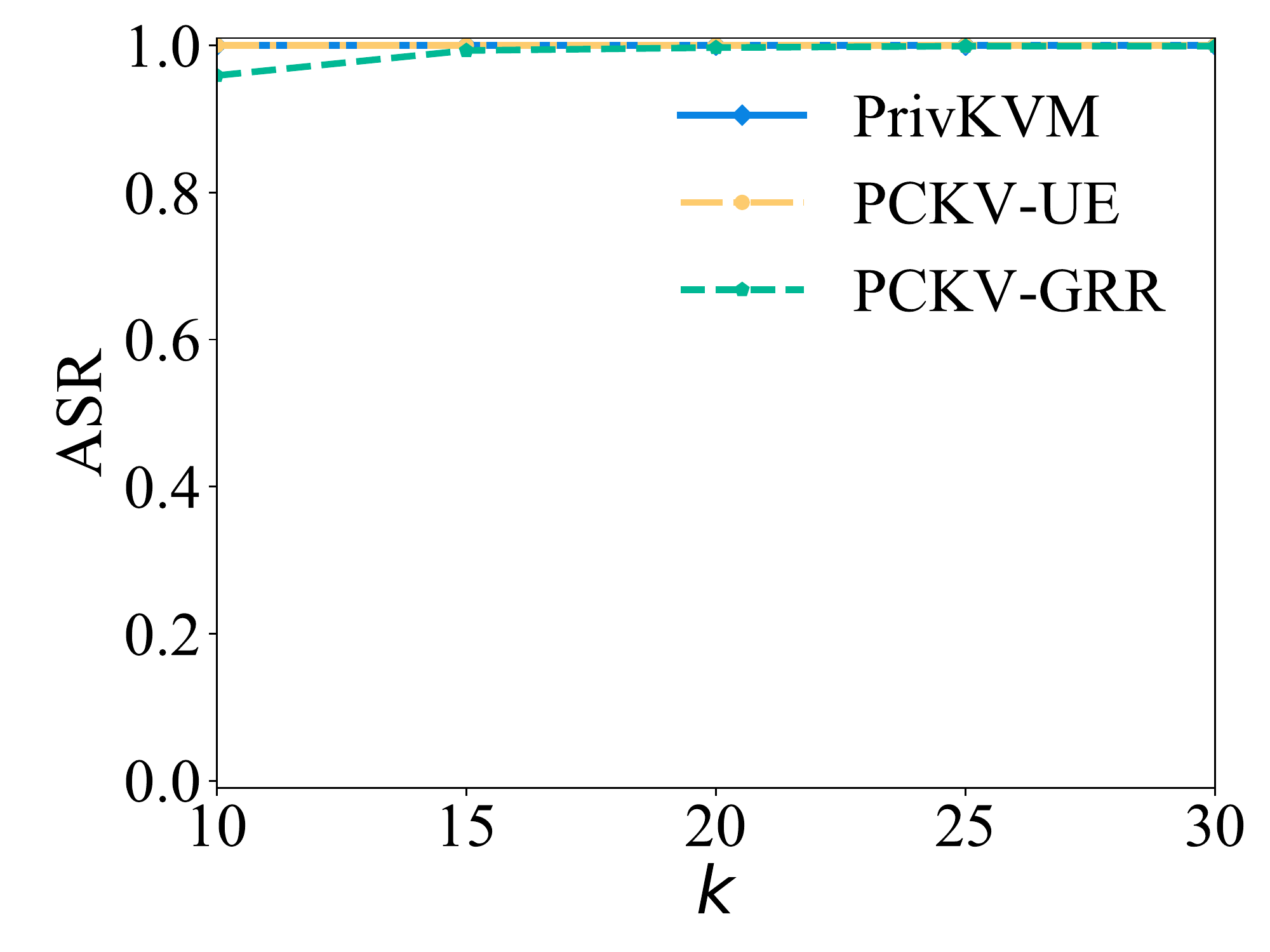}}
\caption{Impact of $\beta$, $r$, $\epsilon$, and $k$ on ASR of M2GA towards recommender systems (first row: Case 1, second row: Case 2, and third row: Case 3). Three LDP protocols and Clothing dataset are used. }
\label{fig:asr_case3}
\end{figure*}

\begin{figure}[!tb]
    \centering 
    \begin{subfigure}{0.4\columnwidth}
      \includegraphics[width=\linewidth]{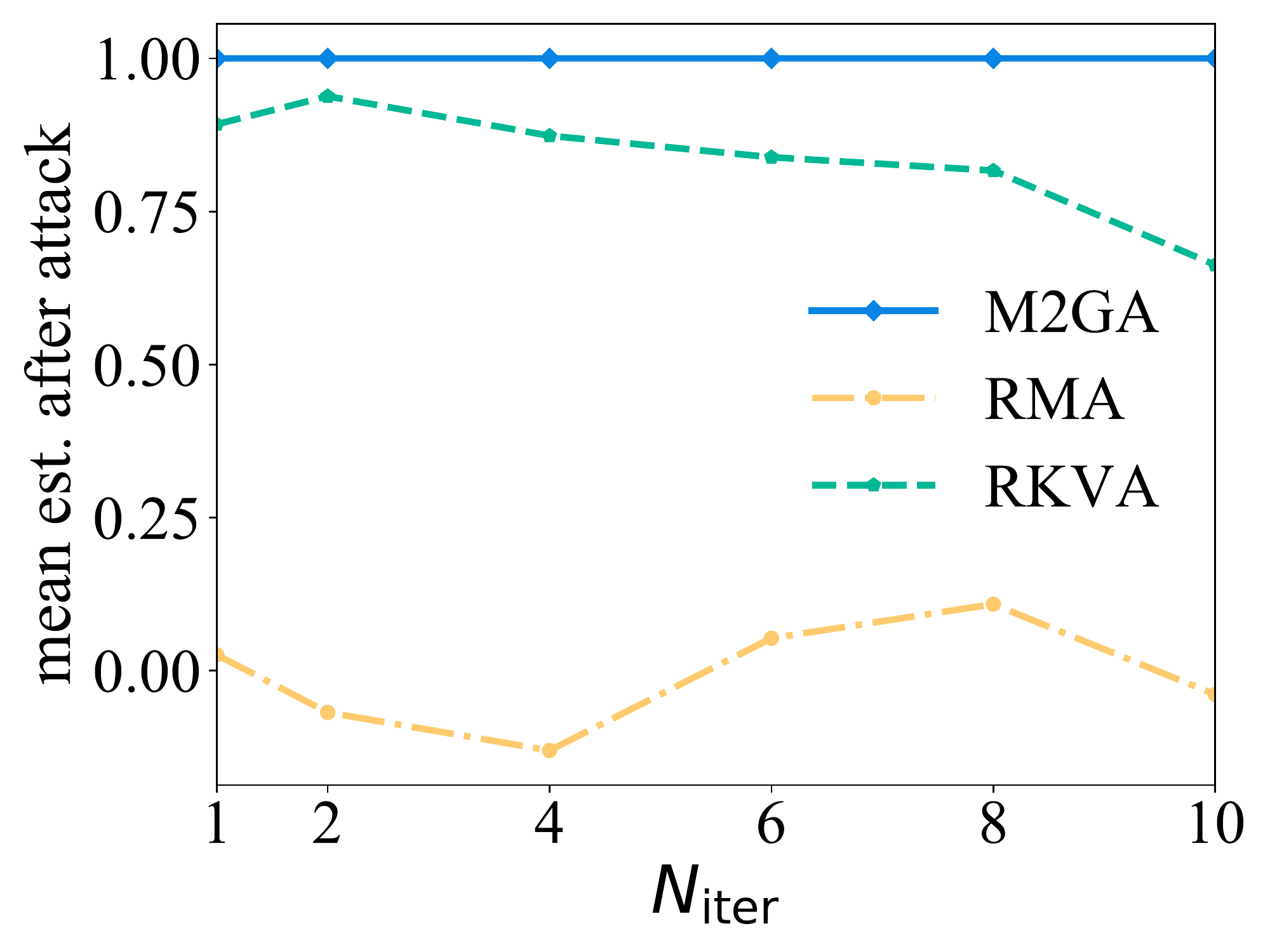}
      \caption{Synthetic}
    \end{subfigure} \hfil
    \begin{subfigure}{0.4\columnwidth}
      \includegraphics[width=\linewidth]{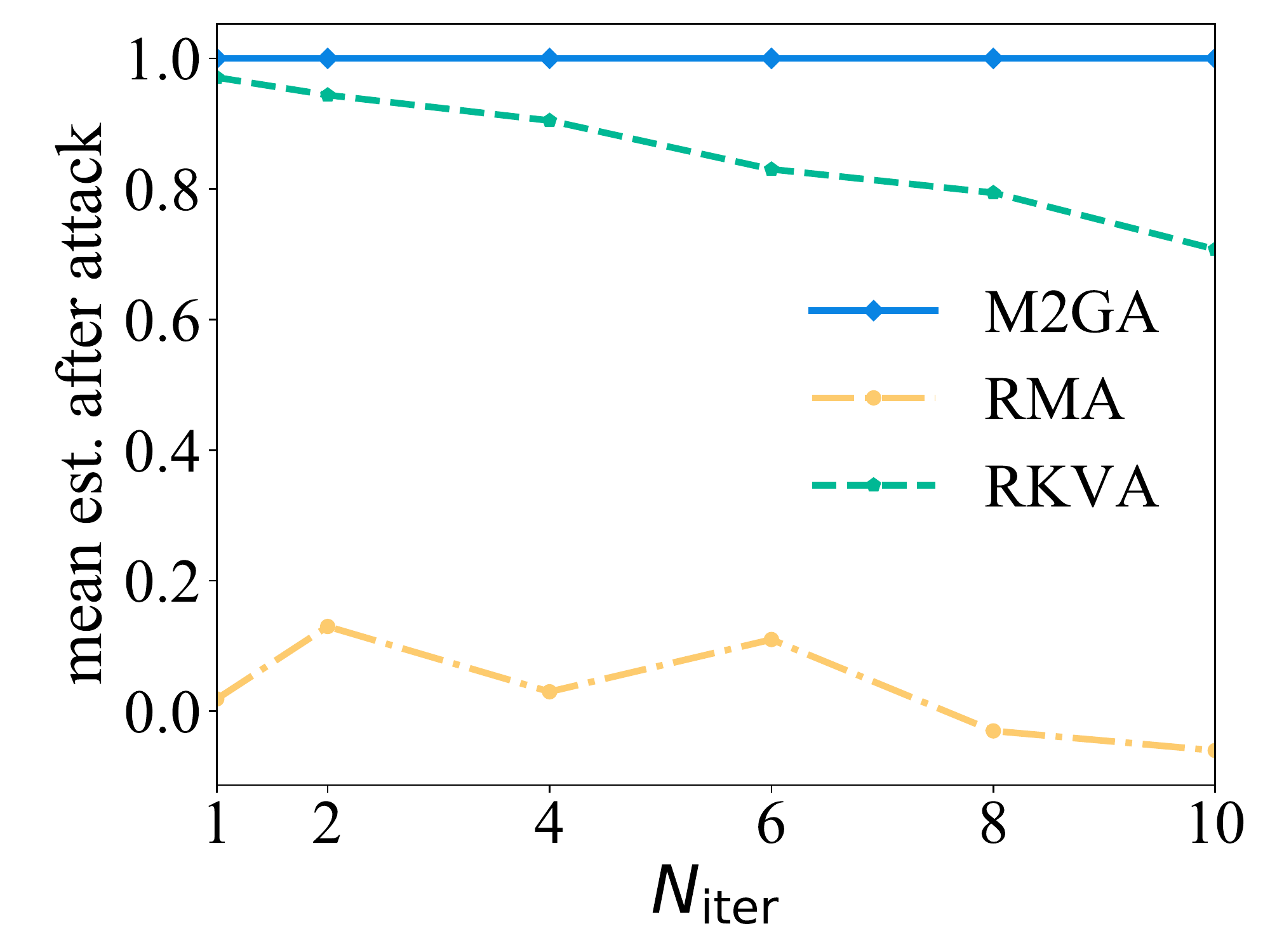}
      \caption{Clothing}
    \end{subfigure} \hfil
    \begin{subfigure}{0.4\columnwidth}
      \includegraphics[width=\linewidth]{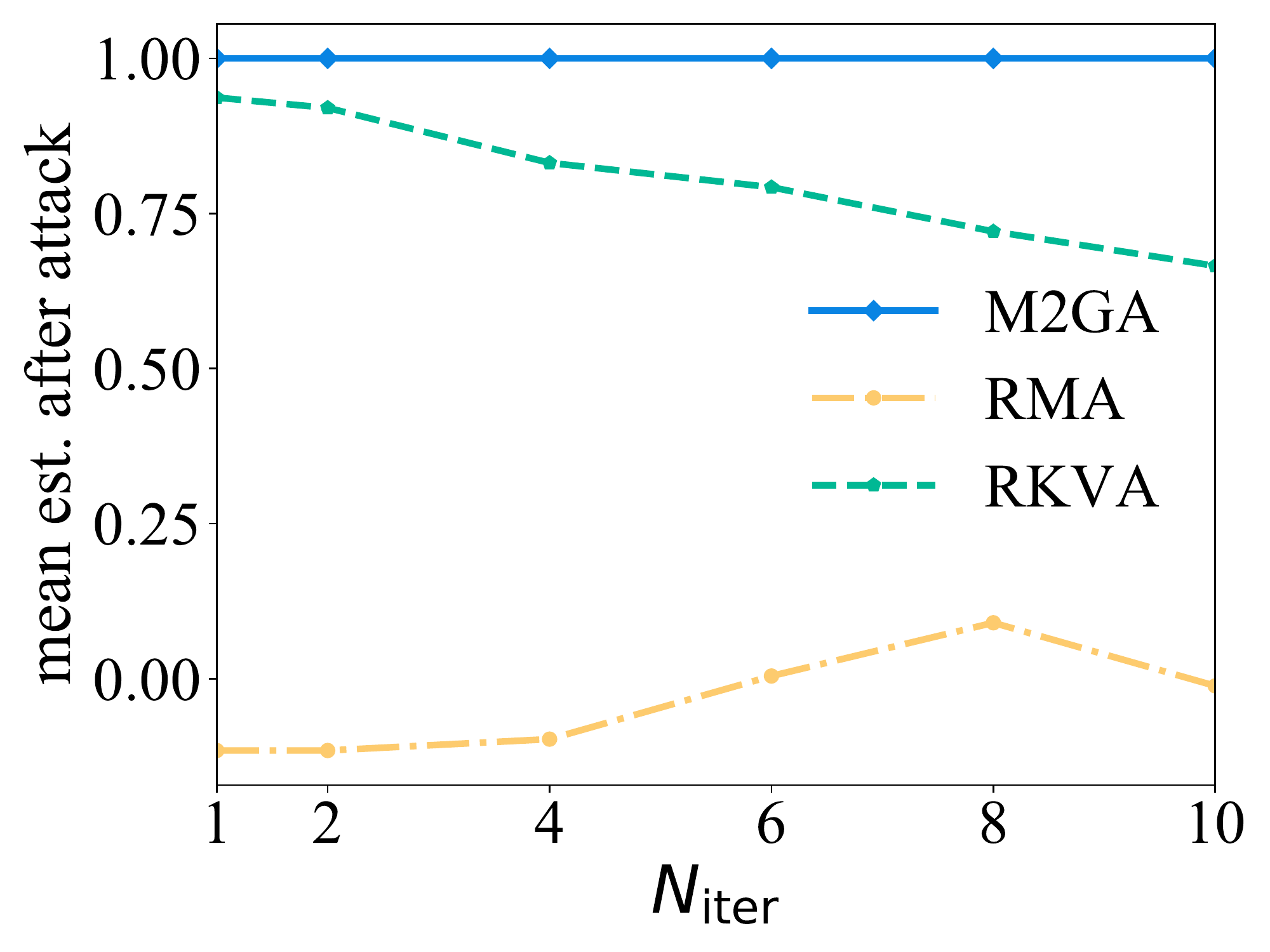}
      \caption{TalkingData}
    \end{subfigure} \hfil
        \begin{subfigure}{0.4\columnwidth}
      \includegraphics[width=\linewidth]{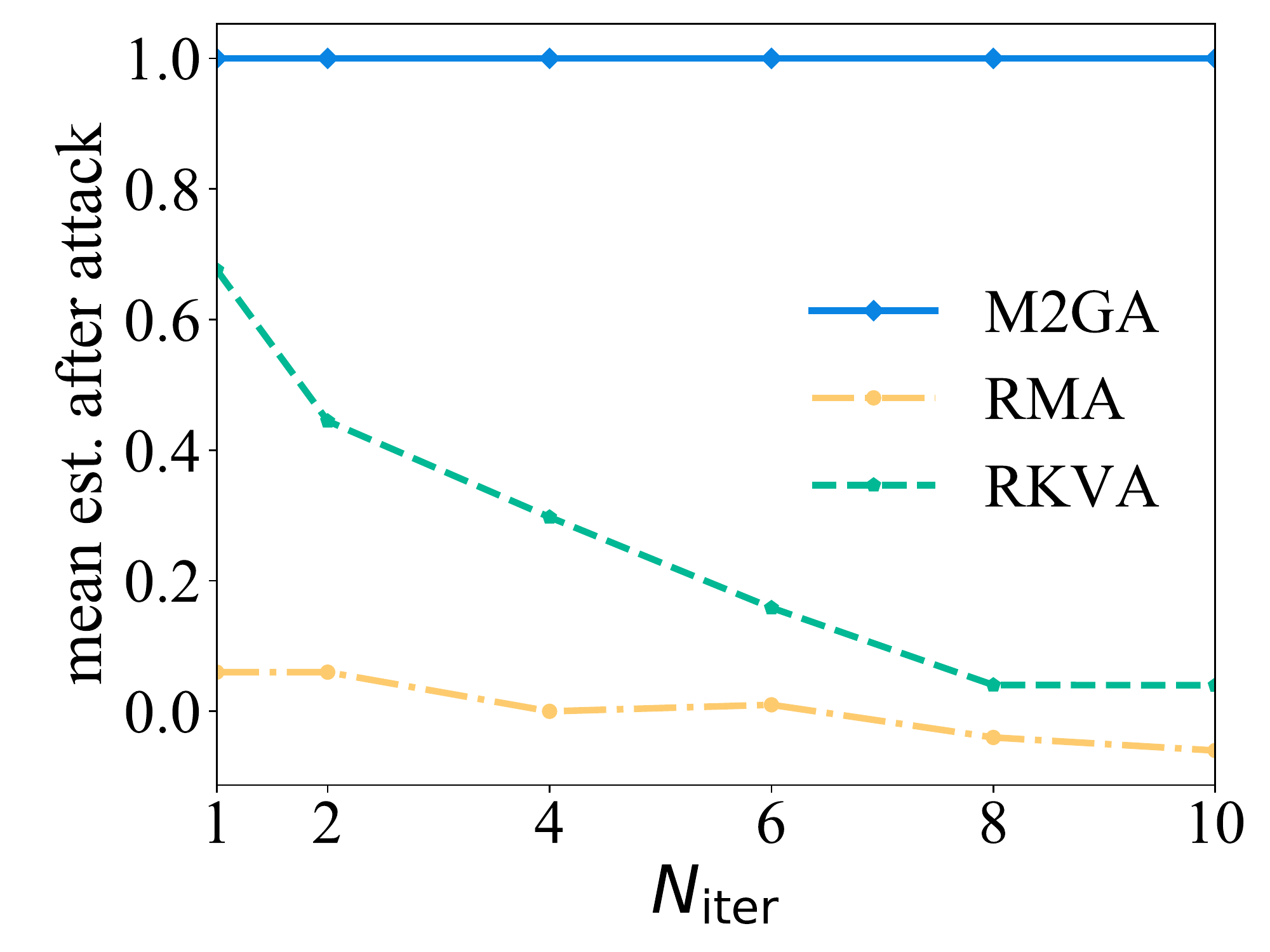}
      \caption{MovieLens-1M}
    \end{subfigure}
\caption{Impact of $N_{\text{iter}}$ on the estimated mean value after attack for PrivKVM on the four datasets.}
\label{fig:exp_privkvm_iterations_1}
\end{figure}

\subsubsection{Parameter Settings}
The parameters involved are $\beta$ (the fraction of fake users), $\epsilon$ (the privacy budget), and $r$ (the number of target keys). PrivKVM further involves  $N_{\text{iter}}$ (the number of rounds), while PCKV-UE and PCKV-GRR further involve $\ell$ (the padding length).  
Unless otherwise mentioned, we set the default values of these parameters as follows: $\beta=0.05$, $\epsilon=1.0$, $r=1$, $N_{\text{iter}}=10$, $\ell=1$ for Synthetic, $\ell=2$ for Clothing, $\ell=20$ for TalkingData, and $\ell=100$ for MovieLens-1M. We set $\ell$ differently for different datasets to  consider their different characteristics, which is suggested by~\cite{gu2020pckv}. 
We set $r=10$ and $t=20$ by default when evaluating our attacks to the recommender system downstream application. 
We randomly sample $r$ keys from the entire dictionary as the target keys for each dataset. We vary one  parameter while keeping the others fixed to their default values, to investigate its impact on the frequency and mean gains. 
We note that we clip the estimated frequencies and support counts in the LDP protocols as we described in Section~\ref{back:pckv}.

\subsection{Experimental Results}
\label{sec:exp_results}
Figure~\ref{fig:exp_attack_freq_synthetics}--Figure~\ref{fig:exp_attack_mean_movielens}  show the frequency gains and mean gains of our attacks on the  four datasets.  
Figure  \ref{fig:asr_case3} shows the ASRs of M2GA to the recommender systems in different cases on Clothing dataset.
Moreover, we also explore the impact of $N_{\text{iter}}$ on our attacks for PrivKVM, and the results are shown in Figure~\ref{fig:exp_privkvm_iterations_1}. Note that we don't show the results of frequency estimation since the frequencies of keys are estimated only in the first round and thus are not affected by $N_{\text{iter}}$. We have the following observations:
\begin{itemize}[leftmargin=*]
    \item  In all scenarios, M2GA  achieves larger  frequency and mean gains than the two baseline attacks (RMA and RKVA). This is because  M2GA is an optimization based attack.
    \item RKVA achieves larger frequency gains than RMA  except PCKV-UE, as RKVA considers target keys.  RMA achieves a  larger frequency gain for PCKV-UE  because the target key RKVA samples  gets perturbed and the perturbed message in PCKV-UE  continues to support this target key with a probability of $1/2$, while  a target key is supported with a probability of $2/3$ in RMA.
    \item M2GA and RKVA achieve larger frequency and mean gains as the  number of fake users (i.e.,  $\beta$) increases. However, the frequency/mean gains of RMA may increase,  not change, or fluctuate as $\beta$ increases in different datasets and for different LDP protocols. 
    
    \item The security-privacy trade-off does not necessarily hold. In particular, we observe security-privacy trade-off with respect to the frequency gains of M2GA (for PrivKVM, this is because we set $r=1$), i.e., the frequency gains of M2GA decrease as $\epsilon$ increases. However, the mean gains of M2GA may increase, fluctuate, or decrease as $\epsilon$ increases in different datasets and for different LDP protocols.

    \item The mean gains of M2GA increase as the number of target keys (i.e., $r$) increases for all the three LDP protocols. The frequency gains of  M2GA decrease  as $r$ increases for PrivKVM. This is because  a fake user can only increase the estimated frequency  for a single target key. The frequency gains of M2GA  increase as  $r$ increases for PCKV-UE. This is because a fake user in M2GA can simultaneously support all the target keys. 
    However, the frequency gains of  M2GA for PCKV-GRR  have different trends on different datasets. 
    We find that this is mainly caused by  clipping the estimated frequencies and support counts in PCKV-GRR.
    
    \item M2GA achieves high ASRs towards the recommender systems in different cases. Specifically, in Case 1 and Case 3, when $\beta\ge 0.01$ and $\epsilon\ge 1$, M2GA achieves close-to-1 ASRs. Our results mean that the recommender system recommends almost all target items under M2GA. In Case 2, M2GA still achieves ASRs that are close to 1 for PrivKVM, while the ASRs for PCKV-UE and PCKV-GRR are close to 0.5, which means that half of the target keys/items are among the recommended $t$ items. The ASRs for PCKV-UE and PCKV-GRR are smaller in Case 2 because the estimated average rating scores of many non-target keys are 1 in PCKV-UE and PCKV-GRR.
    
    \item  Our strongest attack, i.e., M2GA, is effective for different $N_{\text{iter}}$. Specifically, the estimated mean after attack is consistently 1.0 when $N_{\text{iter}}$ ranges from 1 to 10.
\end{itemize}

%%%%%%%%%%%%%%%%%%%%%%%%%%%%

\section{Defenses}
\label{sec:def}

 Cao et al.~\cite{cao2019data} proposed three defenses against poisoning attacks to LDP protocols for categorical data. However, these defenses cannot be directly applied to defend against our attacks. 
This is because these defenses rely on the assumption that each user only holds one single item. In contrast, we consider key-value data, where each user usually has multiple KV pairs. We explore two methods to detect fake users as defenses against our poisoning attacks. For both methods, we assume the server knows the KV pairs sent from each user. For one-class classifier based detection, we further assume the server knows $\lambda$ fraction of genuine users as ground truth.

\subsection{One-class Classifier (OC) based Detection}
\label{sec:oc}
Detecting fake users is essentially an anomaly detection problem, where we aim to distinguish fake users as outliers from the genuine ones. Therefore, we can leverage the one-class machine learning classifiers that are commonly used for anomaly (outlier) detection to detect fake users. Specifically, we treat each user's messages sent to the server as its features. For PrivKVM, we concatenate each user's messages in multiple rounds as a single feature vector. We can then use these features as training data to fit an outlier detection classifier. In our experiments, we use isolation forest \cite{liu2008isolation}. An isolation forest trains an ensemble of randomly partitioned trees to detect outliers. After training, the isolation forest can categorize the users into two groups. We assume the server already knows $\lambda$ fraction of the genuine users as ground truth. Moreover, the server treats the group which includes more ground-truth genuine users as the ``genuine'' group and the other one as the ``fake'' group. The users in the ``fake'' group are considered as fake users and are excluded from aggregation. The server only uses the messages sent by users in the ``genuine'' group to estimate the frequencies and mean values. 
In our experiments, we use the implementation of isolation forest in Scikit-learn \cite{ifimplement}.

\begin{figure*}[!t]
    \centering % <-- added
    {\includegraphics[width=0.3\textwidth]{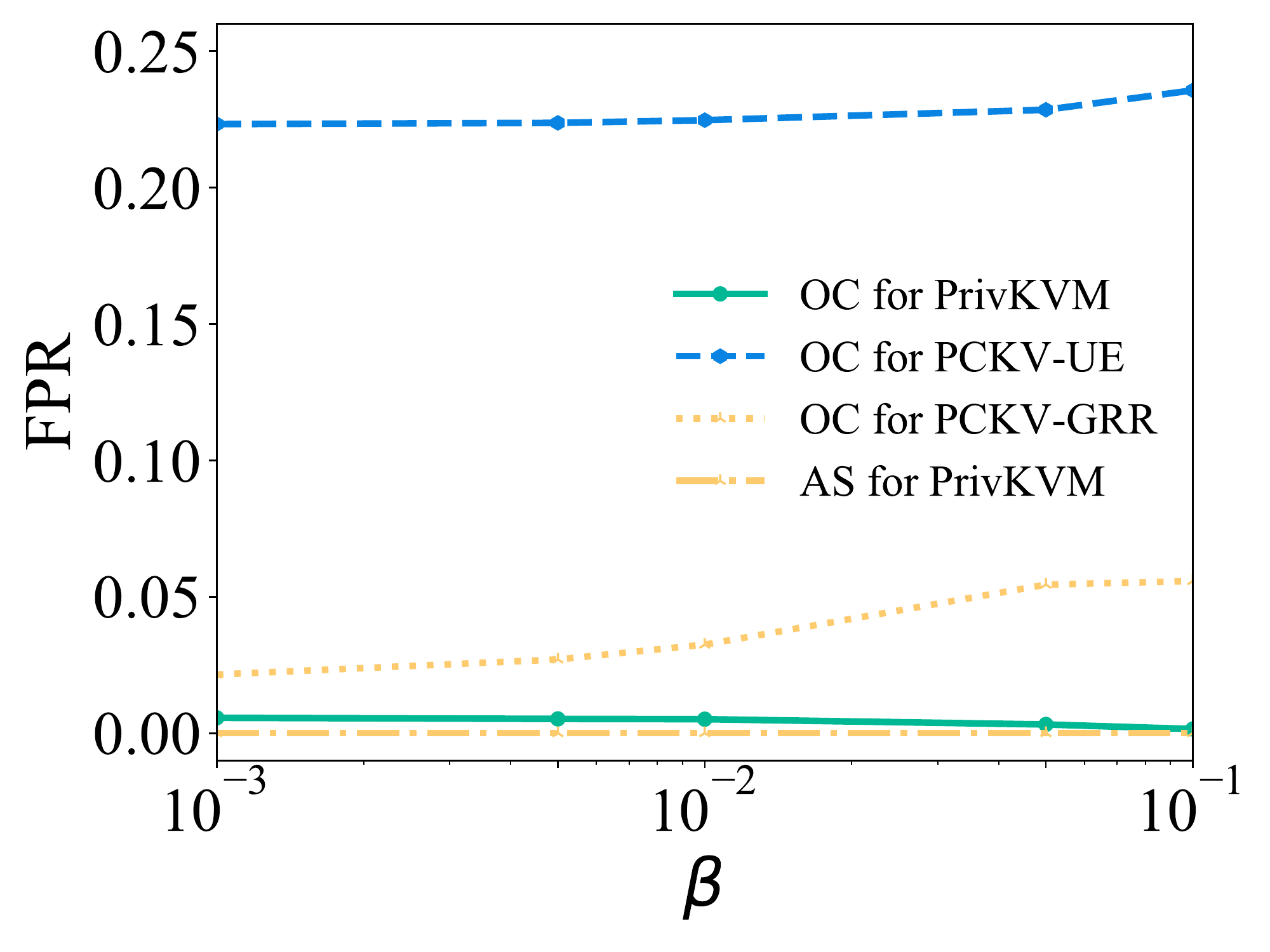}}
    {\includegraphics[width=0.3\textwidth]{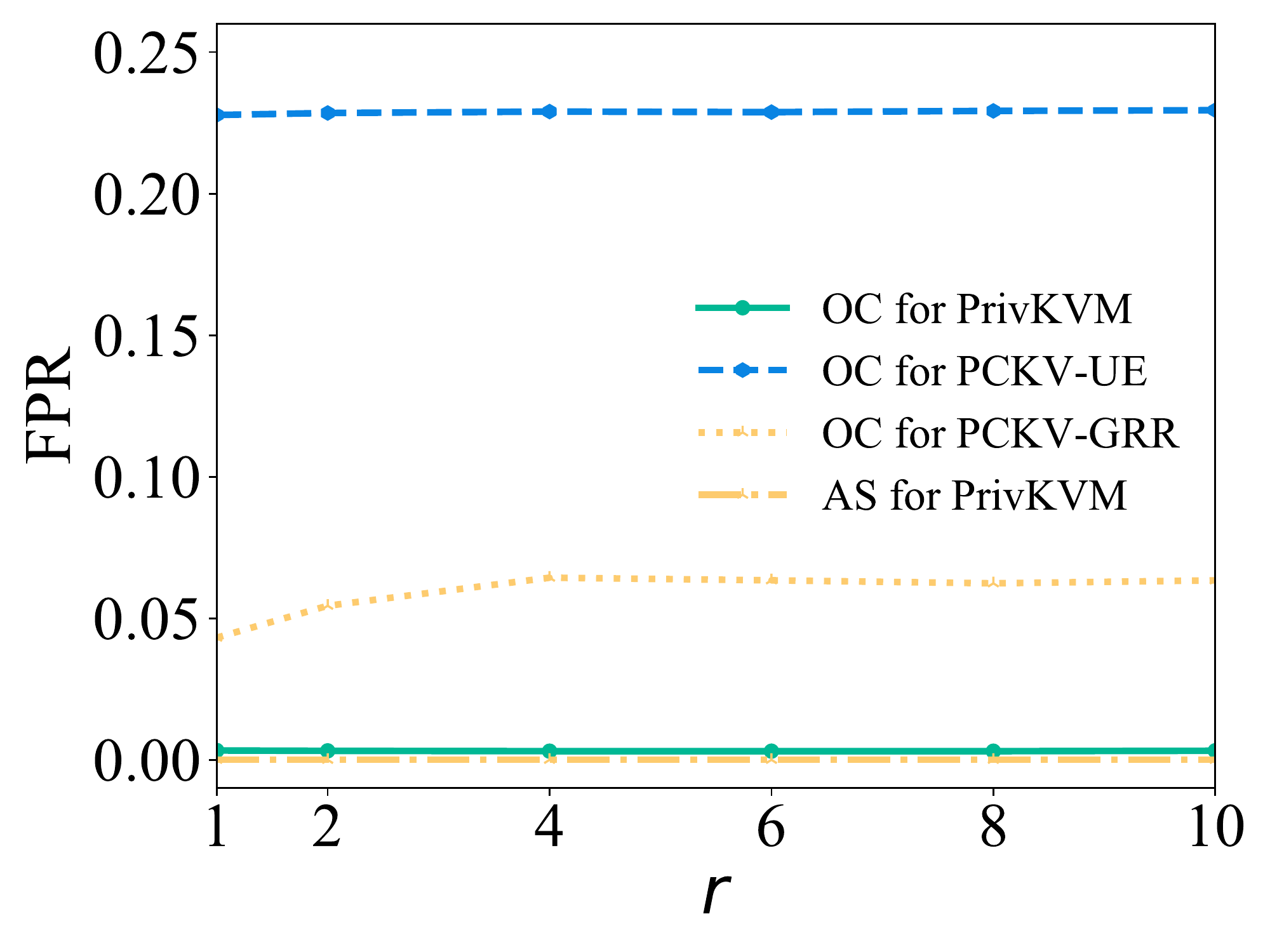}}
    {\includegraphics[width=0.3\textwidth]{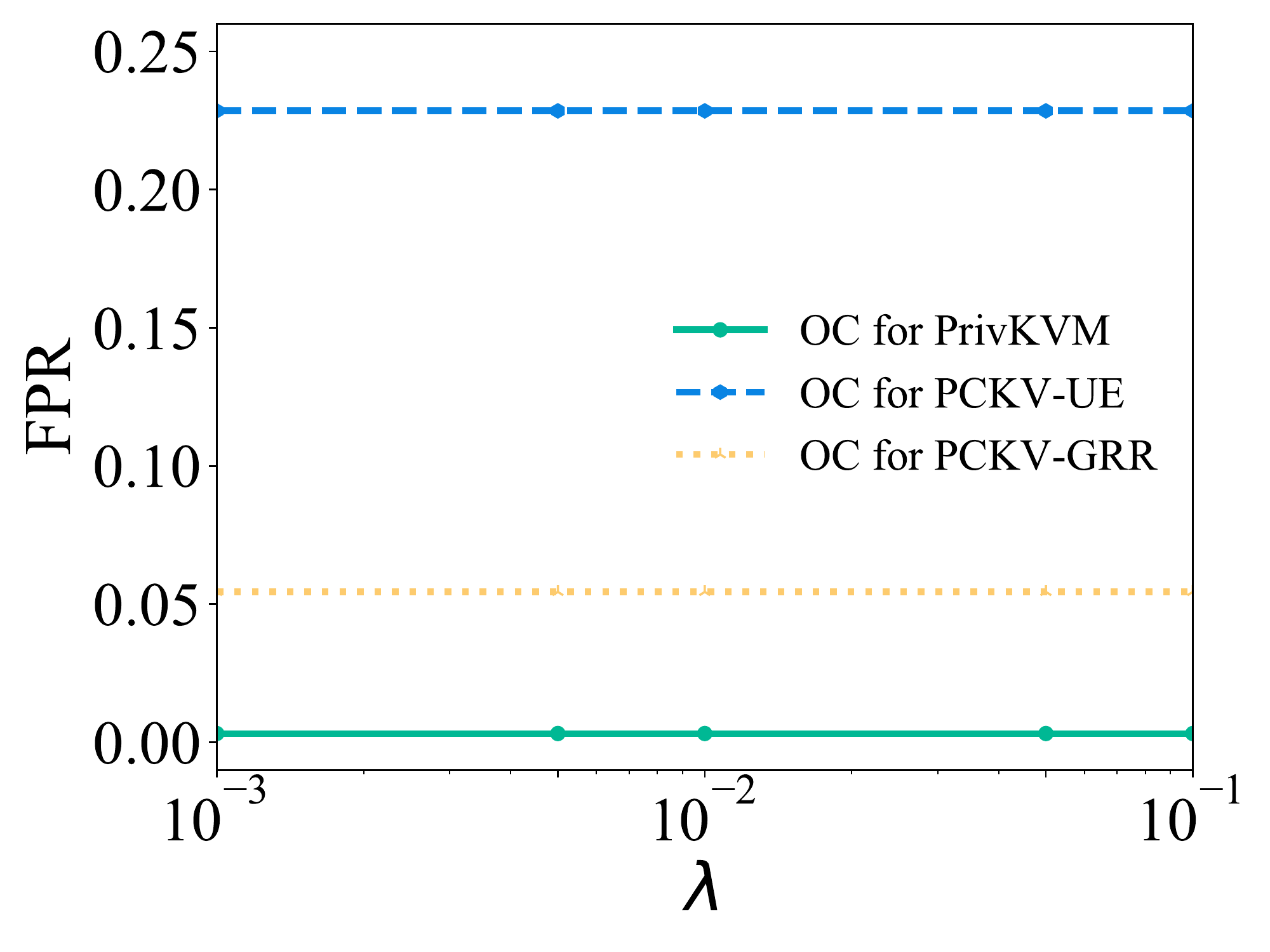}}
    
    {\includegraphics[width=0.3\textwidth]{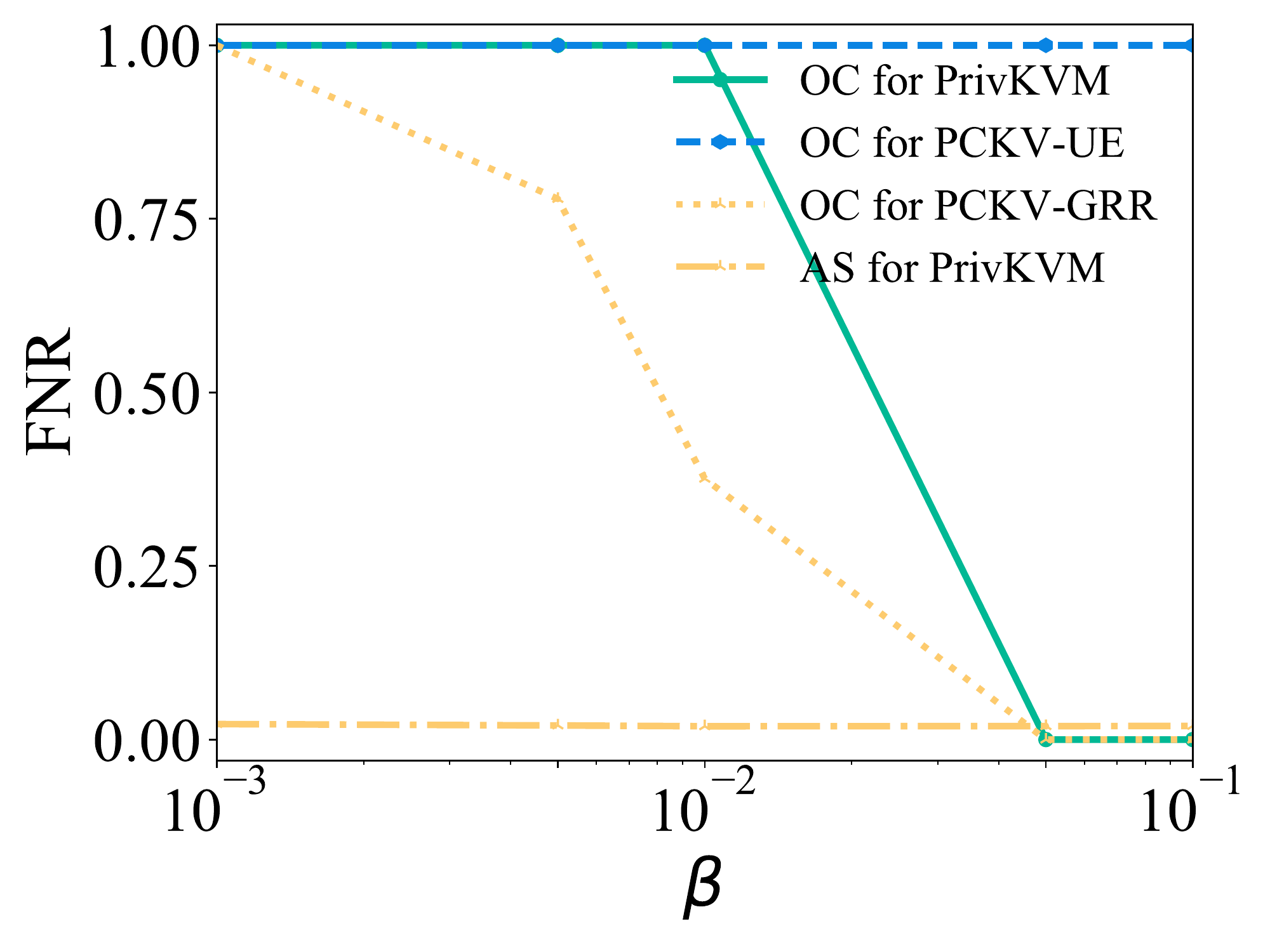}}
    {\includegraphics[width=0.3\textwidth]{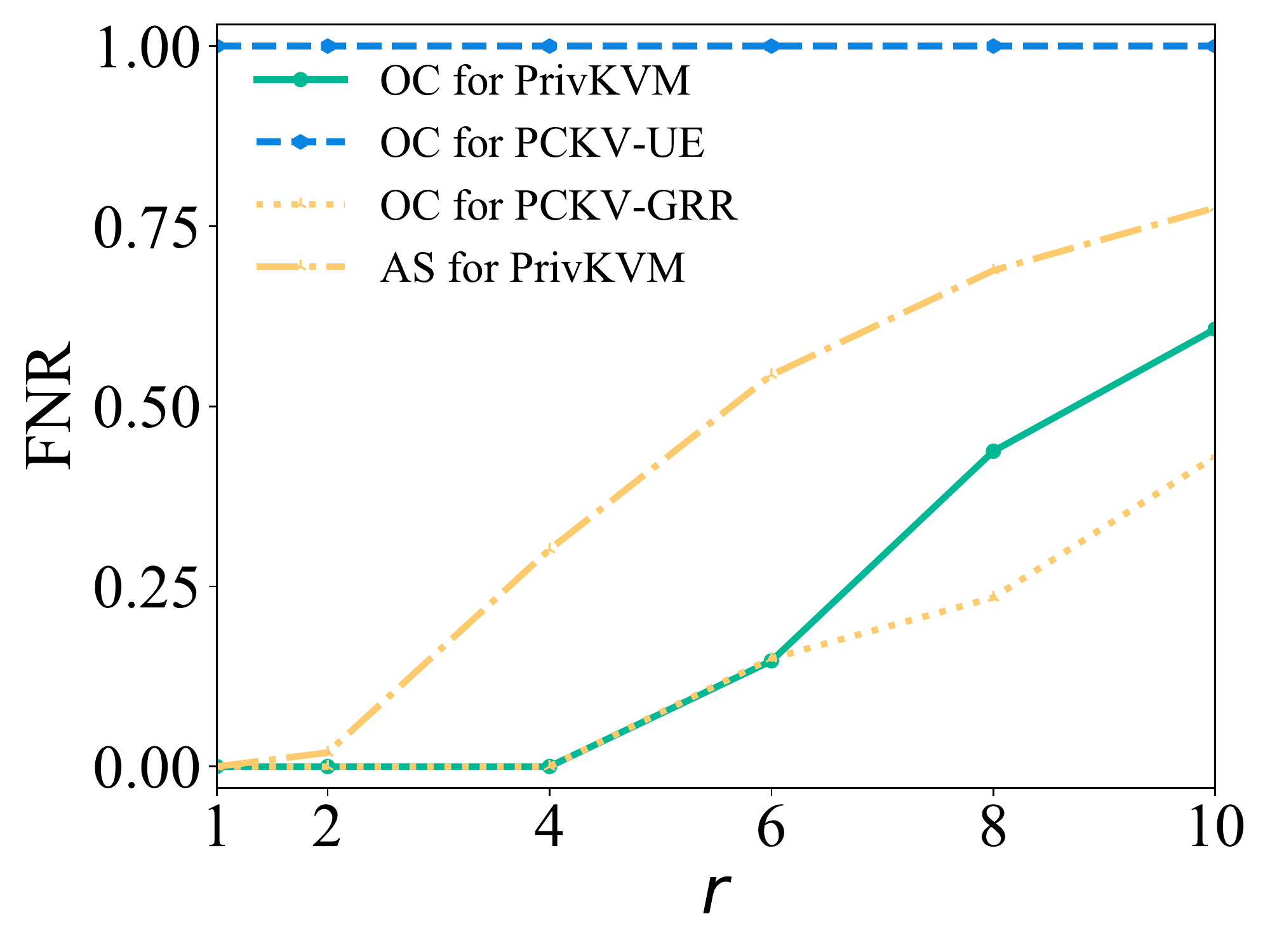}}
    {\includegraphics[width=0.3\textwidth]{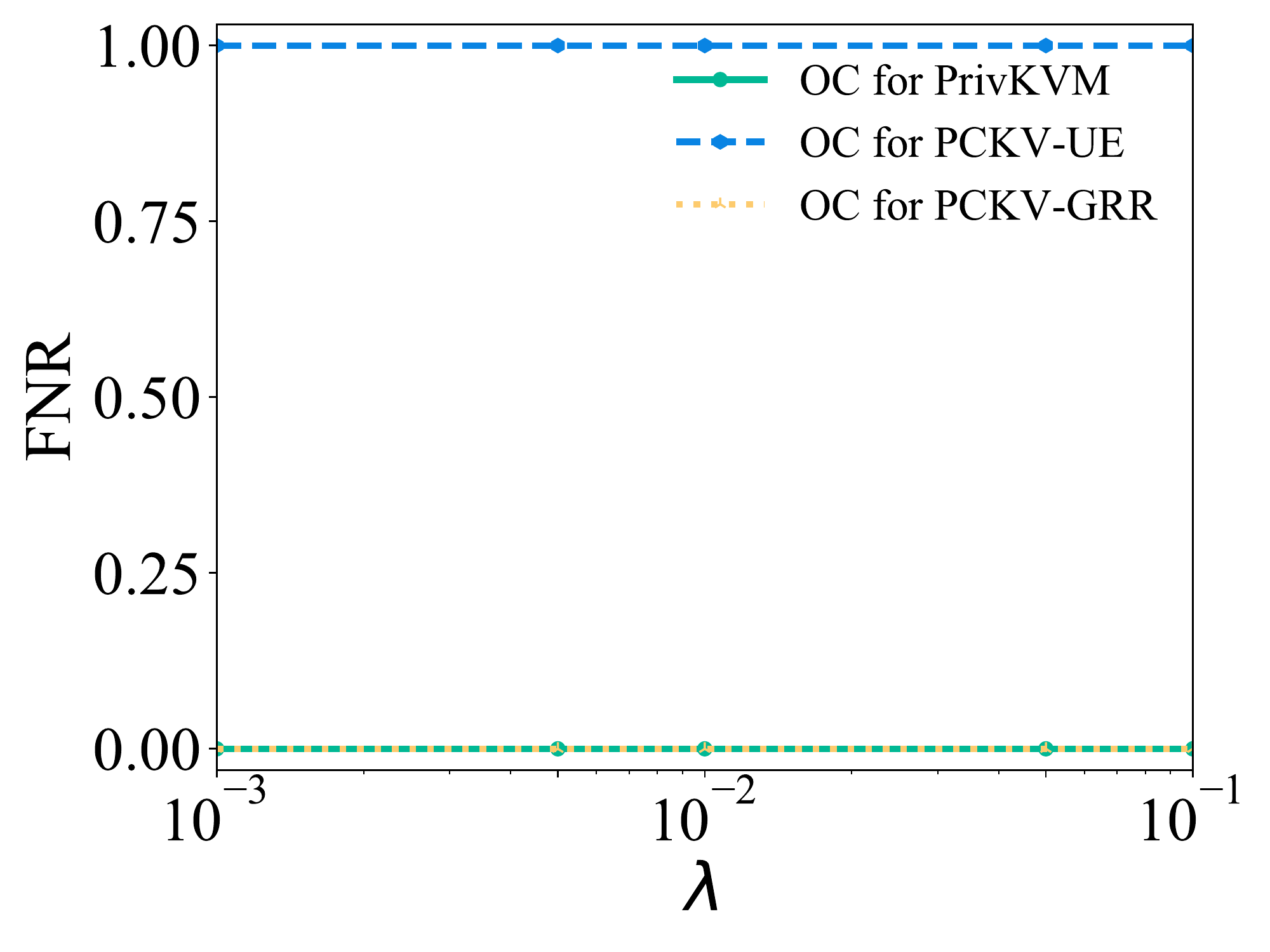}}
\caption{Impact of $\beta$, $r$, and $\lambda$ on FPR (first row) and FNR (second row) of detecting fake users against M2GA on TalkingData.}
\label{fig:exp_defense_talkingdata_m2ga_recall}
\end{figure*}

\begin{figure}[!t]
    \centering % <-- added
    {\includegraphics[width=0.20\textwidth]{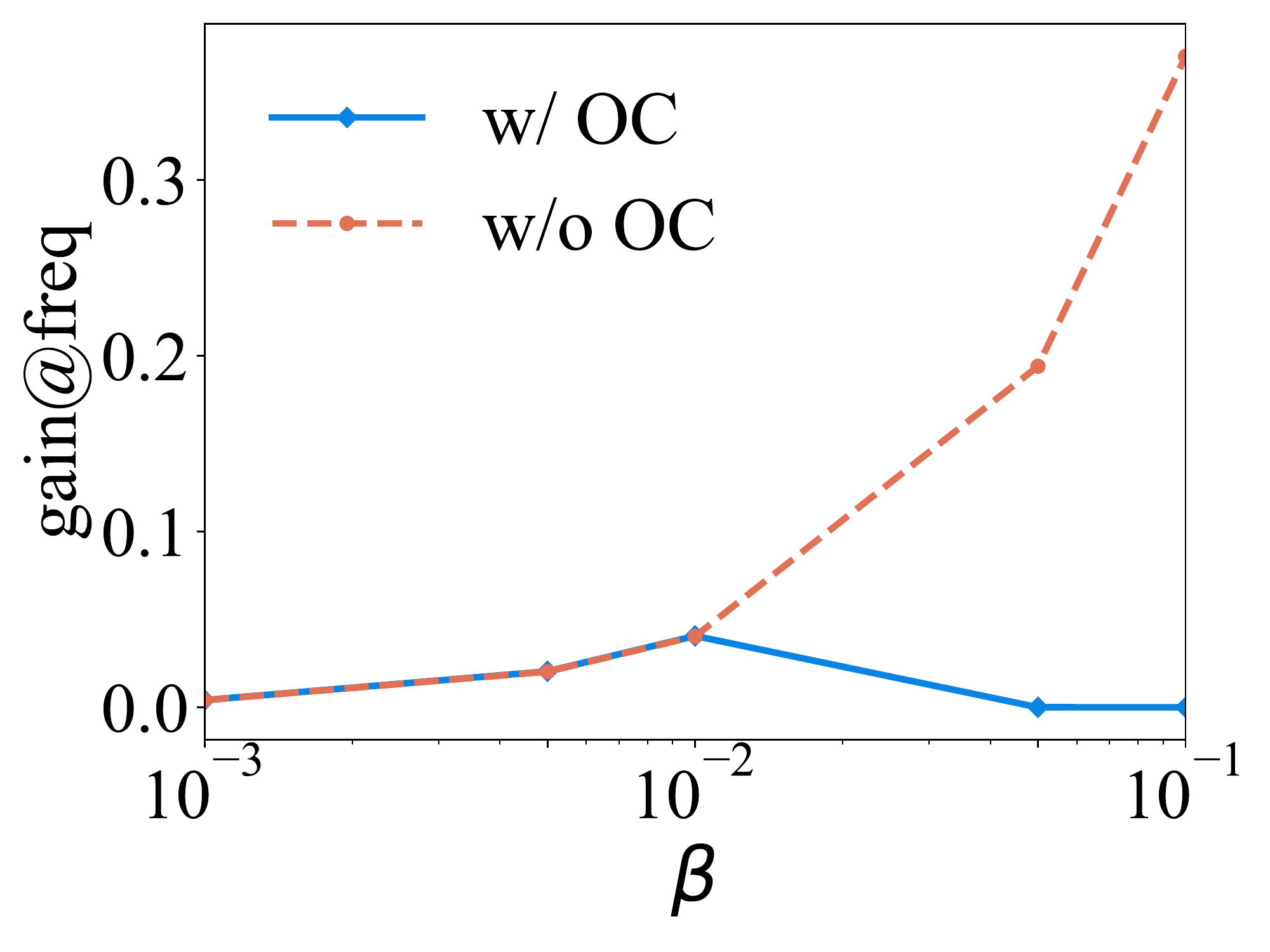}}
    {\includegraphics[width=0.20\textwidth]{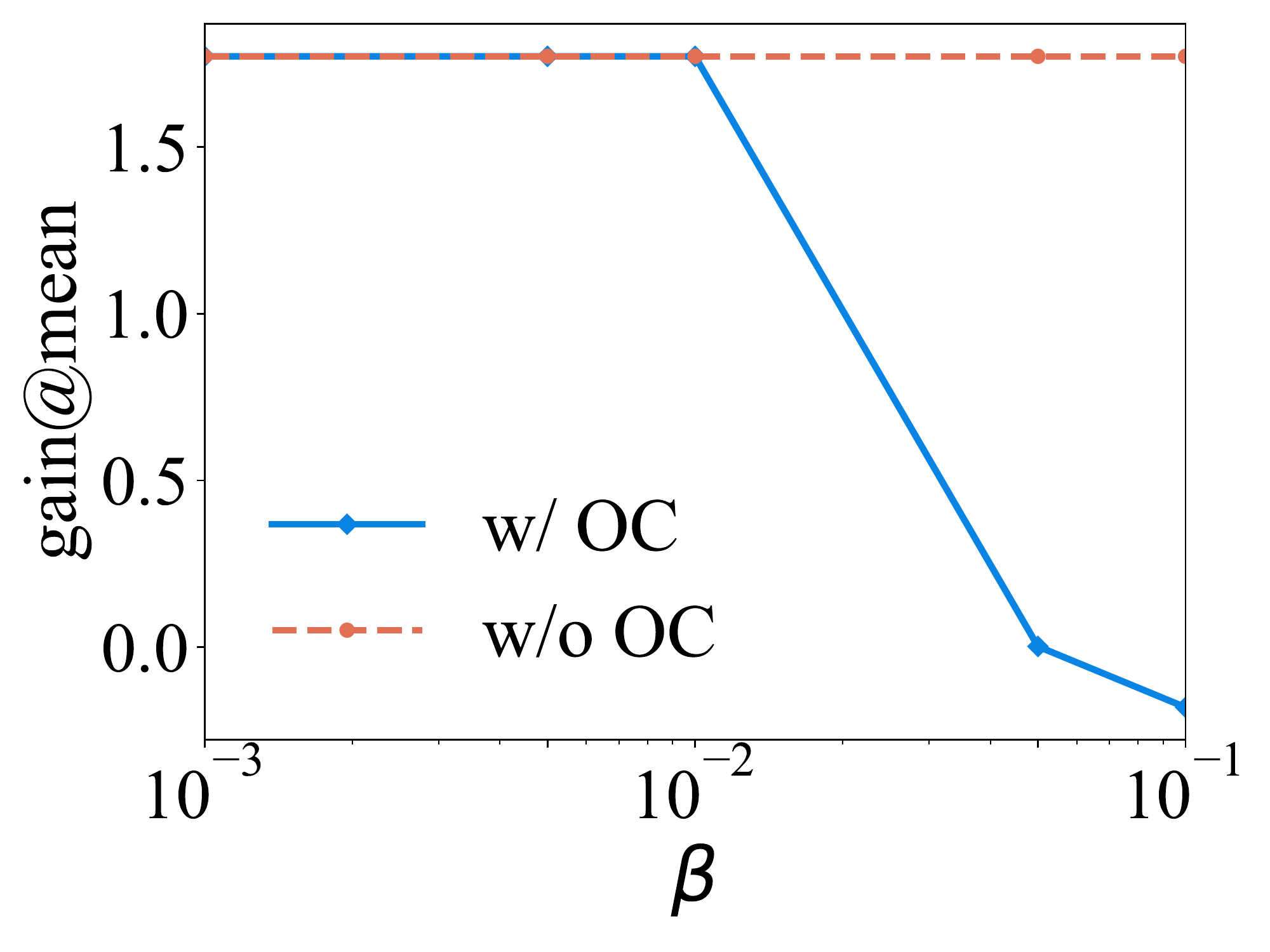}}
    
    {\includegraphics[width=0.20\textwidth]{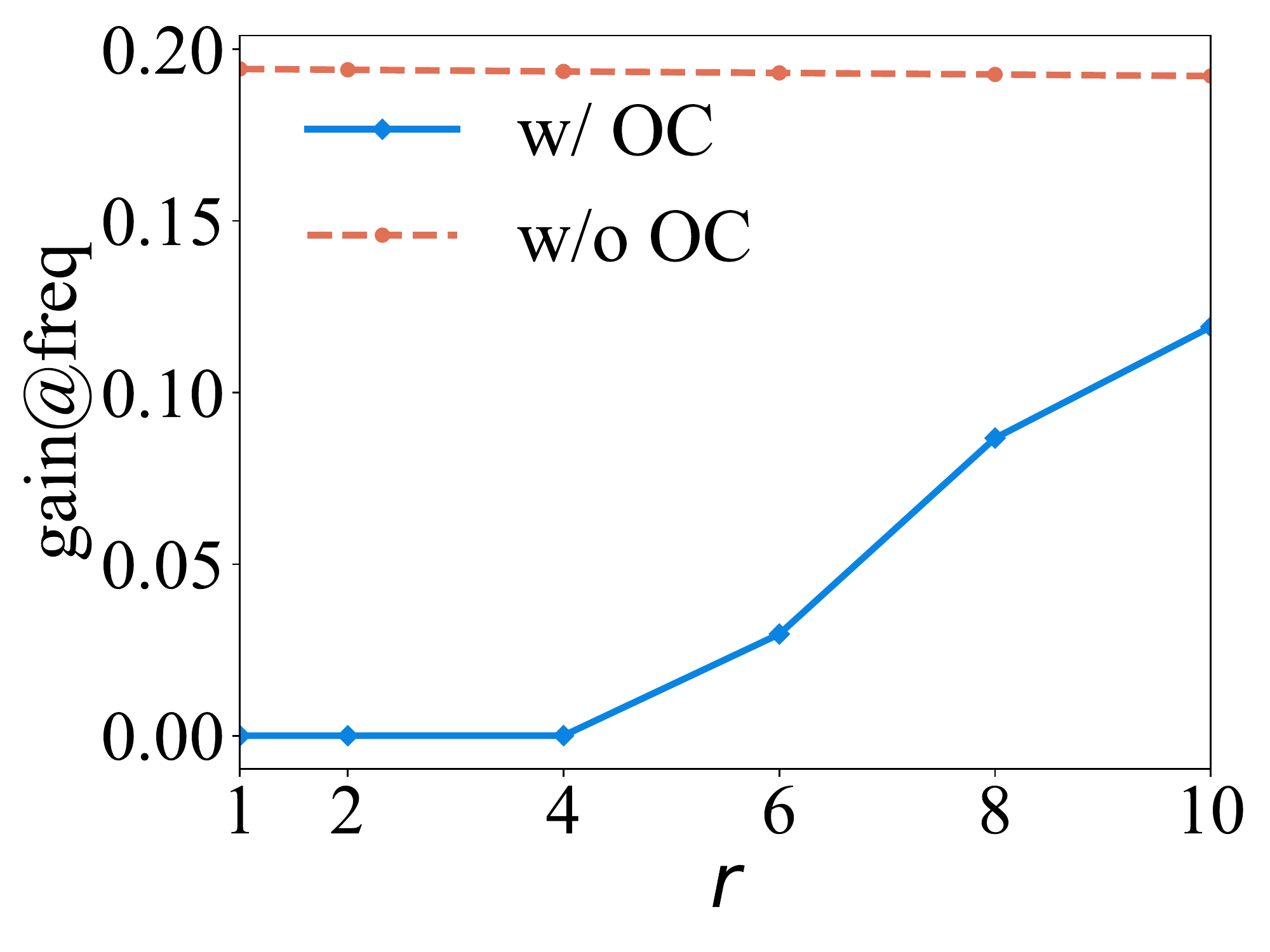}}
    {\includegraphics[width=0.20\textwidth]{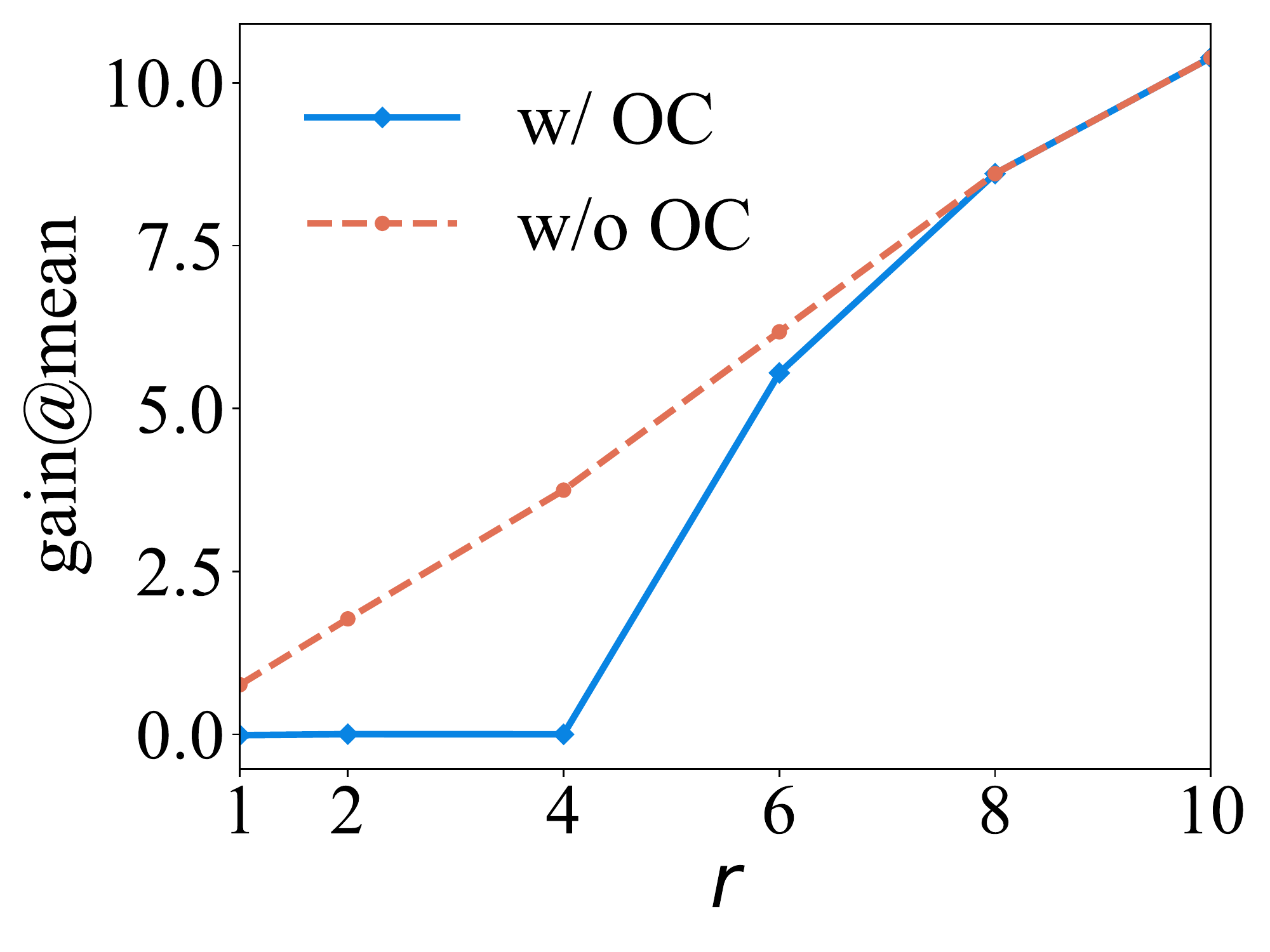}}
\caption{Impact of $\beta$ and $r$ on the defense effectiveness of OC against M2GA for PrivKVM on TalkingData.}
\label{fig:exp_oc_talkingdata_m2ga_privkvm}
\end{figure}

\subsection{Anomaly Score (AS) based Detection}

We note that, multiple rounds of communications are conducted in PrivKVM, allowing us to check the consistency of the messages sent by a user in different rounds. Based on this observation, we propose a method to detect fake users for PrivKVM. 
Recall that, in PrivKVM, each user sends a perturbed KV pair and the index of a key to the server in each round. Since the key is randomly sampled from the large dictionary, it is unlikely that the same key is repeatedly selected in multiple rounds for genuine users. However, since a fake user promotes a target key in each round, it may send the same key to the server in multiple rounds, especially when the number of target keys is small. 

Based on this intuition, we assign an \emph{anomaly score} to each user, which we define as the maximum number of rounds in which the user sends the same index of key to the server.  Specifically, in round $t$, the server computes the number of rounds $N_{k,u}^t$ in which the user $u$ has sent key $k$ to the server. The  anomaly score of user $u$ in round $t$ is the maximum $N_{k,u}^t$ over possible $k$'s. 
If the anomaly score for a user is no smaller than $\eta$ (called \emph{anomaly threshold}), then we mark the user as a fake one. We calculate the anomaly score of each user  and detect fake users in each round. When a user is detected as fake in a certain round, we exclude the user in the subsequent rounds for mean estimation. Moreover, we re-estimate the frequencies of keys based on the messages sent by users in the first round by removing the ones belonging to the detected fake users. 

\begin{figure}[!t]
    \centering % <-- added
    {{\includegraphics[width=0.20\textwidth]{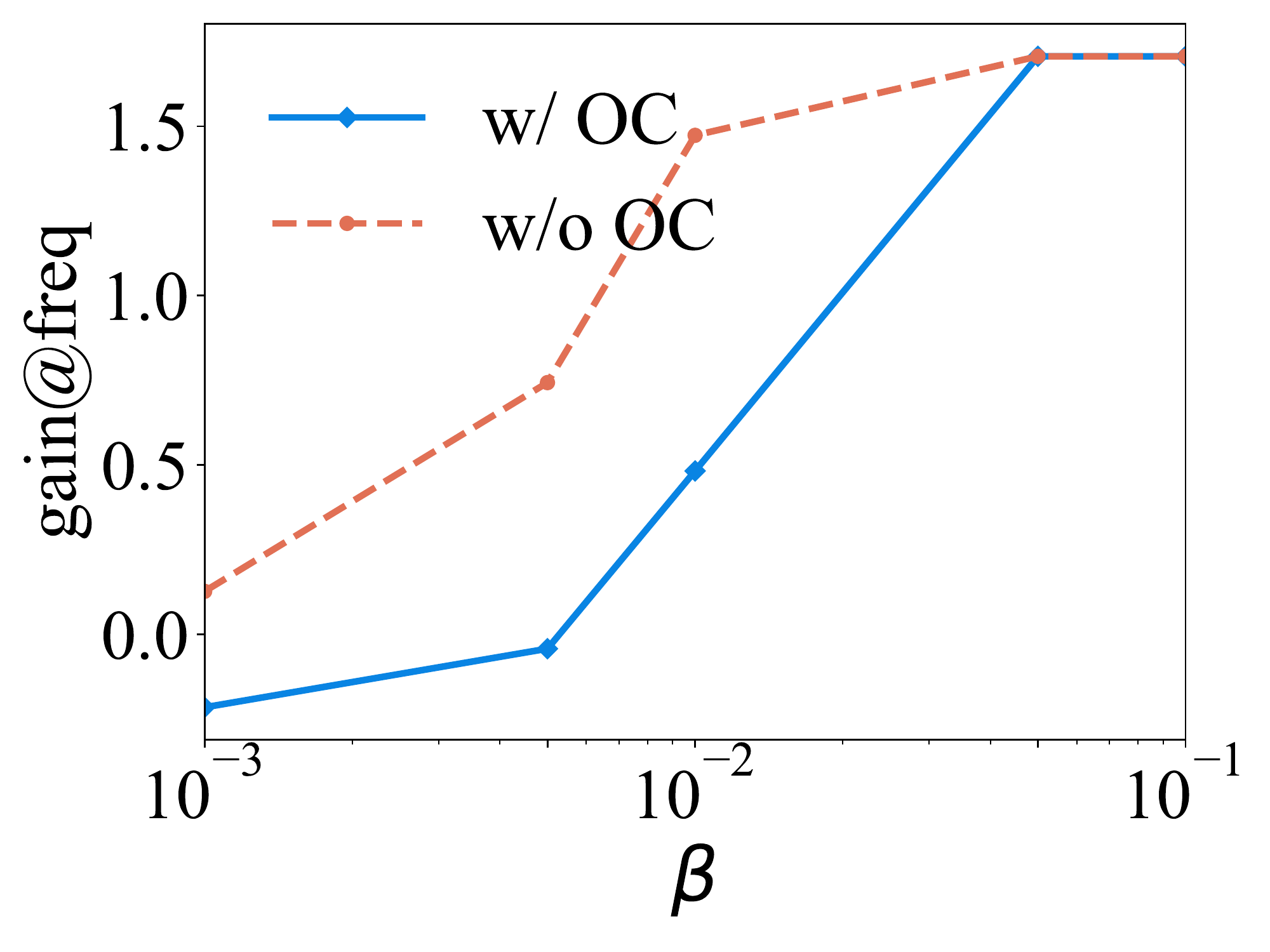}}}
    {{\includegraphics[width=0.20\textwidth]{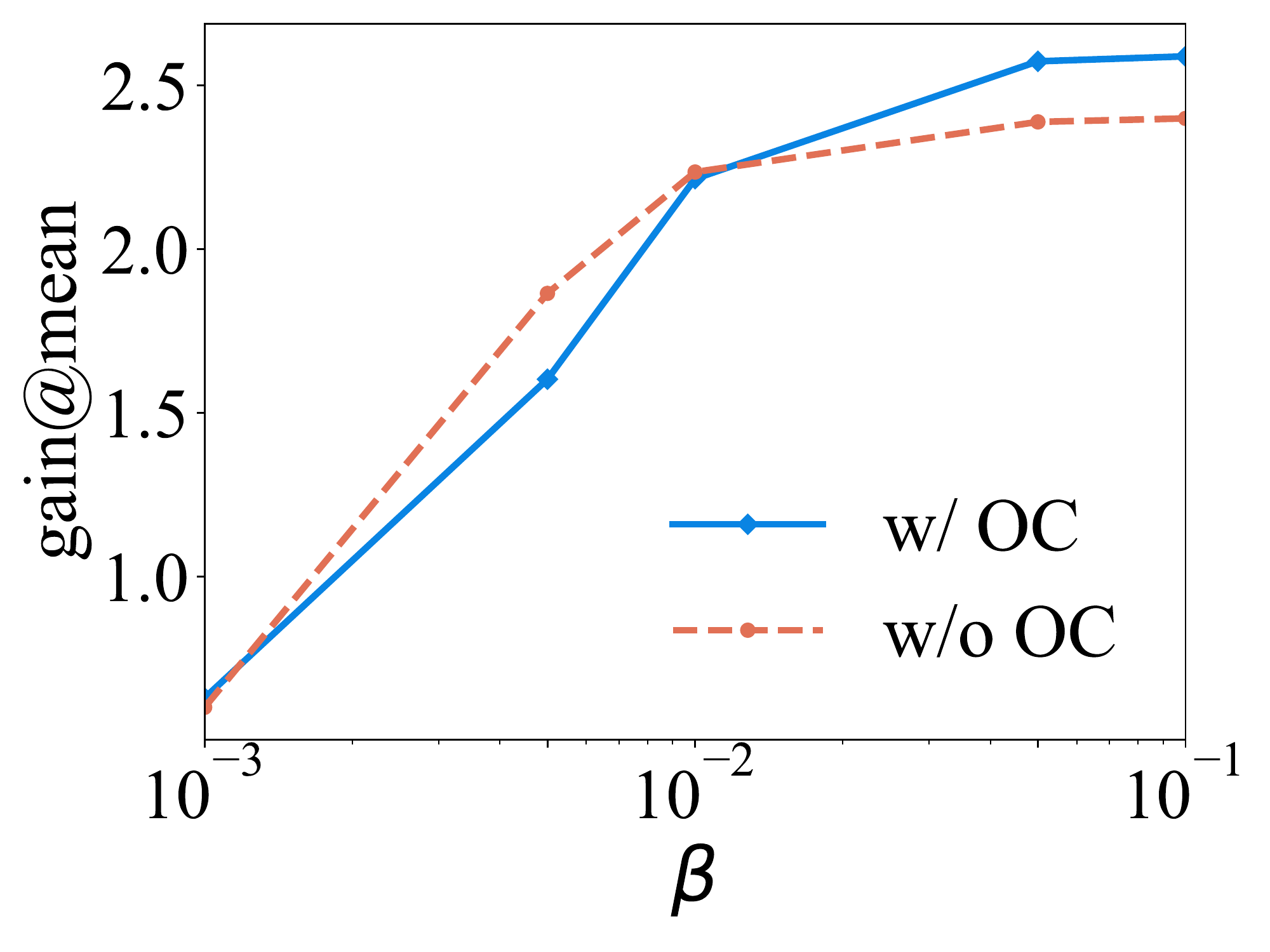}}}
    
    {{\includegraphics[width=0.19\textwidth, height=2.65cm]{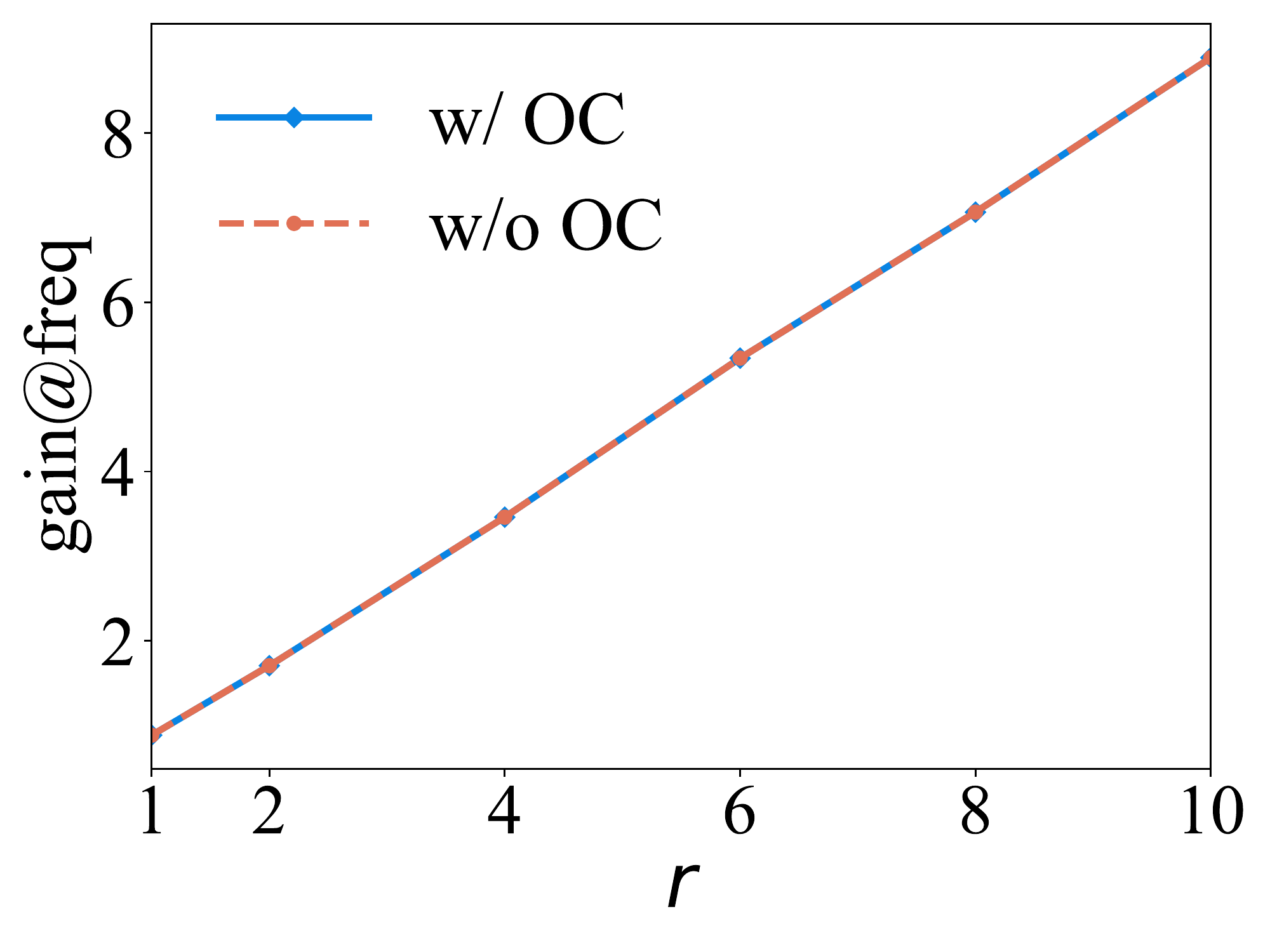}}}
    {{\includegraphics[width=0.20\textwidth]{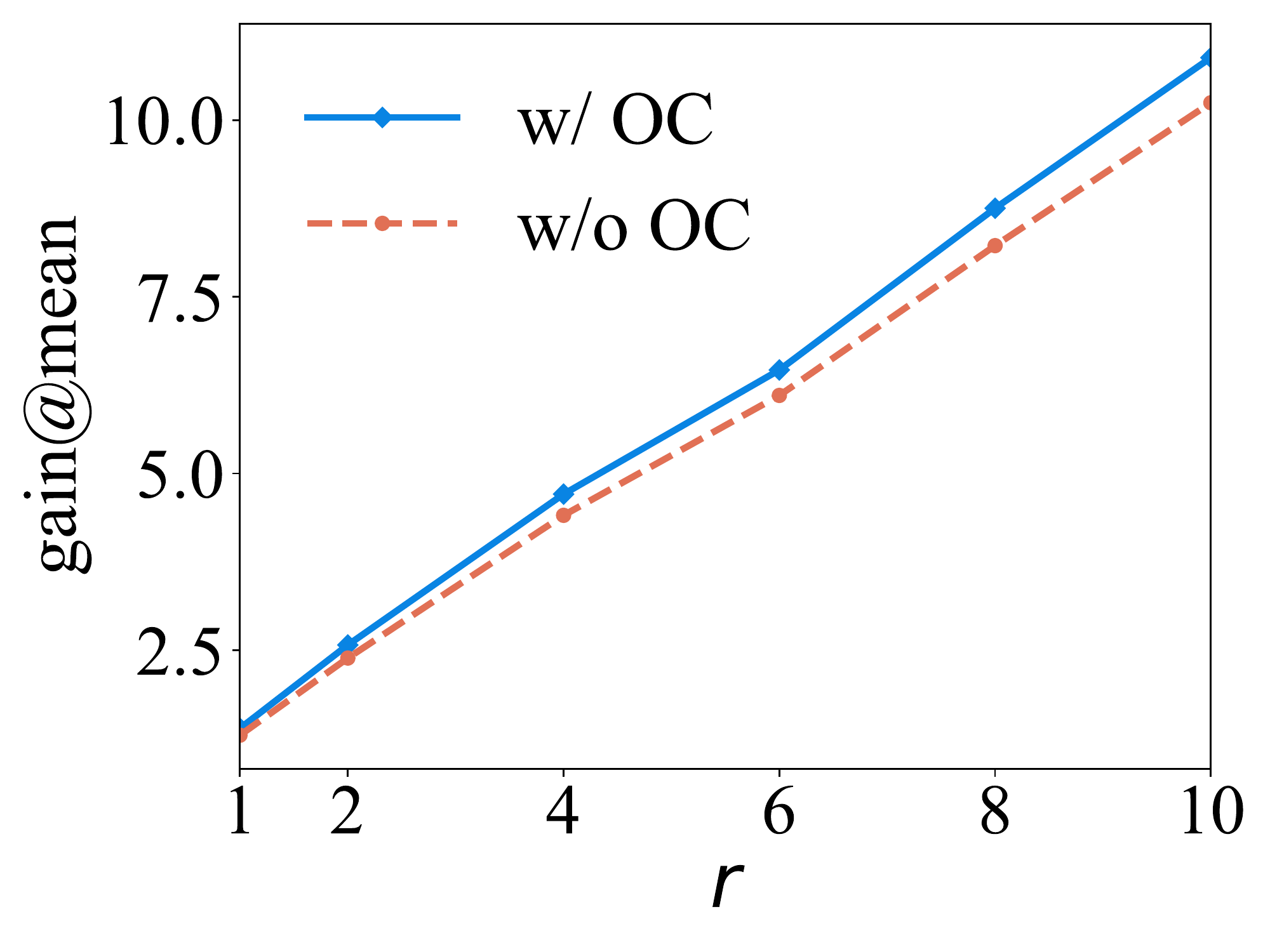}}}
\caption{Impact of $\beta$ and $r$ on the defense effectiveness of OC against M2GA for PCKV-UE on TalkingData.}
\label{fig:exp_oc_talkingdata_m2ga_pckvue}
\end{figure}

\begin{figure}[!t]
    \centering 
    {\includegraphics[width=0.20\textwidth]{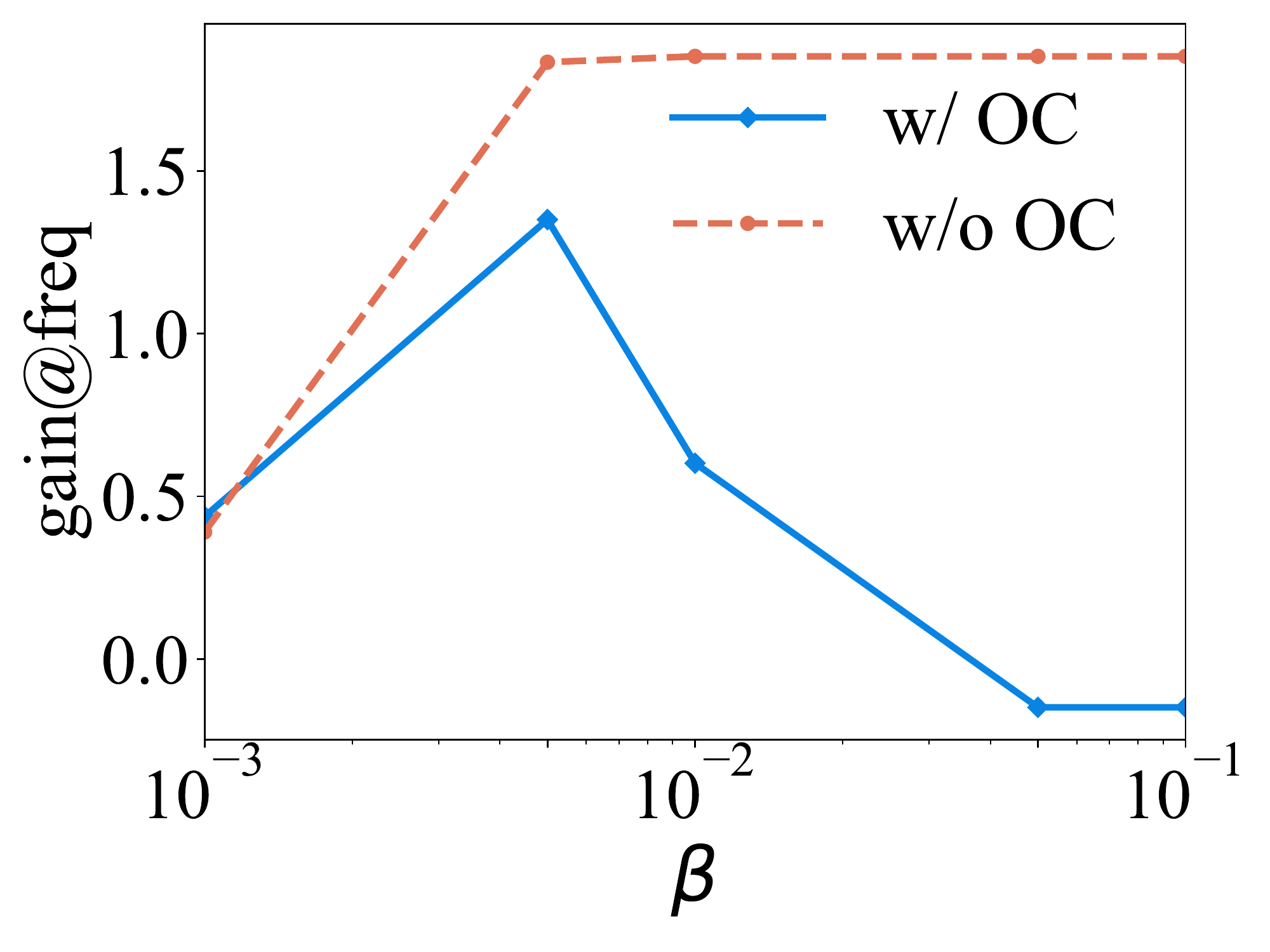}}
    {\includegraphics[width=0.20\textwidth]{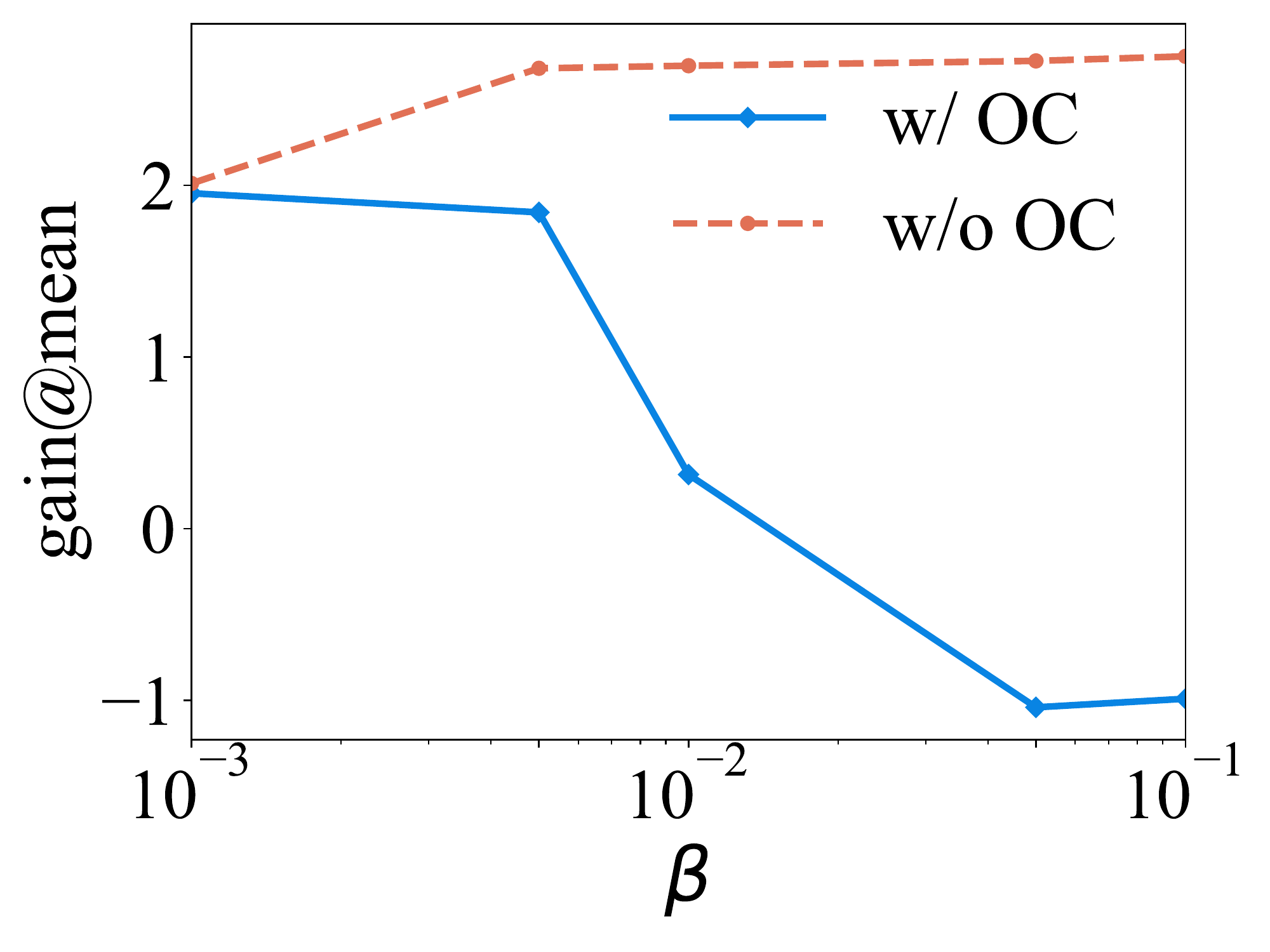}}
    
    {\includegraphics[width=0.20\textwidth, height=2.69cm]{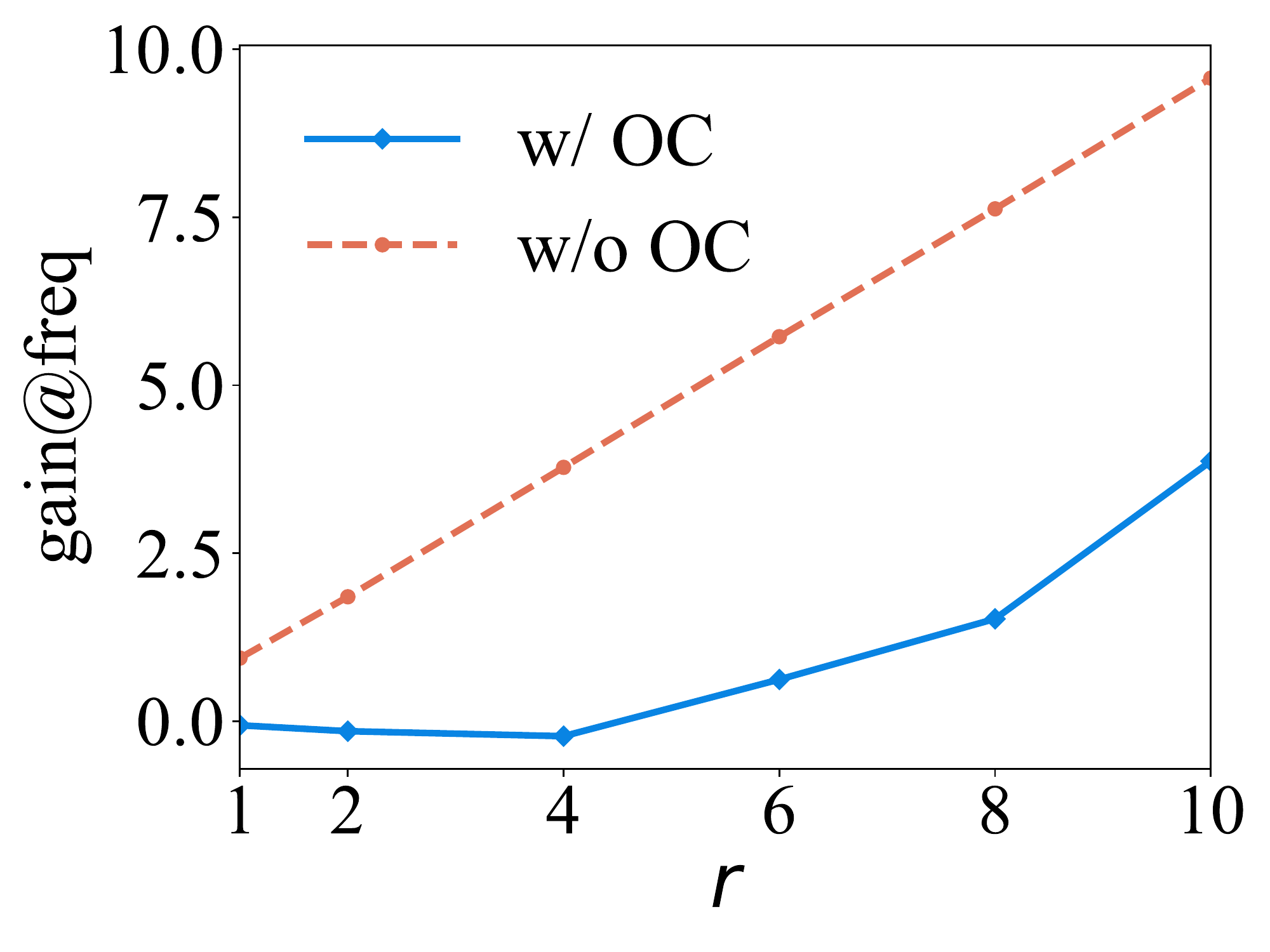}}
    {\includegraphics[width=0.20\textwidth]{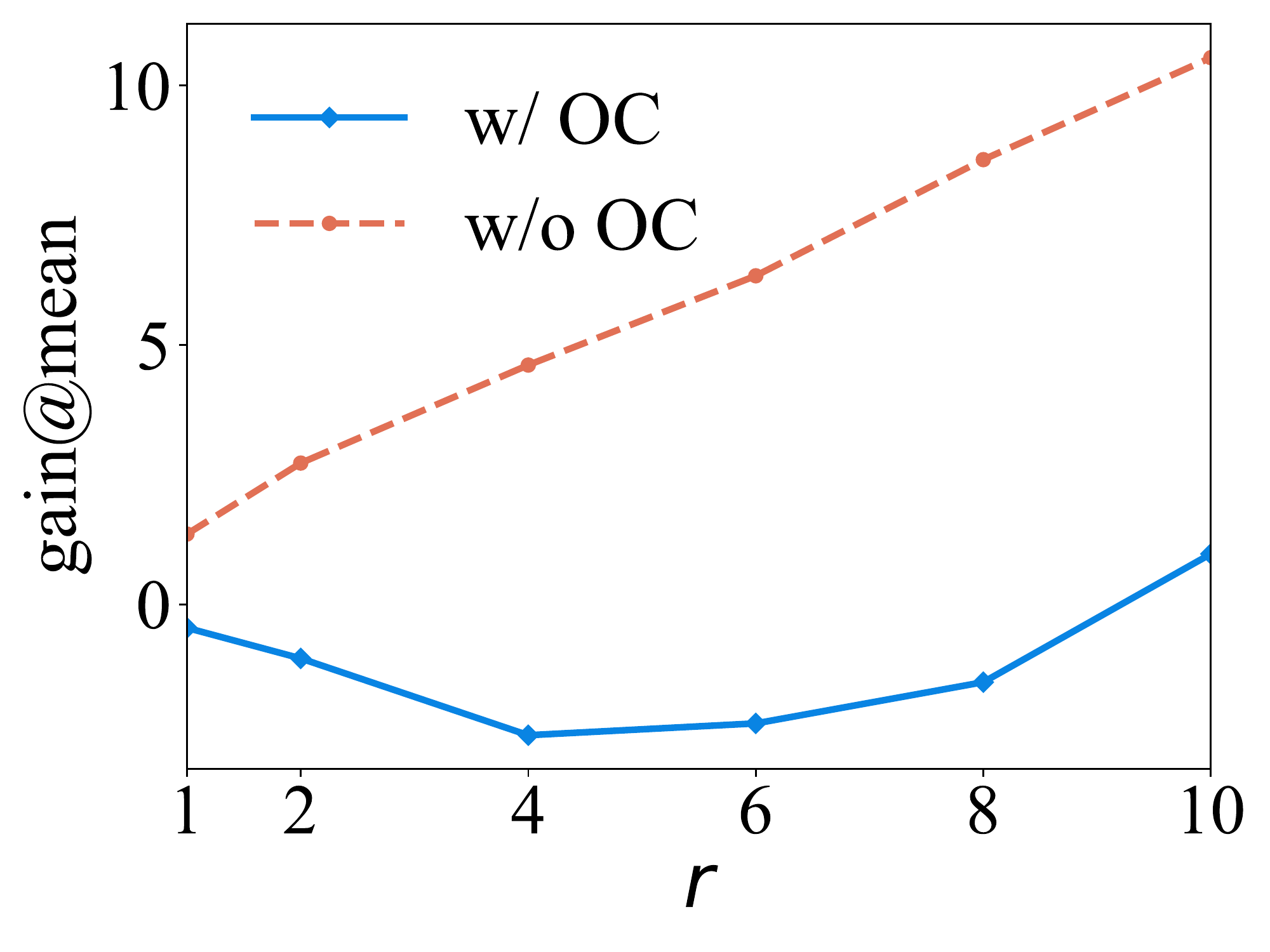}}
\caption{Impact of $\beta$ and $r$ on the defense effectiveness of OC against M2GA for PCKV-GRR on TalkingData.}
\label{fig:exp_oc_talkingdata_m2ga_pckvgrr}
\end{figure}

\subsection{Experiments}
\subsubsection{Experimental Setup}
Unless otherwise mentioned, we adopt the following default parameters: $\beta=0.05$, $r=2$, $\epsilon=1.0$, $N_{\text{iter}}=10$, and $\eta=2$. We adopt \emph{frequency gain (gain@freq)} and \emph{mean gain (gain@mean)} of a set of target keys as the evaluation metrics (please refer to Section~\ref{evaluation_metric_exp_sec} for details). Moreover, we also consider \emph{False Positive Rate (FPR)} (or \emph{False Negative Rate (FNR)}), which is the fraction of genuine (or fake) users that are detected as fake (or genuine). We vary one parameter while keeping the others fixed to their default values to study the impact of it on the effectiveness of our defenses. 
Moreover, we evaluate M2GA since it is the strongest attack.

\subsubsection{Experimental Results}
\label{sec:def_results}
Figure~\ref{fig:exp_defense_talkingdata_m2ga_recall} shows the impact of $\beta$, $r$, and $\lambda$ on the FPRs and FNRs  of OC and AS against M2GA on TalkingData dataset. 
Figure~\ref{fig:exp_oc_talkingdata_m2ga_privkvm},~\ref{fig:exp_oc_talkingdata_m2ga_pckvue}, and~\ref{fig:exp_oc_talkingdata_m2ga_pckvgrr} show the impact of $\beta$ and $r$ on the defense effectiveness of OC against  M2GA  on TalkingData dataset for PrivKVM, PCKV-UE, and PCKV-GRR, respectively. Figure~\ref{fig:exp_as_talkingdata_m2ga} shows the impact of $\beta$ and $r$ on the defense effectiveness of AS against M2GA for PrivKVM  on TalkingData dataset. Note that in Figure~\ref{fig:exp_as_talkingdata_m2ga}, we set $\beta=0.001$ when exploring the impact of $r$ to better illustrate the impact of $r$ (AS is not effective in the default setting $\beta=0.05$ regardless of $r$).  

Our key observation is that the defenses are effective in some scenarios but have limited effectiveness in other scenarios.   For instance, when $\beta$ is small or $r$ is large, OC fails to detect the fake users with high FNR. Moreover, OC has high FPR (e.g., 22\% for PCKV-UE), which results in utility loss as a large fraction of genuine users are excluded from aggregation. For instance, when $\beta=0.05$ and $r=2$, OC for PCKV-GRR has a FPR of 5.5\% and the mean gain is -1.04, which means that the estimated mean decreases by 1.04 compared to the estimated mean without attack and defense. Similarly, when $\beta$ or $r$ is small, AS can detect a large fraction or all of fake users,  and thus  the frequency and mean gains of M2GA under AS become close to 0. However,  when $\beta$ or $r$ is large, the frequency gain and/or mean gain increase substantially. Our results show that new defenses are needed to defend against our attacks. 

\begin{figure}[!t]
    \centering 
    {\includegraphics[width=0.20\textwidth]{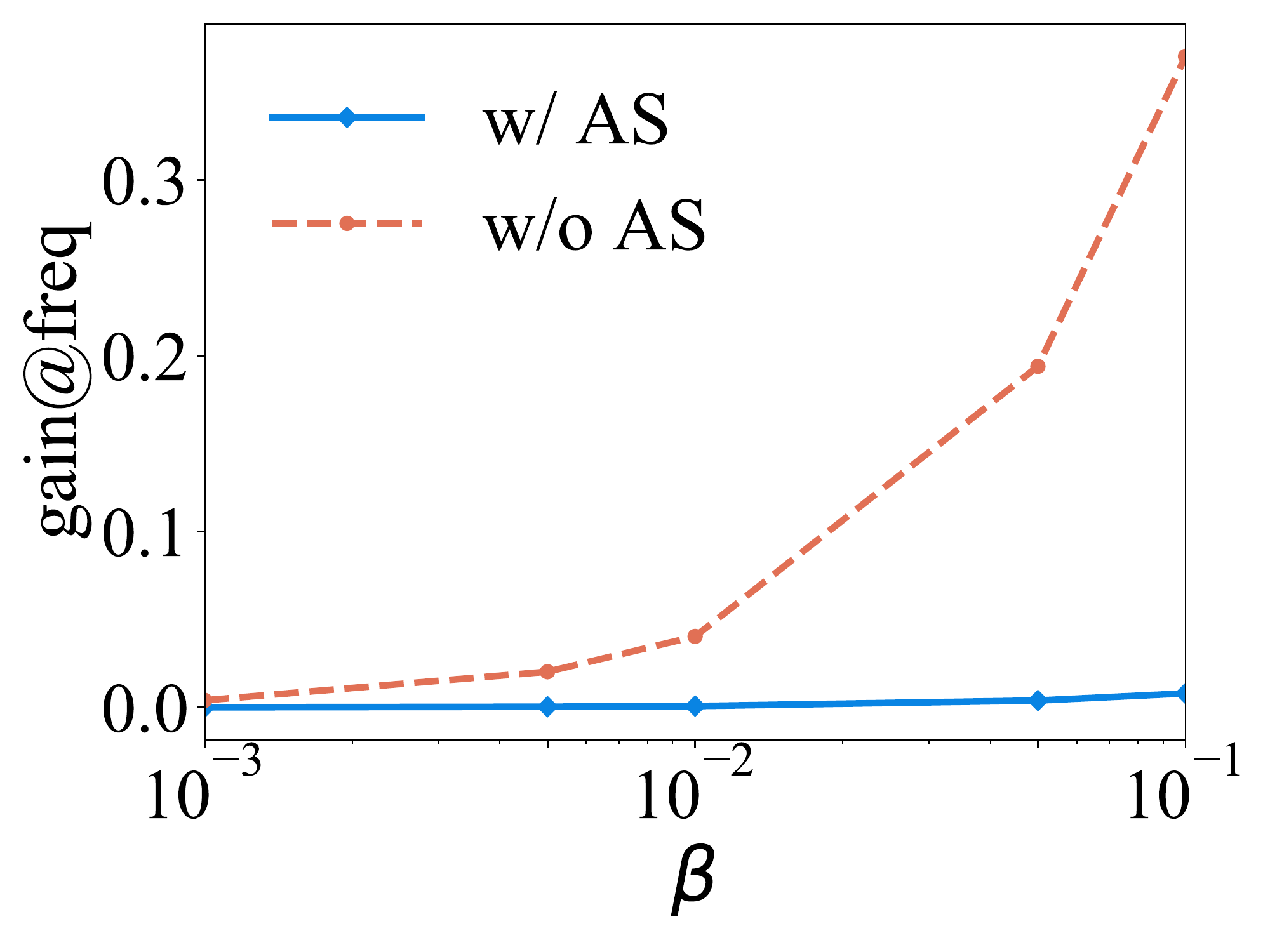}}
    {\includegraphics[width=0.20\textwidth]{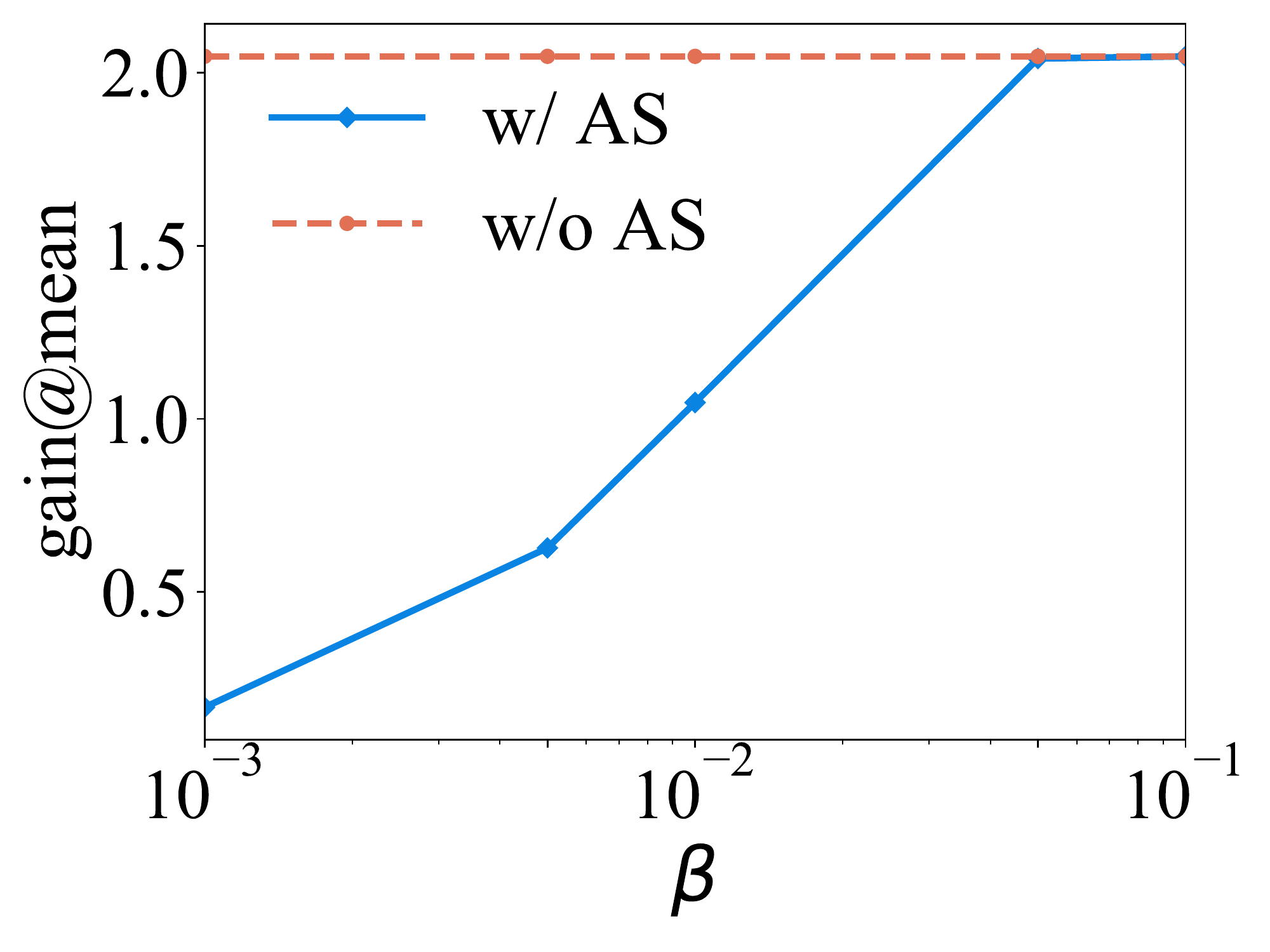}}
    
    {\includegraphics[width=0.20\textwidth]{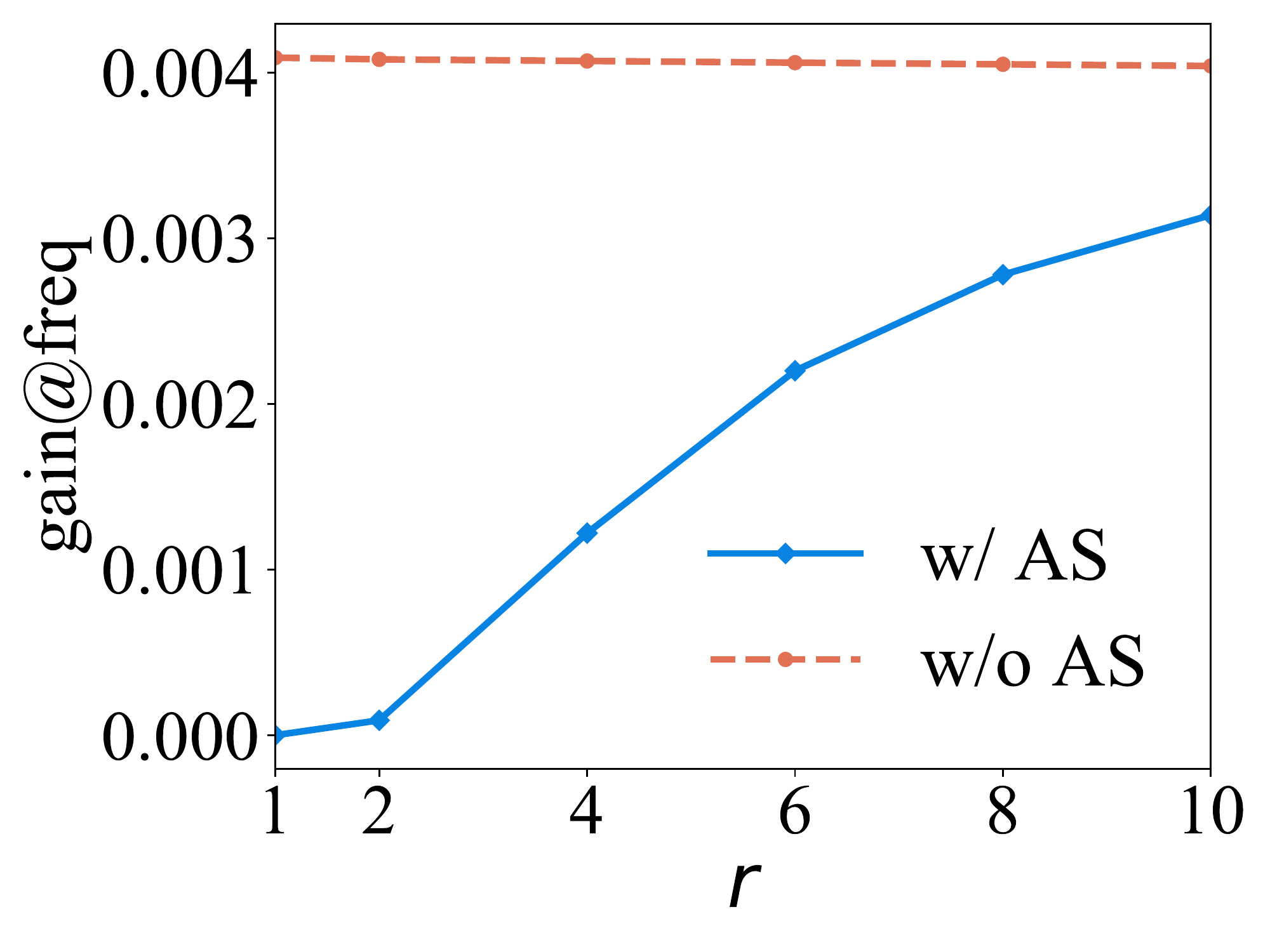}}
    {\includegraphics[width=0.19\textwidth, height=2.65cm]{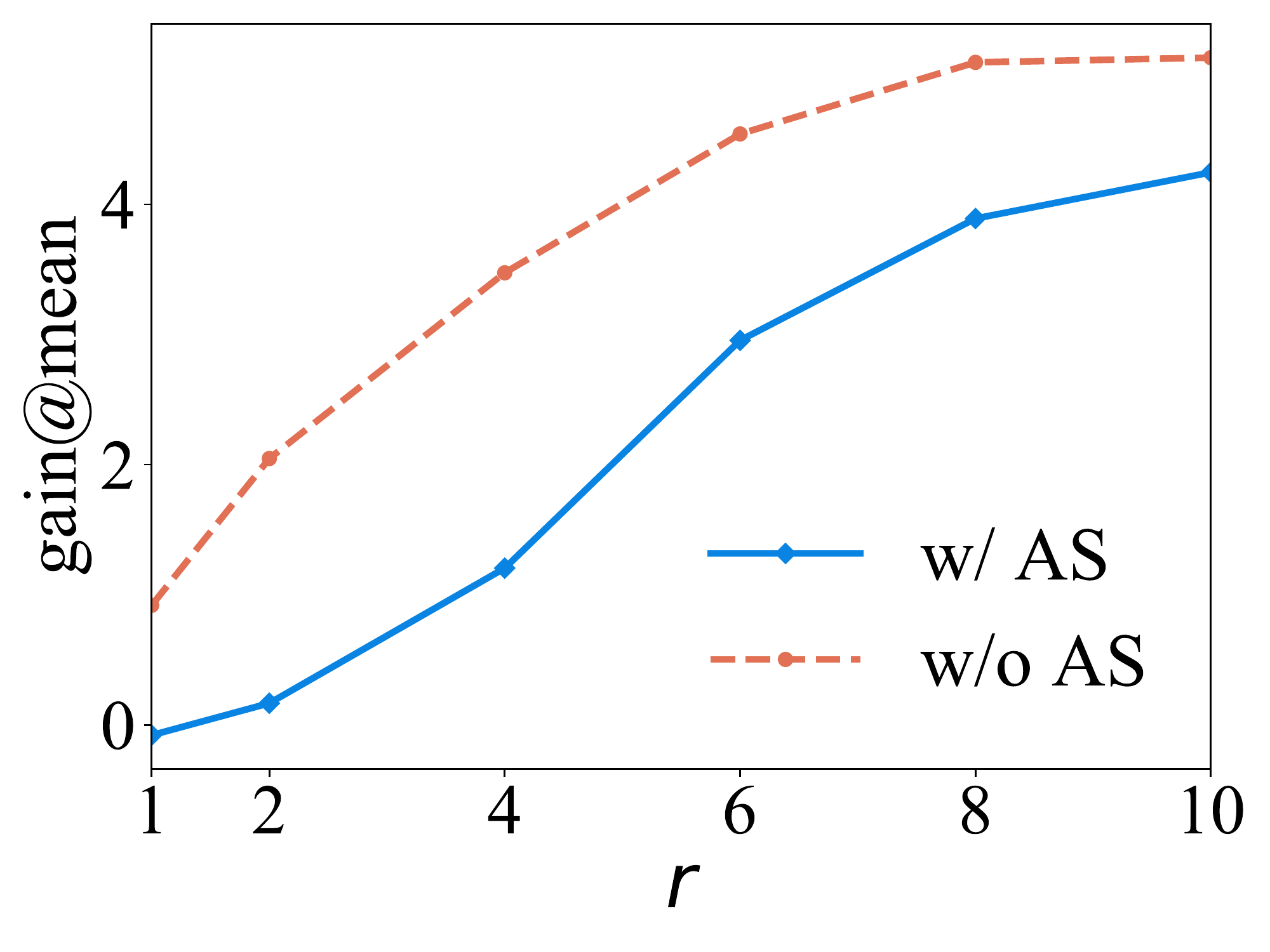}}
\caption{Impact of $\beta$ and $r$ on the defense effectiveness of AS against M2GA for PrivKVM on TalkingData.}
\label{fig:exp_as_talkingdata_m2ga}
\end{figure}

\subsection{Other Defenses}
Another defense is to use verifiable computing in the LDP protocols. For instance, the server may leverage homomorphic encryption \cite{kato2021preventing} when collecting the key-value pairs. However, such methods incur large computational overhead on the user side, downgrading the user experience. Other potential defenses include detecting fake users based on additional information about the users, e.g., their social connections \cite{yu2006sybilguard,danezis2009sybilinfer,gong2014sybilbelief,jia2017random,fu2017robust,wang2017gang,wang19ndss} or registration information~\cite{yuan2019detecting}. Nevertheless, these detection methods are not applicable when the needed information is not available.

%%%%%%%%%%%%%%%%%%%%%%%%%%%%

\section{Conclusion and Future Work}
In this paper, we conduct the first systematic study on  poisoning attacks to LDP protocols for key-value data. We show  such poisoning attacks can be formulated as a two-objective optimization problem. Our results show that an attacker can promote the estimated frequencies and mean values of attacker-chosen target keys. 
We also explore two defenses, which are effective in some scenarios but are ineffective in others. An interesting future work is to study  defenses against our attacks. 

\section*{Acknowledgments}

We thank the anonymous reviewers for their constructive comments.  This work was supported by the National Science Foundation under Grants No. 1937786 and 2112562. Any opinions, findings and conclusions or recommendations expressed in this material are those of the author(s) and do not necessarily reflect the views of the funding
agencies.

%%%%%%%%%%%%%%%%%%%%%%%%%%%%

\bibliographystyle{plain}
\bibliography{reference}

\appendix
\input{appendix}

\end{document}

%% file: appendix.tex
\section{Appendix}
\subsection{Optimality of M2GA}
\label{appendix:proof}
\begin{theorem}
In a given execution of any of the three LDP protocols, M2GA achieves the optimal mean gain if there is only one target key $k$ and $n_1^k\ge n_{-1}^k>\frac{(n+m)b}{2}$.
\end{theorem}

\begin{proof}

When there is only one target key $k$, the mean gain in a given execution of a LDP protocol can be written as 
$G_m(\mathbb{Y})=\tilde{m}_k - \hat{m}_k$, where the second term is irrelevant to the attack. Therefore, $G_m(\mathbb{Y})$ is maximized when $\tilde{m}_k$ is maximized. According to Equation (\ref{eq:mean_estimation_ver2_pckv}), we have the following equation:
\begin{align}
    \tilde{m}_k=\frac{\left(n^k_{1}-n^k_{-1} + \tilde{n}^k_{1}-\tilde{n}^k_{-1} \right)(a-b)}{a(2 p-1)\left(n^k_{1}+n^k_{-1}+ \tilde{n}^k_{1}+\tilde{n}^k_{-1} - (n+m)b\right)},
\end{align}
where $n_1^k$ and $n_{-1}^k$ are constants in a given execution. 
For simplicity, we let $x=n_1^k+n_{-1}^k-(n+m)b$, $y=n^k_1-n^k_{-1}$, and $z=\frac{a-b}{a(2p-1)}$. Then we can rewrite $\tilde{m}_k$ as follows:
\begin{align}
    \tilde{m}_k=z\cdot\frac{y + \tilde{n}^k_{1}-\tilde{n}^k_{-1}}{x+\tilde{n}^k_{1}+\tilde{n}^k_{-1}},
\end{align}
where $z>0$. Taking the partial derivative with respect to $\tilde{n}^k_{1}$ and $\tilde{n}^k_{-1}$, we have the following equations:
\begin{align}
    \frac{\partial \tilde{m}_k}{\partial \tilde{n}^k_{1}}=\frac{z}{(x+\tilde{n}^k_{1}+\tilde{n}^k_{-1})^2}\cdot (2\tilde{n}^k_{-1}+x-y)\\
    \frac{\partial \tilde{m}_k}{\partial \tilde{n}^k_{-1}}=\frac{z}{(x+\tilde{n}^k_{1}+\tilde{n}^k_{-1})^2}\cdot (-2\tilde{n}^k_{1}-x-y)
\end{align}
If $n_1^k\ge n_{-1}^k>\frac{(n+m)b}{2}$, we have $x-y>0$ and $x+y>0$. Since $\tilde{n}^k_{-1}$ and $\tilde{n}^k_{1}$ are both in the range $[0,m]$, we have $\frac{\partial \tilde{m}_k}{\partial \tilde{n}^k_{1}}>0$ and $\frac{\partial \tilde{m}_k}{\partial \tilde{n}^k_{-1}}<0$. Therefore, $\tilde{m}_k$ reaches the maximum value when $\tilde{n}^k_{1}=m$ and $\tilde{n}^k_{-1}=0$, which is what M2GA does. In other words, M2GA maximizes the mean gain $G_m(\mathbb{Y})$ for the given execution.

\end{proof}